\documentclass[CIT,biber]{nowfnt} 

\usepackage[utf8]{inputenc}
\usepackage[tight,footnotesize]{subfigure}

\usepackage{xcolor}
\usepackage{multirow}
\usepackage{amsmath}
\usepackage{amsfonts}
\usepackage{amssymb}
\usepackage{longtable}
\usepackage{hyperref}
\usepackage{indentfirst}
\usepackage{algorithm}
\usepackage{algpseudocode}
\usepackage{amsmath}
\usepackage{amsfonts}
\usepackage{amssymb}
\usepackage{longtable}
\usepackage[mathscr]{eucal}
\usepackage{mathrsfs}
\usepackage{blindtext}
\usepackage{etoolbox}
\usepackage{flushend}
\usepackage{amsthm}
\usepackage{thmtools}
\usepackage{cleveref}
\usepackage{longtable}
\usepackage{subfigure}

\crefname{figure}{\bf Fig.}{\bf Figs.}
\crefname{table}{\bf Tab.}{Tabs.}
\crefname{algorithm}{\bf Alg.}{\bf Algs.}
\crefname{section}{\bf Sect.}{\bf Sects.}
\crefname{example}{\bf Example}{\bf Examples}
\crefname{chapter}{\bf Section}{\bf Sections}
\crefname{theorem}{\bf Theorem}{\bf Theorems}
\crefname{lemma}{\bf Lemma}{\bf Lemmas}
\crefname{corollary}{\bf Corollary}{\bf Corollaries}

\title{Overlapped Arithmetic Codes}

\maintitleauthorlist{
Yong Fang \\
Chang'an University \\
fy@chd.edu.cn
}

\issuesetup
{%
 copyrightowner={Y.~Fang},
 volume        = xx,
 issue         = xx,
 pubyear       = 2025,
 isbn          = xxx-x-xxxxx-xxx-x,
 eisbn         = xxx-x-xxxxx-xxx-x,
 doi           = 10.1561/XXXXXXXXX,
 firstpage     = 1, 
 lastpage      = 18
 }

\usepackage{mwe}

\author{Fang,Yong}

\affil{Chang'an University; fy@chd.edu.cn}

\begin{document}

\makeabstracttitle

\begin{abstract}
In his seminal paper giving birth to information theory, Shannon raised two fundamental coding problems in modern communications systems, \textit{i.e.}, {\em source coding} for data compression and {\em channel coding} for error correction, which form the two main branches of information theory. Since then, many new coding problems appeared. For example, in the 1970s, Slepian and Wolf created a new branch for information theory, \textit{i.e.}, {\em distributed source coding}, which is also called multiple-terminal source coding. 

Arithmetic codes are usually deemed as the most important means to implement lossless source coding, whose principle is mapping every source symbol to a sub-interval in $[0,1)$. For every source symbol, the length of its mapping sub-interval is exactly equal to its probability. With this symbol-interval mapping rule, the interval $[0,1)$ will be fully covered and there is neither {\em overlapped} sub-interval (corresponds to more than one source symbol) nor {\em forbidden} sub-interval (does not correspond to any source symbol).

It is well-known that there is a duality between source coding and channel coding, so every good source code may also be a good channel code meanwhile, and vice versa. Inspired by this duality, arithmetic codes can be easily generalized to address many coding problems beyond source coding by redefining the source-interval mapping rule. If every source symbol is mapped to an enlarged sub-interval, the mapping sub-intervals of different source symbols will be partially {\em overlapped} and we obtain {\em overlapped arithmetic codes}, which can realize {\em distributed source coding}. On the contrary, if every source symbol is mapped to a narrowed sub-interval, there will be one or more {\em forbidden} sub-intervals in $[0,1)$ that do not correspond to any source symbol and we obtain {\em forbidden arithmetic codes}, which can implement {\em joint source-channel coding}. Furthermore, by allowing the coexistence of {\em overlapped} sub-intervals and {\em forbidden} sub-intervals, we will obtain {\em hybrid arithmetic codes}, which can cope with {\em distributed joint source-channel coding}.

The aim of this monograph is to review recent advances in {\em overlapped arithmetic codes} during the past decade. This monograph includes six sections. \cref{c-intro} introduces the principle of arithmetic codes by emphasizing two implementation issues---{\em underflow handling} and {\em bitstream termination}. \cref{c-ccs} reveals that overlapped arithmetic codes are actually a kind of nonlinear coset codes and hence possess two fundamental properties---{\em Coset Cardinality Spectrum} (CCS) and {\em Hamming Distance Spectrum} (HDS). Then \cref{c-ccs} shows how to calculate the CCS of overlapped arithmetic codes theoretically or numerically. On the basis of CCS, \cref{c-ccs} also derives the decoding algorithm of overlapped arithmetic codes. \cref{c-coexist} introduces a powerful analysis tool---{\em coexisting interval} to calculate the error rate of overlapped arithmetic codes. The topic of HDS is addressed in \cref{c-hds}, where a tight connection between CCS and HDS is also built. Then aided by HDS, \cref{c-tailed} explains why overlapped arithmetic codes can be improved by mapping the last few symbols of a block to non-overlapped intervals. Finally, \cref{c-conclusion} concludes this monograph.
\end{abstract}

\chapter{Review on Arithmetic Codes}\label{c-intro}
\vspace{-25ex}%
\section{Theoretical Foundation}
This subsection will briefly review the main theoretical results of four coding problems. We begin with some notation conventions. Let $(X,Y)$ be a pair of discrete random variables defined over the alphabet ${\cal X}\times{\cal Y}$. The joint {\em probability mass function} (pmf) of $X$ and $Y$ is denoted as $p(x,y)$. The marginal pmfs of $X$ and $Y$ are denoted as $p(x)$ and $p(y)$, respectively. Let $p(x|y)$ be the conditional pmf of $X$ given $Y$, and $p(y|x)$ that of $Y$ given $X$. These pmfs are related by
\begin{align}
	p(x,y) = p(x)p(y|x) = p(y)p(x|y).
\end{align}
 
\begin{definition}[Entropy, Conditional Entropy, Joint Entropy, and Mutual Information]
	Let $X$ and $Y$ be two discrete random variables with finite alphabets ${\cal X}$ and ${\cal Y}$, respectively. The entropy of $X$ is defined as 
	\begin{align}
		H(X)\triangleq-\sum_{x\in{\cal X}}p(x)\log_2{p(x)},
	\end{align}	
	and the entropy of $Y$, denoted by $H(Y)$, is defined in a similar way. The conditional entropy of $X$ given $Y$ is defined as
	\begin{align}
		H(X|Y) \triangleq -\sum_{x\in{\cal X}}\sum_{y\in{\cal Y}}p(x,y)\log_2{p(x|y)},
	\end{align}
	and the conditional entropy of $Y$ given $X$, denoted by $H(Y|X)$, can be defined in a similar way. The joint entropy of $X$ and $Y$ is defined as
	\begin{align}
		H(X,Y) \triangleq -\sum_{x\in{\cal X}}\sum_{y\in{\cal Y}}p(x,y)\log_2{p(x,y)},
	\end{align}
	and the mutual information between $X$ and $Y$ is defined as
	\begin{align}
		I(X;Y) \triangleq -\sum_{x\in{\cal X}}\sum_{y\in{\cal Y}}p(x,y)\log_2\frac{p(x)p(y)}{p(x,y)}.
	\end{align}
\end{definition}
According to the above definitions, it is easy to obtain
\begin{align}
	H(X,Y) = H(X) + H(Y|X) = H(Y) + H(X|Y)
\end{align}
and
\begin{align}
	I(X;Y) = H(X)+H(Y)-H(X,Y).
\end{align}
Now we give Shannon's First Theorem \cite{Shannon}, which lays the theoretical foundation for lossless coding of memoryless discrete sources.
\begin{theorem}[Shannon's First Theorem]\label{thm:shannon1}
	Let $X$ be a discrete random variable with finite alphabet 
	${\cal X}$. Let $x^n\triangleq(x_1,\dots,x_n)$ be $n$ independent realizations of $X$. To compress $x^n$ with zero loss, the achievable rate $R$ (bits/symbol) is lower bounded by $H(X)$ as block length $n\to\infty$.
\end{theorem}
We call the random variable $X$ to be compressed as the {\em source}. Since the achievable rate $R\geq H(X)$, lossless source coding is also called {\em entropy coding}. Note that \cref{thm:shannon1} considers only memoryless sources. As for lossless coding of sources with memory, the reader may refer to {\bf Chapter~4} of Cover's textbook \citep{Cover}. 

Then we consider the problem of channel coding. A discrete memoryless transmission {\em channel} can be described by the conditional pmf $p(y|x)$ of the discrete output $Y$ given the discrete input $X$. The most important property of a channel is its {\em capacity}.
\begin{definition}[Channel Capacity]
	Let $p(y|x)$ be a discrete memoryless channel with finite alphabet. Its capacity is defined as
	\begin{align}
		C \triangleq \max_{p(x)}I(X;Y).
	\end{align} 
\end{definition}
Following is Shannon's second theorem, which lays the theoretical foundation for channel coding \cite{Shannon}.
\begin{theorem}[Shannon's Second Theorem]
	Let $(x,y)^n$ be $n$ independent uses of a discrete memoryless channel with finite alphabet. The information amount conveyed by one use of the channel (bits/use) is upper bounded by $C$ as block length $n\to\infty$.
\end{theorem}

\begin{figure}[!t]
	\centering
	\subfigure[]{\includegraphics[width=.5\linewidth]{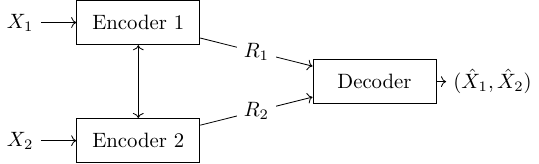}\label{subfig:cscsym}}%
	\subfigure[]{\includegraphics[width=.5\linewidth]{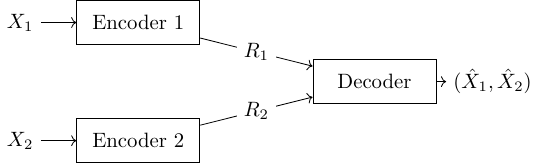}\label{subfig:dscsym}}\\
	\subfigure[]{\includegraphics[width=.5\linewidth]{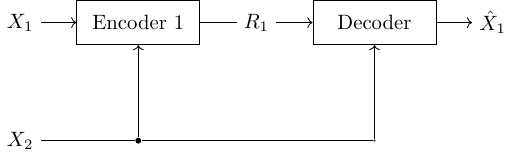}\label{subfig:cscasym}}%
	\subfigure[]{\includegraphics[width=.5\linewidth]{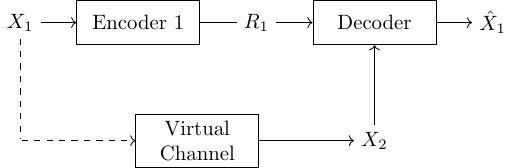}\label{subfig:dscasym}}
	\caption{Four typical frameworks of source coding. (a) Symmetric CSC. (b) Symmetric DSC. (c) Asymmetric CSC. (d) Asymmetric DSC.}
	\label{fig:schemes}
\end{figure}

The third coding problem is {\em multiterminal source coding}, or the so-called {\em Distributed Source Coding} (DSC). Shown in \cref{fig:schemes} are four typical frameworks of source coding. The differences between {\em Conventional/Centralized Source Coding} (CSC) and DSC are quite clear: (a) In the symmetric case, the two DSC encoders do not communicate with each other; and (b) In the asymmetric case, the DSC encoder of $X_1$ does not know the side information $X_2$. Just as CSC, there are two (lossless or lossy) forms of DSC. Following is {\em Slepian-Wolf theorem} \citep{SWC}, which lays the theoretical foundation for lossless DSC.
\begin{theorem}[Slepian-Wolf Theorem]\label{thm:sw}
	Let $X_1$ and $X_2$ be two correlated discrete sources. To compress $X_1$ and $X_2$ with zero loss, the achievable rate region is lower bounded by $R_1\geq H(X_1|X_2)$,  $R_2\geq H(X_2|X_1)$, and  $R_1+R_2\geq H(X_1,X_2)$, no matter whether the encoders of $X_1$ and $X_2$ communicate with each other or not.  
\end{theorem}

After comparing \cref{thm:sw} with \cref{thm:shannon1}, it can be seen that lossless DSC can achieve the same performance as lossless CSC. Lossless DSC is also called {\em Slepian-Wolf coding}. As for lossy DSC, its general form is called {\em Berger-Tung coding} \citep{BTC}, and its asymmetric form with decoder side information is called {\em Wyner-Ziv coding} \citep{WZC}. 

\begin{figure}[!t]
	\centering
	\includegraphics[width=\linewidth]{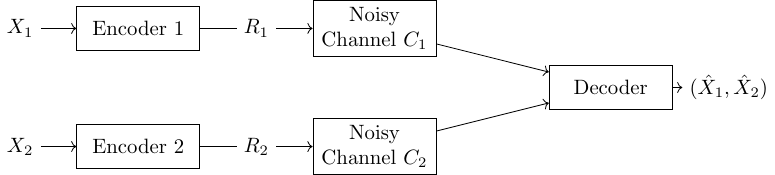}
	\caption{Diagram of the DJSCC problem for two terminals.}
	\label{fig:djscc}
\end{figure}

Finally, we consider the noisy-channel Slepian-Wolf coding problem, which is the lossless version of the {\em Distributed Joint Source-Channel Coding} (DJSCC) problem (cf. \cref{fig:djscc}). According to the Source-Channel Separation Theorem, the following theorem holds.

\begin{theorem}[Noisy-Channel Slepian-Wolf Coding]\label{thm:djscc}
	If two discrete sources $X_1$ and $X_2$ are compressed independently and their bitstreams are conveyed over two noisy channels with capacities $C_1$ and $C_2$, respectively, then the achievable rate region is given by $R_1\geq H(X_1|X_2)/C_1$,  $R_2\geq H(X_2|X_1)/C_2$, and  $C_1R_1+C_2R_2\geq H(X_1,X_2)$.  
\end{theorem}

\section{Huffman Codes for Entropy Coding}
Let $X$ be a memoryless discrete source with finite alphabet ${\cal X}$. To compress $X$ with zero loss, Shannon's first theorem reveals that the achievable rate is lower bounded by $H(X)$, the entropy of $X$, but it does not tell how to achieve $H(X)$. In 1952, Huffman proposed a simple but efficient method for lossless data compression \citep{Huffman}, whose principle can be found in every textbook in information theory. However, the performance of Huffman codes is very poor when the source alphabet is small, as shown by \cref{ex:huffman}. In the extreme case of binary source compression, Huffman codes even do not work at all. 

\begin{example}\label{ex:huffman}
As we know, every DNA molecule consists of two strands wound around each other. There are four types of bases in a DNA molecule: \textit{Adenine} ($A$), \textit{Cytosine} ($C$), \textit{Guanine} ($G$), and \textit{Thymine} ($T$). The two strands of a DNA molecule are held together by the bonds between bases. The fractions of $A$, $C$, $G$, and $T$ bases vary for different organisms. Note that $A$ always pairs with $T$, while $C$ always pairs with $G$, so $A$ and $T$ always have the same fraction, while $C$ and $G$ always have the same fraction. Shown in \cref{tab:dnapmf} are the fractions of bases for human, mycoplasma genitalium (MG), arabidopsis thaliana (AT), and caenorhabditis elegans (CE). The Huffman codes for every organism are listed in \cref{tab:dnahuffman}. It can be seen that except MG, the Huffman codes of other organisms are actually fixed-length codes. 
 
\begin{table}[!t]
	\centering
	\caption{Fractions of Bases in DNA Molecules}
	\begin{tabular}{c||c||c||c||c||c}
		\hline	
		Organism &$A$ &$G$ &$T$ &$C$ &$H(X)$ \\
		\hline\hline
		Human &$.295$ &$.295$ &$.205$ &$.205$ &$1.9765$ \\
		\hline
		MG &$.342$ &$.342$ &$.158$ &$.158$ &$1.9$\\
		\hline
		AT &$.325$ &$.325$ &$.175$ &$.175$ &$1.9341$\\
		\hline
		CE &$.32$ &$.32$ &$.18$ &$.18$ &$1.9427$\\
		\hline
	\end{tabular}
	\label{tab:dnapmf}
\end{table}

\begin{table}[!t]
	\centering
	\caption{Huffman Codes for Bases in DNA Molecules}
	\begin{tabular}{c||c||c||c||c||c||c}
		\hline	
		Organism &$A$ &$G$ &$T$ &$C$ &$H(X)$ &$R$ \\
		\hline\hline
		Human &$00$ &$01$ &$10$ &$11$ &$1.9765$ &$2$ \\
		\hline
		MG &$0$ &$10$ &$110$ &$111$  &$1.9$ &$1.974$\\
		\hline
		AT &$00$ &$01$ &$10$ &$11$ &$1.9341$ &$2$\\
		\hline
		CE &$00$ &$01$ &$10$ &$11$ &$1.9427$ &$2$\\
		\hline
	\end{tabular}
	\label{tab:dnahuffman}
\end{table}
\end{example}

The low efficiency of Huffman codes comes from the property that Huffman codes always map every source symbol to a bit string, so that at least one bit is needed for one source symbol. To allow a fractional number of bits for one source symbol, we must map a block of source symbols as a whole to a bit string, which is the idea of arithmetic codes. Though the idea of arithmetic codes can date back to 1950s, Rissanen is widely known as the inventor of arithmetic codes \citep{ac}. For the history of arithmetic codes, the reader may refer to \citep{Langdon}. In the following, we will first introduce the primitive idea behind arithmetic codes, and then discuss two implementation issues, \textit{i.e.}, {\em finite-precision} issue and {\em prefix-code} issue, encountered in practice, which were comprehensively discussed in Witten's classical article \cite{WittenCACM87}.  

\section{Raw Arithmetic Bitstream}\label{subsec:raw}
\begin{figure}[!t]
	\subfigure[$r=1$]{\includegraphics[width=.5\linewidth]{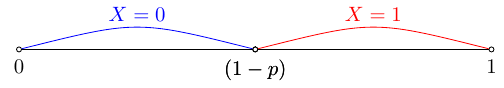}\label{subfig:ac}}%
	\subfigure[$r<1$]{\includegraphics[width=.5\linewidth]{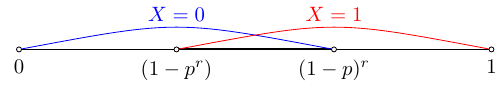}\label{subfig:oac}}
	\subfigure[$r>1$]{\includegraphics[width=.5\linewidth]{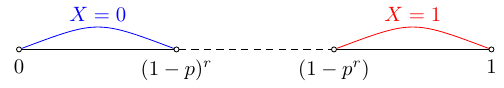}\label{subfig:eac-A}}%
	\subfigure[$r>1$]{\includegraphics[width=.5\linewidth]{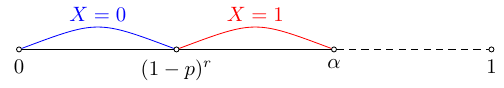}\label{subfig:eac-B}}
	\subfigure[$r>1$]{\includegraphics[width=.5\linewidth]{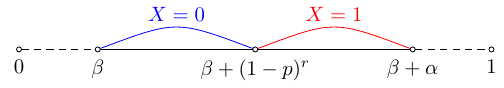}\label{subfig:eac-C}}%
	\subfigure[$r<1$]{\includegraphics[width=.5\linewidth]{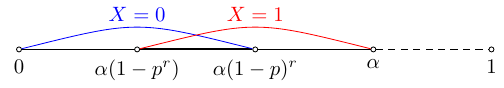}\label{subfig:jscac}}
	\caption{Generalized arithmetic codes for binary sources, where $p\triangleq\Pr(X=1)$. (a) Standard arithmetic codes to handle source coding. (b) Overlapped arithmetic codes to handle DSC (cf. \cref{subfig:dscsym} and \cref{subfig:dscasym}). (c), (d), and (e) Three variants of forbidden arithmetic codes to handle channel coding ($p=1/2$) or joint source-channel coding ($p\neq 1/2$), where $\alpha = (1-p)^r+p^r<1$. (f) Hybrid arithmetic codes to handle the DJSCC problem (cf. \cref{fig:djscc}).}
	\label{fig:gac}
\end{figure}

\begin{figure}[!t]
	\includegraphics[width=\linewidth]{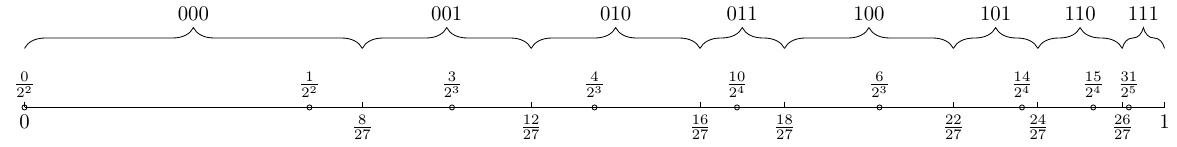}
	\caption{An example of the mapping from $x^n\in\mathbb{B}^n$ to $[l,h)\subset[0,1)$, where $n=3$ and $p=\Pr(X=1)=1/3$.}
	\label{fig:exam}
\end{figure}

Let $X$ be a binary source with bias probability $\Pr(X=1)=p$. Let $x^n\in\mathbb{B}^n$ be $n$ independent realizations of $X$. The arithmetic encoder maps $x^n$ onto an interval $[l,h)\subset[0,1)$. Then from $[l,h)$, the number that can be represented with the fewest bits is binarized as the bitstream of $x^n$. Concretely speaking, the mapping interval $[l,h)$ of $x^n$ is obtained via a recursion, and for every source symbol, the length of its mapping interval is proportional to its probability (cf. \cref{subfig:ac}). Let $[l_i,h_i)$ denote the mapping interval of $x^i\in\mathbb{B}^i$. Initially, $[l_0,h_0)=[0,1)$. According to \cref{subfig:ac}, on knowing $[l_i,h_i)$, we can derive $[l_{i+1},h_{i+1})$ by
\begin{align}\label{eq:I}
	\begin{cases}
		l_{i+1} = l_i + x_{i+1}\cdot(1-p)(h_i-l_i)\\
		h_{i+1} = h_i - (1-x_{i+1})\cdot p(h_i-l_i).
	\end{cases}
\end{align}
Finally, we obtain the mapping interval $[l_n,h_n)$ for $x^n\in\mathbb{B}^n$. Clearly,
\begin{align}
	[l_n,h_n) \subset [l_{n-1},h_{n-1}) \subset \dots \subset [l_1,h_1) \subset [l_0,h_0) = [0,1).
\end{align}
For conciseness, we shorten $[l_n,h_n)$ to $[l,h)$, whose length is
\begin{align}
	(h-l)=p(x^n)\triangleq\Pr(X^n=x^n)=(1-p)^{n-|x^n|}\cdot p^{|x^n|}, 
\end{align}
where $|x^n|$ is the support (the number of nonzero terms) of $x^n$. There are infinitely many real numbers in $[l,h)$, among which the real number that can be represented by the fewest bits is binarized as the bitstream of $x^n$. Obviously, in any interval, integers can be represented by fewer bits than floating-point numbers. However, there is no integer in $[l,h)$. To solve this problem, we introduce the following concept.
\begin{definition}[Normalized Interval]
	Let $[l,h)\subset[0,1)$ be the mapping interval of $x^n\in\mathbb{B}^n$. The normalized mapping interval of $x^n$ is defined as $[l2^m,h2^m)$, where
	\begin{align}\label{eq:m}
		m \triangleq \lceil-\log_2(h-l)\rceil = -\lfloor\log_2(h-l)\rfloor.
	\end{align}	
\end{definition}
\begin{lemma}[Properties of Normalized Interval]
The following properties of $[l2^m,h2^m)$ hold obviously.
\begin{itemize}
	\item $1\leq (h2^m-l2^m) = (h-l)2^m <2$. 
	\item $l2^m$ and $h2^m$ are constrained by
		\begin{align}
			\begin{cases}
				(\lceil{l2^m}\rceil-1)<l2^m\leq \lceil{l2^m}\rceil\\ 
				\lceil{l2^m}\rceil<h2^m<(\lceil{l2^m}\rceil+2).
			\end{cases}
		\end{align}	
	\item If $\log_2(h-l)\in\mathbb{Z}$, then $(h-l)2^m=1$ and there is one and only one integer $\lceil{l2^m}\rceil$ in $[l2^m,h2^m)$. 
	\item If $\log_2(h-l)\notin\mathbb{Z}$, then $1<(h-l)2^m<2$ and there may be one or two integers in $[l2^m,h2^m)$. Further,
	\begin{itemize}
		\item if there are two integers in $[l2^m,h2^m)$, they must be $\lceil{l2^m}\rceil$ and $\lceil{l2^m}\rceil+1$; and 
		\item if there is only one integer in $[l2^m,h2^m)$, it must be $\lceil{l2^m}\rceil$.
	\end{itemize}
\end{itemize}
\end{lemma}

According to the above properties, the integer $\lceil{l2^m}\rceil$ always belongs to $[l2^m,h2^m)$, while the integer $\lceil{l2^m}\rceil+1$ may or may not belong to $[l2^m,h2^m)$. The integer $\lceil{l2^m}\rceil$ in the normalized interval $[l2^m,h2^m)$ corresponds to the real number $\lceil{l2^m}\rceil2^{-m}$ in the original interval $[l,h)$. By binarizing $\lceil{l2^m}\rceil$ into $m$ bits, we obtain a bitstream of $x^n$. If the integer $\lceil{l2^m}\rceil+1$ belongs to $[l2^m,h2^m)$, the corresponding real number $(\lceil{l2^m}\rceil+1)2^{-m}$ must also belong to $[l,h)$. By binarizing $\lceil{l2^m}\rceil+1$ into $m$ bits, we obtain another bitstream of $x^n$. 

\begin{definition}[Raw Arithmetic Bitstream]
	Let $[l,h)\subset[0,1)$ be the mapping interval of $x^n\in\mathbb{B}^n$. Let $m=-\lfloor\log_2{(h-l)}\rfloor$. Apparently, $\lceil{l2^m}\rceil2^{-m}\in[l,h)$ holds always, but $(\lceil{l2^m}\rceil+1)2^{-m}\in[l,h)$ may or may not hold. The integer $\lceil{l2^m}\rceil$ can be represented by a string of $m$ bits, which is called a raw arithmetic bitstream of $x^n$. If $(\lceil{l2^m}\rceil+1)2^{-m}\in[l,h)$, the string of $m$ bits representing the integer $\lceil{l2^m}\rceil+1$ is also a raw arithmetic bitstream of $x^n$.
\end{definition}

According to the definition, for $x^n\in\mathbb{B}^n$, if $(\lceil{l2^m}\rceil+1)2^{-m}\notin[l,h)$, there is only one raw arithmetic bitstream; if $(\lceil{l2^m}\rceil+1)2^{-m}\in[l,h)$, there are two raw arithmetic bitstreams.

\begin{example}[Example of Raw Arithmetic Bitstream]
	Let $X$ be a binary source with bias probability $p=\Pr(X=1)$ and $x^n$ be $n$ independent realizations of $X$. \cref{fig:exam} shows how the infinite-precision arithmetic codec partitions the interval $[0,1)$ into sub-intervals $[l,h)$ according to $x^n$, where $p=1/3$ and $n=3$. The labels under the abscissa are the boundaries of sub-intervals, while the labels above the abscissa are $\frac{\lceil{l2^m}\rceil}{2^m}$ (and $\frac{\lceil{l2^m}\rceil+1}{2^m}$ if it belongs to $[l,h)$), \textit{i.e.}, the real numbers within each sub-interval that can be converted into integers by a stride-$m$ bitwise left shift, where $m \triangleq -\lfloor\log_2(h-l)\rfloor$. For example, the mapping sub-interval of $x^n=000$ is $[l,h)=[0,8/27)$ and $m=2$. There are two real numbers $0/2^2$ and $1/2^2$ in $[l,h)$ that can be converted into integers $0$ and $1$, respectively, by a stride-2 bitwise left shift. That means, if we enlarge $[l,h)=[0,8/27)$ by $2^2$ times, there will be two integers $0$ and $1$ in the normalized sub-interval $[l2^m,h2^m)=[0,32/27)$. The integers $0$ and $1$ can be binarized into $m=2$ bits, \textit{i.e.}, $00_b=0$ and $01_b=1$, both of which are raw arithmetic bitstreams of $x^n=000$. In every other sub-interval, there is one and only one real number that can be converted into an integer by a stride-$m$ bitwise left shift, so there is only one raw arithmetic bitstream. For more details, the reader may refer to the \textit{raw} row of \cref{tab:exam}.
\end{example}

\begin{table}[!t]
	\centering
	\footnotesize
	\caption{Example of Arithmetic Code for $p=1/3$, $n=3$, and $w=8$}
	\begin{tabular}{c||c||c||c||c}
		\hline	
		$x^n$ &$000$ &$001$ &$010$ &$011$\\
		\hline\hline
		$p(x^n)$ &$8/27$ &$4/27$ &$4/27$ &$2/27$\\
		\hline
		$m$ &$2$ &$3$ &$3$ &$4$\\
		\hline
		$[l,h)$ &$[0,\frac{8}{27})$ &$[\frac{8}{27},\frac{12}{27})$ &$[\frac{12}{27},\frac{16}{27})$ &$[\frac{16}{27},\frac{18}{27})$\\
		\hline
		$\frac{\lceil l2^m\rceil(+1)}{2^m}$ &$\frac{0}{2^2}$ or $\frac{1}{2^2}$ &$\frac{3}{2^3}$ &$\frac{4}{2^3}$ &$\frac{10}{2^4}$\\
		\hline
		raw &$00$ or $01$ &$011$ &$100$ &$1010$\\
		\hline
		final $\eta$ 
		&${\color{cyan}0}\underline{{\color{red}1}}\underline{{\color{blue}0}010111}$ 
		&${\color{cyan}01}\underline{{\color{red}1}}\underline{{\color{blue}1}000111}$ 
		&$\underline{{\color{red}1}}{\color{green}00}\underline{{\color{blue}1}011111}$ &${\color{cyan}10}\underline{{\color{red}1}}{\color{green}0}\underline{{\color{blue}1}010111}$\\
		\hline
		final $\lambda$  
		&${\color{cyan}0}\underline{{\color{red}0}}\underline{{\color{blue}0}000000}$ 
		&${\color{cyan}01}\underline{{\color{red}0}}\underline{{\color{blue}0}110000}$ 
		&$\underline{{\color{red}0}}{\color{green}11}\underline{{\color{blue}1}001000}$ &${\color{cyan}10}\underline{{\color{red}0}}{\color{green}1}\underline{{\color{blue}1}000000}$\\
		\hline
		prefix &$001$ &$0101/0110$ &$1000$ &$10100$\\
		\hline
		half-tail &$0$ &$01$ &-- &$10$\\
		\hline
		\hline
		\hline	
		$x^n$ &$100$ &$101$ &$110$ &$111$ \\
		\hline\hline
		$p(x^n)$ &$4/27$ &$2/27$ &$2/27$ &$1/27$ \\
		\hline
		$m$ &$3$ &$4$ &$4$ &$5$ \\
		\hline
		$[l,h)$ &$[\frac{18}{27},\frac{22}{27})$ &$[\frac{22}{27},\frac{24}{27})$ &$[\frac{24}{27},\frac{26}{27})$ &$[\frac{26}{27},1)$\\
		\hline
		$\frac{\lceil l2^m\rceil}{2^m}$ &$\frac{6}{2^3}$ &$\frac{14}{2^4}$ &$\frac{15}{2^4}$ &$\frac{31}{2^5}$\\
		\hline
		raw &$110$ &$1110$ &$1111$ &$11111$\\
		\hline
		final $\eta$ 
		&${\color{cyan}1}\underline{{\color{red}1}}{\color{green}0}\underline{{\color{blue}1}000001}$ 
		&${\color{cyan}11}\underline{{\color{red}1}}{\color{green}0}\underline{{\color{blue}0}011011}$ 
		&${\color{cyan}111}\underline{{\color{red}1}}\underline{{\color{blue}0}110011}$ 
		&${\color{cyan}1111}\underline{{\color{red}1}}\underline{{\color{blue}1}111111}$ \\
		\hline
		final $\lambda$  
		&${\color{cyan}1}\underline{{\color{red}0}}{\color{green}1}\underline{{\color{blue}0}101100}$ 
		&${\color{cyan}11}\underline{{\color{red}0}}{\color{green}1}\underline{{\color{blue}0}000100}$ 
		&${\color{cyan}111}\underline{{\color{red}0}}\underline{{\color{blue}0}011100}$ 
		&${\color{cyan}1111}\underline{{\color{red}0}}\underline{{\color{blue}1}101000}$ \\
		\hline
		prefix &$1011/1100$ &$11011$ &$11101$ &$111110$\\
		\hline
		half-tail &$1$ &$11$ &$111$ &$1111$\\
		\hline		
	\end{tabular}
	\label{tab:exam}
\end{table}

\section{Finite-Precision Issue}
\begin{algorithm}[!t]
	\caption{$({\bf z},j,\lambda,\eta) = {\rm encode\_one\_symbol}(x,p,{\bf z},j,\lambda,\eta)$}
	\Comment{$x$: source symbol to be encoded}\\
	\hspace*{\fill}\Comment{$p$: bias probability of the source}\\
	\hspace*{\fill}\Comment{${\bf z}$: array of arithmetic bitstream}\\
	\hspace*{\fill}\Comment{$j$: number of occupied bits in $\bf z$}
	\begin{algorithmic}[1]
		\State $(\lambda,\eta) = {\rm shrink\_window}(x,p,\lambda,\eta)$
		\State $({\bf z},j,\lambda,\eta) = {\rm renormalize\_window}({\bf z},j,\lambda,\eta)$
	\end{algorithmic}
	\label{alg:win}
\end{algorithm}

\begin{algorithm}[!t]
	\caption{$(\lambda,\eta) = {\rm shrink\_window}(x,p,\lambda,\eta)$}
	\begin{algorithmic}[1]
		\If{$x=0$} 
			\State $\lambda$ remains the same 
			\State $\eta\gets\lambda+\lfloor(1-p)(\eta-\lambda+1)\rceil-1$
		\Else
			\State $\lambda\gets\lambda+\lfloor(1-p)(\eta-\lambda+1)\rceil$
			\State $\eta$ remains the same
		\EndIf
	\end{algorithmic}
	\label{alg:shrink}
\end{algorithm}

\begin{algorithm}[!t]
	\caption{$({\bf z},j,\lambda,\eta) = {\rm renormalize\_window}({\bf z},j,\lambda,\eta)$}
	\begin{algorithmic}[1]
		\While{the MSBs of $\lambda$ and $\eta$ match each other}
		\State $z_{j+1}\gets$ the matched MSB
		\State $j\gets j+1$
		\State both $\lambda$ and $\eta$ are left shifted by 1 place
		\State bit 0 is appended to $\lambda$ and bit 1 is appended to $\eta$
		\EndWhile				
	\end{algorithmic}
	\label{alg:renorm}
\end{algorithm}

So far we totally ignored the issue of precision. There is no infinite-precision codec in practice, so accurate values of $l$ and $h$ are unavailable. This problem can be solved by introducing a sliding window. As we know, $0=(.00\cdots)_b$ and $1=(.11\cdots)_b$. We define a width-$w$ sliding window as $[\lambda:\eta]$, where $\lambda$ and $\eta$ are two $w$-bit buffers. Initially, 
\begin{align}\label{eq:win}
	\begin{cases}
		\eta = (\overbrace{1\cdots1}^{w})_b = 2^w-1\\
		\lambda = (\underbrace{0\cdots0}_{w})_b = 0.
	\end{cases}
\end{align}
For conciseness, we write \eqref{eq:win} as $[\lambda:\eta]=[0:2^w)$. To encode $x_i$, we renew $\lambda$ and $\eta$ according to \cref{alg:win}, which calls \cref{alg:shrink} and \cref{alg:renorm} in sequence. In \cref{alg:shrink}, $\lfloor\cdot\rceil$ denotes the rounding function. In \cref{alg:renorm}, MSB is the acronym of {\em Most Significant Bit}, and similarly, the $i$-th MSB will be shortened to $i$-MSB below. Note that $0=(.00\cdots)_b$ and $1=(.11\cdots)_b$, so there are actually infinitely many $0$s following $\lambda$ and infinitely many $1$s following $\eta$. That is why the interval length should be $(\eta-\lambda+1)$ rather than $(\eta-\lambda)$ in \cref{alg:shrink}. After renormalization, the MSB of $\lambda$ is always $0$ and the MSB of $\eta$ is always $1$. 

\begin{figure}[!t]
	\includegraphics[width=\linewidth]{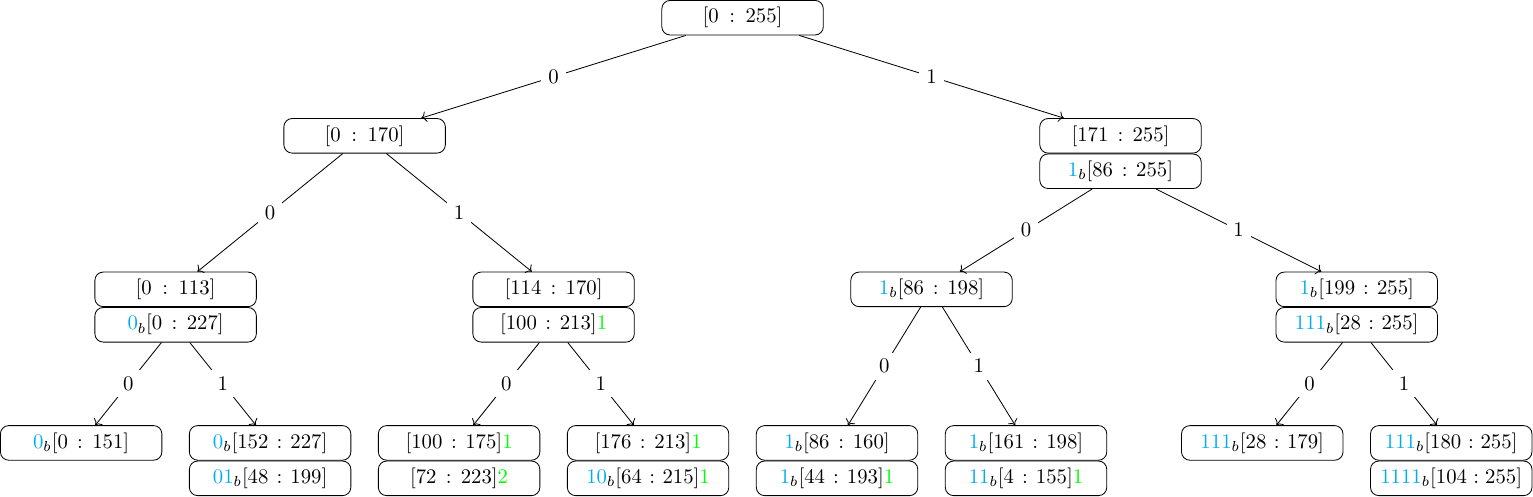}
	\caption{Evolution of sliding window $[\lambda:\eta]$ for different source blocks, where $n=3$, $p=1/3$, and $w=8$. Initially, $[\lambda:\eta]=[0:255]$. Depending on $x_i$, $[\lambda:\eta]$ is shrunk (and then renormalized, if needed). The bitstring on the left of $[\lambda:\eta]$, if any, contains the bits output by the encoder so far, and the number on the right of $[\lambda:\eta]$, if any, is the number of underflow bits currently.}
	\label{fig:enc}
\end{figure}

However, \cref{alg:win} is not enough to solve the finite-precision issue. In the worst case, the sliding window becomes
\begin{align}\label{eq:ufmin}
	\begin{cases}
		\eta = (1\overbrace{0\cdots0}^{(w-1)})_b = 2^{w-1}\\
		\lambda = (0\underbrace{1\cdots1}_{(w-1)})_b = 2^{w-1}-1.
	\end{cases}
\end{align}
We find $(\eta-\lambda+1)=2$! Once this extreme case happens, the arithmetic codec will no longer work.
\begin{definition}[Underflow]
	The phenomenon that the 2-MSB of $\lambda$ is 1 and the 2-MSB of $\eta$ is 0 is called {\em underflow}, meanwhile the 2-MSBs of $\lambda$ and $\eta$ (equal to 1 and 0, respectively) are called underflow bits. 
\end{definition}
If the underflow phenomenon happens, we have 
\begin{align*}
	2\leq(\eta-\lambda+1)\leq2^{w-1}, 
\end{align*}
where $(\eta-\lambda+1)=2$ corresponds to the case of \eqref{eq:ufmin} and $(\eta-\lambda+1)=2^{w-1}$ corresponds to the case of
\begin{align}\label{eq:ufmax}
	\begin{cases}
		\eta = (10\overbrace{1\cdots1}^{(w-2)})_b = 2^{w-1}+2^{w-2}-1\\
		\lambda = (01\underbrace{0\cdots0}_{(w-2)})_b = 2^{w-2}.
	\end{cases}
\end{align}
It can be seen that the underflow phenomenon may make the sliding window $[\lambda:\eta]$ very small. To solve this problem, we define an auxiliary variable $\upsilon$ to record the accumulated number of underflow bits and then introduce \cref{alg:underflow}, where $i$-MSB refers to the $i$-th MSB and LSB is the acronym of {\em Least Significant Bit}. After applying \cref{alg:underflow}, we have 
\begin{align*}
	2^{w-2}+2\leq(\eta-\lambda+1)\leq2^w, 
\end{align*}
where the second equality holds {\em if and only if} (iff) $[\lambda:\eta]=[0:2^w)$ and the first equality holds iff
\begin{align}\label{eq:noufmax}
	\begin{cases}
		\eta = (1z\overbrace{0\cdots0}^{(w-2)})_b = 2^{w-1}+z\cdot2^{w-2}\\
		\lambda = (0z\underbrace{1\cdots1}_{(w-2)})_b = z\cdot2^{w-2}+2^{w-2}-1.
	\end{cases}
\end{align}

\begin{algorithm}[!t]
	\caption{$(\lambda,\eta,\upsilon) = {\rm remove\_underflow\_bits}(\lambda,\eta,\upsilon)$}
	\begin{algorithmic}[1]
		\While{the 2-MSB of $\lambda$ is 1 and the 2-MSB of $\eta$ is 0}
		\State the 2-MSBs of $\lambda$ and $\eta$ are removed
		\State all remaining bits except MSBs, \textit{i.e.}, from 3-MSBs to LSBs, of $\lambda$ and $\eta$ are left shifted by 1 place
		\State bit 0 is appended to $\lambda$ and bit 1 is appended to $\eta$
		\State $\upsilon\gets\upsilon+1$ 
		\EndWhile
	\end{algorithmic}
	\label{alg:underflow}
\end{algorithm}

On the basis of \cref{alg:underflow}, the standard function to encode one source symbol is given by \cref{alg:intact}. To encode $x^n$, we initialize $[\lambda:\eta]=[0:2^w)$ and set $j=\upsilon=0$. Then \cref{alg:intact} is repeated. For every $x_i$, where $1\leq i\leq n$, the encoder calls \cref{alg:shrink} to shrink the sliding window $[\lambda:\eta]$. Afterwards, if the MSBs of $\lambda$ and $\eta$ match each other, the encoder in turn calls \cref{alg:push_underflow_bits} to output the MSB and underflow bits, calls \cref{alg:renorm} to renormalize $[\lambda:\eta]$, and calls \cref{alg:underflow} to remove underflow bits from $[\lambda:\eta]$. Given that underflow bits are removed, the 2-MSB of $\lambda$ will never be greater than the 2-MSB of $\eta$.

\begin{algorithm}[!t]
	\caption{$({\bf z},j,\lambda,\eta,\upsilon) = {\rm encode\_one\_symbol}(x,p,{\bf z},j,\lambda,\eta,\upsilon)$}
	\begin{algorithmic}[1]
		\State $(\lambda,\eta) = {\rm shrink\_window}(x,p,\lambda,\eta)$
		\If{the MSBs of $\lambda$ and $\eta$ match each other}
		\State $({\bf z},j,\lambda,\eta,\upsilon) = {\rm push\_underflow\_bits}({\bf z},j,\lambda,\eta,\upsilon)$
		\State $({\bf z},j,\lambda,\eta) = {\rm renormalize\_window}({\bf z},j,\lambda,\eta)$
		\State $(\lambda,\eta,\upsilon) = {\rm remove\_underflow\_bits}(\lambda,\eta,\upsilon)$
		\EndIf
	\end{algorithmic}
	\label{alg:intact}
\end{algorithm}

\begin{algorithm}[!t]
	\caption{$({\bf z},j,\lambda,\eta,\upsilon) = {\rm push\_underflow\_bits}({\bf z},j,\lambda,\eta,\upsilon)$}
	\begin{algorithmic}[1]
		\If{both MSBs of $\lambda$ and $\eta$ are $0$} 
		\State $(z_{j+1}\dots z_{j+\upsilon+1}) \gets (0\underbrace{1\cdots1}_{\upsilon})$
		\Else 
		\State $(z_{j+1}\dots z_{j+\upsilon+1}) \gets (1\underbrace{0\cdots0}_{\upsilon})$
		\EndIf
		\State $j\gets j+\upsilon+1$
		\State $\upsilon\gets0$
		\State both $\lambda$ and $\eta$ are left shifted by 1 place
		\State bit 0 is appended to $\lambda$ and bit 1 is appended to $\eta$
	\end{algorithmic}
	\label{alg:push_underflow_bits}
\end{algorithm}

\begin{example}[Evolution of Sliding Window]
	An example is given in \cref{fig:enc} to demonstrate how the sliding window $[\lambda:\eta]$ evolves during encoding, where $p=1/3$, $n=3$, and $w=8$. Correspondingly, the final sliding window is given by the $\lambda$ and $\eta$ rows of \cref{tab:exam}, where the $w=8$ bits in the final sliding window are underlined, the matched bits that are left shifted out of the final sliding window are marked in {\color{cyan}cyan}, and underflow bits are marked in {\color{green}green}. For example, if $x^n=101$, two cyan bits $\color{cyan}11$ are output and there is only one underflow bit ($\color{green}0$ for $\eta$ and $\color{green}1$ for $\lambda$). In \cref{tab:exam}, the MSBs and 2-MSBs of the final sliding window are marked in {\color{red}red} and {\color{blue}blue}, respectively.
\end{example}

\section{Prefix-Code Issue}
Even though the finite-precision issue has been perfectly solved, there is still another issue. Actually, when we decode a raw arithmetic bitstream, we implicitly assume that it is followed by an all-zero bitstring. However in practice, a raw arithmetic bitstream may be followed by arbitrary bitstrings, and some special following bitstrings may cause decoding failure. For conciseness, we refer to the string of following bits as {\em tail}. Let us define different kinds of bitstreams.
\begin{definition}[Prefix Bitstream, Zero-tail Bitstream, and Half-tail Bitstream]
	Let $[l,h)\subset[0,1)$ be the mapping interval of $x^n\in\mathbb{B}^n$, and the arithmetic bitstream of $x^n$ is denoted by $z^m\triangleq(z_1\cdots z_m)$. 
	\begin{itemize}
		\item If $(.z_1\cdots z_mz_{m+1}\cdots)_b\in[l,h)$ holds for every tail $(.z_{m+1}\cdots)_b\in[0,1)$, we say that $z^m$ is a prefix bitstream of $x^n$.
		\item If $(.z_1\cdots z_mz_{m+1}\cdots)_b\in[l,h)$ holds for the zero tail $(.z_{m+1}\cdots)_b=(.00\cdots)_b=0$ but not for all tails, we say that $z^m$ is a zero-tail bitstream of $x^n$. 
		\item If $(.z_1\cdots z_mz_{m+1}\cdots)_b\in[l,h)$ holds for the half tail $(.z_{m+1}\cdots)_b=(.100\cdots)_b=(.011\cdots)_b=0.5$ but not for all tails, we say that $z^m$ is a half-tail bitstream of $x^n$.		
	\end{itemize} 
\end{definition}

If $z^m$ is a prefix bitstream of $x^n$, then $x^n$ can always be correctly decoded, no matter what tail $(z_{m+1}\cdots)$ follows $z^m$. On the contrary, if $z^m$ is a zero-tail or half-tail bitstream of $x^n$, there will be a risk of decoding failure for those tails causing $(.z_1\cdots z_mz_{m+1}\cdots)_b\notin[l,h)$. Now we wonder whether raw arithmetic bitstreams are prefix bitstreams.

\begin{theorem}[Raw Arithmetic Bitstream]
	\label{lem:raw}
	Let $[l,h)\subset[0,1)$ be the mapping interval of $x^n$ and $m\triangleq-\lfloor\log_2{(h-l)}\rfloor$. If $\frac{\lceil{l2^m}\rceil+1}{2^m}\geq h$, there is only one raw arithmetic bitstream, which is a zero-tail bitstream representing $\lceil{l2^m}\rceil$; if $\frac{\lceil{l2^m}\rceil+1}{2^m}<h$, there are two raw arithmetic bitstreams, among which the raw arithmetic bitstream representing $\lceil{l2^m}\rceil$ is a prefix bitstream, while the raw arithmetic bitstream representing $(\lceil{l2^m}\rceil+1)$ is a zero-tail bitstream. 
\end{theorem}
\begin{proof}
	Since $l2^m\leq\lceil{l2^m}\rceil<l2^m+1$, we have $l\leq\lceil{l2^m}\rceil2^{-m}<l+2^{-m}$. Since $2^{-m}=2^{\lfloor\log_2{(h-l)}\rfloor}\leq2^{\log_2{(h-l)}} = (h-l)$, we have $l+2^{-m}\leq l+(h-l)=h$. Since $l\leq\lceil{l2^m}\rceil2^{-m}<h$, the string of $m$ bits for $\lceil{l2^m}\rceil$ is always a raw arithmetic bitstream of $x^n$. 
	
	Then we proceed to the string of $m$ bits representing $(\lceil{l2^m}\rceil+1)$. Since $2^{-m}=2^{\lfloor\log_2{(h-l)}\rfloor}>2^{\log_2{(h-l)}-1} = (h-l)/2$, we have $l+2^{1-m}>l+(h-l)=h$. In a word, $l+2^{-m}\leq h < l+2^{1-m}$. As we know,
	\begin{align*}
		0\leq(.\underbrace{0\cdots0}_{m}z_{m+1}\cdots)_b<2^{-m},
	\end{align*}
	we have
	\begin{align*}
		(.z_1\cdots z_m)_b \leq (.z_1\cdots z_mz_{m+1}\cdots)_b < (.z_1\cdots z_m)_b+2^{-m},
	\end{align*}
	where $(.z_1\cdots z_m)_b=\lceil{l2^m}\rceil2^{-m}$ or $(\lceil{l2^m}\rceil+1)2^{-m}$. From $l+2^{-m}\leq(\lceil{l2^m}\rceil+1)2^{-m}<l+2^{1-m}$, it is easy to get the following points.
	\begin{itemize}
		\item If $l<l+2^{-m}\leq (\lceil{l2^m}\rceil+1)2^{-m}<h$, there must be two raw arithmetic bitstreams, representing $\lceil{l2^m}\rceil$ and $(\lceil{l2^m}\rceil+1)$, respectively. Since $\lceil{l2^m}\rceil2^{-m}+2^{-m}<h$, the raw arithmetic bitstream representing $\lceil{l2^m}\rceil$ is a prefix bitstream; since $(\lceil{l2^m}\rceil+1)2^{-m}+2^{-m}\geq l+2^{1-m}>h$, the raw arithmetic bitstream representing $(\lceil{l2^m}\rceil+1)$ is a zero-tail bitstream.
		
		\item If $h\leq (\lceil{l2^m}\rceil+1)2^{-m}<l+2^{1-m}$, the string of $m$ bits representing $(\lceil{l2^m}\rceil+1)$ is not a raw arithmetic bitstream and hence there is only one raw arithmetic bitstream representing $\lceil{l2^m}\rceil$, which is a zero-tail bitstream because $\lceil{l2^m}\rceil2^{-m}+2^{-m}\geq h$.
	\end{itemize}
	Now this theorem holds.
\end{proof}

As shown by \cref{lem:raw}, raw arithmetic bitstreams are usually not prefix bitstreams, and their successful decoding can be guaranteed only if they are followed by the zero tail. There is a classical solution to this problem proposed in \citep{WittenCACM87}, as described by \cref{alg:prefix}, which makes use of the MSBs and 2-MSBs of the final sliding window. The reader can easily verify that \cref{alg:prefix} always produces prefix bitstreams. Up to now, we are ready to give the intact pseudo code of arithmetic encoder in \cref{alg:encoder}, where $z^m\in\mathbb{B}^m$ is the arithmetic bitstream of $x^n$ produced by a finite-precision encoder. Note that $m$ in \cref{alg:prefix} and \cref{alg:encoder} is the length of real-world arithmetic bitstreams, rather than $-\lfloor\log_2{(h-l)}\rfloor$, the length of raw arithmetic bitstreams.

\begin{algorithm}[!t]
	\caption{$({\bf z},m) = \mathrm{end\_bitstream}({\bf z},j,\lambda,\eta,\upsilon)$}
	\begin{algorithmic}[1]
		\If{both the 2-MSB of $\lambda$ and the 2-MSB of $\eta$ are $0$} 
		\State 
		$(z_{j+1}\dots z_{j+\upsilon+2}) \gets (0\underbrace{1\cdots1}_{\upsilon}1)$
		\ElsIf{both the 2-MSB of $\lambda$ and the 2-MSB of $\eta$ are $1$}
		\State 
		$(z_{j+1}\dots z_{j+\upsilon+2}) \gets (1\underbrace{0\cdots0}_{\upsilon}0)$
		\Else\Comment{the 2-MSB of $\lambda$ must be $0$ and the 2-MSB of $\eta$ must be $1$ because underflow bits have been removed}
		\State 
		$(z_{j+1}\dots z_{j+\upsilon+2}) \gets (0\underbrace{1\cdots1}_{\upsilon}1)$ or $(1\underbrace{0\cdots0}_{\upsilon}0)$
		\EndIf
		\State $m \gets j+\upsilon+2$
	\end{algorithmic}
	\label{alg:prefix}
\end{algorithm}

\begin{algorithm}[!t]
	\caption{$z^m = \mathrm{arithmetic\_encoder}(x^n,p,w,\text{\color{red}isPrefix})$} 
	\Comment{$m$: length of arithmetic bitstream}\\
	\hspace*{\fill}\Comment{$\text{\color{red}isPrefix}$: flag whether arithmetic bitstream is prefix or not}
	\begin{algorithmic}[1]
		\State Allocate a bulk of memory ${\bf z}$ large enough to accommodate the arithmetic bitstream of $x^n$
		\State $\upsilon\gets0$, $j\gets0$, and $[\lambda:\eta]\gets[0:2^w)$
		\For{$i\gets 1$ to $n$}
		\State $({\bf z},j,\lambda,\eta,\upsilon) = \mathrm{encode\_one\_symbol}(x_i,p,{\bf z},j,\lambda,\eta,\upsilon)$
		\EndFor
		\If{\text{\color{red}isPrefix}}
			\State $({\bf z},m) = \mathrm{end\_bitstream}({\bf z},j,\lambda,\eta,\upsilon)$ \Comment{prefix bitstream}
		\Else
			\State $m\gets j$ \Comment{half-tail bitstream}
		\EndIf
	\end{algorithmic}
	\label{alg:encoder}
\end{algorithm}

The classical solution \cref{alg:prefix} \citep{WittenCACM87} is actually based on such an assumption: The length of arithmetic bitstream $m$ is unknown at the decoder, so it is unpredictable what bits will be appended to $z^m$ at the decoder. However, most of data are nowadays packetized, so the length of arithmetic bitstream can be easily inferred at the decoder from overhead information. If $m$ is available at the decoder, both zero-tail and half-tail bitstreams can always be correctly decoded, because the decoder can append $(00\cdots)$ or $(10\cdots)$ to $z^m$. We have known that raw arithmetic bitstreams are zero-tail bitstreams. It was proven in \citep{FangTIT20} and can also be verified by \cref{tab:exam} that given $m$ available at the decoder, \cref{alg:prefix} is not needed (see the false $\text{\color{red}isPrefix}$ branch of \cref{alg:encoder}), if only $(z_{m+1}\cdots) = (10\cdots)$ is appended to $z^m$. If \cref{alg:prefix} is bypassed, \cref{alg:encoder} will output a half-tail arithmetic bitstream. From \cref{tab:exam}, we can easily get the following finding, which was also proved in \cite{FangTIT20}.

\begin{theorem}[Length of Arithmetic Bitstream]
	For any $x^n\in\mathbb{B}^n$, its prefix bitstream is the longest, its half-tail bitstream is the shortest, and the length of its raw bitstream is in between.
\end{theorem}

However, note that for both raw bitstream and half-tail bitstream, the length of arithmetic bitstream should transmitted to the decoder as side information; otherwise the decoding may fail.

\section{General Arithmetic Codes}
Arithmetic codes were originally proposed to implement lossless source coding. There is a well-known duality between source coding and channel coding, so every good source code may also be a good channel code meanwhile, and vice versa. For example, {\em Low-Density Parity-Check} (LDPC) codes and polar codes can also be used for data compression. Inspired by this duality, arithmetic codes can be easily generalized to address many coding problems beyond source coding.  We notice that with the symbol-interval mapping rule defined by \cref{subfig:ac}, the interval $[0,1)$ will be fully covered and there is neither {\em overlapped} sub-interval (corresponds to more than one source symbol) nor {\em forbidden} sub-interval (does not correspond to any source symbol). Instead, if we redefine a different symbol-interval mapping rule, arithmetic codes can also realize {\em channel coding}, {\em distributed source coding}, {\em joint source-channel coding}, {\em etc.}, just as shown below. 
\begin{itemize}
	\item As shown by \cref{subfig:oac}, if every source symbol is mapped to an enlarged sub-interval, we will obtain {\em overlapped arithmetic codes}, which can implement {\em distributed source coding} \cite{GrangettoCL07,GrangettoTSP09}. With overlapped arithmetic codes, the mapping sub-intervals of different source symbols will be partially {\em overlapped}, as shown by the thick line segment in \cref{subfig:oac}.
	
	\item As shown by \cref{subfig:eac-A}, \cref{subfig:eac-B}, and \cref{subfig:eac-C}, if every source symbol is mapped to a narrowed sub-interval, we will obtain {\em forbidden arithmetic codes}, which can implement {\em joint source-channel coding} \cite{554275}. Especially for uniformly distributed sources, {\em forbidden arithmetic codes} actually degenerate into {\em error-correcting codes} for {\em channel coding}. With {\em forbidden arithmetic codes}, there will be one or more sub-intervals that do not correspond to any source symbol, which are called {\em forbidden} sub-intervals, as shown by the dashed line segments in \cref{subfig:eac-A}, \cref{subfig:eac-B}, and \cref{subfig:eac-C}).
	
	\item As shown by \cref{subfig:jscac}, by allowing the coexistence of {\em overlapped} sub-intervals and {\em forbidden} sub-intervals, we will obtain {\em hybrid arithmetic codes}, which can implement {\em distributed joint source-channel coding} \cite{5652684}.
\end{itemize}

\section{Arrangement of Remaining Sections}
DSC is an important branch of network information theory featured by separate encoding and joint decoding. Compared with conventional or centralized source coding, DSC has mainly two advantages: (a) The heavy computation burden is shifted from encoder to decoder, making DSC suited to those scenarios requiring lightweight terminals, \textit{e.g.}, wireless sensor network \cite{1328091}; (b) The flow of data across terminals is avoided for privacy protection, making DSC suited to those scenarios where the data to be compressed are highly confidential, \textit{e.g.}, distributed genome data compression \cite{6543135}.

In essence, lossless DSC, also known as Slepian-Wolf coding, is equivalent to channel coding \cite{5281737,5319758} and thus naturally implemented with channel codes \cite{1184140}, {\em e.g.}, turbo codes \cite{957380}, LDPC codes \cite{1042242}, and polar codes \cite{6284254}. However, channel codes are originally designed for the purpose of error correction, so may not be well suited to data compression for several reasons: (a) Most channel codes are binary codes, while real-world sources, \textit{e.g.}, images and videos, are typically nonbinary; (b) The statistical features of real-world sources are usually time-varying, while it is difficult for block channel codes to adapt to nonstationary source characteristics.   

As the most famous method for lossless source coding, arithmetic codes can also be adapted to other coding problems by simply redefining symbol-interval mapping rules. For example, Slepian-Wolf coding can be implemented with overlapped arithmetic codes, a variant of arithmetic codes mapping source symbols to partially overlapped intervals. The following sections of this monograph will review recent advances in {\em overlapped arithmetic codes} during the past decade, while leaving the discussion on {\em forbidden arithmetic codes} and {\em hybrid arithmetic codes} as future work. Besides overlapped arithmetic codes, overlapped {\em quasi}-arithmetic codes were also proposed in \cite{4379079} for the same purpose. In addition, there is another similar scheme that makes use of arithmetic codes to implement Slepian-Wolf coding by bit puncturing rather than interval overlapping, which is named as {\em punctured arithmetic codes} \cite{4957108}. This monograph will discuss only overlapped arithmetic codes with half-tail bitstreams, while ignoring other similar works. Finally, we would like to remind the reader of that for clarity, this monograph treats with only uniformly distributed binary sources. As for the works on non-uniformly distributed binary sources or non-binary sources, the reader may refer to the references recommended in \cref{c-conclusion}.

\chapter{Coset Cardinality Spectrum}\label{c-ccs}
\vspace{-25ex}%
Note that this and the following sections will treat with only uniformly distributed binary sources for clarity and conciseness. 

This section will utilize probability theory to build a mathematical model for overlapped arithmetic codes. It will be found that overlapped arithmetic codes are actually a kind of nonlinear coset codes partitioning source space into unequal-sized cosets. Hence the distribution of coset cardinality, or the so-called {\em Coset Cardinality Spectrum} (CCS), is a fundamental property of overlapped arithmetic codes. This section will give methods to calculate the CCS of overlapped arithmetic codes theoretically or numerically. On the basis of CCS, this section will also derive the decoding algorithm of overlapped arithmetic codes. Experimental results will show that aided by CCS, the performance of low-complexity decoder can be significantly improved. 

\section{Math Foundation}
To begin with, let us introduce an important theorem, which will lay the math foundation for the rest of this monograph. 
\begin{definition}[Counting Function]
	Let $\wedge$ denote boolean AND. Let $\omega=(a_1,a_2,\dots)$ be a sequence of real numbers. The counting function is defined as the number of terms $a_i$'s, where $1\leq i\leq n$, for which $\{a_i\}\triangleq a_i-\lfloor{a_i}\rfloor\in{\cal I}\subseteq[0,1)$. In mathematical language,
	\begin{align}
		c({\cal I};n;\omega) \triangleq |\{a_i: (\{a_i\}\in{\cal I})\wedge(1\leq i\leq n)\}|.
	\end{align}
\end{definition}

With the help of counting function, we introduce the following concept (cf. \textbf{Definition 1.1} of \cite{bible}).
\begin{definition}[u.d. mod 1 Sequence]\label{def:udm1}
	The sequence $\omega = (a_1,a_2,\dots)$ is {\em uniformly distributed modulo 1} (u.d. mod 1) if for every pair of $l$ and $h$	with $0\leq l<h\leq 1$, we have 
	\begin{align}
		\lim_{n\to\infty}\frac{c([l,h);n;\omega)}{n} = h-l.
	\end{align}
\end{definition}

An interesting finding is that almost every geometric sequence is a u.d. mod 1 sequence, as stated by the following lemma (cf. \textbf{Corollary 4.1} and \textbf{Corollary 4.2} of \cite{bible}).
\begin{lemma}[Property of Geometric Sequence]\label{prop:gp}
	The geometric sequence $(a,ar,ar^2,\dots)$, where $a\neq 0$, is u.d. mod 1 for almost every non-integral $r>1$, and the exceptional set has Lebesgue measure zero.
	\label{prop:gpud}
\end{lemma}

\begin{lemma}[u.d. mod 1 Sequence]
	\label{lem:um1}
	Let $(a_1,a_2,\dots)$ be a u.d. mod 1 sequence and $(X_1,X_2,\dots)$ a sequence of {\em independent and identically distributed} (i.i.d.) discrete random variables. Let $E_n\triangleq\{a_1X_1+\cdots+a_nX_n\}$, where $\{\cdot\}$ denotes the fractional part of a real number. Then $E_n$ will be {\em uniformly distributed} (u.d.) over $[0,1)$ as $n\to\infty$. 
\end{lemma}

\begin{theorem}[Weighted Sum of Geometric Sequence]
	\label{thm:gm}
	Let $(X_1,X_2,\dots)$ be a sequence of i.i.d. discrete random variables. Let $E_n\triangleq\{aX_1+arX_2+\cdots+ar^{n-1}X_n\}$, where $a\neq 0$ and $\{\cdot\}$ denotes the fractional part of a real number. Then for almost every non-integral $r>1$, $E_n$ will be u.d. over $[0,1)$ as $n\to\infty$. 		
\end{theorem}

\section{Asymptotic Spectrum}\label{sec:asympt}
We begin with infinite-length overlapped arithmetic codes. Let $X$ be a binary random variable with bias probability $\Pr(X=1)=1/2$ and $X^n=X_1^n\triangleq(X_1,\dots,X_n)$ be $n$ independent copies of $X$. The encoder recursively shrinks the initial interval $[0,1)$ according to $X^n$. Let $[l(X^i),h(X^i))$ be the mapping interval of $X^i$, and initially, $[l(X^i),h(X^i))=[0,1)$. According to \cref{subfig:oac}, for a rate-$r$ overlapped arithmetic code, we have $h(X^i)-l(X^i)=2^{-ir}$ and
\begin{align}
	l(X^i) 
	&= l(X^{i-1}) + X_i\cdot(1-2^{-r})\cdot\left(h(X^{i-1})-l(X^{i-1})\right)\nonumber\\
	&= l(X^{i-1}) + X_i\cdot(1-2^{-r})\cdot2^{-(i-1)r}\label{eq:lXirecusion}.
\end{align}
After rearrangement, we obtain
\begin{align}
	l(X^i) 
	&= (1-2^{-r})\sum_{i'=1}^{i}{\left(X_{i'}\cdot2^{-(i'-1)r}\right)}\nonumber\\
	&= (2^r-1)\sum_{i'=1}^{i}{\left(X_{i'}\cdot2^{-i'r}\right)}.\label{eq:lXi}
\end{align}
For $r>0$, we have
\begin{align}
	\lim_{n\to\infty}{(h(X^n)-l(X^n))} = \lim_{n\to\infty}{2^{-nr}} = 0. 
\end{align}
In other word, as $n$ increases, the mapping interval of $X^n$ will converge to a real number $U$ in $[0,1)$, \textit{i.e.},
\begin{align}\label{eq:U}
	U &= \lim_{n\to\infty}l(X^n)=\lim_{n\to\infty}h(X^n)\nonumber\\
	  &= (2^r-1)\sum_{i=1}^{\infty}{X_i2^{-ir}}.
\end{align}
A natural problem is: How is the random variable $U$ distributed over $[0,1)$? Let $f(u)$ be the \textit{probability density function} (pdf) of $U$. We temporarily name $f(u)$ as the \textit{Asymptotic Spectrum} of overlapping arithmetic codes. According to the properties of pdf, we have
\begin{align}\label{eq:definitionfield}
	\begin{cases}
		f(u) \geq 0,& 0 \leq u < 1\\
		f(u) = 0, 	& u\notin[0,1)
	\end{cases}
\end{align}
and
\begin{align}\label{eq:integral}
	\int_{-\infty}^{\infty}{f(u)\,du} = \int_{0}^{1}{f(u)\,du} = 1.
\end{align}

\begin{theorem}[Implicit Form of Asymptotic Spectrum]\label{thm:asympt}
	For a rate-$r$ overlapped arithmetic code defined by \cref{subfig:oac}, its asymptotic spectrum, denoted as $f(u)$, is given by 
	\begin{align}\label{eq:asympt}
		2^{1-r}f(u) &= f(u2^r) + f((u-(1-2^{-r}))2^r)\\
					&= f(u2^r) + f(1-(1-u)2^r).
	\end{align}
\end{theorem}

\begin{figure}[!t]
	\subfigure[$q=1/\sqrt{2}$]{\includegraphics[width=.5\linewidth]{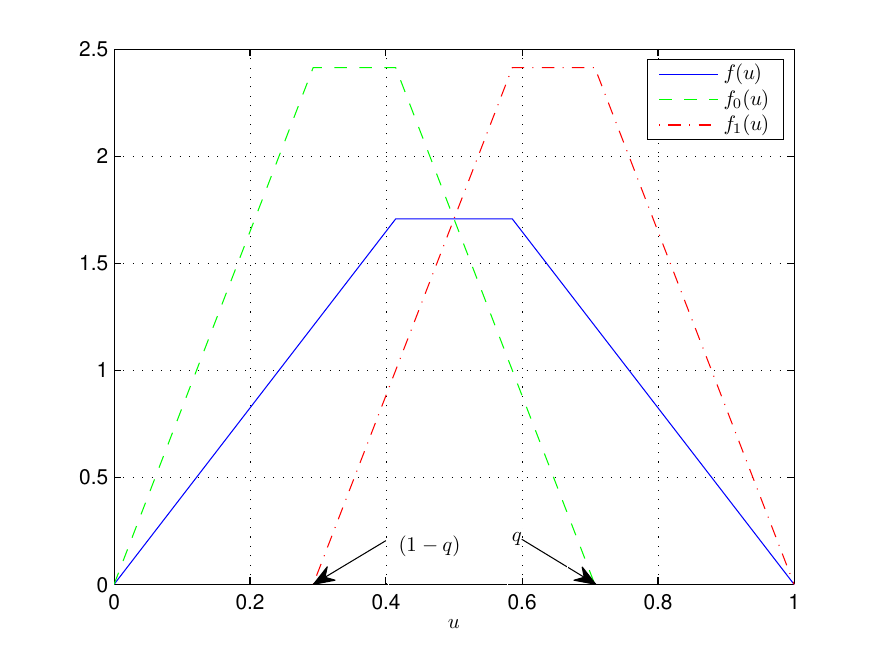}\label{subfig:sqrt2}}%
	\subfigure[$1/2<q\leq(\sqrt{5}-1)/2$]{\includegraphics[width=.5\linewidth]{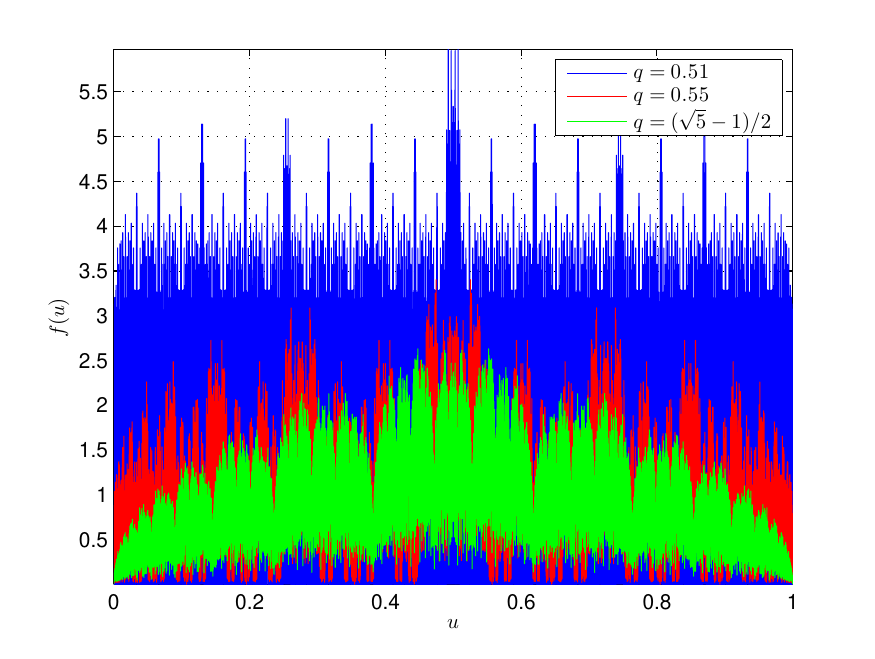}\label{subfig:sqrt5}}\\
	\subfigure[$(\sqrt{5}-1)/2<q<1/\sqrt{2}$]{\includegraphics[width=.5\linewidth]{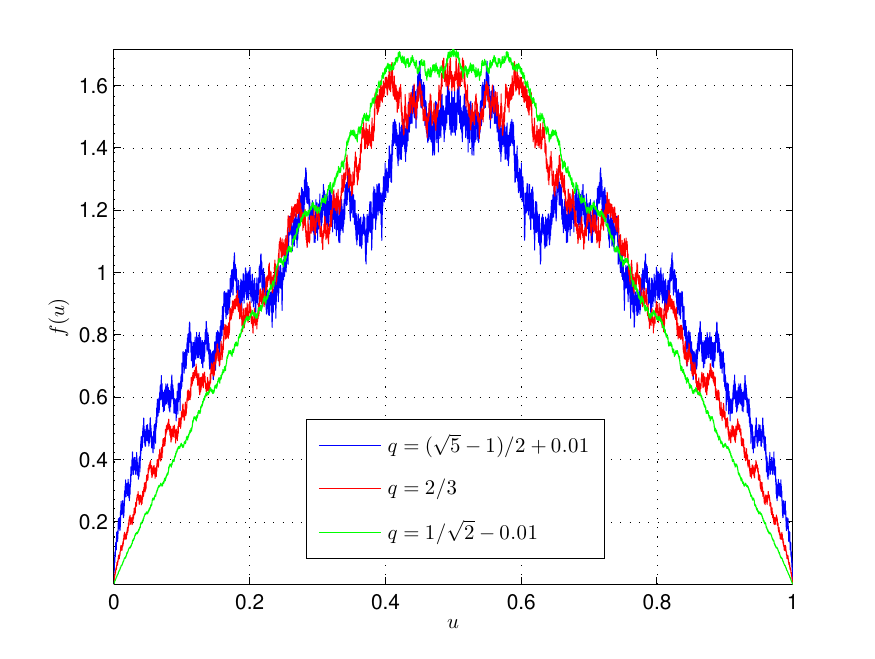}\label{subfig:sqrt52}}%
	\subfigure[$q>1/\sqrt{2}$]{\includegraphics[width=.5\linewidth]{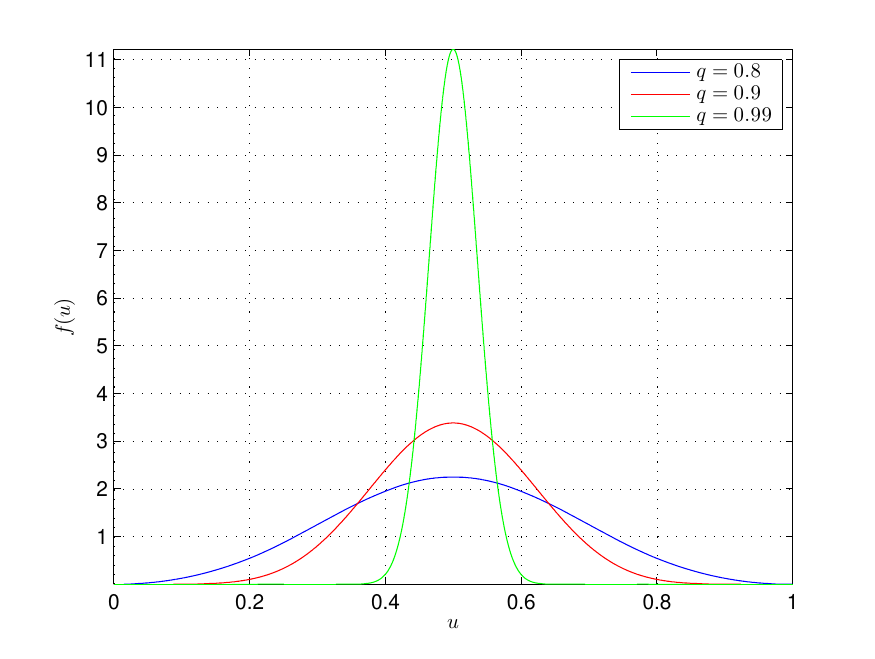}\label{subfig:to1}}
	\caption{Examples of asymptotic spectrum, where $q=2^{-r}=2^{-1/2}$. (a) An illustration of the relations between $f(u)$, $f_0(u)$, and $f_1(u)$. Both $f_0(u)$ and $f_1(u)$ have the same shape as $f(u)$ because $f_0(u)$ is obtained by scaling $f(u)$ by $q$ times along the $u$-axis and $f_1(u)$ is obtained by shifting $f_0(u)$ right by $(1-q)$. It can also be found that $f(\frac{1}{2}) = f_0(\frac{1}{2}) = f_1(\frac{1}{2})$. (b) When $1/2<q\leq(\sqrt{5}-1)/2$, $f(u)$ has many zero points. (c) When $(\sqrt{5}-1)/2<q<1/\sqrt{2}$, $f(u)$ has no zero point, but is still not a smooth function. (d) When $q>1/\sqrt{2}$, $f(u)$ is a smooth function.}
	\label{fig:asympt}
\end{figure}

\begin{proof}
	If $X_1=0$, then $X_1^\infty$ will be mapped to a real number in $[0,2^{-r})$; if $X_1=1$, then $X_1^\infty$ will be mapped to a real number in $[(1-2^{-r}),1)$. Let $f_x(u)$ be the conditional pdf of $U$ given $X_1=x\in\mathbb{B}$. According to \cref{subfig:sqrt2}, it is easy to obtain
	\begin{align}\label{eq:fu}
		f(u) 
		&= \Pr(X_1=0)f_0(u) + \Pr(X_1=1)f_1(u)\nonumber\\
		&= (f_0(u)+f_1(u))/2.
	\end{align}
	Since both $X_2^\infty$ and $X_1^\infty$ are infinite-length sequences of i.i.d. random variables, $f(u)$, $f_0(u)$, and $f_1(u)$ must be similar to each other, and the only difference between them is that they are defined in different domains, \textit{i.e.}, $f(u)$ is defined over $[0,1)$, while $f_0(u)$ is defined over $[0,2^{-r})$ and $f_1(u)$ is defined over $[(1-2^{-r}),1)$ (cf. \cref{subfig:sqrt2}). Based on this finding, we can obtain $f_0(u)\propto f(u/q)=f(u2^r)$ and thus
	\begin{align}\label{eq:f0u}
		f_0(u) = \frac{f(u2^r)}{\int_{0}^{2^{-r}}{f(v2^r)dv}} = 2^rf(u2^r).
	\end{align}
	Similarly, we have $f_1(u)=f_0(u-(1-q))=f_0(u-(1-2^{-r}))$ and thus
	\begin{align}\label{eq:f1u}
		f_1(u) = 2^rf((u-(1-2^{-r}))2^r).
	\end{align}
	Finally, substituting \eqref{eq:f0u} and \eqref{eq:f1u} into \eqref{eq:fu} gives \eqref{eq:asympt}.
\end{proof}

It can be seen that \eqref{eq:asympt} is a functional equation {\em with respect to} (w.r.t.) $f(u)$. After solving \eqref{eq:asympt}, we can obtain the closed form of $f(u)$. According to \eqref{eq:definitionfield} and \eqref{eq:asympt}, we can obtain
\begin{align}\label{eq:nonoverlap}
	2^{1-r}f(u) = 
	\begin{cases}
		f(u2^r), &0\leq u<(1-2^{-r})\\
		f(1-(1-u)2^r), &2^{-r}\leq u<1.
	\end{cases}
\end{align}
From the first branch of \eqref{eq:nonoverlap}, we can get $2^{1-r}f(0)=f(0)$, and hence $f(0)=f(1)=0$ for $r<1$ \citep{FangSPL09}. The following corollary then follows.

\begin{corollary}[Symmetry of Asymptotic Spectrum]
	For $r<1$, $f(u)$ is strictly symmetric around $u=\frac{1}{2}$, \textit{i.e.}, $f(u)=f(1-u)$.
\end{corollary}

\begin{corollary}[Asymptotic Spectrum of Classical Arithmetic Codes]\label{corollary:classicAC}
	The asymptotic spectrum of classical arithmetic codes is a shifted unit gate function, \textit{i.e.} if $r=1$, then 
	$f(u+\frac{1}{2}) = G_1(u)$ \cite{FangSPL09}, where the gate function $G_T(u)$ is defined as
	\begin{align}
		G_T(u) = \begin{cases}
			1/T, 	& -\frac{T}{2} \leq u < \frac{T}{2}\\
			0,  	& u\notin[-\frac{T}{2},\frac{T}{2}).
		\end{cases}
	\end{align}
\end{corollary}

\begin{corollary}[Zero Points]\label{corollary:zeros}
	If $\frac{1}{2}<2^{-r}\leq\frac{\sqrt{5}-1}{2}$, then $\forall k \in \mathbb{N}$,
	\begin{align}
		f(\tfrac{2^{-kr}}{2^{-r}+1}) = f(1-\tfrac{2^{-kr}}{2^{-r}+1}) = 0.
	\end{align}
\end{corollary}

\begin{corollary}[Limit of Asymptotic Spectrum]\label{corollary:limitInitial}
	As $r$ goes to $0$, $f(u+1/2)$ will converge to the unit Dirac delta function $\delta(u)$.
\end{corollary}

\begin{theorem}[Explicit Form of Asymptotic Spectrum]
	Let $\mathscr{F}^{-1}$ denote inverse Fourier transform. For uniform binary sources, the asymptotic spectrum of rate-$r$ overlapped arithmetic code is given by \cite{FangTCOM13}
	\begin{align}\label{eq:initialExplicitForm}
		f(u+1/2) = \mathscr{F}^{-1} 
		\left\{\prod_{k=0}^{\infty}{\cos(\omega(1-2^{-r})2^{-kr-1})}\right\}.
	\end{align}
\end{theorem}

\begin{corollary}[Expansion of sinc Function]
	\begin{align}\label{eq:sinc}
		\mathrm{sinc}(\omega) = \frac{\sin(\omega)}{\omega} = \prod_{k=1}^{\infty}{\cos(\omega/2^k)}.
	\end{align}
\end{corollary}
However, we have to point out that it is not the first time that \eqref{eq:sinc} is found. In fact, Euler was the first one who discovered \eqref{eq:sinc}. From (\ref{eq:sinc}), we can get the special closed form of $f(u)$ for rate-$\frac{1}{m}$ overlapped arithmetic code, where $m\in\mathbb{Z}^+$.

\begin{corollary}[Special Closed Form of Asymptotic Spectrum]\label{corollary:closedForm}
	Let $\otimes$ be convolution operation and $T_m = (1-2^{-r})2^{mr}$. The closed 
	form of $f(u)$ for rate-$\frac{1}{m}$ overlapped arithmetic code, where $m\in\mathbb{Z}^{+}$, is
	\begin{align}\label{eq:specialClosedForm}
		f(u+1/2) = \bigotimes_{m'=1}^{m}{G_{T_{m'}}(u)}.
	\end{align}
\end{corollary}
Especially, when $r=1/2$, it can be deduced from \eqref{eq:specialClosedForm} that $f(u+1/2) = G_{T_{1}}(u) \otimes G_{T_{2}}(u)$, where $T_1=\sqrt{2}-1$ and $T_2=2-\sqrt{2}$. Hence we have the following classical CCS \citep{FangTCOM13}
\begin{align}\label{eq:closedForm_halfRate}
	f(u) = \begin{cases}
		\frac{u}{3\sqrt{2}-4}, 		& 0 \leq u < \sqrt{2}-1\\
		\frac{1}{2-\sqrt{2}}, 		& \sqrt{2}-1 \leq u < 2-\sqrt{2}\\
		\frac{1-u}{3\sqrt{2}-4}, 	& 2-\sqrt{2} \leq u < 1.%
	\end{cases}
\end{align}

From \cref{corollary:classicAC}, \cref{corollary:limitInitial}, and \cref{corollary:closedForm}, it can be found that as $r$ decreases from $1$ to $0$, $f(u)$ will be gradually concentrated around $u=1/2$ (changing gradually from the uniform distribution to a Dirac delta spike at $u=1/2$).

\begin{example}[Asymptotic Spectrum]
\cref{fig:asympt} gives some examples of $f(u)$ to show how the shape of $f(u)$ changes w.r.t. $r$ or $q=2^{-r}$. It can be observed that $q=\frac{\sqrt{5}-1}{2}$ and $q = \frac{1}{\sqrt{2}}$ are two watersheds of $f(u)$ because $f(u)$ shows very different shapes when $q$ falls into $(\frac{1}{2},\frac{\sqrt{5}-1}{2}]$, $(\frac{\sqrt{5}-1}{2},\frac{1}{\sqrt{2}})$, and $[\frac{1}{\sqrt{2}},1)$. \cref{subfig:sqrt5} confirms the zero points of $f(u)$ when $\frac{1}{2}<q\leq\frac{\sqrt{5}-1}{2}$, as predicted by \cref{corollary:zeros}. It can be seen that though $q$ differs by only $0.01$, the shape of $f(u)$ is very different for $q=\frac{\sqrt{5}-1}{2}$ and $q=\frac{\sqrt{5}-1}{2}+0.01$. \cref{subfig:to1} shows that $f(u)$ becomes smooth when $\frac{1}{\sqrt{2}}\leq q<1$. It can also be found from \cref{subfig:to1} that $f(u)$ does converge to the Dirac delta spike centered at $u=0.5$ as $q$ goes to $1$ (or $r\to0$), as predicted by \cref{corollary:limitInitial}.
\end{example}

\section{Source Space Partitioning}\label{sec:spacepar}
Now we proceed to finite-length overlapped arithmetic codes. Still consider uniform binary sources. For simplicity, we assume $nr\in\mathbb{Z}$ below. To generate the bitstream, the final mapping interval $[l(X^n),h(X^n))$ is enlarged by $2^{nr}$ times to obtain a normalized interval $[s(x^n),s(x^n)+1)$ with unit length, where
\begin{align}\label{eq:ell}
	s(x^n) 
	&\triangleq 2^{nr}l(X^n) = (2^r-1)\sum_{i=1}^{n}{X_i2^{(n-i)r}}\nonumber\\
	&\overset{d}{=} (1-2^{-r})\sum_{i=1}^{n}{X_i2^{ir}},
\end{align}
where $\stackrel{d}{=}$ denotes two random variables equal in distribution. Clearly, there is one and only one integer in $[s(x^n),s(x^n)+1)$, which is just
\begin{align}\label{eq:mXn}
	M = m(X^n) \triangleq \lceil s(x^n) \rceil.
\end{align}
Obviously, $0\le s(x^n)\le(2^{nr}-1)$ and $M$ is drawn from $[0:2^{nr})\triangleq\{0,1,\dots,(2^{nr}-1)\}$, so $M$ can be represented by $nr\leq n$ bits, which is just the raw arithmetic bitstream of $X^n$. In this sense, overlapped arithmetic codes realize a many-to-one nonlinear mapping from $\mathbb{B}^n$ to $[0:2^{nr})$ or $\mathbb{B}^{nr}$, which partitions source space $\mathbb{B}^n$ into $2^{nr}$ cosets, and $M$ is just the index of the coset containing $X^n$. Note that source space partitioning by overlapped arithmetic codes differs from source space partitioning by linear block codes in that the former results in unequal-sized cosets, while the latter results in equal-sized cosets. 

As an analogy, each source block $x^n\in\mathbb{B}^n$ can be taken as a ball and each coset can be taken as a bin. Clearly, source space partitioning is just a procedure of ball binning. Let us give a simple example to demonstrate this analog. Assume $n=4$ and $r=0.5$. For every $x^n\in\mathbb{B}^n$, we calculate $s(x^n)$ according to \eqref{eq:ell}. The location of $s(x^n)$ along the real number axis is illustrated by \cref{subfig:mapping}. As shown by \cref{subfig:binning}, all source blocks (balls) falling into the interval $(m-1,m]$ will be classified into the $m$-th coset/bin ${\cal C}_m$. The codewords in each coset are structured into an incomplete binary tree, as shown by \cref{subfig:tree}.

\begin{figure}[!t]
	\centering
	\subfigure[]{\includegraphics[width=.8\linewidth]{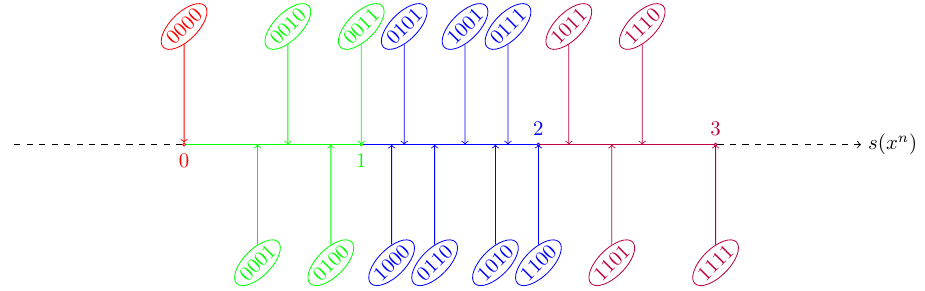}\label{subfig:mapping}}\\
	\subfigure[]{\includegraphics[width=.8\linewidth]{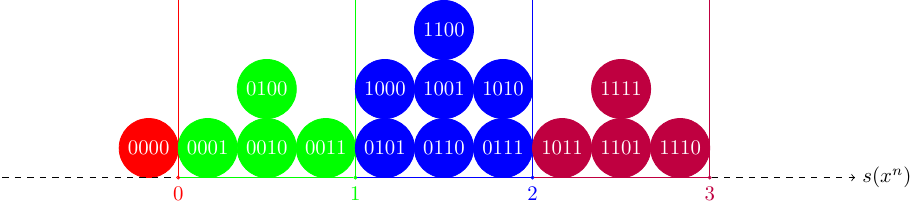}\label{subfig:binning}}
	\subfigure[]{\includegraphics[width=.8\linewidth]{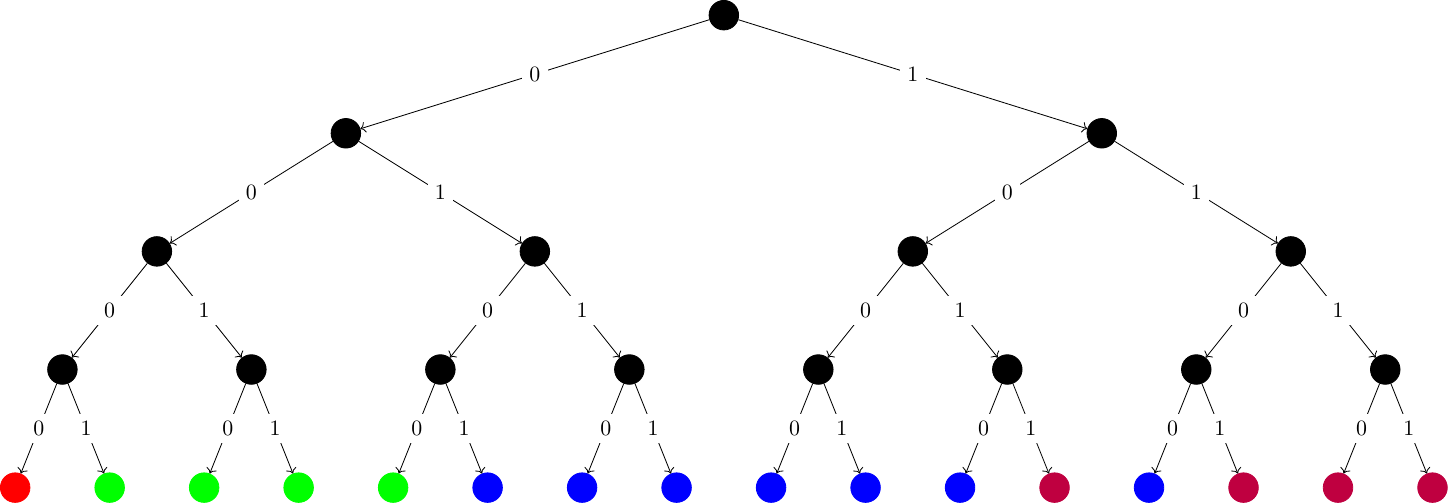}\label{subfig:tree}}
	\caption{An example of coset partitioning of source space with overlapped arithmetic codes for $n=4$ and $r=0.5$. (a) Mapping from $x^n\in\mathbb{B}^n$ to $s(x^n)\in\mathbb{R}$, where $s(x^n)$ is calculated by \eqref{eq:ell}. In this example, $0\le s(x^n)\le (2^{nr}-1)=3$. (b) A ball-binning analogy of overlapped arithmetic codes. Each source block $x^n\in\mathbb{B}^n$ can be taken as a ball. All $2^n=16$ balls in $\mathbb{B}^n$ are classified into $2^{nr}=4$ bins according to $s(x^n)$. Each bin corresponds to a coset and the $m$-th coset is denoted by ${\cal C}_m$. In general, the $m$-th bin is just the interval $(m-1,m]$. There is only one exception, \textit{i.e.}, the $0$-th bin ${\cal C}_0$ is a point $\{0\}$ and there is always one and only one ball $0^n$ in ${\cal C}_0$. (c) Tree structures of cosets.}
	\label{fig:part}
\end{figure}

\section{Definition of CCS}\label{sec:ccs}
As we know, linear block codes always partition source space into equal-sized cosets so that the distribution of coset cardinality is never a problem. On the contrary, as shown by \cref{subfig:binning}, overlapped arithmetic codes partition source space into unequal-sized cosets, so the distribution of coset cardinality is an important problem. Let $|{\cal C}_m|$ be the cardinality of ${\cal C}_m$, where $0\le m< 2^{nr}$. For example, in \cref{subfig:binning}, $|{\cal C}_0|=1$, $|{\cal C}_1|=|{\cal C}_3|=4$, and $|{\cal C}_2|=7$. It is very difficult to deduce the distribution of $|{\cal C}_m|$ for finite $n$ as $M$ is a discrete random variable. To solve this problem, we introduce the following concept.
\begin{definition}[Bitstream Projection]
	Let $M=m(X^n)=\lceil s(X^n)\rceil\sim p(m)$ for $m\in[0:2^{nr})$, where $s(X^n)$ is defined by \eqref{eq:ell}, be the bitstream of a length-$n$ and rate-$r$ overlapped arithmetic code. The projection of $M$ onto the interval $[l(X^i),h(X^i))$, where $l(X^i)$ is given by \eqref{eq:lXi} and $h(X^i)=l(X^i)+2^{-ir}$, is defined as
	\begin{align}\label{eq:Uin}
		U_{i,n} \triangleq \frac{\overbrace{2^{-nr}M}^{U_{0,n}}-l(X^i)}{\underbrace{h(X^i)-l(X^i)}_{2^{-ir}}} = 2^{ir}\left(U_{0,n}-l(X^i)\right).
	\end{align}
\end{definition}
Evidently, $U_{i,n}$ is defined over $[0,1)$ for every $i\in[0:n]\triangleq\{0,1,\dots,n\}$. After observing \eqref{eq:U} and \eqref{eq:Uin}, we find 
\begin{align}
	U &= \lim_{n\to\infty}{U_{0,n}} = \lim_{n\to\infty}{l(X^n)} = \lim_{n\to\infty}{h(X^n)}\nonumber\\
	&= (2^r-1)\sum_{i=1}^{\infty}{X_i2^{-ir}}.
\end{align}
For conciseness, $U_{i,n}$ can be shortened to $U_i$ without causing ambiguity. Especially, $U_0=2^{-nr}M$ is called the \textit{initial} projection and $U_n$ is called the \textit{final} projection. It is easy to obtain
\begin{align}\label{eq:Un}
	U_n = M-2^{nr}l(X^n) = \lceil s(x^n)\rceil-s(x^n), 
\end{align}
showing that $U_n$ is the ceiling error of $s(x^n)$. From \eqref{eq:Uin}, we have $U_i=2^{ir}\left(U_0-l(X^i)\right)$, and conversely, 
\begin{align}\label{eq:U=}
	U_0 
	&= 2^{-ir}U_i + l(X^i) = 2^{-(i+1)r}U_{i+1} + l(X^{i+1})\nonumber\\
	&\stackrel{(a)}{=} 2^{-(i+1)r}U_{i+1} + l(X^i) + X_{i+1}(1-2^{-r})2^{-ir},
\end{align}
where $(a)$ comes from \eqref{eq:lXirecusion}. After removing $2^{-ir}$ at both sides, \eqref{eq:U=} can be rewritten as 
\begin{align*}
	U_{i-1} = 2^{-r}U_i + X_i(1-2^{-r}), 
\end{align*}
which can be illustrated by \cref{fig:Uibackward}. Conversely, we have 
\begin{align*}
	U_i = 2^r\left(U_{i-1} - X_i(1-2^{-r})\right). 
\end{align*}
If $M=m(X^n)$ is decoded along the path $X^n$, we will obtain the sequence $U_0^n\triangleq(U_0,\dots,U_n)$ via the forward recursion. With the help of \cref{fig:Uibackward}, the following theorem holds obviously.
\begin{theorem}[Markov Chain]\label{thm:markov}
	The sequence $(U_{0},\dots,U_{n})$ forms a Markov chain, \textit{i.e.}, $U_{0} \to \dots \to U_{n}$. The sequence $(U_{0},\dots,U_{i-1},X_i)$ also forms a Markov chain, \textit{i.e.}, $U_{0} \to \dots \to U_{i-1} \to X_i$. The relation between $U_{i-1}$, $X_i$, and $U_i$ can be summarized by
	\begin{align}\label{eq:Ui}
		\begin{cases}
			U_{i-1} = 2^{-r}U_i + X_i(1-2^{-r}),\\
			U_i = 2^r\left(U_{i-1} - X_i(1-2^{-r})\right).
		\end{cases}
	\end{align}
\end{theorem}

\begin{figure}
	\begin{align}
		\begin{array}{ccccccccccc}
			U_0 & \leftarrow\cdots\leftarrow & U_{i-1} & \leftarrow & U_i & \leftarrow & U_{i+1} & \leftarrow\cdots\leftarrow & U_{n-1} & \leftarrow & U_n\\
			\uparrow & \cdots & \uparrow &  & \uparrow &  & \uparrow & \cdots & \uparrow\\
			X_1 & \cdots & X_i &  & X_{i+1} &  & X_{i+2} & \cdots & X_n\\
		\end{array}\nonumber
	\end{align}
	\caption{Backward-recursion deduction of $U_i$.}
	\label{fig:Uibackward}
\end{figure}

\begin{definition}[Coset Cardinality Spectrum]
	The pdf of $U_{i,n}$, denoted as $f_{i,n}(u)$, is called the level-$i$ {\em Coset Cardinality Spectrum} (CCS).
\end{definition}
For conciseness, $f_{i,n}(u)$ can be abbreviated to $f_i(u)$ without causing ambiguity. Especially, $f_0(u)$ is called the {\em initial} CCS and $f_n(u)$ is called the {\em final} CCS. It is now clear that the so-called {\em asymptotic spectrum} $f(u)$ that is temporarily named in \cref{sec:asympt} should be formally called {\em asymptotic initial CCS}.
\begin{corollary}[Alternative Explanation of Asymptotic Spectrum]
	We define $l(x^n)$ as \eqref{eq:lXi}. Let $\delta(u)$ denote the Dirac delta function. For uniform binary sources, we have
	\begin{align}\label{eq:futrivial}
		f(u) = \lim_{n\rightarrow\infty}f_{0,n}(u) = \lim_{n\rightarrow\infty}2^{-n}\sum_{x^n\in\mathbb{B}^n}{\delta(u-l(x^n))}.
	\end{align}
\end{corollary}
To deduce $f_{i,n}(u)$, we should begin with $f_{n,n}(u)$ and then go back to $f_{0,n}(u)$ recursively \citep{FangTCOM16b} (cf. \cref{fig:Uibackward}). According to \eqref{eq:Un}, $U_n$ is the ceiling error of $s(x^n)$. Furthermore, $s(x^n)$ is the weighted sum of $n$ leading terms drawn from the geometric sequence $(a,a2^r,a2^{2r},\dots)$, where $a=(2^r-1)$ (cf. \eqref{eq:ell}). Hence from \cref{thm:gm}, we can obtain the following theorem.  
\begin{theorem}[Final CCS]\label{thm:finalccs}
	For almost every code rate $r<1$, given $nr\in\mathbb{Z}$, we have
	\begin{align}\label{eq:Pi}
		\lim_{n\to\infty}f_{n,n}(u) = \Pi(u) \triangleq 
		\begin{cases}
			0, 	&u\notin[0,1)\\
			1,	&0\le u<1.
		\end{cases}
	\end{align}
\end{theorem}
After knowing $f_{n,n}(u)$, we can then deduce $f_{i,n}(u)$, where $0\leq i<n$, via a backward recursion. To achieve this goal, let us first introduce the concept of conditional CCS.
\begin{definition}[Conditional CCS]
	The conditional pdf of $U_{i,n}$ given $X_j=x\in\mathbb{B}$ is called the conditional CCS and denoted by $f_{i|j,n}(u|x)$, or just $f_{i|j}(u|x)$ for conciseness.
\end{definition}
\begin{lemma}[Conditional CCS]
	The conditional CCS is given by \cite{FangTCOM16b}
	\begin{align}\label{eq:condccs}
		f_{i-1|i}(u|x) = 2^rf_i((u-x(1-2^{-r}))2^r).
	\end{align}
\end{lemma}
\begin{theorem}[Backward Deduction of CCS]\label{thm:ccs}
	The level-$i$ CCS $f_{i,n}(u)$ can be deduced via the following backward recursion
	\begin{align}\label{eq:ccs}
		f_{i-1}(u) 
		&= \Pr(X_i=0)f_{i-1|i}(u|0) + \Pr(X_i=1)f_{i-1|i}(u|1)\nonumber\\
		&= 2^{r-1}\left(f_i(u2^r) + f_i((u-(1-2^{-r}))2^{r})\right).
	\end{align}
\end{theorem}
\begin{corollary}[Conditional Probability of Source Symbol Given CCS]
	Let $p_{j|i}(x|u) \triangleq \Pr(X_j=x|U_i=u)$. Then we have
	\begin{align}\label{eq:pxu}
		p_{i|i-1}(x|u) 
		&= \frac{\Pr(X_i=x)f_{i-1|i}(u|x)}{f_{i-1}(u)} = \frac{f_{i-1|i}(u|x)}{2f_{i-1}(u)}\nonumber\\
		&= 2^{r-1}\frac{f_i((u-x(1-2^{-r}))2^r)}{f_{i-1}(u)}.
	\end{align}
	\label{lem:condccs}
\end{corollary}

An example will be helpful for us to understand \cref{thm:ccs} and \cref{lem:condccs}. Consider the following classical example. For uniform binary sources and $r=0.5$, as $n\to\infty$, $f_0(u)$ will converge to $f(u)$ given by \eqref{eq:closedForm_halfRate}. Actually, $f_i(u)=f(u)$ for any $i$ satisfying $(n-i)=\infty$. We plot $f_{i-1}(u)$ in \cref{subfig:fux}. Given \eqref{eq:closedForm_halfRate}, it is easy to get
\begin{align}
	f_{i-1|i}(u|0) = \begin{cases}
		\frac{2u}{3\sqrt{2}-4}, 		&0\leq u<(1-\frac{1}{\sqrt{2}})\\
		\frac{\sqrt{2}}{2-\sqrt{2}},	&(1-\frac{1}{\sqrt{2}})\leq u<(\sqrt{2}-1)\\
		\frac{\sqrt{2}-2u}{3\sqrt{2}-4},&(\sqrt{2}-1)\leq u<\frac{1}{\sqrt{2}}\\
		0,								&\frac{1}{\sqrt{2}}\leq u<1.
	\end{cases}
\end{align}
We also plot $f_{i-1|i}(u|0)$ and $f_{i-1|i}(u|1)=f_{i-1|i}((u-(1-\frac{1}{\sqrt{2}}))|0)$ in \cref{subfig:fux}. In turn, we obtain  
\begin{align}
	p_{i|i-1}(0|u) = 
	\begin{cases}
		1, &0\leq u<(1-\frac{1}{\sqrt{2}})\\
		\frac{\frac{\sqrt{2}}{2-\sqrt{2}}}{\frac{\sqrt{2}}{2-\sqrt{2}}+\frac{\sqrt{2}-2(1-u)}{3\sqrt{2}-4}}, &(1-\frac{1}{\sqrt{2}})\leq u<(\sqrt{2}-1)\\
		\frac{\frac{\sqrt{2}-2u}{3\sqrt{2}-4}}{\frac{\sqrt{2}-2u}{3\sqrt{2}-4}+\frac{\sqrt{2}-2(1-u)}{3\sqrt{2}-4}}, &(\sqrt{2}-1)\leq u<(2-\sqrt{2})\\
		\frac{\frac{\sqrt{2}-2u}{3\sqrt{2}-4}}{\frac{\sqrt{2}-2u}{3\sqrt{2}-4}+\frac{\sqrt{2}}{2-\sqrt{2}}}, &(2-\sqrt{2})\leq u<\frac{1}{\sqrt{2}}\\ 
		0, 			&\frac{1}{\sqrt{2}}\leq u<1.
	\end{cases}
\end{align}
After simplification, we get
\begin{align}
	p_{i|i-1}(0|u) = \begin{cases}
		1, &0\leq u<(1-\frac{1}{\sqrt{2}})\\
		\frac{2-\sqrt{2}}{2u}, &(1-\frac{1}{\sqrt{2}})\leq u<(\sqrt{2}-1)\\
		\frac{\sqrt{2}-2u}{2(\sqrt{2}-1)}, &(\sqrt{2}-1)\leq u<(2-\sqrt{2})\\
		1-\frac{2-\sqrt{2}}{2(1-u)}, &(2-\sqrt{2})\leq u<\frac{1}{\sqrt{2}}\\ 
		0, &\frac{1}{\sqrt{2}}\leq u<1.
	\end{cases}
\end{align}
We plot $p_{i|i-1}(0|u)$ and $p_{i|i-1}(1|u)=1-p_{i|i-1}(0|u)$ in \cref{subfig:pxu}. It can be found that in the overlapped interval $[(1-\frac{1}{\sqrt{2}}),\frac{1}{\sqrt{2}})$, the probability $p_{i|i-1}(0|u)$ is monotonously decreasing from $1$ to $0$, while the probability $p_{i|i-1}(1|u)$ is monotonously increasing from $0$ to $1$. Hence in the overlapped interval $[(1-\frac{1}{\sqrt{2}})$, $\frac{1}{\sqrt{2}})$, we have $p_{i|i-1}(1|u)\neq p_{i|i-1}(0|u)$, except at $u=0.5$. That means: Even though $U_{i-1}$ falls into the overlapped interval, it still provides partial information of $X_i$, and the two branches are not equiprobable.
\begin{figure}[!t]
	\centering
	\subfigure[]{\includegraphics[width=.5\linewidth]{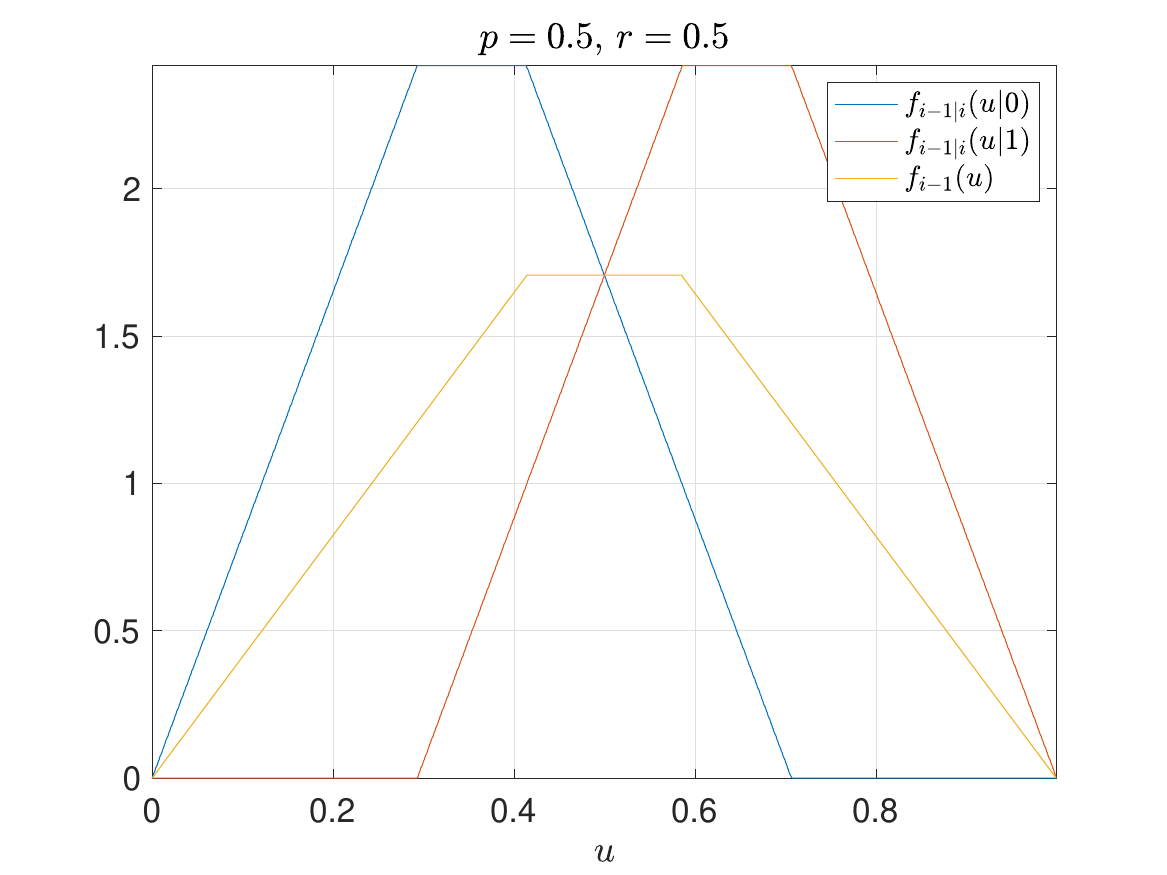}\label{subfig:fux}}%
	\subfigure[]{\includegraphics[width=.5\linewidth]{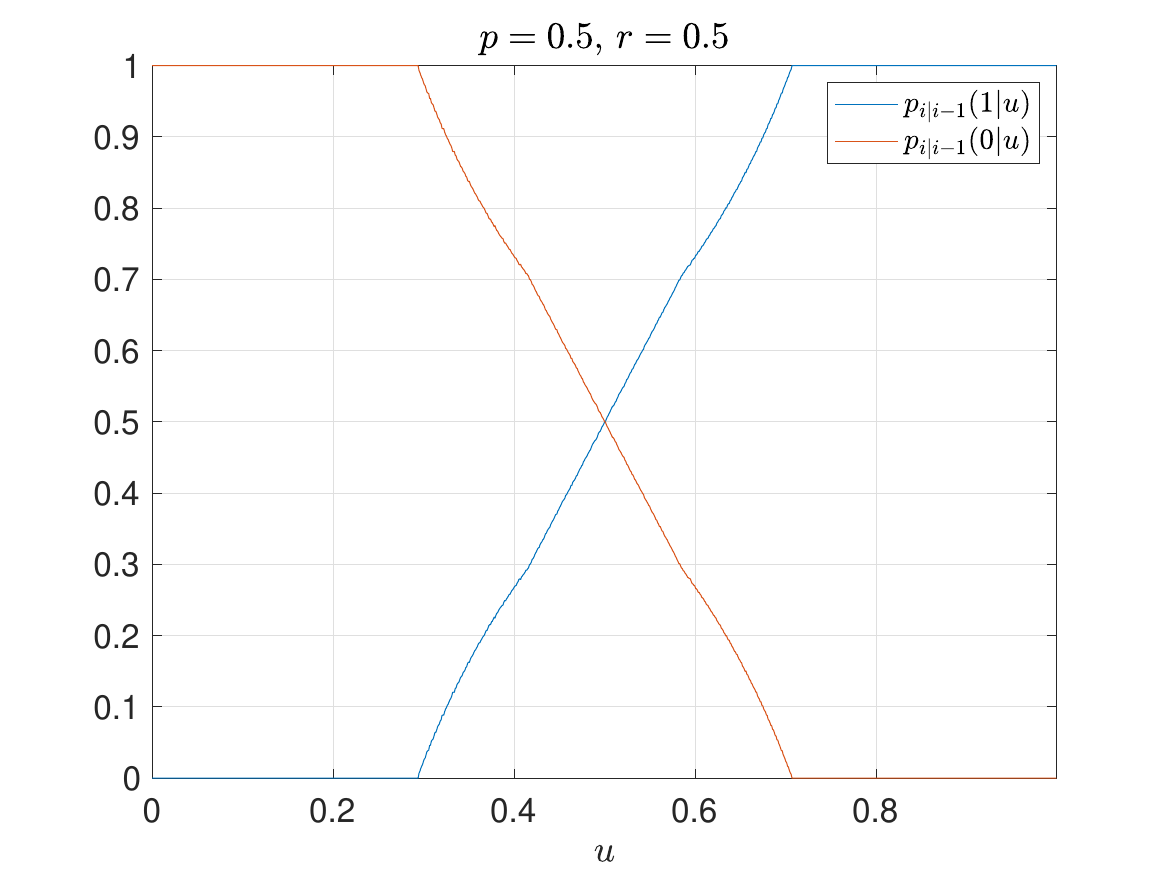}\label{subfig:pxu}}\\
	\caption{(a) Example of CCS and conditional CCS. (b) Conditional probability of source symbol given CCS.}
	\label{fig:fuxpxu}
\end{figure}

\section{Numerical Algorithm of CCS}\label{sec:numerical}
For conciseness, $f_{i,n}(u)$ will be abbreviated to $f_i(u)$ below. It is usually very hard to deduce the exact closed form of $f_i(u)$ according to \eqref{eq:ccs}, so we have to resort to a numerical algorithm. 

\subsection{Rounding Numerical Algorithm}
Let us discretize the interval $[0,1)$ into $N$ segments of equal length. We take $\hat{f}_i(j)$ as the approximation of $f_i(j/N)$. Initially, $\hat{f}_n(j)=1$ for every $j\in[0:N)$ (cf. \cref{thm:finalccs}). Let us define
\begin{align}\label{eq:lambda}
	\begin{cases}
		\lambda_0\triangleq j2^r\\
		\lambda_1\triangleq(j-N(1-2^{-r}))2^r.
	\end{cases}
\end{align}
Let $\lfloor\cdot\rceil$ denote the rounding function. Then \eqref{eq:ccs} can be discretized by
\begin{align}\label{eq:furounding}
	\hat{f}_{i-1}(j) = 2^{r-1}\left(\hat{f}_i(\lfloor\lambda_0\rceil) + \hat{f}_i(\lfloor\lambda_1\rceil)\right).
\end{align}
Sometimes we have $\lfloor\lambda_0\rceil\notin[0:N)$ or $\lfloor\lambda_1\rceil\notin[0:N)$. It does not matter because $\hat{f}_i(j)\equiv0$ for any $j\notin[0:N)$ according to the definition of CCS. Since the core of \eqref{eq:furounding} is the rounding function, this method is referred to as \textit{rounding} numerical algorithm.

\subsection{Linear Numerical Algorithm}
Let $\lfloor\cdot\rfloor$ and $\lceil\cdot\rceil$ denote the floor and ceiling functions, respectively. The rounding numerical algorithm \eqref{eq:furounding} can be slightly improved by
\begin{align}\label{eq:linear}
	\hat{f}_{i-1}(j) = 2^{r-1}\left(g_0(\lambda_0) + g_1(\lambda_1)\right),
\end{align}
where
\begin{align}\label{eq:inter}
	\begin{cases}
		g_0(\lambda_0) = (\lceil\lambda_0\rceil-\lambda_0)\hat{f}_{i}(\lfloor\lambda_0\rfloor) + (\lambda_0-\lfloor\lambda_0\rfloor)\hat{f}_{i}(\lceil\lambda_0\rceil)\\
		g_1(\lambda_1) = (\lceil\lambda_1\rceil-\lambda_1)\hat{f}_{i}(\lfloor\lambda_1\rfloor) + (\lambda_1-\lfloor\lambda_1\rfloor)\hat{f}_{i}(\lceil\lambda_1\rceil).
	\end{cases}
\end{align}
Especially, if $\lambda_0\in\mathbb{Z}$, then $g_0(\lambda_0)=\hat{f}_i(\lambda_0)$; if $\lambda_1\in\mathbb{Z}$, then $g_1(\lambda_1)=\hat{f}_i(\lambda_1)$. Since \eqref{eq:inter} is similar to the linear interpolation formula, \eqref{eq:linear} is named as \textit{linear} numerical algorithm.

\subsection{Fine Numerical Algorithm}
It must be pointed out that for both rounding and linear numerical algorithms, a renormalization step is needed after \eqref{eq:furounding} and \eqref{eq:linear}:
\begin{align}
	\hat{f}_i(j) = \left.\hat{f}_i(j)\middle/\sum_{j'=0}^{N-1}\hat{f}_i(j')\right..
\end{align}
To remove this renormalization step, the \textit{fine} numerical algorithm was proposed in \citep{Fine}. Now let $\hat{f}_i(j)$ be an approximation of the scaled-up probability of $U_i$ falling into the interval $[j/N,(j+1)/N))$, \textit{i.e.}, 
\begin{align}\label{eq:hatfij}
	\hat{f}_i(j)\approx N\int_{j/N}^{(j+1)/N}f_i(u)\,du. 
\end{align}
According to \eqref{eq:ccs}, we have
\begin{align}\label{eq:fi1u}
	\int_{j/N}^{(j+1)/N}{f_{i-1}(u)\,du} 
		=\; &2^{r-1}\int_{j/N}^{(j+1)/N}{\left(f_i(u2^r) + \right.}\nonumber\\
		  &\qquad\quad\;\;{\left.f_i((u-(1-2^{-r}))2^r)\right)\,du}.
\end{align}
Besides $\lambda_0$ and $\lambda_1$ defined by \eqref{eq:lambda}, we also define
\begin{align}
\begin{cases}
	\eta_0\triangleq (j+1)2^r = \lambda_0+2^r\\
	\eta_1 \triangleq (j+1-N(1-2^{-r}))2^r = \lambda_1+2^r.
\end{cases}
\end{align}
Let $v\triangleq u2^r$ and $w\triangleq (u-(1-2^{-r}))2^r$. Then \eqref{eq:fi1u} will become
\begin{align}
	\int_{\frac{j}{N}}^{\frac{j+1}{N}}{f_{i-1}(u)\,du} 
	&= \left(\int_{\frac{\lambda_0}{N}}^{\frac{\eta_0}{N}}{f_i(v)\,d{v}} + \int_{\frac{\lambda_1}{N}}^{\frac{\eta_1}{N}}{f_i(w)\,d{w}}\right)/2\nonumber\\
	&= \left(\int_{\frac{\lambda_0}{N}}^{\frac{\eta_0}{N}}{f_i(u)\,d{u}} + \int_{\frac{\lambda_1}{N}}^{\frac{\eta_1}{N}}{f_i(u)\,d{u}}\right)/2,
\end{align}
which can be written as 
\begin{align}\label{eq:refined}
	\hat{f}_{i-1}(j) = \left(g_0(j) + g_1(j)\right)/2,
\end{align}
where $\hat{f}_{i-1}(j)$ is defined by \eqref{eq:hatfij} and
\begin{align}
	\begin{cases}
		g_0(j)\approx N\int_{\lambda_0/N}^{\eta_0/N}{f_i(u)\,du}\\ 	
		g_1(j)\approx N\int_{\lambda_1/N}^{\eta_1/N}{f_i(u)\,du}.
	\end{cases}
\end{align}
In plain words, $g_0(j)$ is an approximation of the scaled-up probability of $U_i$ falling into $[\lambda_0/N, \eta_0/N)$ and $g_1(j)$ is an approximation of the scaled-up probability of $U_i$ falling into $[\lambda_1/N, \eta_1/N)$.

For $0<r<1$, we have $1<2^r<2$. Hence, $\lfloor\eta_0\rfloor-\lceil\lambda_0\rceil=0$ or $1$. It is easy to obtain
\begin{align}
	g_0(j) =\; 
	&N\int_{\lambda_0/N}^{\lceil\lambda_0\rceil/N}{f_i(u)\,du} + N\int_{\lceil\lambda_0\rceil/N}^{\lfloor\eta_0\rfloor/N}{f_i(u)\,du} \;+ \nonumber \\
	&N\int_{\lfloor\eta_0\rfloor/N}^{\eta_0/N}{f_i(u)\,du}.
\end{align}
For $N$ sufficiently large, $f_i(u)$ will be approximately uniform over the interval $[j/N,(j+1)/N)$. Hence we have
\begin{align}
	N\int_{\lambda_0/N}^{\lceil\lambda_0\rceil/N}{f_i(u)\,du} 
	&\approx (\lceil\lambda_0\rceil-\lambda_0)N\int_{(\lceil\lambda_0\rceil-1)/N}^{\lceil\lambda_0\rceil/N}{f_i(u)\,du}\nonumber\\ 
	&\approx (\lceil\lambda_0\rceil-\lambda_0)\hat{f}_i(\lceil\lambda_0\rceil-1)
\end{align}
and
\begin{align}
	N\int_{\lfloor\eta_0\rfloor/N}^{\eta_0/N}{f_i(u)\,du} 
	&\approx (\eta_0-\lfloor\eta_0\rfloor)N\int_{\lfloor\eta_0\rfloor/N}^{(\lfloor\eta_0\rfloor+1)/N}{f_i(u)\,du}\nonumber\\
	&\approx (\eta_0-\lfloor\eta_0\rfloor)\hat{f}_i(\lfloor\eta_0\rfloor).
\end{align}
According to the above analysis, we can obtain
\begin{align}
	g_0(j) =\;&(\lceil\lambda_0\rceil-\lambda_0)\hat{f}_i(\lceil\lambda_0\rceil-1) + (\lfloor\eta_0\rfloor-\lceil\lambda_0\rceil)\hat{f}_i(\lceil\lambda_0\rceil)\;+ \nonumber\\
	&(\eta_0-\lfloor\eta_0\rfloor)\hat{f}_i(\lfloor\eta_0\rfloor).
\end{align}
Especially, if $\lfloor\eta_0\rfloor=\lceil\lambda_0\rceil$, 
\begin{align}
	g_0(j) = (\lceil\lambda_0\rceil-\lambda_0)\hat{f}_i(\lceil\lambda_0\rceil-1) + (\eta_0-\lfloor\eta_0\rfloor)\hat{f}_i(\lfloor\eta_0\rfloor).
\end{align}
Similarly, 
\begin{align}
	g_1(j) =\;&(\lceil\lambda_1\rceil-\lambda_1)\hat{f}_i(\lceil\lambda_1\rceil-1) + (\lfloor\eta_1\rfloor-\lceil\lambda_1\rceil)\hat{f}_i(\lceil\lambda_1\rceil)\;+ \nonumber\\
	&(\eta_1-\lfloor\eta_1\rfloor)\hat{f}_i(\lfloor\eta_1\rfloor),
\end{align}
and especially, if $\lfloor\eta_1\rfloor=\lceil\lambda_1\rceil$, 
\begin{align}
	g_1(j) = (\lceil\lambda_1\rceil-\lambda_1)\hat{f}_i(\lceil\lambda_1\rceil-1) + (\eta_1-\lfloor\eta_1\rfloor)\hat{f}_i(\lfloor\eta_1\rfloor).
\end{align}

\section{Information Theoretic Understanding}\label{sec:infthm}
With the asymptotic initial CCS $f(u)$, which is given by \eqref{eq:asympt}, we can have a better understanding of overlapped arithmetic codes. 

\subsection{Expected Coset Cardinality}
Let $M=\lceil s(X^n)\rceil\sim p(m)$, where $m\in[0:2^{nr})$, be the coset index of $X^n$. 
Let ${\cal C}_m$ be the $m$-th coset, whose cardinality is denoted by $|{\cal C}_m|$. Since $\sum_{m=0}^{2^{nr}-1}{|{\cal C}_m|}=|\mathbb{B}^n|=2^n$ and $U_{0,n}=2^{-nr}M$, where $U_{0,n}$ is the initial bitstream projection, we have
\begin{align*}
	f_{0,n}(u) = 2^{-n}\sum_{m=0}^{2^{nr}-1}|{\cal C}_m|\cdot\delta(u-m2^{-nr}),
\end{align*}
which is followed by
\begin{align*}
	|{\cal C}_m| = 2^n\int_{(m2^{-nr})_{-}}^{((m+1)2^{-nr})_{-}}f_{0,n}(u)\,du = 2^n\int_{(m2^{-nr})_{-}}^{(m2^{-nr})_{+}}f_{0,n}(u)\,du.
\end{align*}
For $n$ sufficiently large, we have
\begin{align*}
	2^n\int_{(m2^{-nr})_{-}}^{((m+1)2^{-nr})_{-}}f_{0,n}(u)\,du 
	&\approx 2^n\int_{m2^{-nr}}^{(m+1)2^{-nr}}f(u)\,du\\
	&\approx 2^n\cdot f(m2^{-nr})\cdot 2^{-nr},
\end{align*}
where $f(u)$ is the asymptotic initial CCS defined by \eqref{eq:asympt}, and thus
\begin{align}\label{eq:C_m}
	|{\cal C}_m|\approx f(m2^{-nr}) \cdot 2^{n(1-r)}. 
\end{align}
The probability of ${\cal C}_m$ is
\begin{align*}
	p(m) = \Pr(X^n\in{\cal C}_m) = \sum_{x^n\in{\cal C}_m}{p(x^n)} \overset{(a)}{=} |{\cal C}_m|\cdot 2^{-n},
\end{align*}
where $(a)$ comes from $p(x^n) \equiv 2^{-n}$ for uniform binary sources. From \eqref{eq:C_m}, for $n$ sufficiently large, we have
\begin{align}\label{eq:pm}
	p(m) \approx f(m2^{-nr}) \cdot 2^{-nr}.
\end{align}
Though the accurate value of $p(m)$ changes with $n$ and thus is hard to find, it can be well approximated by $f(m2^{-nr})2^{-nr}$ for large $n$.

\begin{definition}[Expected Coset Cardinality (ECC)]
	We refer to $\mathbb{E}[|{\cal C}_M|]$, the expectation of $|{\cal C}_M|$, as the ECC.
\end{definition}

It is easy to find that the ECC is actually equal to the expected number of codewords searched by a full-search decoder. 
According to the definition of ECC, we have
\begin{align*}
	\mathbb{E}[|{\cal C}_M|] = \sum_{m=0}^{2^{nr}-1}{p(m)\cdot|{\cal C}_m|}.
\end{align*}
From \eqref{eq:C_m} and \eqref{eq:pm}, for $n$ sufficiently large, we can obtain
\begin{align}\label{eq:ecc}
	\mathbb{E}[|{\cal C}_M|] 	
	&\approx 2^{n(1-r)}\sum_{m=0}^{2^{nr}-1}{f^2(m2^{-nr})\cdot2^{-nr}}\nonumber\\
	&\approx 2^{n(1-r)}\int_{0}^{1}{f^2(u)\,du}.
\end{align}
Given $f(u)\geq 0$ and $\int_{0}^{1}{f(u)\,du}=1$, we have $\int_{0}^{1}{f^2(u)\,du}\geq 1$ and thus
\begin{align*}
	\mathbb{E}[|{\cal C}_M|]\geq 2^{n(1-r)}. 
\end{align*}
This point can be explained in an intuitive way: The cosets of larger cardinalities are more likely to appear and further, more codewords must be searched for those cosets of larger cardinalities, so the complexity of full-search decoder should be scaled by $\int_0^1{f^2(u)du}\geq 1$.

\begin{definition}[Normalized Expected Coset Cardinality]
	We refer to $2^{-n(1-r)}\cdot\mathbb{E}[|{\cal C}_M|]$ as the {\em normalized} ECC.
\end{definition}

According to \eqref{eq:ecc}, it is easy to obtain
\begin{align}
	\lim_{n\rightarrow\infty}\frac{\mathbb{E}[|{\cal C}_M|]}{2^{n(1-r)}} = \int_{0}^{1}{f^2(u)\,du} \geq 1.
\end{align}
For $r=1$, we have $f(u)=\Pi(u)$, where $\Pi(u)$ is defined by \eqref{eq:Pi}, and $\int_{0}^{1}{f^2(u)\,du}=1$.
As $r\to0$, $f(u)$ will converge to $\delta(u-1/2)$, where $\delta(u)$ is the Dirac delta function, and $\int_{0}^{1}{f^2(u)\,du}\to\infty$. In a word, as $r$ decreases from $1$ to $0$, the normalized ECC will increase from $1$ to $\infty$.

\subsection{Block-wise Rate Loss}
As we know, the asymptotic initial spectrum $f(u)$ is not uniform over $[0,1)$ for $r<1$, hence the coset index $M$ is not uniformly distributed over $[0:2^{nr})$, which implies a possible rate loss. 

\begin{figure}[!t]
	\centering
	\subfigure[$H(X^n)=n$]{\includegraphics[width=.5\linewidth]{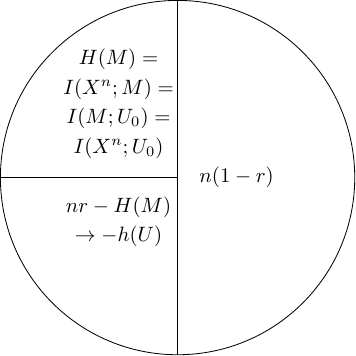}\label{subfig:block}}%
	\subfigure[$H(X_i)=1$]{\includegraphics[width=.5\linewidth]{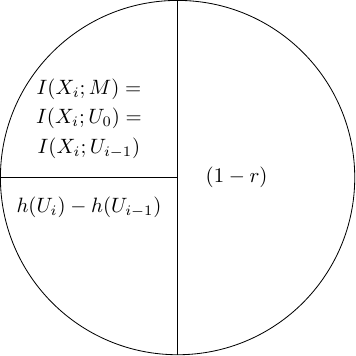}\label{subfig:symbol}}
	\caption{(a) The entropy of $X^n$ is $n$. Though $X^n$ is encoded at rate $r$, the information amount of $X^n$ provided by $M$ is smaller than $nr$ and the remaining uncertainty of $X^n$ is higher than $n(1-r)$. The block-wise rate loss will converge to $-h(U)$ as $n\to\infty$. (b) The entropy of $X_i$ is $1$. Though $X_i$ is encoded at rate $r$, the information amount of $X_i$ provided by $M$ is smaller than $r$ and the remaining uncertainty of $X_i$ is larger than $(1-r)$. The symbol-wise rate loss is $(h(U_i)-h(U_{i-1}))$.}
\end{figure}

Since $U_0=2^{-nr}M$ is a bijective function, we have
\begin{align}
	\begin{cases}
		H(X^n|M) = H(X^n|U_0)\\
		I(X^n;M) = I(X^n;U_0) = I(M;U_0) = H(M)
	\end{cases}
\end{align}
and
\begin{align}
	H(X^n) 
	&= H(M)+H(X^n|M)\nonumber\\
	&= I(X^n;M)+H(X^n|M)\nonumber\\
	&= I(X^n;U_0)+H(X^n|U_0) = n,
\end{align}
just as illustrated by \cref{subfig:block}.

\begin{theorem}[Block-wise Rate Loss]\label{thm:blkloss}
	The rate loss of coding $X^n$ with a rate-$r$ overlapped arithmetic code is
	\begin{align}\label{eq:blkloss}
		\lim_{n\to\infty}(nr-H(M)) 
		&= \lim_{n\to\infty}(H(X^n|M)-n(1-r))\nonumber\\
		&= -h(U),
	\end{align}
	where $U$ is the asymptotic initial projection defined by \eqref{eq:U}.
\end{theorem}
\begin{proof}
	The key is to compute the entropy of $M$:
\begin{align}
	H(M) 	
	&= -\sum_{m=0}^{2^{nr}-1}{p(m)\log_2{p(m)}}\nonumber\\
	&\to -\sum_{m=0}^{2^{nr}-1}{p(m)\log_2{(f(m2^{-nr})2^{-nr})}}\nonumber\\
	&= nr - \sum_{m=0}^{2^{nr}-1}{f(m2^{-nr})\log_2{f(m2^{-nr})}\cdot2^{-nr}}.
\end{align}
Since $H(X^n)=H(M)+H(X^n|M)=n$, we have
\begin{align}\label{eq:totalrateloss}
	\lim_{n\rightarrow\infty}(H(X^n|M)-n(1-r)) 
	&= \lim_{n\rightarrow\infty}{(nr-H(M))}\nonumber\\
	&= \int_{0}^{1}{f(u)\log_2{f(u)}\,du} = -h(U).
\end{align}
\end{proof}
Because $m\in[0:2^{nr})$, we have $H(M)\leq nr$ and the equality holds iff $M$ is uniformly distributed over $[0:2^{nr})$. As $n\to\infty$, the block rate loss will tend to a constant $-h(U)$, where $U$ is the asymptotic initial bitstream projection defined by \eqref{eq:U}. As we know, $h(U)\leq 0$, where the equality holds iff $f(u)$ is uniform over $[0,1)$. The rate loss per source symbol is $\frac{nr-H(M)}{n}$. For $r>0$, we have $h(U)>-\infty$ and $\frac{nr-H(M)}{n}\to0$ as $n\to\infty$, so the rate loss per source symbol will vanish as code length increases. As for $r=0$, we have $f(u)=\delta(u-1/2)$ and $h(U)=-\infty$, so the rate loss per source symbol is unknown.

\subsection{Symbol-wise Rate Loss}
Now we consider the symbol-wise rate loss of overlapped arithmetic codes. As shown by \cref{thm:markov}, $U_0\to\dots\to U_{i-1} \to X_i$, so
\begin{align}
	\begin{cases}
		H(X_i|M) = H(X_i|U_0) = H(X_i|U_{i-1})\\
		I(X_i;M) = I(X_i;U_0) = I(X_i;U_{i-1})
	\end{cases}
\end{align}
and 
\begin{align}
	H(X_i) 
	&= I(X_i;M)+H(X_i|M)\nonumber\\
	&= I(X_i;U_0)+H(X_i|U_0)\nonumber\\
	&= I(X_i;U_{i-1})+H(X_i|U_{i-1})=1,
\end{align}
as illustrated by \cref{subfig:symbol}. By the chain rule, we can get
\begin{align}
	H(X^n|U_0) 
	&= \sum_{i=1}^{n}{H(X_i|U_0,X^{i-1})}\nonumber\\
	&\overset{(a)}{=} \sum_{i=1}^{n}{H(X_i|U_0)} = \sum_{i=1}^{n}{H(X_i|M)},
\end{align}
where $(a)$ is because $X^n$ is an i.i.d. random process. Therefore,
\begin{align}
	H(X^n|M) = H(X^n|U_0) 
	&= \sum_{i=1}^{n}{H(X_i|M)}\nonumber\\
	&= \sum_{i=1}^{n}{H(X_i|U_{i-1})}.
\end{align}

\begin{theorem}[Symbol-wise Rate Loss]\label{thm:symloss}
	For a rate-$r$ overlapped arithmetic code, as $n\to\infty$, the rate loss of coding $X_i$ is
	\begin{align}\label{eq:HXiUi}
		r-I(X_i;M) 
		&= H(X_i|M)-(1-r)\nonumber\\
		&= H(X_i|U_{i-1})-(1-r)\nonumber\\
		&= h(U_i)-h(U_{i-1}).
	\end{align}
\end{theorem}
\begin{proof}
	According to \eqref{eq:condccs} and the properties of differential entropy, $h(U_{i-1}|X_i)=h(U_i)-r$ as $n\to\infty$. Hence,
	\begin{align}\label{eq:HXiUi}
		H(X_i|U_{i-1}) 
		&= H(X_i) - I(X_i;U_{i-1})\nonumber\\
		&= 1 - (h(U_{i-1})-h(U_{i-1}|X_i)) \nonumber\\
		&= 1 - h(U_{i-1}) + h(U_i) - r \nonumber\\
		&= (1-r) + h(U_i) - h(U_{i-1}).
	\end{align}
	Since $X_i$ is coded at rate $r$, the rate loss is $(h(U_i)-h(U_{i-1}))$.
\end{proof}

\begin{corollary}
	The sequence $(h(U_0),\dots,h(U_n))$ is a monotonously increasing sequence, \textit{i.e.}, $h(U_0) \leq \cdots \leq h(U_n)\to 0$, where $h(U_n)\to 0$ is because $f_n(u)$ tends to be uniform over $[0,1)$ as $n\to\infty$.
\end{corollary}

\begin{remark}[A Neat Proof of \cref{thm:blkloss}]
	According to \cref{thm:symloss}, we can obtain a neat proof of \cref{thm:blkloss}. The block rate loss is equal to the sum of symbol losses, \textit{i.e.},
	\begin{align}
		\sum_{i=1}^{n}(h(U_i)-h(U_{i-1})) = h(U_n)-h(U_0) \overset{(a)}{\to} -h(U),
	\end{align}
	where $(a)$ is because $h(U_n)\to 0$ and $U_0\to U$ as $n\to\infty$.
\end{remark}

As $(n-i)\rightarrow\infty$, both $f_{i-1,n}(u)$ and $f_{i,n}(u)$ will converge to $f(u)$ and thus $(h(U_{i-1,n})-h(U_{i,n}))\to0$. Hence, there is hardly any symbol-wise rate loss for those symbols far away from the end of a block, and the block-wise rate loss comes mainly from ending symbols of a block.

\begin{remark}[An Intuitive Explanation of Rate Loss]
The rate loss of overlapped arithmetic codes can be explained intuitively with an example. We first compress $X^n$ with a rate-$r$ overlapped arithmetic encoder to obtain bitstream $M$ and then compress $X^n$ with a rate-$1$ standard arithmetic encoder parameterized with the probability sequence $(p_{1|0}(x|u),\dots,p_{n|n-1}(x|u))$ to obtain another bitstream $\hat{M}$. Obviously, the length of $M$ is always $nr$ and the expected length of $\hat{M}$ is about $n(1-r)-h(U_0)$. Thus, the total length of $M$ and $\hat{M}$ is $n-h(U_0)\geq n$. It is easy to see that the error-free recovery of $X^n$ is achievable by the interaction between the overlapped arithmetic decoder of $M$ and the standard arithmetic decoder of $\hat{M}$. First, according to $U_0$, the overlapped arithmetic decoder can obtain $p_{1|0}(x|u)$, which is then used by the standard arithmetic decoder to recover $X_1$ exactly. Next according to $U_0$ and $X_1$, the overlapped arithmetic decoder can obtain $U_1$ and $p_{2|1}(x|u)$, which is then used by the standard arithmetic decoder to recover $X_2$ exactly. Such operations are repeated until $X_n$ is recovered. This example shows that $X^n$ can be exactly represented by two bitstreams $M$ and $\hat{M}$, whose total length is slightly larger than $H(X^n)=n$, and the rate loss tends to $-h(U)$ as $n\to\infty$.
\end{remark}

\section{Decoding Algorithm}
\begin{figure}[!t]
	\includegraphics[width=\linewidth]{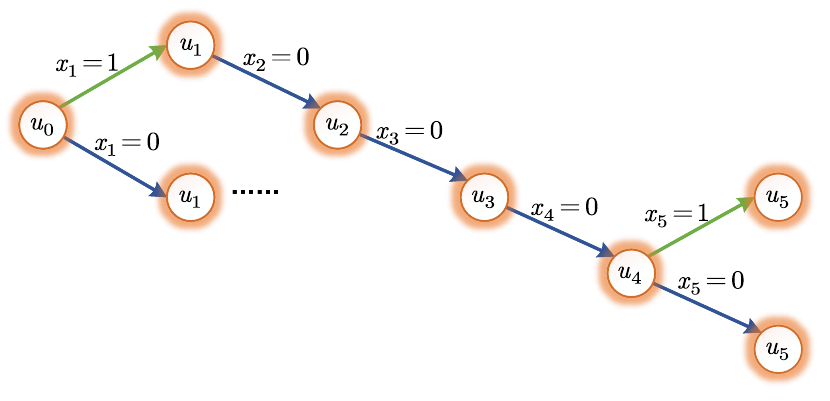}
	\caption{An incomplete binary tree generated by the decoder.}
	\label{fig:dec}
\end{figure}

Up to now, the reader may wonder why we need to study CCS or what we can gain from CCS? One reason is that the decoder can make use of CCS to achieve a better performance. 

Let $u_0$ be a realization of $U_0$. From $u_0$, a depth-$n$ incomplete binary tree ${\cal T}(u_0,n)$ will be created by the decoder \citep{FangTCOM13}, as shown by \cref{fig:dec}. In ${\cal T}(u_0,n)$, each length-$i$ path corresponds to a length-$i$ binary block. If $i=n$, it is called a {\em full} path; otherwise, {\em i.e.}, if $i<n$, it is called a {\em partial} path. We refer to the path corresponding to binary block $x^i\in\mathbb{B}^i$ as the $x^i$-path, and its end node is called the $x^i$-node. Let us define
\begin{equation}\label{eq:uxi}
	u(x^i) \triangleq \frac{u_0-l(x^i)}{h(x^i)-l(x^i)} = 2^{ir}(u_0-l(x^i)).
\end{equation}
Initially, $u(x^0)=u_0$. By comparing \eqref{eq:uxi} with \eqref{eq:Uin}, it can be found that $u(x^i)$ is actually a realization of $U_i$. From \eqref{eq:Ui}, we can obtain
\begin{align}\label{eq:uxi_recursive}
	u(x^i) = 2^r\left(u(x^{i-1}) - x_i(1-2^{-r})\right),
\end{align}
showing that $u(x^i)$ is a function w.r.t. $x_i$ and $u(x^{i-1})$. With the help of side information $y^n\in\mathbb{B}^n$, the decoder searches through the incomplete binary tree ${\cal T}(u_0,n)$ to find the best binary block $\hat{x}^n$ such that
\begin{align}
	\hat{x}^n &= \arg\max_{x^n\in{\cal T}(u_0,n)}p(y^n|x^n)\label{eq:ml}\\
			  &\overset{(a)}{=} \arg\max_{x^n\in{\cal T}(u_0,n)}p(x^n|y^n),\label{eq:map}
\end{align}
where $(a)$ holds because the maximum {\em a posteriori} (MAP) decoder \eqref{eq:map} is equivalent to the maximum likelihood (ML) decoder \eqref{eq:ml} for uniformly distributed $X$.

There are mainly two tree searching strategies: {\em breath}-first search, {\em i.e.}, the $M$-algorithm, and {\em metric}-first search, {\em i.e.}, the stack algorithm. This monograph omits the details of these two algorithms, as they can be found in many textbooks. Instead, we put an emphasis on how to calculate path metrics. We model the correlation between $X$ and $Y$ as a binary symmetric channel with crossover probability $\epsilon=\Pr(X\neq Y)$. Let $\oplus$ denote {\em eXclusive OR} (XOR). By intuition, the metric of the $x^i$-path should be calculated by
\begin{align}\label{eq:eta}
	\varsigma(x^i) 
	&\triangleq p(x^i|y^n) \overset{(a)}{=} p(x^i|y^i) \overset{(b)}{=} \prod_{i'=1}^{i}{p(x_{i'}|y_{i'})}\nonumber\\ 
	&= \epsilon^{|z^i|}\cdot(1-\epsilon)^{i-|z^i|},
\end{align}
where $(a)$ and $(b)$ are because both $X^n$ and $Y^n$ are i.i.d. random processes, and $|z^i|$ denotes the number of nonzero terms in $z^i=x^i\oplus y^i$. Equivalently in the log domain,
\begin{align}\label{eq:logeta}
	\varsigma(x^i) = |z^i|\cdot\log{\epsilon} + (i-|z^i|)\cdot\log(1-\epsilon).
\end{align}

However, \eqref{eq:eta} totally ignores the properties of $M$. As we know, the nonuniform distribution of $M$ will cause rate loss. Hence, a gain can be achieved if $M$ is also taken into consideration during decoding. Now we define the {\em overall} metric of the $x^i$-path as
\begin{align}
	\sigma(x^i) 
	&\triangleq p(x^i|y^n,m) = p(x^i|y^i,m)\nonumber\\
	&\overset{(a)}{\propto} p(x^i|y^i)\cdot p(x^i|m) = \varsigma(x^i)\cdot p(x^i|m),
\end{align}
where $(a)$ is due to the Markov chain $Y^n\to X^n\to M$. We call $\varsigma(x^i)$ the {\em extrinsic} metric and $p(x^i|m)$ the {\em intrinsic} metric. 

\begin{lemma}[Intrinsic Metric]
	The intrinsic metric of the $x^i$-path is
	\begin{align}\label{eq:pxm}
		p(x^i|m) = 2^{i(r-1)}\frac{f_i(u(x^i))}{f_0(u_0)} \propto f_i(u(x^i)).
	\end{align}
\end{lemma}
\begin{proof}
	Due to the i.i.d. nature of $X^n$ and $U_{0} \to \dots \to U_{i-1} \to X_i$,
\begin{align*}
	p(x^i|m) = p(x^i|u_0) = \prod_{i'=1}^{i}p(x_{i'}|u_0) = \prod_{i'=1}^{i}p(x_{i'}|u(x^{i'-1})).
\end{align*}
According to \eqref{eq:pxu} and \eqref{eq:uxi_recursive}, we have
\begin{align*}
	p(x_i|u(x^{i-1})) = 2^{r-1} \frac{f_i(u(x^i))}{f_{i-1}(u(x^{i-1}))}.
\end{align*}
Therefore,
\begin{align*}
	p(x^i|m) 
	&= 2^{i(r-1)}\frac{f_1(u(x_1))}{f_0(u_0))}\times\cdots\times\frac{f_i(u(x^i))}{f_{i-1}(u(x^{i-1}))}\\
	&= 2^{i(r-1)}\frac{f_i(u(x^i))}{f_0(u_0)} \overset{(a)}{\propto}f_i(u(x^i)),
\end{align*}
where $(a)$ is because $2^{i(r-1)}/f_0(u_0)$ is a constant for every path. 
\end{proof}
\begin{theorem}[Overall Metric]\label{thm:metric}
	The overall metric of the $x^i$-path is
	\begin{align}\label{eq:musi}
		\sigma(x^i) \propto \epsilon^{|z^i|}\cdot(1-\epsilon)^{i-|z^i|} \cdot f_i(u(x^i)),
	\end{align}
	where $z^i=x^i\oplus y^i$, and equivalently in the log domain,
	\begin{align}\label{eq:logmusi}
		\sigma(x^i) = |z^i|\cdot\log{\epsilon} + (i-|z^i|)\cdot\log(1-\epsilon) +\log{f_i(u(x^i))}.
	\end{align}
\end{theorem}

\begin{remark}[An Intuitive Explanation on Intrinsic Path Metric]\label{subsec:extrinsic}
Let us rewrite \eqref{eq:pxm} as
\begin{align}\label{eq:pxm_variant}
	p(x^i|m) = {2^{i(r-1)}}\frac{f_i(u(x^i))}{f_0(u_0)} = \frac{f_i(u(x^i))\cdot{2^{(n-i)(1-r)}}}{f_0(u_0)\cdot{2^{n(1-r)}}}.
\end{align}
The denominator of \eqref{eq:pxm_variant} is in fact the total number of leaf nodes of the tree ${\cal T}(u_0,n)$, and the numerator of \eqref{eq:pxm_variant} is in fact the number of leaf nodes of the subtree grown from the $x^i$-node given $U_i=u(x^i)$. For uniform binary sources, all full paths are actually of the same probability $2^{-n}$, so the intrinsic metric of the $x^i$-path must be equal to the ratio of the number of leaf nodes of the subtree grown from the $x^i$-node to the total number of leaf nodes of the tree grown from the root node, which perfectly coincides with our intuition.

Continue the above analysis. It has been proved that as $n\to\infty$, the final CCS $f_n(u)$ will become uniform over $[0,1)$ for uniform sources \cite{FangTCOM16b}. Thus, \eqref{eq:pxm_variant} will become 
\begin{align}
	p(x^n|m) = \frac{f_n(u(x^n))}{f_0(u_0)\cdot2^{n(1-r)}} = \frac{1}{f_0(u_0)\cdot2^{n(1-r)}},
\end{align}
{\em i.e.}, all full paths in the tree ${\cal T}(u_0,n)$ have the same intrinsic metric. This is because for uniform binary sources, all codewords $x^n\in\mathbb{B}^n$ are of the same probability and the tree ${\cal T}(u_0,n)$ contains $f_0(u_0)2^{n(1-r)}$ leaf nodes. It again perfectly coincides with our intuition.
\end{remark}

\begin{remark}[Breath-First Search]
In the original papers on overlapped arithmetic codes \cite{GrangettoCL07,GrangettoTSP09}, the decoder was realized with the breadth-first search, {\em i.e.}, the $M$-algorithm. Let $n$ be source block length. The decoding process includes $n$ stages and the decoder maintains a list of paths. At each stage, the list of paths is traversed and one source symbol is decoded for each path. After each stage, at most $M$ optimal paths with larger metrics are retained by pruning others, and then at the next stage, at most $2M$ branches are created. Hence the decoder should allocate a memory block of $2M$ nodes for each stage to avoid memory overflow, and the order of memory size of the $M$-algorithm is ${\cal O}(2nM)$. Such process is repeated and finally after the $n$-th stage, the best length-$n$ path is selected. It can be found that the $M$-algorithm decoder of overlapped arithmetic codes is quite similar to the famous {\em Successive-Cancellation List} (SCL) decoder of polar codes \cite{7055304}. 
\end{remark}

\begin{figure}
	\centering
	\subfigure[]{\includegraphics[width=.5\linewidth]{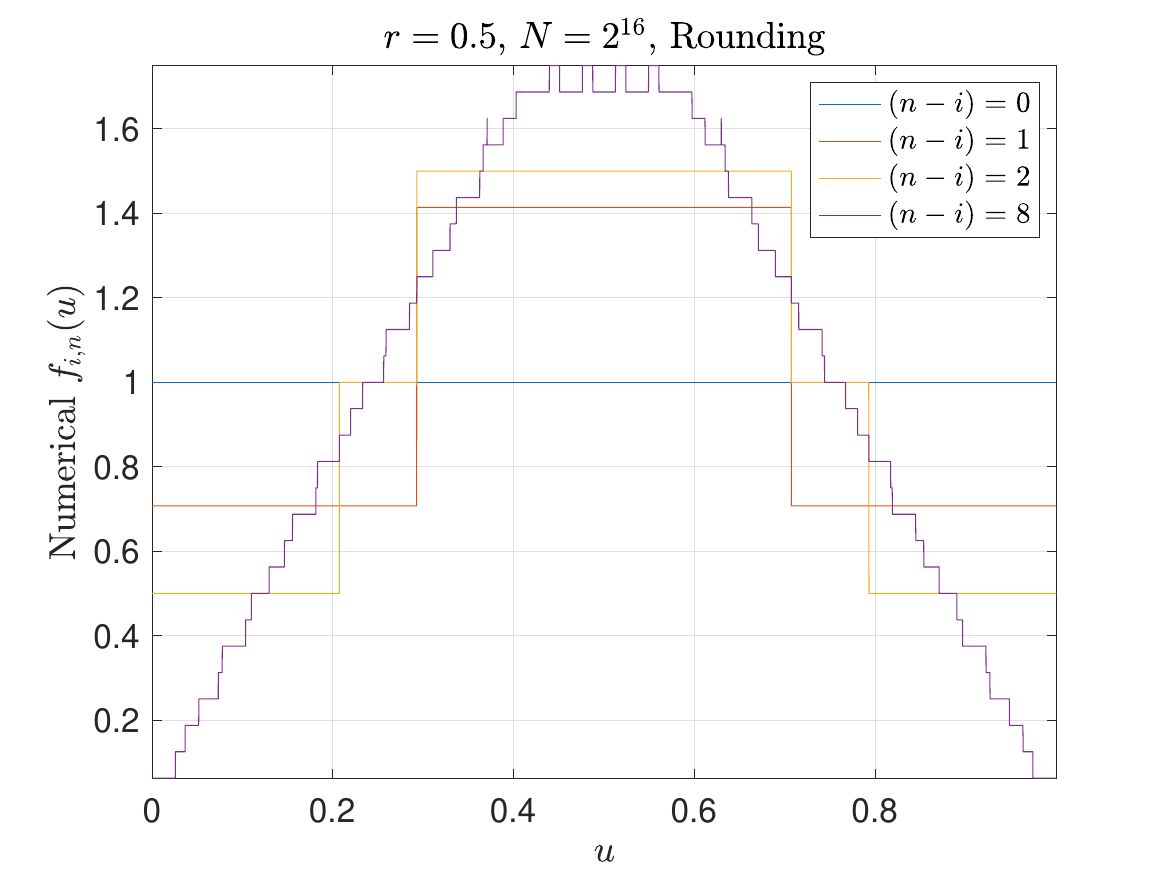}\label{subfig:roundsmall}}%
	\subfigure[]{\includegraphics[width=.5\linewidth]{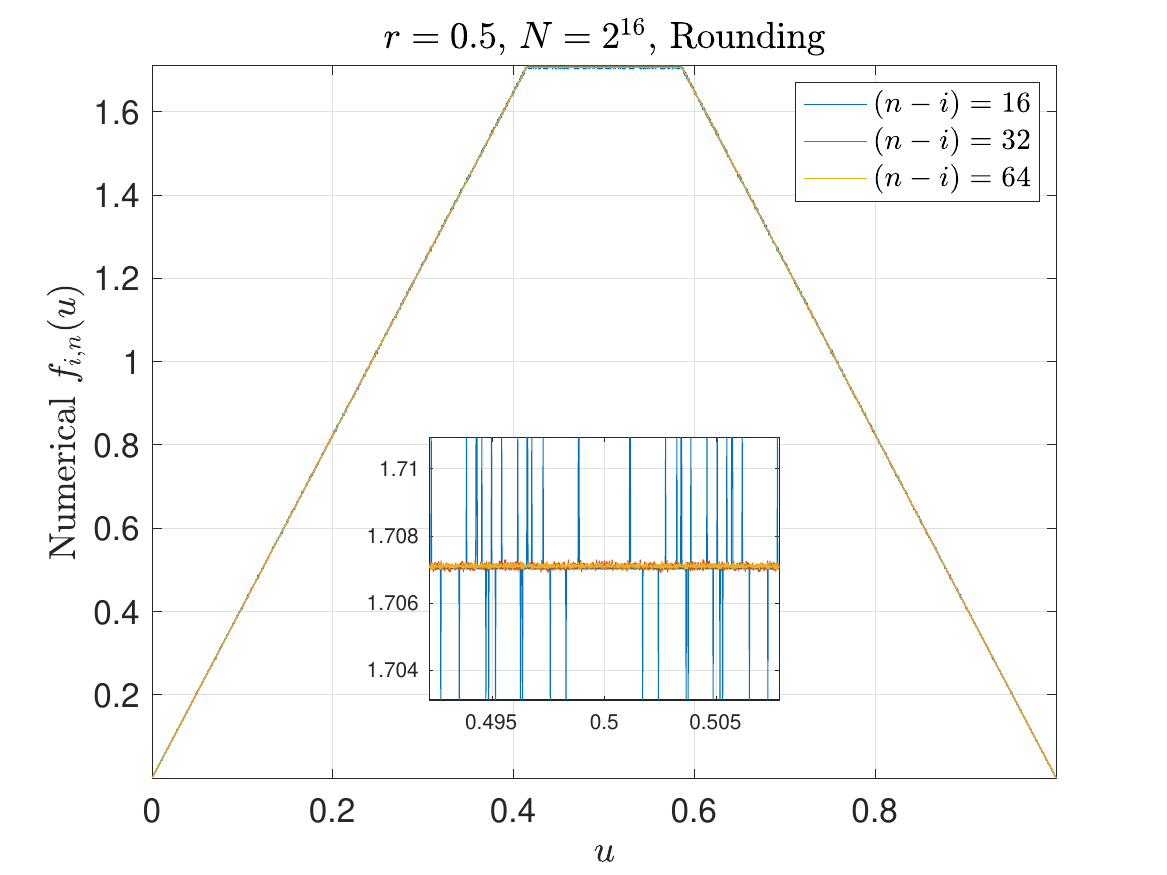}\label{subfig:roundlarge}}\\
	\subfigure[]{\includegraphics[width=.5\linewidth]{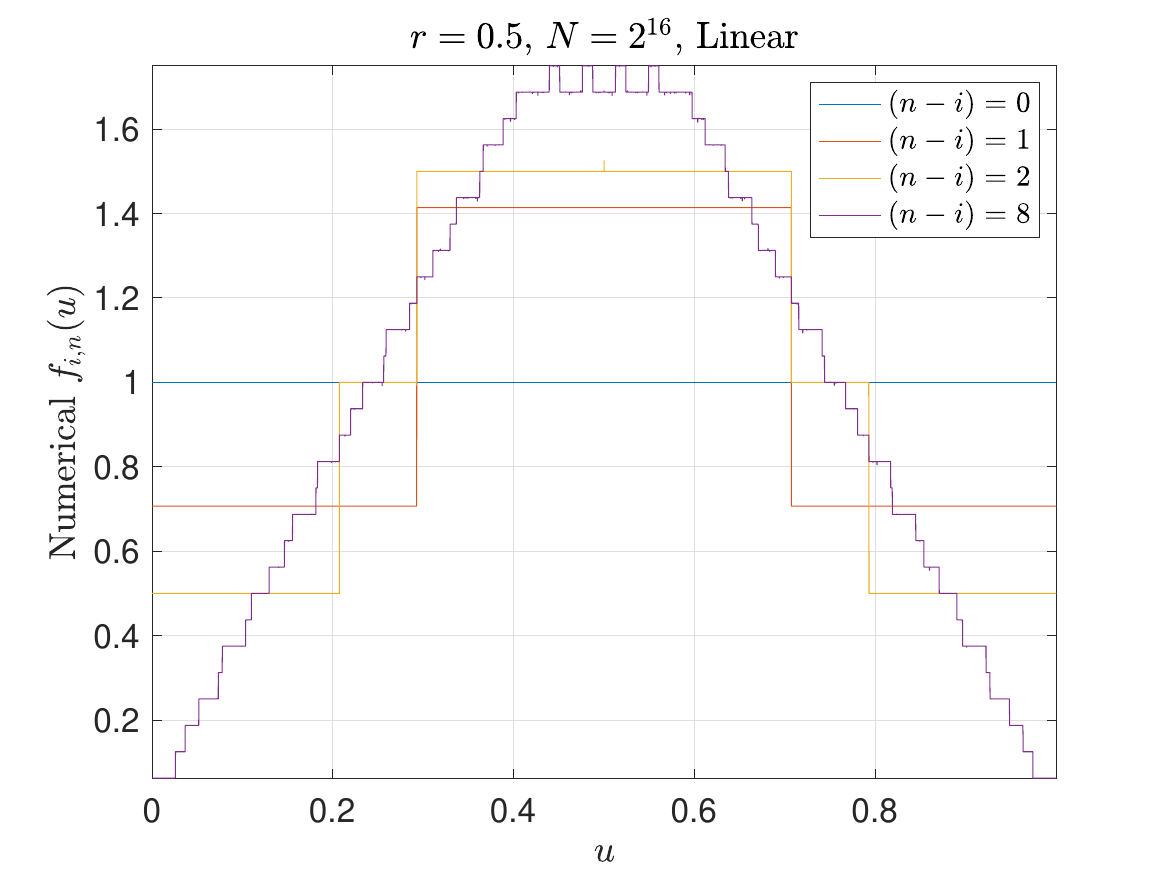}\label{subfig:linearsmall}}%
	\subfigure[]{\includegraphics[width=.5\linewidth]{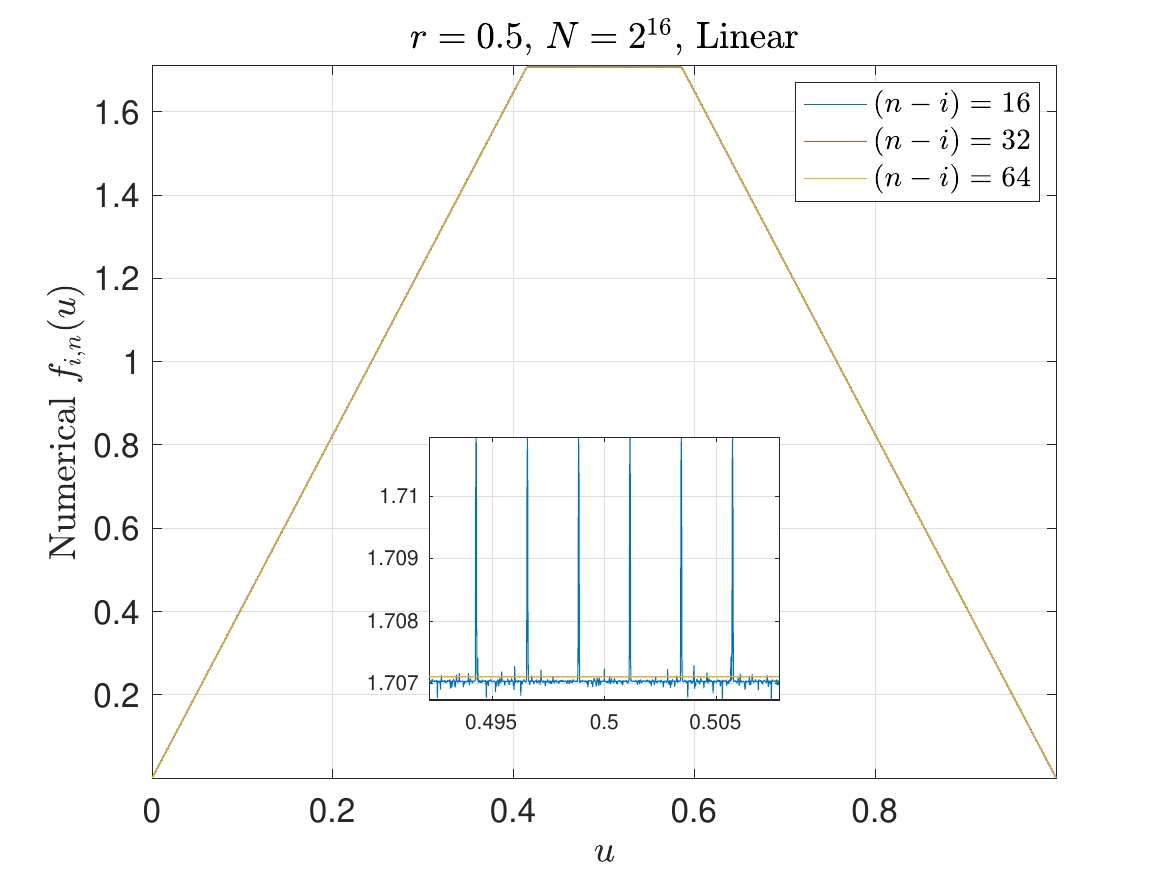}\label{subfig:linearlarge}}\\
	\subfigure[]{\includegraphics[width=.5\linewidth]{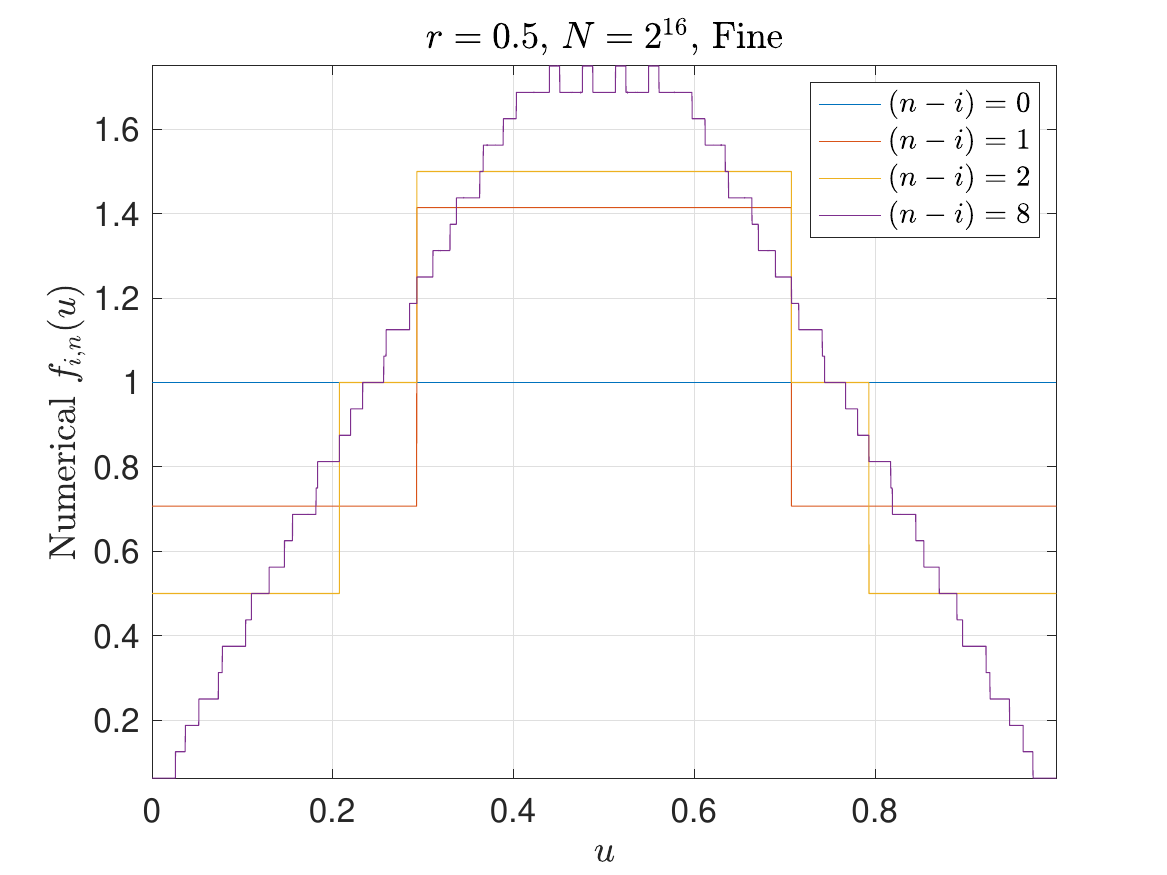}\label{subfig:finesmall}}%
	\subfigure[]{\includegraphics[width=.5\linewidth]{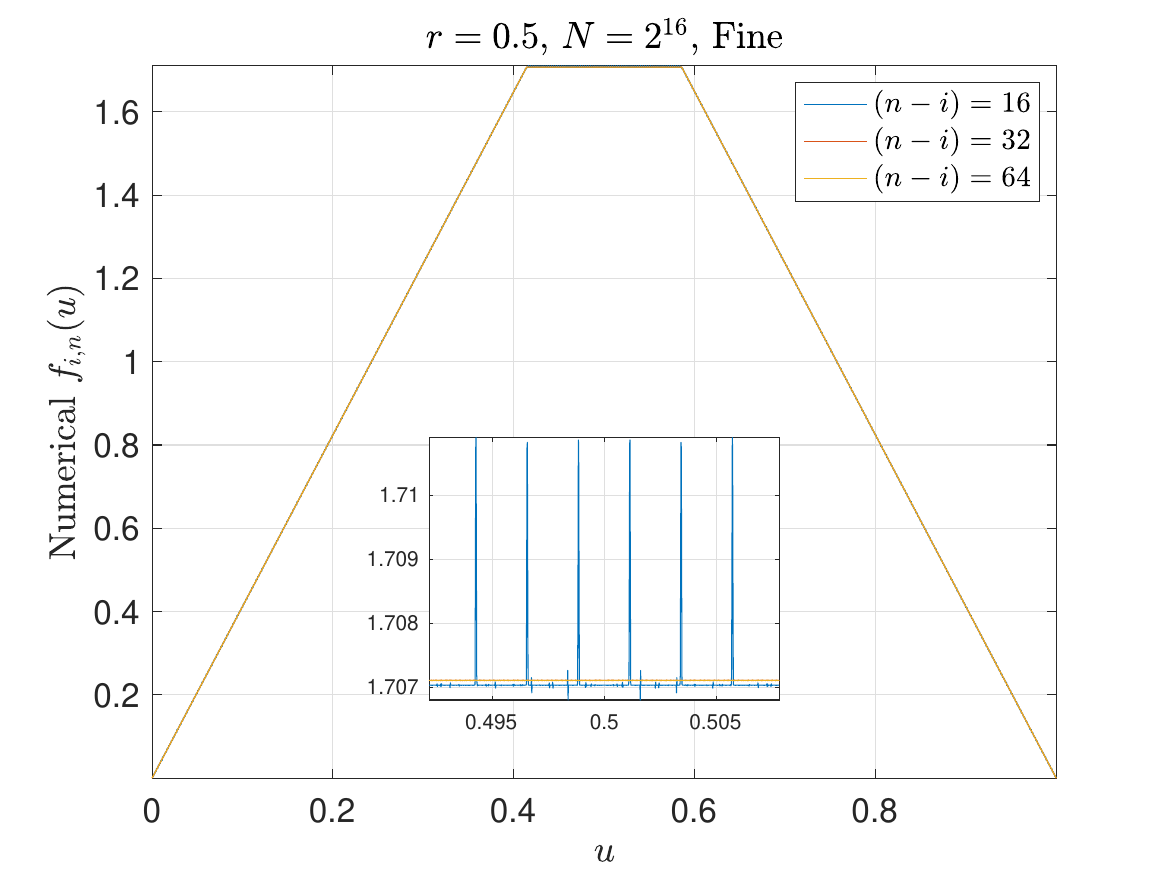}\label{subfig:finelarge}}
	\caption{The interval $[0,1)$ is divided into $N=2^{16}$ segments. (a) and (b) are some examples of the rounding numerical algorithm. It has two spikes when $(n-i)=8$ and dithers along the real CCS when $(n-i)=16$. (c) and (d) are some examples of the linear numerical algorithm. It has one spike when $(n-i)=2$ and dithers when $(n-i)=8$. (e) and (f) are some examples of the fair numerical algorithm. It overcomes all weaknesses of the rounding and linear numerical algorithms.}
	\label{fig:numerical}
\end{figure}

\begin{remark}[Backward-Replacing Algorithm]	
As we know, the decoding tree of overlapped arithmetic codes is incomplete, so fewer than $2M$ branches are usually created at each stage, which means that partial memory will be wasted unavoidably. In \cite{FangTCOM13}, the concept of {\em Expansion Factor} is defined as $\mu_i \triangleq 1+\int_{1-2^{-r}}^{2^{-r}}{f_i(u)}\,du \in [1,2)$. If $M$ paths are retained after the $i$-th stage, then the number of branches created at the $(i+1)$-th stage is $M\mu_i$. For example, if $r=0.5$, we have $\mu_i=2-\frac{\sqrt{2}}{4}\approx 1.65$ and $M\mu_i\approx 1.65M$, meaning that $(2-1.65)/2>1/6$ of the memory is wasted. If $r=1$, we have $\mu_i=1$, and as $r\to 0$, we have $\mu_i\to 2$. Hence, the problem of memory waste will become especially serious at high rates. To solve this problem, we propose the {\em Backward-Replacing Algorithm}: The decoder allocates a block of memory for $nM$ nodes in total and maintains at most $M$ paths at each stage. 

The following example shows how the backward-replacing algorithm works. Let $M=6$. After the $i$-th stage, we have only $5$ paths, sorted in the descending order of metric and numbered from $1$ to $5$. These $5$ paths are stored as a list $\{1,2,3,4,5,\_\}$, where $\_$ denotes an unoccupied cell. The decoder begins with the $1$-st path, which has two children. After the $1$-st path is coped with, the list becomes $\{1a,2,3,4,5,1b\}$, where $1a$ and $1b$ denote the two children of the $1$-st path. Then the $2$-nd path is dealt with, which has only one child. Thus the list becomes $\{1a,2*,3,4,5,1b\}$, where $2*$ denotes the sole child of the $2$-nd path. Then the $3$-rd path is tackled, which has two children and the list becomes $\{1a,2*,3a,4,3b,1b\}$, {\em i.e.}, the $5$-th path is replaced directly without being handled. In turn, the $4$-th path is treated, which has two children. However, there is no more space, so only the better child is retained, while the other is discarded. If the $4b$ child is better than the $4a$ child, the list finally becomes $\{1a,2*,3a,4b,3b,1b\}$. This list is sorted and the above actions are repeated for the next stage.
\end{remark}

\section{Experimental Verification}
\subsection{Numerical Algorithm}\label{subsec:example}
We use the classical CCS when $r=0.5$ to compare three numerical algorithms proposed in \cref{sec:numerical}. When $r=0.5$, for $n$ sufficiently large, as $i$ decreases from $n$ to $0$, $f_{i,n}$ will tend to $f(u)$ as given by \eqref{eq:closedForm_halfRate}. For simplicity, we set the final CCS $f_n(u)=1$ for $u\in[0,1)$. Some results are included in \cref{fig:numerical}, where $N=2^{16}$. It can be observed that for small $(n-i)$, the rounding and linear numerical algorithms have some spikes, which are however eliminated by the fine numerical algorithm. For large $(n-i)$, the rounding and linear numerical algorithms are sometimes in a dither, which is removed or alleviated by the fine numerical algorithm. In general, there results confirm that the {\em fine} numerical algorithm is superior to the {\em linear} numerical algorithm, while the {\em linear} numerical algorithm is superior to the {\em rounding} numerical algorithm.

\subsection{Decoding Algorithm}\label{subsec:decoding}
\begin{figure}[!t]
	\centering
	\subfigure[]{\includegraphics[width=.5\linewidth]{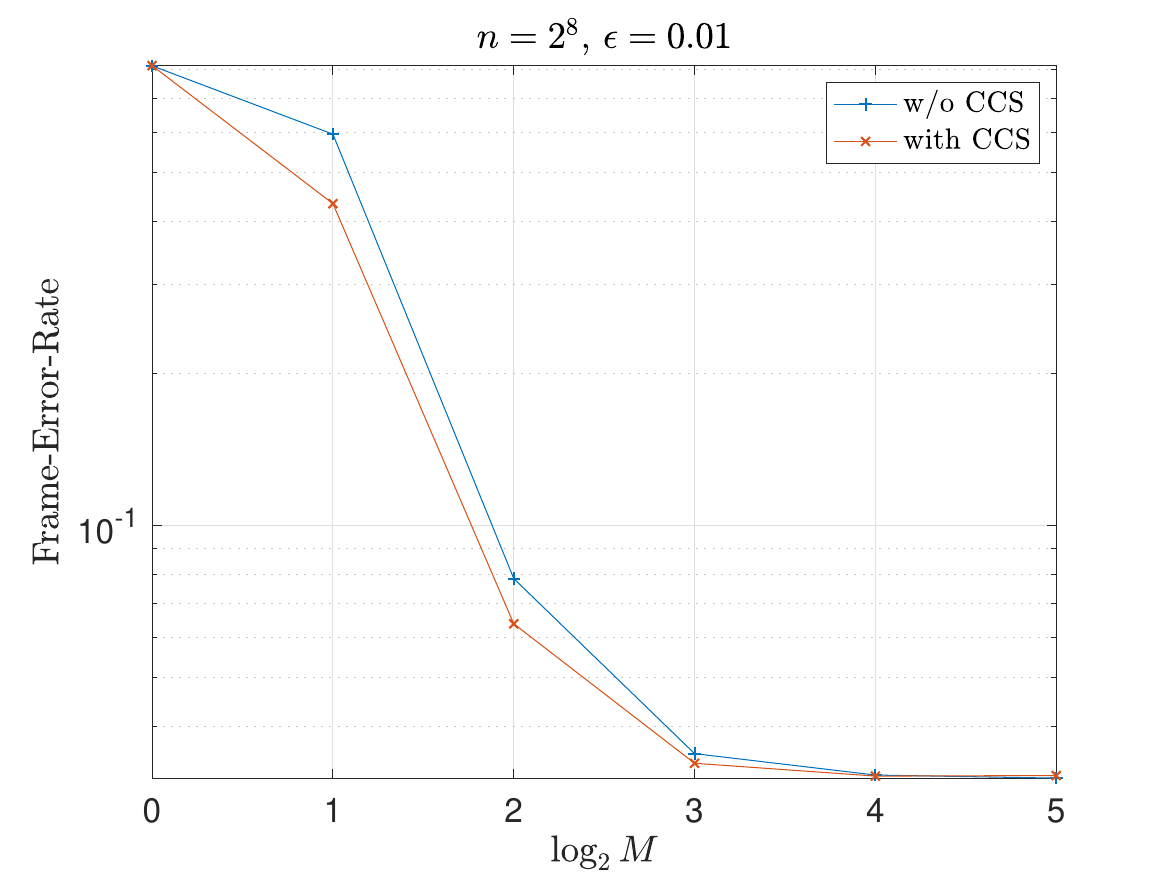}\label{subfig:fer.01}}%
	\subfigure[]{\includegraphics[width=.5\linewidth]{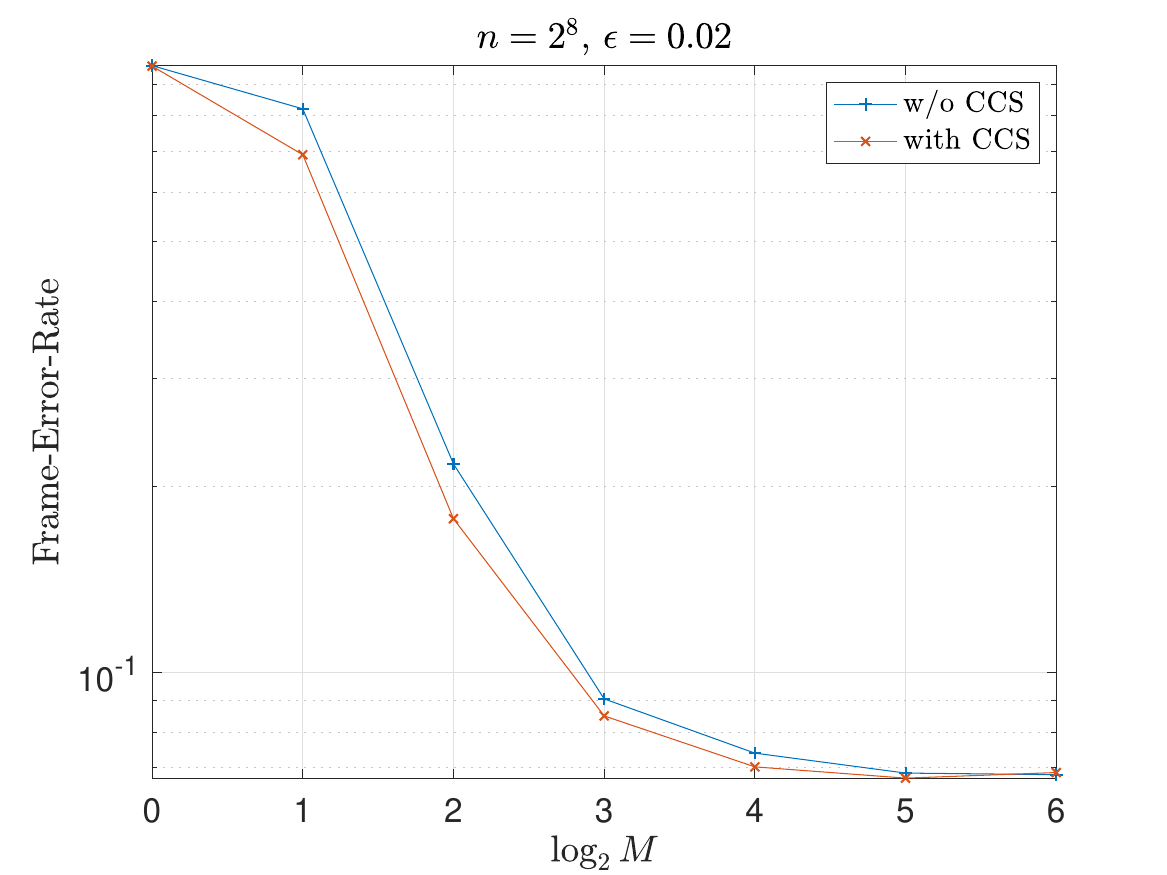}\label{subfig:fer.02}}\\
	\subfigure[]{\includegraphics[width=.5\linewidth]{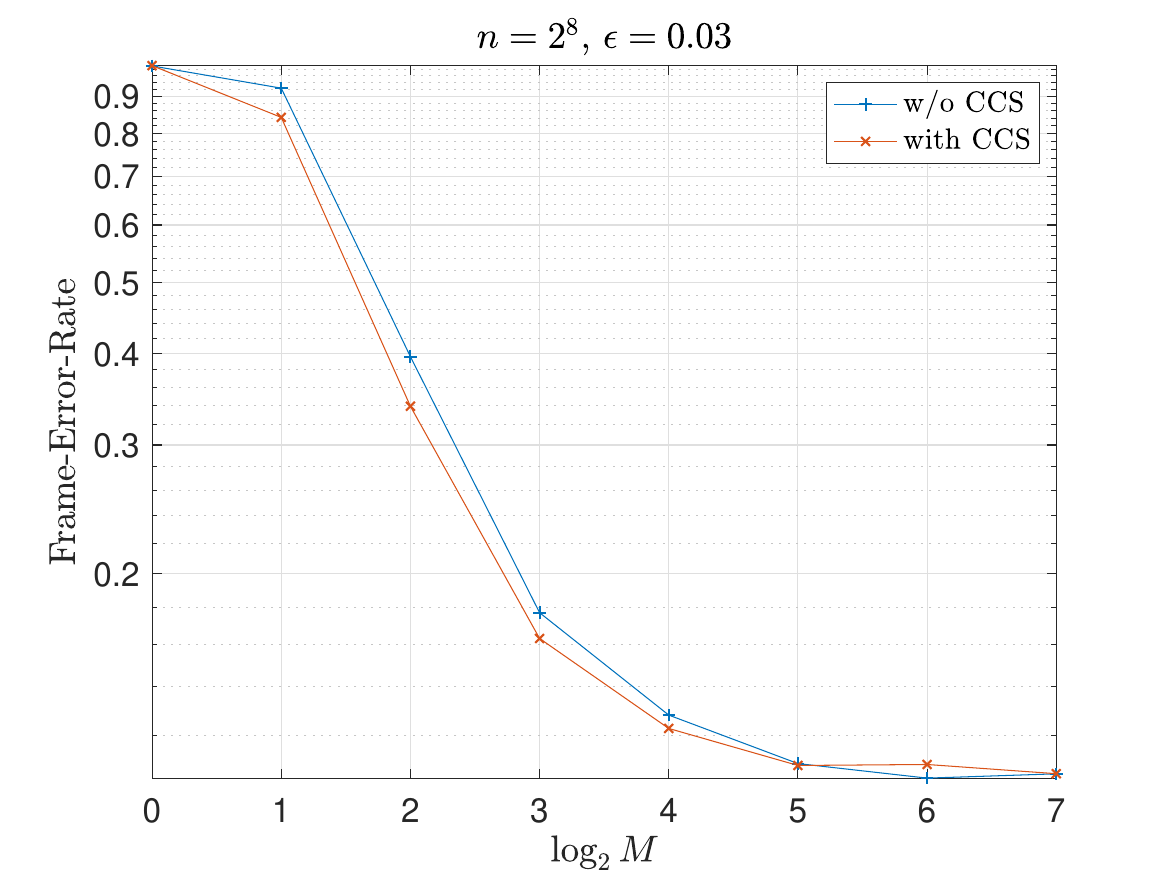}\label{subfig:fer.03}}%
	\subfigure[]{\includegraphics[width=.5\linewidth]{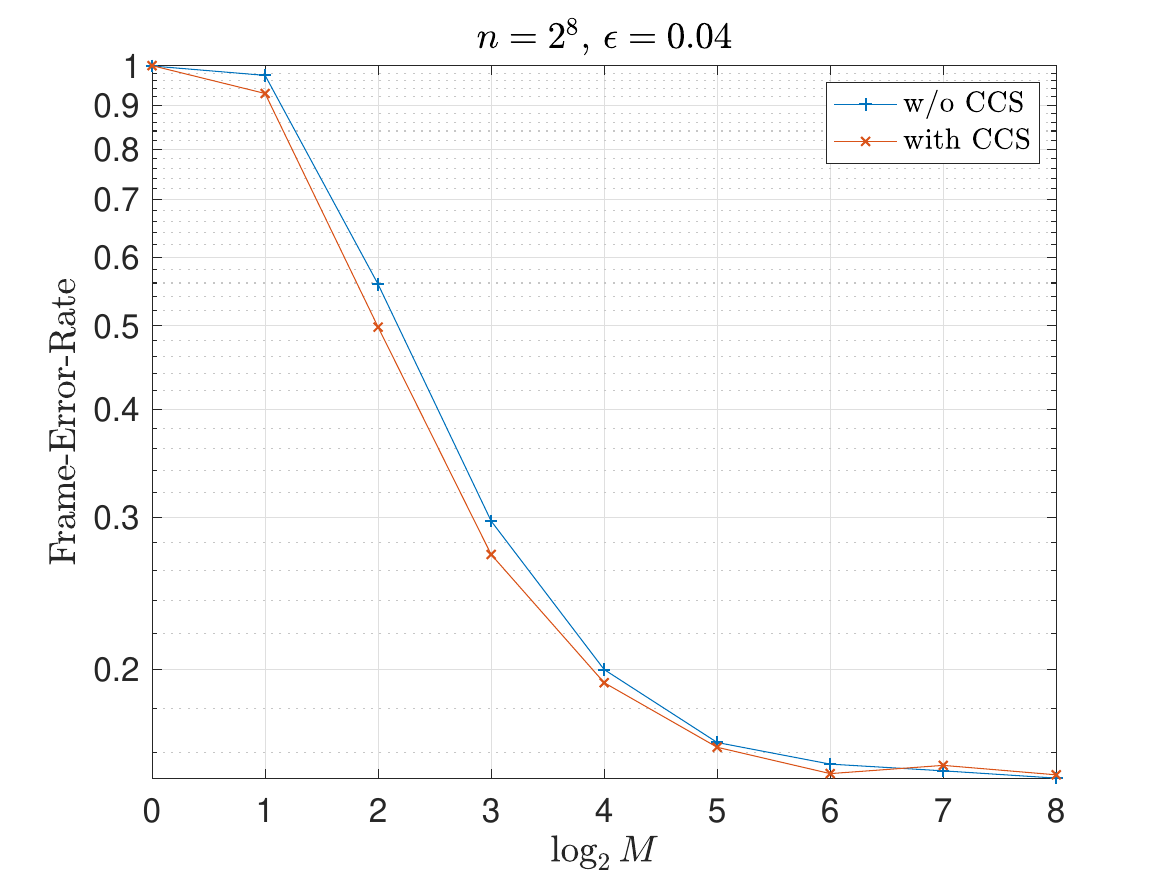}\label{subfig:fer.04}}\\
	\caption{Decoder without (w/o) CCS versus decoder with CCS, where $M$ is the maximum number of survival paths after pruning.}
	\label{fig:ccsdec}
\end{figure}
\cref{fig:ccsdec} gives some examples to verify the superiority of \eqref{eq:musi} over \eqref{eq:eta}, where $M$ is the maximum number of survival paths after pruning. As predicted by theoretical analyses, the decoder with CCS is beneficial to the proper path so that the proper path is more likely to be retained after pruning and the mis-pruning risk of the proper path is reduced. \cref{fig:ccsdec} shows that a lower frame error rate is achieved by making use of CCS for $M=2^1$ to $2^4$. However, \cref{fig:ccsdec} also shows that the advantage of the decoder with CCS will gradually vanish for large $M$ (greater than $2^5$). There are two reasons for this phenomenon: (a) For uniform binary sources, all codewords are of the same probability $2^{-n}$; (b) The $M$-algorithm decoder will be more like the full-search decoder as $M$ increases. Hence, the decoder with CCS is more useful in the lightweight (small $M$) case.

\chapter{Coexisting Interval and Error Rate Analysis}\label{c-coexist}
\vspace{-25ex}%
With {\em Coset Cardinality Spectrum} (CCS), we have known how overlapped arithmetic codes partition source space into unequal-sized cosets. This section will introduce the concept of {\em Coexisting Interval}, which is a powerful analysis tool of overlapped arithmetic codes and has found applications to the deductions of \textit{Hamming Distance Spectrum} (HDS) \cite{FangTCOM16a,FangTCOM16b} and block error rate \cite{FangCOMML21}. This section will also show how to deduce the block error rate of overlapped arithmetic codes with coexisting interval, while the next section will make use of coexisting interval to deduce the HDS of overlapped arithmetic codes. Note that though coexisting interval was formally defined in \cite{FangTCOM16b}, it is actually originated from {\em Risky Interval} defined in \cite{FangTCOM16a}.

\section{Definition of Coexisting Interval}\label{subsec:coexist}
For conciseness, we will not distinguish random variables equal in distribution below. If $X^n$ is an i.i.d. random process, then 
\begin{align}
	\sum_{i=1}^{n}{X_i\cdot2^{(n-(i-1))r}} \stackrel{d}{=} \sum_{i=1}^{n}{X_i2^{ir}},\nonumber
\end{align}
where $\stackrel{d}{=}$ means that two random variables are equal in distribution. Therefore, according to \eqref{eq:ell},
\begin{align}\label{eq:ellequal}
	s(x^n) \overset{d}{=} (1-2^{-r})\sum_{i=1}^{n}{X_i2^{ir}} = (2^r-1)\sum_{i=0}^{n-1}{X_{i+1}2^{ir}}.
\end{align}
For any $x^n\in\mathbb{B}^n$, it is easy to know
\begin{align}\label{eq:bound}
	0 = s(0^n) \leq s(x^n) \leq s(1^n) = (2^r-1)\frac{2^{nr}-1}{2^r-1} = 2^{nr} - 1.\nonumber
\end{align}
The following lemmas hold obviously.
\begin{lemma}[Relation Between $s$-Function and Coset]\label{lem:coset}
	The necessary and sufficient condition for the event that $x^n\in\mathbb{B}^n$ belongs to the $m$-th coset ${\cal C}_m$ is $s(x^n)\in(m-1,m]$, where $m\in[0:2^{nr})$. Let $\{\cdot\}\leftrightarrow\{\cdot\}$ denote the equivalence between two events. Then
	\begin{align}
		\{x^n\in{\cal C}_m\} \leftrightarrow \{s(x^n)\in(m-1,m]\}. 
	\end{align}	
\end{lemma}
\begin{lemma}[Coexistence of Codewords]\label{lem:coexist}
	The necessary and sufficient condition for the event that two codewords $x^n\in\mathbb{B}^n$ and $y^n\in\mathbb{B}^n$ coexist in the same coset ${\cal C}_m$ is $\lceil{s(x^n)}\rceil=\lceil{s(y^n)}\rceil=m$. 
\end{lemma}
For $0\leq d\leq n$, we define $b^d\triangleq(b_1,\dots,b_d)\in\mathbb{B}^d$ and
\begin{align}\label{eq:jd}
	j^d\triangleq\{j_1,\dots,j_d\}\subseteq[n]\triangleq\{1,\dots,n\}, 
\end{align}
where 
\begin{align*}
	1\leq j_1<j_2<\dots<j_d\leq n. 
\end{align*}
Especially, $j^0\equiv\emptyset$ and $j^n\equiv[n]$. Further, we define
\begin{align}\label{eq:Jnd}
	{\cal J}_{n,d} \triangleq \{j^d:1\leq j_1<j_2<\dots<j_d\leq n\}.
\end{align}
The following properties of ${\cal J}_{n,d}$ are obvious.
\begin{itemize}
	\item ${\cal J}_{n,0}=\{\emptyset\}$, ${\cal J}_{n,1}=\{\{1\},\dots,\{n\}\}$, and ${\cal J}_{n,n}=\{[n]\}$.
	\item The cardinality $|{\cal J}_{n,d}| = \binom{n}{d}$. Especially, $|{\cal J}_{n,0}|=|{\cal J}_{n,n}|=1$.
	\item $\sum_{d=0}^{n}{|{\cal J}_{n,d}|} = 2^n$ and $\sum_{d=0}^{n}{\left(2^d|{\cal J}_{n,d}|\right)} = 3^n$.
\end{itemize}

\begin{definition}[Shift Function]
	For $j^d\in{\cal J}_{n,d}$ and $b^d\in\mathbb{B}^d$, we define the shift function as \cite{FangTCOM16b} 
	\begin{align}\label{eq:tau}
		\tau(j^d,b^d) \triangleq (1-2^{-r})\sum_{d'=1}^d{(1-2b_{d'})2^{rj_{d'}}}\in\mathbb{R}.
	\end{align}
\end{definition}

\begin{lemma}[Properties of Shift Function]
	For $j^d\in{\cal J}_{n,d}$ and $b^d\in\mathbb{B}^d$,  
	\begin{itemize}
		\item $\tau(j^d,b^d\oplus1^d) = -\tau(j^d,b^d)$; and
		\item $|\tau(j^d,b^d)|\leq (2^{nr}-1)<2^{nr}$.
	\end{itemize}
\end{lemma}
\begin{proof}
	Since $(1-2b)=-(1-2(b\oplus1))$ for $b\in\mathbb{B}$, we have 
	\begin{align*}
		\tau(j^d,b^d\oplus1^d) = -\tau(j^d,b^d). 
	\end{align*}
	We rewrite \eqref{eq:tau} as
	\begin{align}\label{eq:tauvar}
		\tau(j^d,b^d)
		&= \underbrace{(1-2^{-r})\sum_{d'=1}^{d}{2^{rj_{d'}}}}_{c(j^d)} - 2\underbrace{(1-2^{-r})\sum_{d'=1}^{d}{b_{d'}2^{rj_{d'}}}}_{v(j^d,b^d)}\nonumber\\
		&= c(j^d) - 2v(j^d,b^d).
	\end{align}
	Since $b^d\in\mathbb{B}^d$, we have 
	\begin{align*}
		0 = v(j^d,0^d) \leq v(j^d,b^d)\leq v(j^d,1^d) = c(j^d).
	\end{align*}
	Therefore, 
	\begin{align*}
		-c(j^d) = \tau(j^d,1^d) \leq \tau(j^d,b^d) \leq \tau(j^d,0^d) = c(j^d).
	\end{align*}
	Since $(1-2^{-r})>0$ for $0<r<1$, we have
	\begin{align*}
		c(j^d)\leq c(j^n)=s(1^n)=(2^{nr}-1). 
	\end{align*}
	Therefore,
	\begin{align*}
		-(2^{nr}-1) = \tau(j^n,1^n) \leq \tau(j^d,b^d) \leq \tau(j^n,0^n) = (2^{nr}-1).
	\end{align*}
	Hence, $|\tau(j^d,b^d)|\leq (2^{nr}-1)<2^{nr}$.
\end{proof}

\begin{lemma}[Physical Meaning of Shift Function]\label{prop:tau}
	Let $x^n\in\mathbb{B}^n$, $y^n\in\mathbb{B}^n$, and $z^n=x^n\oplus y^n\in\mathbb{B}^n$. We define $x_{j^d}\triangleq(x_{j_1},\dots,x_{j_d})\in\mathbb{B}^d$. In a similar way, $y_{j^d}$ and $z_{j^d}$ are defined. If $z_{j^d}=1^d$ and $z_{[n]\setminus j^d}=0^{n-d}$, then
	\begin{align}\label{eq:ell_yn}
		\begin{cases}
			s(y^n) = s(x^n) + \tau(j^d,x_{j^d})\\
			s(x^n) = s(y^n) + \tau(j^d,y_{j^d}).
		\end{cases}
	\end{align}
\end{lemma}
\begin{proof}
	For $1\leq i\leq n$, the following two branches hold obviously:
	\begin{itemize}
		\item If $(x_i,y_i)=(0,1)$, then $x_i2^{ir}=0$ and $y_i2^{ir}=2^{ir}=x_i2^{ir}+2^{ir}$;  
		\item If $(x_i,y_i)=(1,0)$, then $x_i2^{ir}=2^{ir}$ and $y_i2^{ir}=0=x_i2^{ir}-2^{ir}$. 
	\end{itemize}
	The above two branches can be merged as 
	\begin{align*}
		y_i2^{ir}=x_i2^{ir} + (1-2x_i)2^{ir}. 
	\end{align*}
	Then \eqref{eq:ell_yn} follows immediately.
\end{proof}

\begin{definition}[Coexisting Interval]
	Let $\{(a,b]+\tau\} \triangleq (a+\tau,b+\tau]\subset\mathbb{R}$. For $m\in[1:2^{nr})$, the $m$-th coexisting interval associated with $j^d\in{\cal J}_{n,d}$ and $b^d\in\mathbb{B}^d$ is defined as
	\begin{align}\label{eq:Im}
		{\cal I}_m{(j^d,b^d)} &\triangleq \{(m-1,m]-\tau(j^d,b^d)\} \cap (m-1,m]\nonumber\\
		&\subseteq (m-1,m].
	\end{align}
\end{definition}

We refer to $(m-1,m]$ as the $m$-th unit interval. Note that the above definition of coexisting interval does not consider $m=0$. Of course, we can define the $0$-th coexisting interval as $(-1,0]\cap[0,2^{nr})=[0,0]=\{0\}$, which includes only one point $0$ in $\mathbb{R}$. However, there is always one and only one codeword $0^n$ in the $0$-th coset ${\cal C}_0$ (cf. \cref{subfig:binning}), so we will no longer discuss the $0$-th coexisting interval.

\begin{lemma}[Concrete Form of Coexisting Interval]
	Depending on the value of $\tau(j^d,b^d)$, the $m$-th coexisting interval associated with $j^d\in{\cal J}_{n,d}$ and $b^d\in\mathbb{B}^d$ has different forms:
	\begin{align}\label{eq:frakI}
		{\cal I}_m(j^d,b^d) = 
		\begin{cases}
			\emptyset, 				& |\tau(j^d,b^d)|\geq 1\\
			(m-1, m-\tau(j^d,b^d)],	& 0\leq\tau(j^d,b^d)<1\\
			(m-1-\tau(j^d,b^d), m],	& -1<\tau(j^d,b^d)\leq0.
		\end{cases}
	\end{align}
	Consequently, the complement of ${\cal I}_m(j^d,b^d)$ is
	\begin{align}\label{eq:frakI_bar}
		\bar{\cal I}_m(j^d,b^d) 
		&\triangleq (m-1,m]\setminus{\cal I}_m(j^d,b^d)\nonumber\\
		&= \begin{cases}
			(m-1,m], 				& |\tau(j^d,b^d)|\geq 1\\
			(m-\tau(j^d,b^d),m],	& 0\leq\tau(j^d,b^d)<1\\
			(m-1,m-1-\tau(j^d,b^d)],& -1<\tau(j^d,b^d)\leq0.
		\end{cases}
	\end{align}
\end{lemma}
\begin{proof}
	According to \eqref{eq:Im}, it is easy to obtain \eqref{eq:frakI} with the help of \cref{fig:coexist}, while \eqref{eq:frakI_bar} is a direct result of \eqref{eq:frakI}. 
\end{proof}

\begin{figure}[!t]
	\subfigure[$0\leq\tau(j^d,b^d)<1$]{\includegraphics[width=.5\linewidth]{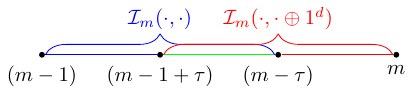}\label{subfig:pos}}%
	\subfigure[$-1<\tau(j^d,b^d)\leq0$]{\includegraphics[width=.5\linewidth]{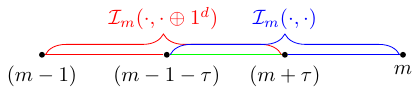}\label{subfig:neg}}
	\caption{The $m$-th pair of mirror coexisting intervals associated with $j^d\in{\cal J}_{n,d}$ and $b^d\in\mathbb{B}^d$, where $\tau$ is a shorthand of $\tau(j^d,b^d)$, ${\cal I}_m(\cdot,\cdot)$ is a shorthand of ${\cal I}_m(j^d,b^d)$, and ${\cal I}_m(\cdot,\cdot\oplus 1^d)$ is a shorthand of ${\cal I}_m(j^d,b^d\oplus 1^d)$.}
	\label{fig:coexist}
\end{figure}

\begin{lemma}[Length of Coexisting Interval]
	Let $|{\cal I}|$ be the length of continuous interval ${\cal I}$. Then 
	\begin{align}\label{eq:len}
		|{\cal I}_m(j^d,b^d)| 
		&= \left(1-|\tau(j^d,b^d)|\right)^+ \nonumber\\
		&\triangleq \max(0,1-|\tau(j^d,b^d)|)\in[0,1].
	\end{align}
\end{lemma}

\begin{definition}[Mirror Coexisting Interval]
	We say that ${\cal I}_m(j^d,b^d)$ and ${\cal I}_m(j^d,b^d\oplus1^d)$ form the $m$-th pair of mirror coexisting intervals associated with $j^d\in{\cal J}_{n,d}$ and $b^d\in\mathbb{B}^d$.
\end{definition}

\cref{fig:coexist} shows a pair of mirror coexisting intervals associated with $j^d\in{\cal J}_{n,d}$ and $b^d\in\mathbb{B}^d$ for $|\tau(j^d,b^d)|<1$.

\begin{lemma}[Mirror Coexisting Intervals]\label{prop:mirror} 
	With the help of \cref{fig:coexist}, we have the following properties of (mirror) coexisting intervals.
	\begin{itemize}
		\item The mirror of a coexisting interval can be obtained by a shift:
		\begin{align}\label{eq:mirror}
			{\cal I}_m(j^d,b^d\oplus1^d) = {\cal I}_m(j^d,b^d) + \tau(j^d,b^d).
		\end{align}
		\item A pair of mirror coexisting intervals must belong to the same unit interval. More concretely, the $m$-th pair of mirror coexisting intervals must belong to the $m$-th unit interval.
		\item The $m$-th pair of mirror coexisting intervals are almost symmetric around the point $(m-1/2)$ (except at the end points of coexisting intervals). 
	\end{itemize}
\end{lemma}
\begin{proof}
	If $\tau(j^d,b^d)\geq0$, then $\tau(j^d,b^d\oplus1^d)=-\tau(j^d,b^d)\leq0$. According to \eqref{eq:frakI}, we have ${\cal I}_m(j^d,b^d) = (m-1, m-\tau(j^d,b^d)]$ and
	\begin{align}\label{eq:pos}
		{\cal I}_m(j^d,b^d\oplus1^d) 
		&= (m-1-\tau(j^d,b^d\oplus1^d), m]\nonumber\\
		&= (m-1+\tau(j^d,b^d), m].
	\end{align}
	If $\tau(j^d,b^d)\leq0$, then $\tau(j^d,b^d\oplus1^d)=-\tau(j^d,b^d)\geq0$. According to \eqref{eq:frakI}, we have ${\cal I}_m(j^d,b^d) = (m-1-\tau(j^d,b^d),m]$ and
	\begin{align}\label{eq:neg}
		{\cal I}_m(j^d,b^d\oplus1^d) 
		&= (m-1,m-\tau(j^d,b^d\oplus1^d)]\nonumber\\
		&= (m-1,m+\tau(j^d,b^d)].
	\end{align}
	Combining the above two cases, we can obtain \eqref{eq:mirror} immediately\footnote{Note that $(a,b]=\emptyset$ if $a\geq b$, and $\emptyset$ is a subset of every set.}. 
	
	The second bullet point holds obviously, so its proof is omitted to avoid verbosity. As for the symmetry between ${\cal I}_m(j^d,b^d)$ and ${\cal I}_m(j^d,b^d\oplus1^d)$, it can be found from \eqref{eq:pos} and \eqref{eq:neg} that the sum of the infimum of ${\cal I}_m(j^d,b^d)$ and the maximum of ${\cal I}_m(j^d,b^d\oplus1^d)$ is always $(2m-1)=2(m-1/2)$, and so is the sum of the maximum of ${\cal I}_m(j^d,b^d)$ and the infimum of ${\cal I}_m(j^d,b^d\oplus1^d)$. Thus the symmetry holds.
\end{proof}

\begin{definition}[Coexisting Interval Set] 
	The set of coexisting intervals associated with $j^d\in{\cal J}_{n,d}$ and $b^d\in\mathbb{B}^d$ is
	\begin{align}\label{eq:Ijdbd}
		{\cal I}(j^d,b^d) \triangleq \left\{{\cal I}_m(j^d,b^d):m\in[1:2^{nr})\right\}.
	\end{align}
\end{definition}

\begin{theorem}[Necessary and Sufficient Condition for Coexistence]\label{thm:equiv}
	Consider two binary blocks $x^n\in\mathbb{B}^n$ and $y^n\in\mathbb{B}^n$. Let $z^n=x^n\oplus y^n\in\mathbb{B}^n$. If $z_{j^d}=1^d$ and $z_{[n]\setminus j^d}=0^{n-d}$, the necessary and sufficient condition for the event that $x^n$ and $y^n$ coexist in the same coset is $s(x^n)\in{\cal I}(j^d,x_{j^d})$, or equivalently $s(y^n)\in{\cal I}(j^d,y_{j^d})={\cal I}(j^d,x_{j^d}\oplus 1^d)$. That is
	\begin{align}\label{eq:equiv}
		\{\lceil s(x^n)\rceil = \lceil s(y^n)\rceil\} 
		&\leftrightarrow \{s(x^n)\in{\cal I}(j^d,x_{j^d})\}\nonumber\\
		&\leftrightarrow \{s(y^n)\in{\cal I}(j^d,y_{j^d})\}.
	\end{align}
\end{theorem}
\begin{proof}
	According to \eqref{eq:ell_yn} of \cref{prop:tau}, we have 
	\begin{align*}
		s(y^n)=s(x^n)+\tau(j^d,x_{j^d}), 
	\end{align*}
	where $j^d\in{\cal J}_{n,d}$ and $x_{j^d}\in\mathbb{B}^d$. If 
	\begin{align*}
		s(x^n)\in{\cal I}_m(j^d,x_{j^d})\subseteq(m-1,m], 
	\end{align*}
	then according to \eqref{eq:mirror} of \cref{prop:mirror}, we have
	\begin{align*}
		s(y^n) \in {\cal I}_m(j^d,x_{j^d}) + \tau(j^d,x_{j^d}) 
		&= {\cal I}_m(j^d,x_{j^d}\oplus 1^d)\nonumber\\
		&= {\cal I}_m(j^d,y_{j^d})\subseteq(m-1,m].
	\end{align*}
	Since $s(x^n)\in(m-1,m]$ and $s(y^n)\in(m-1,m]$, both $x^n$ and $y^n$ must belong to the $m$-th coset. After generalization for all $m\in[1:2^{nr})$, if $s(x^n)\in{\cal I}(j^d,x_{j^d})$, or equivalently $s(y^n)\in{\cal I}(j^d,y_{j^d})={\cal I}(j^d,x_{j^d}\oplus 1^d)$, both $x^n$ and $y^n$ must belong to the same coset.
	
	{\bf Converse}. If $s(x^n)\in\bar{\cal I}_m(j^d,x_{j^d})\subseteq(m-1,m]$, where $\bar{\cal I}_m(j^d,x_{j^d})$ is the complement of ${\cal I}_m(j^d,x_{j^d})$ given by \eqref{eq:frakI_bar}, then according to $s(y^n)=s(x^n)+\tau(j^d,x_{j^d})$, we have
	\begin{align*}
		s(y^n) \in
		\begin{cases}
			(m-1+\tau(j^d,x_{j^d}),m+\tau(j^d,x_{j^d})], & |\tau(j^d,x_{j^d})|\geq 1\\
			(m,m+\tau(j^d,x_{j^d})],	& 0\leq\tau(j^d,x_{j^d})<1\\
			(m-1+\tau(j^d,x_{j^d}),m-1],& -1<\tau(j^d,x_{j^d})\leq0.
		\end{cases}
	\end{align*}
	Clearly, in every of the above three cases, $s(y^n)\notin(m-1,m]$. While as we know, $s(x^n)\in(m-1,m]$. According to \cref{lem:coset}, it is evident that $x^n$ and $y^n$ do not coexist in the same coset.
\end{proof}

\begin{corollary}[Necessary Condition for Coexistence]\label{corol:neccoe}
	Let $x^n\in\mathbb{B}^n$, $y^n\in\mathbb{B}^n$, and $z^n=x^n\oplus y^n$. If $z_{j^d}=1^d$ and $z_{[n]\setminus j^d}=0^{n-d}$, the event that $x^n$ and $y^n$ coexist in the same coset happens only if
	$|\tau(j^d,x_{j^d})|<1$.
\end{corollary}

\section{Probability of Coexisting Interval}
Now we ponder over such an important problem: Given $X_{j^d}=b^d\in\mathbb{B}^d$, how possible will $s(x^n)$ fall into ${\cal I}(j^d,b^d)$? To answer this question, we introduce the following important random variable.

\begin{definition}[Conditional Value of $s(X^n)$]
	For $j^d\in{\cal J}_{n,d}$ and $b^d\in\mathbb{B}^d$, we define the conditional value of $s(X^n)$ given $X_{j^d}=b^d$ as
	\begin{align}\label{eq:E}
		E{(j^d,b^d)}\triangleq s(X^n|X_{j^d}=b^d).
	\end{align}	
\end{definition}

According to \eqref{eq:ellequal}, we have
\begin{align}\label{eq:Evar}
	E{(j^d,b^d)} 
	&\overset{d}{=} 
	\underbrace{(1-2^{-r})\sum_{d'=1}^{d}{b_{d'}2^{rj_{d'}}}}_{{\rm deterministic~constant}~c(j^d,b^d)} + \underbrace{(1-2^{-r})\sum_{i\in[n]\setminus j^d}{X_i2^{ir}}}_{{\rm random~variable}~V([n]\setminus j^d)}\nonumber\\
	&= c(j^d,b^d) + V([n]\setminus j^d),
\end{align}
where $c(j^d,b^d)$ is a deterministic constant parameterized by $j^d$ and $b^d$, while $V([n]\setminus j^d)$ is a function w.r.t. $(n-d)$ random variables $X_{[n]\setminus j^d}$. Consequently, $E{(j^d,b^d)}$ is also a function w.r.t. $(n-d)$ random variables $X_{[n]\setminus j^d}$. If $d<n$, then $V([n]\setminus j^d)$ is itself a random variable and in turn $E{(j^d,b^d)}$ is also a random variable; otherwise, if $d=n$, then $V([n]\setminus j^d)=0$ and $E{(j^d,b^d)}$ actually degenerates into a deterministic constant parameterized by $b^n$: 
\begin{align*}
	E{(j^n,b^n)} = s(b^n) = c(j^n,b^n) = (1-2^{-r})\sum_{i=1}^{n}{b_i2^{ir}}.
\end{align*}

\begin{remark}[Comparison of $E(j^d,b^d)$ with $\tau(j^d,b^d)$]
	After observing \eqref{eq:tauvar} and \eqref{eq:Evar}, the reader may find that $\tau(j^d,b^d)$ and $E(j^d,b^d)$ are very similar to each other, where 
	\begin{align*}
		\begin{cases}
			\tau(j^d,b^d) = c(j^d)-2v(j^d,b^d)\\
			E(j^d,b^d) = c(j^d,b^d)+V([n]\setminus j^d).
		\end{cases}
	\end{align*}
	However, the reader should notice the differences between them. 
	\begin{itemize}
		\item $c(j^d)$ is a special case of $c(j^d,b^d)$ when $b^d=1^d$.
		\item $v(j^d,b^d)$ is a deterministic constant, while $V([n]\setminus j^d)$ is a random variable (for $d<n$).
		\item $v(j^d,b^d)$ is the sum of $d$ terms, while $V([n]\setminus j^d)$ is the sum of $(n-d)$ terms.
		\item $\tau(j^d,b^d)$ is a deterministic constant over $(-2^{nr},2^{nr})$, while $E(j^d,b^d)$ is a random variable (for $d<n$) over $[0,2^{nr})$.
	\end{itemize}
\end{remark}

\begin{theorem}[Asymptotic Probability of Coexisting Interval]\label{thm:prob}
	For almost every $0<r<1$, we have
	\begin{align}\label{eq:prob}
		\lim_{(n-d)\to\infty}\Pr\left\{E{(j^d,b^d)}\in{\cal I}(j^d,b^d)\right\} = \left(1-|\tau(j^d,b^d)|\right)^+.
	\end{align}
\end{theorem}
\begin{proof}	
	As shown by \eqref{eq:Evar}, $E(j^d,b^d)$ is the sum of a deterministic constant $c(j^d,b^d)$ and a random variable $V([n]\setminus j^d)$ (for $d<n$). Let us pay attention to $V([n]\setminus j^d)$. The cardinality of $[n]\setminus j^d$ will go to infinity as $(n-d)\to\infty$. According to \cref{thm:gm}, for almost every $0<r<1$, $V([n]\setminus j^d)$ will be u.d. mod 1 as $(n-d)\to\infty$, and in turn $E(j^d,b^d)$ will also be u.d. mod 1. In other words, if $E{(j^d,b^d)}\in(m-1,m]$, then $E(j^d,b^d)$ will be u.d. over $(m-1,m]$ as $(n-d)\to\infty$. Therefore,
	\begin{align}
		\lim_{(n-d)\to\infty}\Pr\left\{E{(j^d,b^d)}\in{\cal I}_m(j^d,b^d)\middle|E{(j^d,b^d)}\in(m-1,m]\right\}\nonumber\\ 
		= \frac{|{\cal I}_m(j^d,b^d)|}{|(m-1,m]|} \stackrel{(a)}{=} \left(1-|\tau(j^d,b^d)|\right)^+,
	\end{align}
	where $(a)$ comes from \eqref{eq:len}. After averaging over all $m\in[1:2^{nr})$, we will obtain \eqref{eq:prob} (cf. \eqref{eq:Ijdbd} for the definition of ${\cal I}(j^d,b^d)$). 
\end{proof}

\begin{remark}[A Note about the Proof of \cref{thm:prob}]
	If $(n-d)=\infty$, the cardinality $|[n]\setminus j^d|=\infty$ and for almost every $0<r<1$, $V([n]\setminus j^d)$ will be a continuous random variable u.d. mod 1, and so is $E(j^d,b^d)$. If $(n-d)<\infty$, then $V([n]\setminus j^d)$ will be a discrete random variable, and so is $E(j^d,b^d)$. Finally, if $d=n$, then $V([n]\setminus j^d)=0$ and hence $E(j^d,b^d)$ is a deterministic constant rather than a random variable.
\end{remark}

\section{Error Rate for One Unknown Ending Symbol}\label{sec:review}
Equipped with coexisting interval, we are now ready to compute the frame error rate of overlapped arithmetic codes. Due to the difficult nature of this problem, we will start with some simple cases. It was found in \cite{FangTCOM16a} that the residual symbol errors of overlapped arithmetic codes after decoding usually happen at the end of each block, so we will begin with the simplest case that {\em all but the last} symbols of each block are known at the decoder.

Let $X^n$ be the source and $Y^n$ be the side information known only at the decoder. The correlation between $X^n$ and $Y^n$ is modeled as a binary symmetric channel with crossover probability $\epsilon$. Let $Z^n=X^n\oplus Y^n$. Then the event $X_i\neq Y_i$ is equivalent to the event $Z_i=1$. On receiving the bitstream $M=\lceil{s(X^n)}\rceil$, the decoder will search through the coset ${\cal C}_M$ to find the codeword closest (in Hamming distance) to $Y^n$, which is denoted by $\hat{X}^n$, and take it as the estimate of $X^n$. The {\em frame-error-rate} (FER) after decoding is denoted by $\Pr(e)$, where $e$ denotes the event $\{\hat{X}^n\neq X^n{|Y^n}\}$. If $X^{n-1}$ is known at the decoder, then $\hat{X}^{n-1}\equiv X^{n-1}$. The conditional FER given $X^{n-1}$ known at the decoder is
\begin{align*}
	\Pr(e|X^{n-1}) = \Pr(\hat{X}_n\neq X_n{|Y_n}) < \Pr(e).
\end{align*}

\begin{theorem}[FER Given $X^{n-1}$]
	For almost every $0<r<1$, 
	\begin{align}\label{eq:Pexn1}
		\lim_{n\to\infty}{\Pr(e|X^{n-1})} = (2-2^r)\epsilon.
	\end{align}
\end{theorem}
\begin{proof}
Let us define $\tilde{X}^n\triangleq(X^n\oplus(0^{n-1}1))$, which is called a {\em spare} codeword of $X^n$. From \cref{lem:coexist}, we have the following findings.
\begin{itemize}
	\item If $\lceil{s(X^n)}\rceil\neq\lceil{s(\tilde{X}^n)}\rceil$, then $X^n$ and $\tilde{X}^n$ do not coexist in the same coset, so decoding will always succeed, regardless of $Y^n$.
	\item If $\lceil{s(X^n)}\rceil=\lceil{s(\tilde{X}^n)}\rceil$, then $X^n$ and $\tilde{X}^n$ coexist in the same coset, so decoding correctness purely depends on $Y_n$.
\end{itemize}
Now it is obvious that
\begin{align*}
	\Pr(\hat{X}_n\neq X_n{|Y_n}) 
	&= \Pr(\lceil{s(X^n)}\rceil=\lceil{s(\tilde{X}^n)}\rceil) \cdot \Pr(Z_n=1)\nonumber\\
	&= \epsilon \cdot \Pr(\lceil{s(X^n)}\rceil=\lceil{s(\tilde{X}^n)}\rceil).
\end{align*}
The key is how to calculate $\Pr(\lceil{s(X^n)}\rceil=\lceil{s(\tilde{X}^n)}\rceil)$. Since $X^n\oplus\tilde{X}^n=(0^{n-1}1)$, the Hamming distance between $X^n$ and $\tilde{X}^n$ is $d=1$. By \eqref{eq:equiv} of \cref{thm:equiv}, the equivalent event of $\lceil{s(X^n)}\rceil=\lceil{s(\tilde{X}^n)}\rceil$ is $s(X^n)\in{\cal I}(j^d,X_{j^d})$, where $j^d=\{n\}$ and $X_{j^d}=X_n$. That is,
\begin{align*}
	\{\lceil{s(X^n)}\rceil=\lceil{s(\tilde{X}^n)}\rceil\} \leftrightarrow \{s(X^n)\in{\cal I}(n,X_n)\},
\end{align*}
where ${\cal I}(n,X_n)={\cal I}(\{n\},X_n)$ for conciseness. In turn,
\begin{align*}
	&\Pr(s(X^n)\in{\cal I}(n,X_n)) \nonumber\\
	&\qquad= \sum_{b\in\mathbb{B}}\Pr(X_n=b)\cdot\Pr(s(X^n|X_n=b)\in{\cal I}(n,b))\nonumber\\
	&\qquad\overset{(a)}{=} \tfrac{1}{2}\sum_{b\in\mathbb{B}}\Pr(E(n,b)\in{\cal I}(n,b)),
\end{align*}
where $E(n,b)=E(\{n\},b)$ for conciseness and $(a)$ comes from \eqref{eq:E}. By \eqref{eq:prob} of \cref{thm:prob}, we have
\begin{align*}
	\lim_{n\to\infty}\Pr\left\{E{(n,b)}\in{\cal I}(n,b)\right\} 
	&= \left(1-|\tau(n,b)|\right)^+\nonumber\\
	&= 1-(1-2^{-r})2^r = 2-2^r,	
\end{align*}
where $\tau(n,b)=\tau(\{n\},b)$ for conciseness. Hence \eqref{eq:Pexn1} holds \cite{FangTCOM14}.
\end{proof}

\section{Error Rate for Two Unknown Ending Symbols}\label{sec:double}
After understanding the simplest case, now we consider a slightly more complex case: {\em All but the last two} symbols of each block are known at the decoder. It will be seen that the analysis in the second simplest case is much more difficult than that in the simplest case, and the extension is not straightforward. Let $\Pr(e|X^{n-2})$ be the conditional FER given $X^{n-2}$ at the decoder, then
\begin{align*}
	\Pr(e|X^{n-2}) = \Pr(\hat{X}_{n-1}^n \neq X_{n-1}^n{|Y_{n-1}^n}),
\end{align*}
where $X_{n-1}^n=(X_{n-1},X_n)$. Let us define $(2^2-1)=3$ spare codewords of $X^n$ as 
\begin{align*}
	\begin{cases}
		\tilde{X}^n_{01} = (X^n\oplus (0^{n-2}01))\\
		\tilde{X}^n_{10} = (X^n\oplus (0^{n-2}10))\\
		\tilde{X}^n_{11} = (X^n\oplus (0^{n-2}11)).
	\end{cases}
\end{align*}
Let ${\cal E}_{01}\triangleq\{\lceil{s(\tilde{X}^n_{01})}\rceil = \lceil{s(X^n)}\rceil\}$ and $\bar{{\cal E}}_{01}\triangleq\{\lceil{s(\tilde{X}^n_{01})}\rceil \neq \lceil{s(X^n)}\rceil\}$. In a similar way, ${\cal E}_{10}$, $\bar{{\cal E}}_{10}$, ${\cal E}_{11}$, and $\bar{{\cal E}}_{11}$ are also defined. Then
\begin{align}\label{eq:events}
	\begin{cases}
		{\cal E}_{01} \leftrightarrow \{s(X^n)\in{\cal I}(n,X_n)\}\\
		{\cal E}_{10} \leftrightarrow \{s(X^n)\in{\cal I}(n-1,X_{n-1})\}\\
		{\cal E}_{11} \leftrightarrow \{s(X^n)\in{\cal I}(\{n-1,n\},X_{n-1}^n)\}.
	\end{cases}
\end{align}

\subsection{Enumeration of Event Combinations}
There are $2^{2^2-1}=2^3=8$ event combinations in total.
\begin{enumerate}
	\item $\bar{\cal E}_{01}\wedge\bar{\cal E}_{10}\wedge\bar{\cal E}_{11}$: The decoding succeeds always.
	
	\item ${\cal E}_{01}\wedge\bar{\cal E}_{10}\wedge\bar{\cal E}_{11}$: If $Z_n=X_n\oplus Y_n=1$, \textit{i.e.}, $X_n\neq Y_n$, the decoding will fail, so the failure probability is $\epsilon$.
	
	\item $\bar{\cal E}_{01}\wedge{\cal E}_{10}\wedge\bar{\cal E}_{11}$: Just as case 2, the failure probability is $\epsilon$.
	
	\item $\bar{\cal E}_{01}\wedge\bar{\cal E}_{10}\wedge{\cal E}_{11}$: If $Z_{n-1}^n=(X_{n-1}^n\oplus Y_{n-1}^n)=1^2$, the decoding will fail; if $Z_{n-1}^n=10$ or $01$, since $|Y^n\oplus\tilde{X}_{11}^n|=|Y^n\oplus X^n|=1$, the failure probability is $0.5$. The overall failure probability is
	\begin{align*}
		\Pr(Z_{n-1}^n=1^2) + \tfrac{1}{2}\left(\Pr(Z_{n-1}^n=10)+\Pr(Z_{n-1}^n=01)\right)\nonumber\\
		=\epsilon^2+\tfrac{1}{2}(\epsilon(1-\epsilon)+(1-\epsilon)\epsilon) = \epsilon.
	\end{align*}
	
	\item ${\cal E}_{01}\wedge\bar{\cal E}_{10}\wedge{\cal E}_{11}$: If $Y_n\neq X_n$, the decoding will fail; if $Z_{n-1}^n=10$, we have $|Y^n\oplus\tilde{X}_{01}^n|=2$ and $|Y^n\oplus\tilde{X}_{11}^n|=|Y^n\oplus X^n|=1$, so the failure probability is $0.5$. The overall failure probability is
	\begin{align*}
		\Pr(Y_n\neq X_n) + \tfrac{1}{2}\Pr(Z_{n-1}^n=10)\nonumber\\
		\qquad=\epsilon + \epsilon(1-\epsilon)/2 = \epsilon(3-\epsilon)/2.
	\end{align*}
	
	\item $\bar{\cal E}_{01}\wedge{\cal E}_{10}\wedge{\cal E}_{11}$: This case is a mirror image of case 5 by swapping the roles of $(X_{n-1},Y_{n-1})$ and $(X_n,Y_n)$, so the failure probability is also $\epsilon(3-\epsilon)/2$.
	
	\item ${\cal E}_{01}\wedge{\cal E}_{10}\wedge\bar{\cal E}_{11}$: The decoding succeeds only if $X_{n-1}^n=Y_{n-1}^n$, so the failure probability is
	\begin{align*}
		1-\Pr(Z_{n-1}^n=0^2) = 1-(1-\epsilon)^2 = \epsilon(2-\epsilon).
	\end{align*}
	
	\item ${\cal E}_{01}\wedge{\cal E}_{10}\wedge{\cal E}_{11}$: The decoding succeeds only if $X_{n-1}^n=Y_{n-1}^n$, so the failure probability is also $\epsilon(2-\epsilon)$.	
\end{enumerate}
\begin{table}
	\centering
	\caption{Event Combinations and Their Failure Probabilities}
	\begin{tabular}{c||c}
		\hline	
		Error Prob.& Event Combinations\\
		\hline\hline
		$0$ &$\bar{\cal E}_{01}\wedge\bar{\cal E}_{10}\wedge\bar{\cal E}_{11}$\\
		\hline
		$\epsilon$ &${\cal E}_{01}\wedge\bar{\cal E}_{10}\wedge\bar{\cal E}_{11}$, 
		$\bar{\cal E}_{01}\wedge{\cal E}_{10}\wedge\bar{\cal E}_{11}$, 
		$\bar{\cal E}_{01}\wedge\bar{\cal E}_{10}\wedge{\cal E}_{11}$\\
		\hline
		$\epsilon(3-\epsilon)/2$ &${\cal E}_{01}\wedge\bar{\cal E}_{10}\wedge{\cal E}_{11}$,
		$\bar{\cal E}_{01}\wedge{\cal E}_{10}\wedge{\cal E}_{11}$\\
		\hline
		$\epsilon(2-\epsilon)$ &${\cal E}_{01}\wedge{\cal E}_{10}\wedge\bar{\cal E}_{11}$, 
		${\cal E}_{01}\wedge{\cal E}_{10}\wedge{\cal E}_{11}$\\
		\hline
	\end{tabular}
	\label{tab:error}
\end{table}
We summarize eight event combinations and their failure probabilities in \cref{tab:error}. The failure probabilities in case 7 and case 8 are the same, so they can be merged into
\begin{align}
	({\cal E}_{01}\wedge{\cal E}_{10}\wedge{\cal E}_{11}) \vee ({\cal E}_{01}\wedge{\cal E}_{10}\wedge\bar{\cal E}_{11}) = {\cal E}_{01}\wedge{\cal E}_{10}.
\end{align}
Finally, we will obtain
\begin{align}\label{eq:exn2}
	&\Pr(e|X^{n-2}) = 
	\epsilon(2-\epsilon)\Pr({\cal E}_{01}\wedge{\cal E}_{10})\,+\nonumber\\
	&\epsilon\left(\Pr({\cal E}_{01}\wedge\bar{\cal E}_{10}\wedge\bar{\cal E}_{11}) + 
	\Pr(\bar{\cal E}_{01}\wedge{\cal E}_{10}\wedge\bar{\cal E}_{11}) + 
	\Pr(\bar{\cal E}_{01}\wedge\bar{\cal E}_{10}\wedge{\cal E}_{11})\right)\nonumber\\  
	&+\tfrac{\epsilon(3-\epsilon)}{2}\left(\Pr({\cal E}_{01}\wedge\bar{\cal E}_{10}\wedge{\cal E}_{11})  
	+ \Pr(\bar{\cal E}_{01}\wedge{\cal E}_{10}\wedge{\cal E}_{11})\right).
\end{align}
The problem now is how to determine the occurrence probabilities of event combinations. To achieve this goal, we should consider different $b^2$. Denote the FER given $X^{n-2}$ at the decoder and $X_{n-1}^n=b^2$ by
\begin{align*}
	\Pr(e|b^2)\triangleq\Pr(e|X^{n-2},X_{n-1}^n=b^2). 
\end{align*}
Then $\Pr(e|X^{n-2}\}$ can be expanded to
\begin{align*}
	\Pr(e|X^{n-2}\} = \tfrac{1}{4}\sum_{b^2\in\mathbb{B}^2}\Pr(e|b^2).
\end{align*}
Since $\tau(j^d,b^d\oplus1^d)=-\tau(j^d,b^d)$, we have $\Pr(e|0^2)=\Pr(e|1^2)$ and $\Pr(e|01)=\Pr(e|10)$. Therefore,
\begin{align*}
	\Pr(e|X^{n-2}) = \tfrac{1}{2}\left(\Pr(e|0^2)+\Pr(e|01)\right).
\end{align*}

\subsection{Conditional FER Given $b^2=0^2$}
\begin{figure}[!t]
	\centering
	\subfigure[$1<2q\leq(\sqrt{5}-1)$]{
		\includegraphics[width=\linewidth]{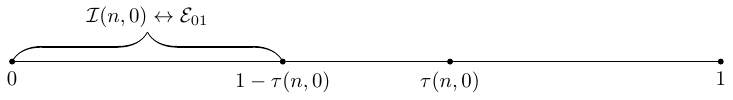}%
		\label{subfig:golden00}
	}\\
	\subfigure[$(\sqrt{5}-1)<2q\leq\sqrt{2}$]{
		\includegraphics[width=\linewidth]{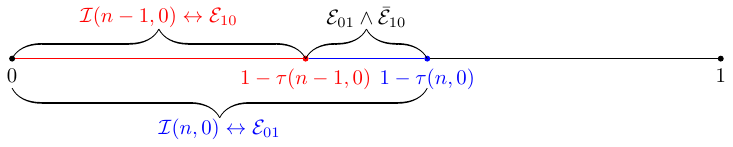}%
		\label{subfig:sqrt200}
	}\\
	\subfigure[$\sqrt{2}<2q<2$]{
		\includegraphics[width=\linewidth]{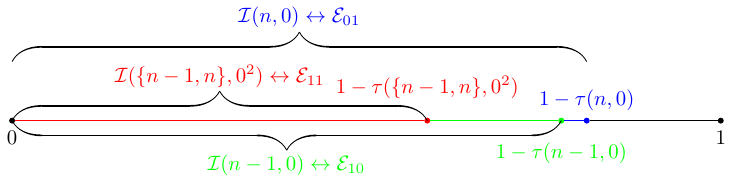}%
		\label{subfig:large00}	
	}
	\caption{Examples of $\tau(j^d,b^d)$ and ${\cal I}(j^d,b^d)$ for $b^2=0^2$.}
	\label{fig:err00}
\end{figure}

\begin{lemma}[$b^2=0^2$]
	Given $X^{n-2}$ known at the decoder and $X_{n-1}^n=0^2$, for almost every $0<r<1$, we have
	\begin{align}\label{eq:e00}
		\lim_{n\to\infty}\Pr(e|0^2) = \epsilon(1-\tau(n,0)) + \epsilon(1-\epsilon)(1-\tau(n-1,0))^+.
	\end{align}
\end{lemma}
\begin{proof}
Given $X_{n-1}^n=0^2$, \eqref{eq:events} becomes
\begin{align*}
	\begin{cases}
		{\cal E}_{01} \leftrightarrow \{E(n,0)\in{\cal I}(n,0)\}\\
		{\cal E}_{10} \leftrightarrow \{E(n-1,0)\in{\cal I}(n-1,0)\}\\
		{\cal E}_{11} \leftrightarrow \{E(\{n-1,n\},0^2)\in{\cal I}(\{n-1,n\},0^2)\},
	\end{cases}
\end{align*}
where $E{(j^d,b^d)}$ is defined by \eqref{eq:E}. Let $q\triangleq 2^{-r}$. From \eqref{eq:tau}, we have
\begin{align}
	\begin{cases}
		\tau(n,0) = (1-q)q^{-1} = 2^r-1 > 0\\
		\tau(n-1,0) = (1-q)q^{-2} = 2^{2r}-2^r  > 0\\
		\tau(\{n-1,n\},0^2) = \tau(n,0)+\tau(n-1,0)>0.
	\end{cases}
	\label{eq1}
\end{align}
For $0<r<1$, we have $1<2q<2$, which is followed by
\begin{align*}
	\begin{cases}
		0<\tau(n,0)<1\\
		0<\tau(n-1,0)<2\\
		0<\tau(\{n,n-1\},0^2)<3
	\end{cases}
\end{align*}
and
\begin{align*}
	0 < \tau(n,0) < \tau(n-1,0) < \tau(\{n,n-1\},0^2) < 3.
\end{align*}
Hence,
\begin{align*}
	{\cal I}(\{n,n-1\},0^2) \subset {\cal I}(n-1,0) \subset {\cal I}(n,0) \neq \emptyset,
\end{align*}
as shown by \cref{fig:err00}. Therefore, 
\begin{align*}
	{\cal E}_{11}\wedge{\cal E}_{10}\wedge{\cal E}_{01} = {\cal E}_{11}\wedge{\cal E}_{10} = {\cal E}_{11},
\end{align*}
which is followed by
\begin{align*}
	\begin{cases}
		{\cal E}_{01}\wedge{\cal E}_{10} = {\cal E}_{10}\\
		{\cal E}_{01}\wedge\bar{\cal E}_{10}\wedge\bar{\cal E}_{11} = 
		{\cal E}_{01}\wedge\bar{\cal E}_{10}\\
		\bar{\cal E}_{01}\wedge{\cal E}_{10}\wedge\bar{\cal E}_{11} = 
		\bar{\cal E}_{01}\wedge\bar{\cal E}_{10}\wedge{\cal E}_{11} = \emptyset\\
		{\cal E}_{01}\wedge\bar{\cal E}_{10}\wedge{\cal E}_{11} = 
		\bar{\cal E}_{01}\wedge{\cal E}_{10}\wedge{\cal E}_{11} = \emptyset.
	\end{cases}
\end{align*}
Consequently, given $X_{n-1}^n=0^2$, \eqref{eq:exn2} is reduced to
\begin{align*}
	\Pr(e|0^2) = \epsilon\Pr({\cal E}_{01}\wedge\bar{\cal E}_{10}) + \epsilon(2-\epsilon)\Pr({\cal E}_{10}).
\end{align*}
Since ${\cal I}(n-1,0) \subset {\cal I}(n,0)$, we have
\begin{align*}
	\Pr({\cal E}_{01}\wedge\bar{\cal E}_{10}) = \Pr({\cal E}_{01})  - \Pr({\cal E}_{10}) 
\end{align*}
and hence
\begin{align*}
	\Pr(e|0^2) = \epsilon\Pr({\cal E}_{01}) + \epsilon(1-\epsilon)\Pr({\cal E}_{10}).
\end{align*}
According to \cref{thm:prob}, as $n\to\infty$,
\begin{align*}
	\begin{cases}
		\Pr({\cal E}_{10}) = (1-\tau(n-1,0))^+\\
		\Pr({\cal E}_{01}) = (1-\tau(n,0))^+ = 1-\tau(n,0).
	\end{cases}
\end{align*}
Finally, we obtain \eqref{eq:e00}.
\end{proof}

\begin{remark}[Discussion on $b^2=0^2$]
It is easy to know that ${\cal I}(n,0)$ is always non-empty, but ${\cal I}(n-1,0)$ and ${\cal I}(\{n,n-1\},0^2)$ may or may not be empty, depending on the value of $q$. By solving $\tau(n-1,0)=1$, we get $2q=(\sqrt{5}-1)$; by solving $\tau(\{n,n-1\},0^2)=1$, we get $2q=\sqrt{2}$. Hence we divide the interval $2q\in(1,2)$ into three sub-intervals.
\begin{itemize}
	\item{$1<2q\leq(\sqrt{5}-1)$}: It is easy to get 
	\begin{align*}
		1\leq\tau(n-1,0)<\tau(\{n,n-1\},0^2)
	\end{align*}
	and hence 
	\begin{align*}
		{\cal I}(n-1,0)={\cal I}(\{n,n-1\},0^2)=\emptyset,
	\end{align*}
	which is followed by $\Pr({\cal E}_{10})=\Pr({\cal E}_{11})=0$. Then \eqref{eq:e00} is reduced to $\Pr(e|0^2) = \epsilon(1-\tau(n,0))$. This case is illustrated by \cref{subfig:golden00}.
	
	\item{$(\sqrt{5}-1)<2q\leq\sqrt{2}$}: It is easy to get 
	\begin{align*}
		\tau(n-1,0)<1\leq\tau(\{n,n-1\},0^2),
	\end{align*}
	so ${\cal I}(n-1,0)\neq\emptyset$ and ${\cal I}(\{n,n-1\},0^2)=\emptyset$. Then \eqref{eq:e00} can be reduced to $\Pr(e|0^2) = \epsilon(1-\tau(n,0)) + \epsilon(1-\epsilon)(1-\tau(n-1,0))$. A special example of this case when $2q=\sqrt{2}$ is shown in \cref{subfig:sqrt200}, where $\tau(n,0)+\tau(n-1,0)=\tau(\{n,n-1\},0^2)=1$. 
	
	\item{$\sqrt{2}<2q<2$}: It is easy to get 
	\begin{align*}
		\tau(n-1,0)<\tau(\{n,n-1\},0^2)<1,
	\end{align*}
	so ${\cal I}(n-1,0)\neq\emptyset$ and ${\cal I}(\{n,n-1\},0^2)\neq\emptyset$. Then \eqref{eq:e00} is reduced to $\Pr(e|0^2) = \epsilon(1-\tau(n,0)) + \epsilon(1-\epsilon)(1-\tau(n-1,0))$. This case is illustrated by \cref{subfig:large00}.
\end{itemize}
\end{remark}

\subsection{Conditional FER Given $b^2=01$}
\begin{figure}
	\centering
	\subfigure[$1<2q\leq(\sqrt{5}-1)$]{
		\includegraphics[width=\linewidth]{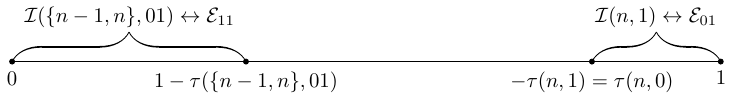}%
		\label{subfig:small01}
	}\\
	\subfigure[$2q=(\sqrt{5}-1)$]{
		\includegraphics[width=\linewidth]{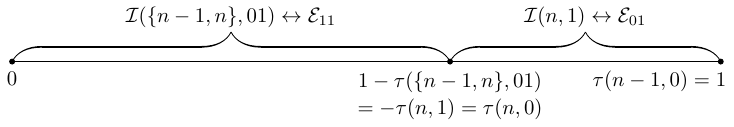}%
		\label{subfig:golden01}
	}\\
	\subfigure[$(\sqrt{5}-1)<2q\leq\sqrt{2}$]{
		\includegraphics[width=\linewidth]{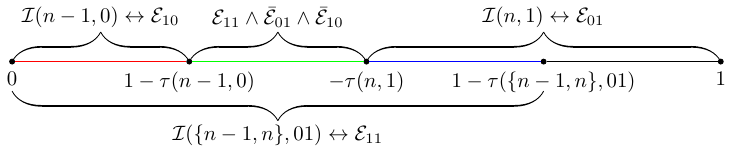}%
		\label{subfig:medium01}
	}\\
	\subfigure[$2q=\sqrt{2}$]{
		\includegraphics[width=\linewidth]{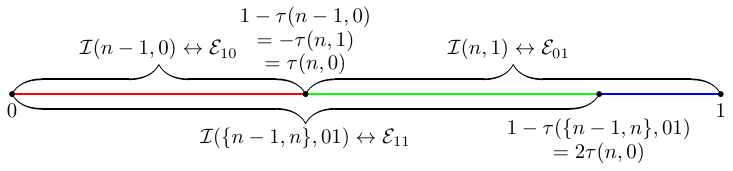}%
		\label{subfig:sqrt201}
	}\\
	\subfigure[$\sqrt{2}<2q<2$]{
		\includegraphics[width=\linewidth]{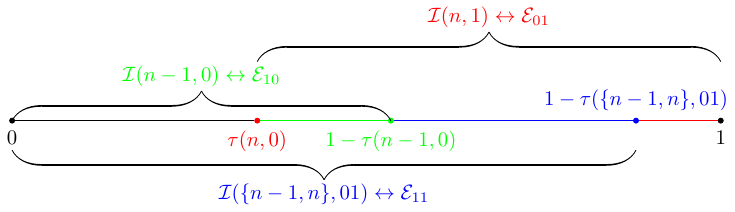}%
		\label{subfig:large01}	
	}
	\caption{Examples of $\tau(j^d,b^d)$ and ${\cal I}(j^d,b^d)$ for $b^2=01$.}
	\label{fig:err01}
\end{figure}

\begin{lemma}[$b^2=01$]
	Given $X^{n-2}$ known at the decoder and $X_{n-1}^n=01$, for almost every $0<r<1$, we have
	\begin{align}\label{eq:e01}
		\lim_{n\to\infty}\Pr(e|01) = \epsilon(2-\tau(n-1,0)) - \epsilon^2(1-\tau(n-1,0))^+.
	\end{align}
\end{lemma}

\begin{proof}
It is easy to know
\begin{align*}
	\begin{cases}
		\tau(n,1) = -\tau(n,0) = -(1-q)q^{-1}<0\\
		\tau(\{n-1,n\},01) = \tau(n-1,0)-\tau(n,0)>0,
	\end{cases}
	\label{eq3}
\end{align*}
where $\tau(n-1,0)$ and $\tau(n,0)$ are given by \eqref{eq1}. For $1<2q<2$, 
\begin{align*}
	0<\tau(\{n-1,n\},01)<\tau(n,0)<\min(1,\tau(n-1,0)).
\end{align*}
With the help of \cref{fig:err01}, we have ${\cal I}(n,1)\neq\emptyset$ and
\begin{align*}
	{\cal I}(n-1,0)\subset{\cal I}(\{n-1,n\},01)\neq\emptyset,
\end{align*}
which is followed by ${\cal E}_{10}\wedge{\cal E}_{11}={\cal E}_{10}$, $\bar{\cal E}_{10}\wedge\bar{\cal E}_{11}=\bar{\cal E}_{11}$, and ${\cal E}_{10}\wedge\bar{\cal E}_{11}=0$. Therefore, some event combinations can be simplified to
\begin{align*}
	\begin{cases}
		{\cal E}_{01}\wedge\bar{\cal E}_{10}\wedge\bar{\cal E}_{11} = {\cal E}_{01}\wedge\bar{\cal E}_{11}\\
		\bar{\cal E}_{01}\wedge{\cal E}_{10}\wedge{\cal E}_{11} = \bar{\cal E}_{01}\wedge{\cal E}_{10}\\
		\bar{\cal E}_{01}\wedge{\cal E}_{10}\wedge\bar{\cal E}_{11} = 0,
	\end{cases}
\end{align*}
and \eqref{eq:exn2} is reduced to
\begin{align}
	\Pr(e|01) 
	= \epsilon(2-\epsilon)\Pr({\cal E}_{01}\wedge{\cal E}_{10}) + \nonumber\\
	\epsilon\left(\Pr({\cal E}_{01}\wedge\bar{\cal E}_{11}) + 
	\Pr(\bar{\cal E}_{01}\wedge{\cal E}_{11}\wedge\bar{\cal E}_{10})\right) + \nonumber\\  
	\tfrac{\epsilon(3-\epsilon)}{2}\left(\Pr({\cal E}_{01}\wedge{\cal E}_{11}\wedge\bar{\cal E}_{10}) + 
	\Pr(\bar{\cal E}_{01}\wedge{\cal E}_{10})\right).
\end{align}
According to
\begin{align*}
	\begin{cases}
		(1-\tau(n-1,0))\leq\tau(n,0), &2q\leq\sqrt{2}\\
		(1-\tau(n-1,0))>\tau(n,0), 	  &2q>\sqrt{2}
	\end{cases}
\end{align*}
and
\begin{align*}
	\begin{cases}
		\tau(n-1,0)\geq 1, 	&2q\leq (\sqrt{5}-1)\\
		\tau(n-1,0)<1, 		&2q>(\sqrt{5}-1),
	\end{cases}
\end{align*}
event combinations can be further simplified depending on $q$.
\begin{itemize}
	\item{$1<2q\leq(\sqrt{5}-1)$}: It is easy to know $1\leq\tau(n-1,0)$ and
	\begin{align*}
		1-\tau(\{n-1,n\},01) = 1-\tau(n-1,0)+\tau(n,0) \leq \tau(n,0).
	\end{align*}
	Hence, we have (cf. \cref{subfig:small01})
	\begin{align*}
		{\cal I}(n-1,0) = {\cal I}(n,1)\cap{\cal I}(\{n-1,n\},01)=\emptyset.
	\end{align*}
	Equivalently, ${\cal E}_{10} = {\cal E}_{01}\wedge{\cal E}_{11}=0$. Now \eqref{eq3} is reduced to
	\begin{align*}
		\Pr(e|01) 
		&= \epsilon\left(\Pr({\cal E}_{01}\wedge\bar{\cal E}_{11}) + 
		\Pr(\bar{\cal E}_{01}\wedge{\cal E}_{11})\right)\\
		&= \epsilon\left(\Pr({\cal E}_{01}) + \Pr({\cal E}_{11})\right).
	\end{align*}
	According to \cref{thm:prob}, we have
	\begin{align*}
		\lim_{n\to\infty}\left(\Pr({\cal E}_{01}) + \Pr({\cal E}_{11})\right) = (2-\tau(n-1,0)).
	\end{align*}
	Finally, we obtain
	\begin{align}\label{eq:pre01small}
		\lim_{n\to\infty}\Pr(e|01) = \epsilon(2-\tau(n-1,0)).
	\end{align}	
	Especially, when $2q=(\sqrt{5}-1)$, we have $\tau(n-1,0)=1$ and
	\begin{align*}
		{\cal I}(n,1)\cup{\cal I}(\{n-1,n\},01)={\cal U}\triangleq(0,2^{nr}-1], 
	\end{align*}
	so $\bar{\cal E}_{01}\wedge\bar{\cal E}_{11}=\emptyset$ and $\Pr(e|01)=\epsilon$ (cf. \cref{subfig:golden01}). 
	
	\item{$(\sqrt{5}-1)<2q\leq\sqrt{2}$}: As shown by \cref{subfig:medium01}, we have	
	\begin{align*}
		1-\tau(n-1,0)\leq \tau(n,0) < 1-\tau(\{n-1,n\},01),
	\end{align*}
	which is followed by 
	\begin{align*}
		\begin{cases}
			\emptyset\neq{\cal I}(n-1,0)\subset{\cal I}(\{n-1,n\},01)\\
			{\cal I}(n,1)\cap{\cal I}(\{n-1,n\},01)\neq\emptyset\\
			{\cal I}(n,1)\cap{\cal I}(n-1,0)=\emptyset.
		\end{cases}
	\end{align*}
	Therefore, ${\cal E}_{01}\wedge{\cal E}_{10}=\bar{\cal E}_{01}\wedge\bar{\cal E}_{11}=0$. Now \eqref{eq3} is reduced to
	\begin{align*}
		\Pr(e|01) =\; 
		&\epsilon\left(\Pr(\bar{\cal E}_{11}) + 
		\Pr(\bar{\cal E}_{01}\wedge\bar{\cal E}_{10})\right) + \\  
		&\tfrac{\epsilon(3-\epsilon)}{2}
		\left(\Pr({\cal E}_{01}\wedge{\cal E}_{11})  
		+ \Pr({\cal E}_{10})\right).
	\end{align*}
	According to \cref{thm:prob}, we have
	\begin{align*}
		\lim_{n\to\infty}\Pr(\bar{\cal E}_{11}) 
		&= \tau(\{n-1,n\},01)\nonumber\\
		&= \tau(n-1,0)-\tau(n,0),
	\end{align*}
	\begin{align*}
		\lim_{n\to\infty}\Pr(\bar{\cal E}_{01}\wedge\bar{\cal E}_{10})
		&= \tau(n,0) - (1-\tau(n-1,0))\nonumber\\
		&= \tau(n-1,0)+\tau(n,0) - 1,
	\end{align*}
	and
	\begin{align*}
		\lim_{n\to\infty}\Pr({\cal E}_{10}) 
		&= \lim_{n\to\infty}\Pr({\cal E}_{01}\wedge{\cal E}_{11})\\
		&= 1-\tau(\{n-1,n\},01) - \tau(n,0)\\
		&= 1-\tau(n-1,0).
	\end{align*}
	Consequently, we obtain 	
	\begin{align}
		\lim_{n\to\infty}\Pr(e|01) 
		&= \epsilon(2\tau(n-1,0)-1) + \epsilon(3-\epsilon)(1-\tau(n-1,0))\nonumber\\
		&= (2-\tau(n-1,0))\epsilon - (1-\tau(n-1,0))\epsilon^2.\label{eq4}
	\end{align}

	Especially, when $2q=\sqrt{2}$, we have $(1-\tau(n-1,0))=\tau(n,0)$ and the universal set $(0,2^{nr}-1]$ is fully covered by ${\cal I}(n,1)$ and ${\cal I}(n-1,0)$, as shown by \cref{subfig:sqrt201}. Hence $\bar{\cal E}_{01}\wedge\bar{\cal E}_{10}=0$ and
	\begin{align*}
		\Pr(e|01) = \epsilon\Pr(\bar{\cal E}_{11}) +  
		\tfrac{\epsilon(3-\epsilon)}{2}\left(\Pr({\cal E}_{01}\wedge{\cal E}_{11}) + 
		\Pr({\cal E}_{10})\right).
	\end{align*}	
	By observing \cref{subfig:sqrt201}, we can find 
	\begin{align*}
		\Pr({\cal E}_{01}\wedge{\cal E}_{11}) + \Pr({\cal E}_{10}) = \Pr({\cal E}_{11}). 
	\end{align*}
	Therefore,
	\begin{align*}
		\Pr(e|01) 
		&= \epsilon\Pr(\bar{\cal E}_{11}) + \tfrac{\epsilon(3-\epsilon)}{2}\Pr({\cal E}_{11})\\
		&= \epsilon + \tfrac{\epsilon(1-\epsilon)}{2}\Pr({\cal E}_{11}).
	\end{align*}
	According to \cref{thm:prob}, we have
	\begin{align*}
		\lim_{n\to\infty}\Pr({\cal E}_{11}) = 2(1-\tau(n-1,0)) = 2\tau(n,0).
	\end{align*}
	
	\item{$\sqrt{2}<2q<2$}: It is easy to know $(1-\tau(n-1,0))>\tau(n,0)$. Hence, we have ${\cal I}(n,1)\cap{\cal I}(n-1,0)\neq\emptyset$, $\bar{{\cal E}}_{01}\wedge\bar{{\cal E}}_{10}=0$, and
	\begin{align*}
		\Pr(e|01) =\; 
		&\epsilon(2-\epsilon)\Pr({\cal E}_{01}\wedge{\cal E}_{10}) + 
		\epsilon\Pr(\bar{\cal E}_{11})\;+\\  
		&\tfrac{\epsilon(3-\epsilon)}{2}
		\left(\Pr({\cal E}_{01}\wedge\bar{\cal E}_{10}\wedge{\cal E}_{11})  
		+ \Pr(\bar{\cal E}_{01}\wedge{\cal E}_{10})\right).
	\end{align*}	
	By observing \cref{subfig:large01}, as $n\to\infty$,
	\begin{align*}
		\begin{cases}
			\Pr(\bar{\cal E}_{01}\wedge{\cal E}_{10}\} = \Pr({\cal E}_{01}\wedge\bar{\cal E}_{10}\wedge{\cal E}_{11}\} \to \tau(n,0)\\
			\Pr({\cal E}_{01}\wedge{\cal E}_{10}) \to 1-\tau(n-1,0) - \tau(n,0)\\
			\Pr(\bar{\cal E}_{11}) \to \tau(n-1,0)-\tau(n,0).
		\end{cases}
	\end{align*}
	Therefore, 
	\begin{align*}
		\lim_{n\to\infty}\Pr(e|01) =\; 
		&\epsilon(\tau(n-1,0)-\tau(n,0)) + \epsilon(3-\epsilon)\tau(n,0)\;+\nonumber\\ 
		&\epsilon(2-\epsilon)(1-\tau(n-1,0)-\tau(n,0)),
	\end{align*}
	which can be reduced to	
	\begin{align}
		\lim_{n\to\infty}\Pr(e|01) = (2-\tau(n-1,0))\epsilon - (1-\tau(n-1,0))\epsilon^2.\label{eq5}
	\end{align}	
\end{itemize}
We can merge \eqref{eq:pre01small}, \eqref{eq4}, and \eqref{eq5} into \eqref{eq:e01}.
\end{proof}

\subsection{FER for Two Unknown Ending Symbols}
\begin{theorem}[FER for Two Unknown Ending Symbols]
	Let $\alpha\triangleq(1+2^r-2^{2r})^+$. For almost every $0<r<1$, we have
	\begin{align}\label{eq:ex2}
		\lim_{n\to\infty}\Pr(e|X^{n-2}) = (4-2^{2r}+\alpha)\epsilon/2 - \alpha\epsilon^2.
	\end{align}	
\end{theorem}
\begin{proof}
According to \eqref{eq:e00} and \eqref{eq:e01}, we have
\begin{align*}
	\lim_{n\to\infty}\left(\Pr(e|0^2)+\Pr(e|01)\right) = \beta\epsilon - 2\epsilon^2(1-\tau(n-1,0))^+,
\end{align*}
where
\begin{align*}
	\beta = 3-\tau(n,0)-\tau(n-1,0)+(1-\tau(n-1,0))^+.
\end{align*}
Therefore,
\begin{align*}
	\lim_{n\to\infty}\Pr(e|X^{n-2}) = \beta\epsilon/2 - \epsilon^2(1-\tau(n-1,0))^+.
\end{align*}
Since $\tau(n,0) = 2^r-1$ and $\tau(n-1,0) = 2^r\tau(n,0)=2^{2r}-2^r$, we have
\begin{align*}
	\begin{cases}
		3-\tau(n,0)-\tau(n-1,0) = 4-2^{2r}\\
		1-\tau(n-1,0) = 1+2^r-2^{2r}.
	\end{cases}
\end{align*}
Hence, $\beta=(4-2^{2r}+\alpha)$. Now \eqref{eq:ex2} follows immediately.
\end{proof}

\begin{remark}[General Case]
	After understanding the above analyses, we can extend the method to the general case. Given $X^{n-t}$ known at the decoder, there are $2^t-1$ spare codewords and we can enumerate $2^{2^t-1}$ event combinations. For each event combination, we enumerate $2^{t-1}$ symbol combinations. Finally, we make use of \cref{thm:prob} to calculate the FER. Clearly, the order of complexity is $O(2^{2^t-1}\cdot 2^{t-1}) = O(2^{2^t+t-2})$. As $t$ increases, the complexity will go up sharply. However, as shown by the above analysis, this complexity can be significantly reduced by making use of the properties of coexisting interval, and a concise expression of FER can be deduced.
\end{remark}

\section{Experimental Verification}\label{sec:example}
\begin{figure}[!t]
	\centering
	\subfigure[]{\includegraphics[width=.5\linewidth]{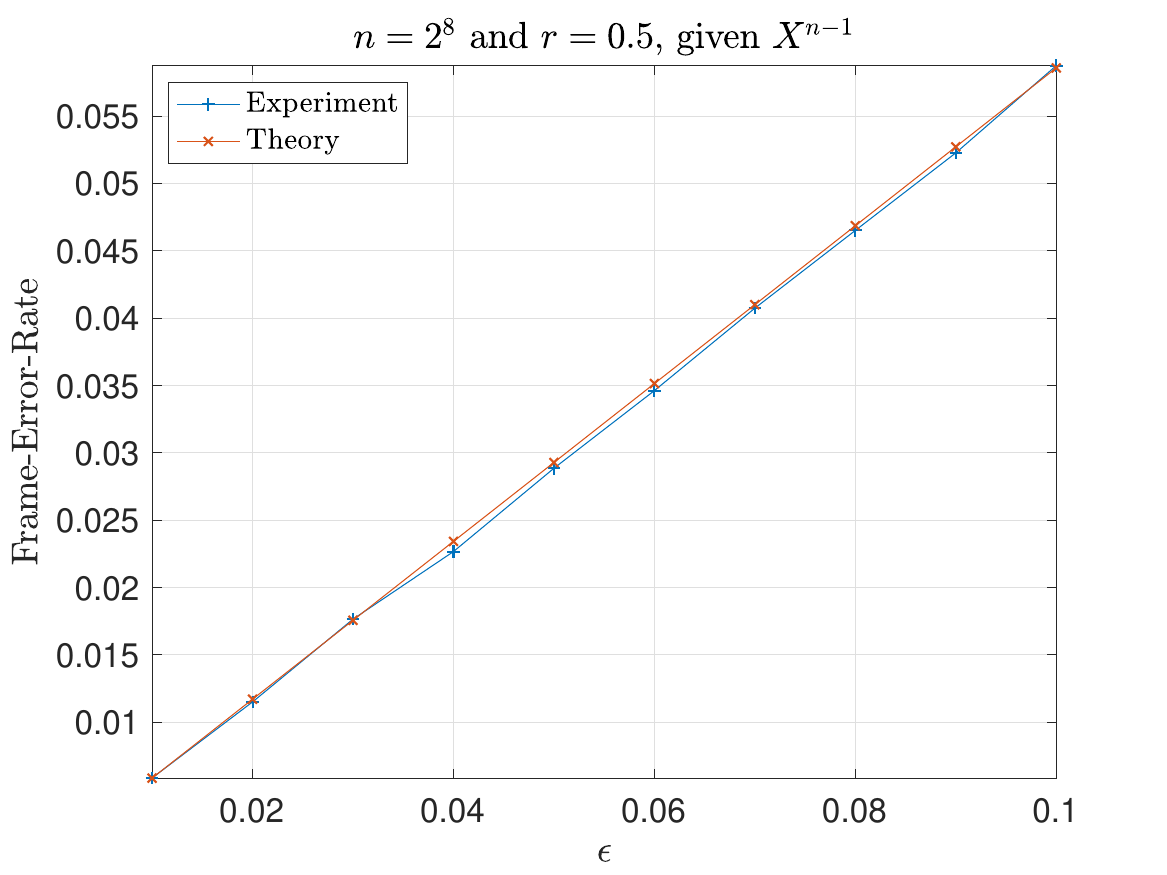}\label{subfig:fer1.5}}%
	\subfigure[]{\includegraphics[width=.5\linewidth]{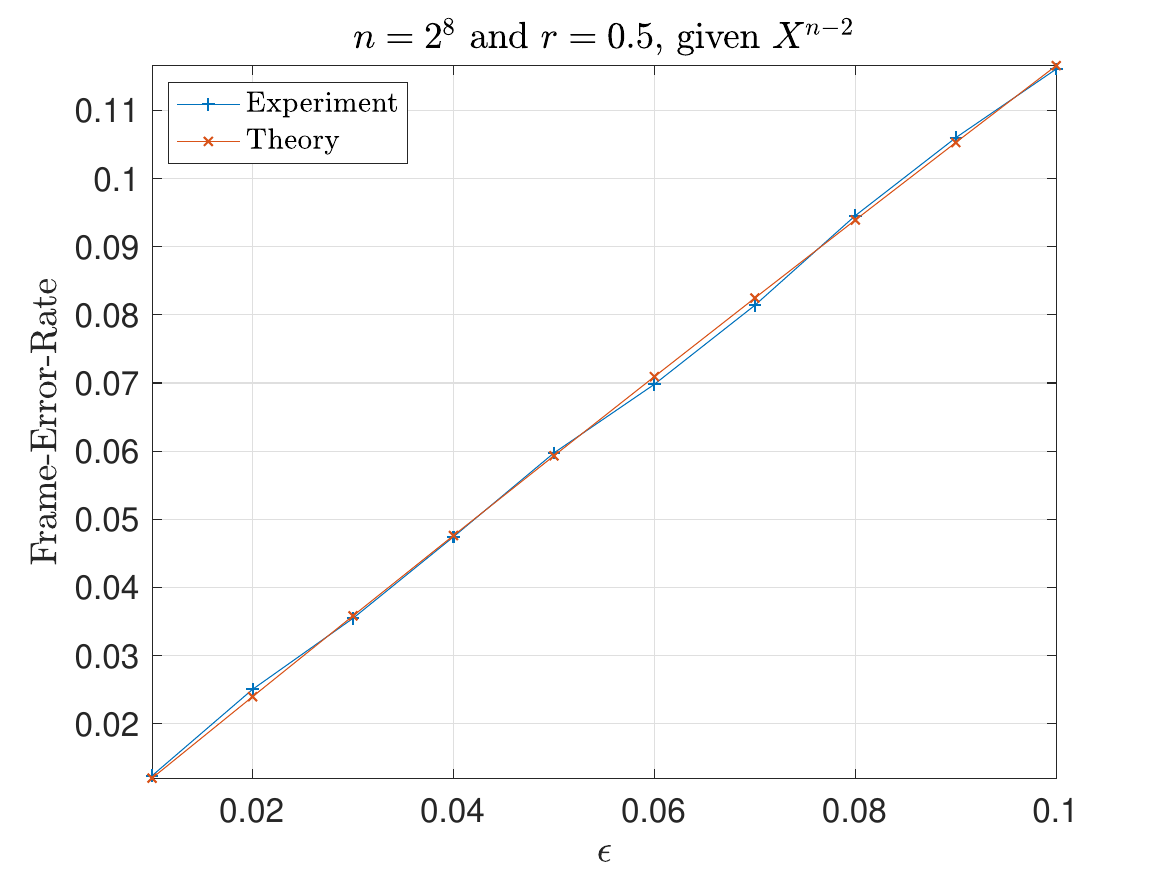}\label{subfig:fer2.5}}\\
	\subfigure[]{\includegraphics[width=.5\linewidth]{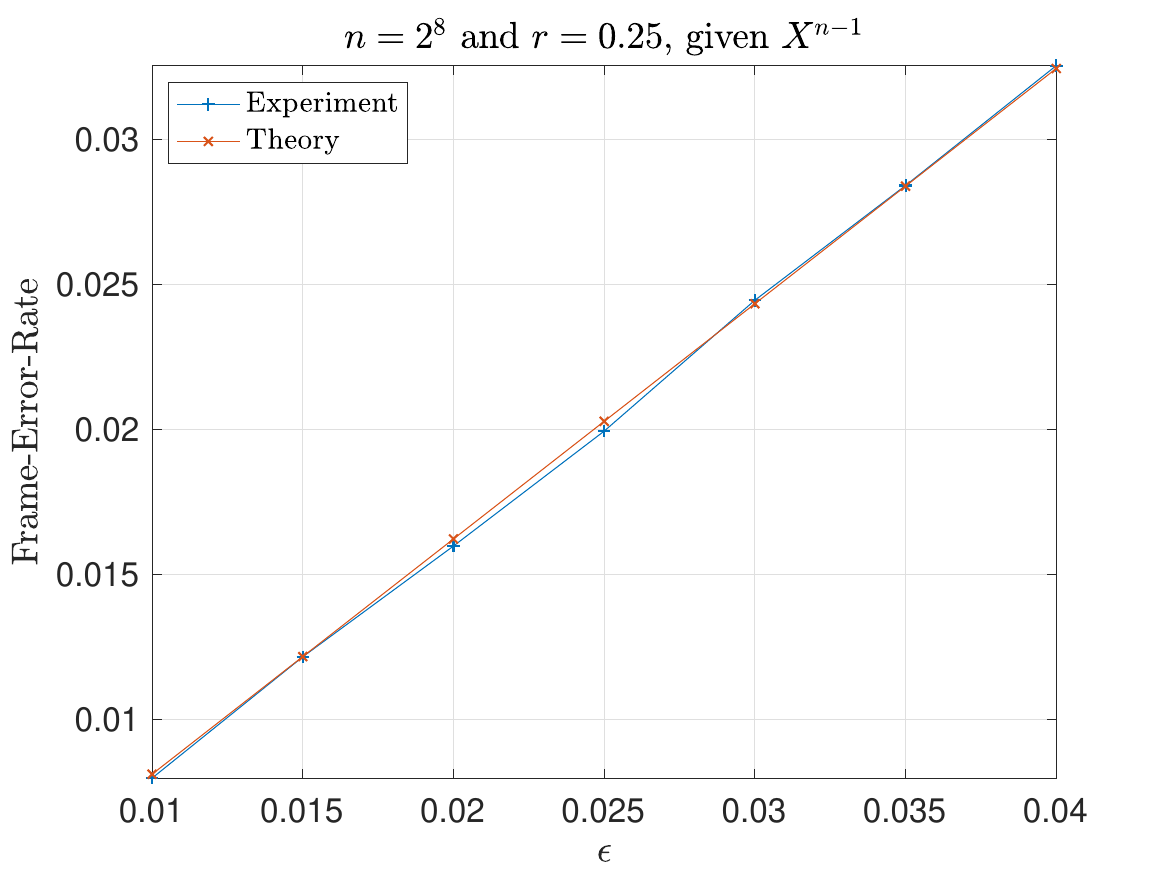}\label{subfig:fer1.25}}%
	\subfigure[]{\includegraphics[width=.5\linewidth]{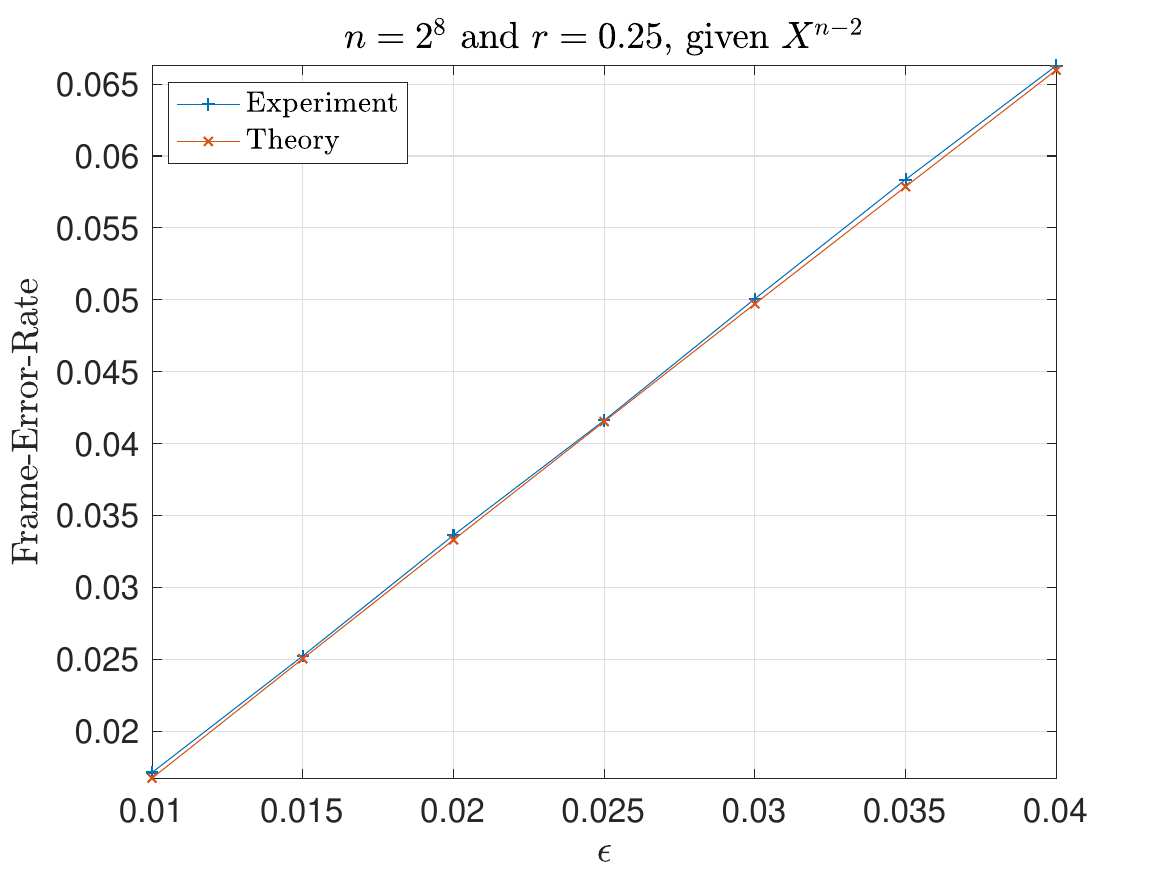}\label{subfig:fer2.25}}\\
	\caption{Theoretical and empirical FER results of overlapped arithmetic codes. (a) and (c) $\Pr(\hat{X}^n\neq X^n|Y^n,X^{n-1})$. (b) and (d) $\Pr(\hat{X}^n\neq X^n|Y^n,X^{n-2})$.}
	\label{fig:results}
\end{figure}

We consider two rates $r=1/2$ and $r=1/4$. For $r=1/2$, according to \eqref{eq:Pexn1}, we have
\begin{align*}
	\lim_{n\to\infty}\Pr(e|X^{n-1}) = (2-\sqrt{2})\epsilon\approx0.5858\epsilon; 
\end{align*}
while from \eqref{eq:ex2}, we can obtain (note $1+2^r-2^{2r}>0$ for $r=1/2$) 
\begin{align*}
	\lim_{n\to\infty}\Pr(e|X^{n-2}) 
	&= (1+\sqrt{2})\epsilon/2 - (\sqrt{2}-1)\epsilon^2\\
	&\approx 1.2071\epsilon - 0.4142\epsilon^2.
\end{align*}
When $r=1/4$, we have
\begin{align*}
	\lim_{n\to\infty}\Pr(e|X^{n-1}) = (2-2^{1/4})\epsilon\approx0.8108\epsilon
\end{align*}
and (note $1+2^r-2^{2r}>0$ for $r=1/4$)
\begin{align*}
	\lim_{n\to\infty}\Pr(e|X^{n-2}) 
	&= (5-2\sqrt{2}+2^{1/4})\epsilon/2 - (1-\sqrt{2}+2^{1/4})\\
	&\approx 1.6804\epsilon - 0.775\epsilon^2.
\end{align*}
To verify the correctness of the above analyses, we compare theoretical results with empirical ones in \cref{fig:results}. Experimental settings are shown in figure titles. It can be seen that the theoretical results closely match the empirical ones, providing strong supporting evidences for the correctness of our theoretical analyses.

\chapter{Hamming Distance Spectrum}\label{c-hds}
\vspace{-25ex}%
We have known in \cref{c-ccs} that overlapped arithmetic codes are actually a kind of nonlinear coset codes partitioning source space into unequal-sized cosets, and {\em Coset Cardinality Spectrum} (CCS) is a fundamental property of overlapped arithmetic codes that can describe the distribution of coset cardinality. However, this is not enough. Just as channel codes, another fundamental property of overlapped arithmetic codes is how far (in Hamming distance) the codewords within the same coset keep away from each other. To avoid verbosity, we introduce the terminology of \textit{mate codewords} as below.
\begin{definition}[Mate Codeword]
	We refer to the codewords belonging to the same coset as mate codewords, or just mates for short. In mathematical language, if $m(x^n)=m(y^n)$, where
	\begin{align*}
		m(x^n) = \lceil s(x^n)\rceil = \lceil (1-2^{-r})\sum_{i=1}^{n}x_i2^{ir} \rceil,
	\end{align*}
	as defined by \eqref{eq:ell} and \eqref{eq:mXn}, then $x^n$ and $y^n$ are mates.	
\end{definition}
We introduce the {\em Hamming Distance Spectrum} (HDS) to quantitatively measure the distribution of inter-mate Hamming distance, which will be helpful for us to understand overlapped arithmetic codes.

\section{Definition of HDS}\label{subsec:toy}
\begin{definition}[Codeword HDS]
	Let $n$ be code length. For $d\in[0:n]\triangleq\{0,1,\dots,n\}$, the HDS of codeword $x^n\in\mathbb{B}^n$ is defined as 
	\begin{align*}
		k_d(x^n) \triangleq \left|\left\{y^n: (y^n\in\mathbb{B}^n)\wedge (m(x^n)=m(y^n))\wedge(|x^n\oplus y^n|=d)\right\}\right|,
	\end{align*}
	where $|x^n\oplus y^n|=d$ means that $y^n$ is $d$-away from $x^n$.
\end{definition}
In plain words, $k_d(x^n)$ is the number of mate codewords $d$-away from $x^n$. Actually, for linear block codes, every codeword has the same HDS, \textit{i.e.}, $k_d(x^n)\equiv k_d(0^n)$ for any $x^n\in\mathbb{B}^n$. However, for overlapped arithmetic codes, every codeword may have different HDS.
\begin{lemma}[Properties of Codeword HDS]
	Let ${\cal C}_m$ be the coset containing $x^n$, whose cardinality is denoted by $|{\cal C}_m|$, where $m\in[0:2^{nr})$. It is easy to obtain the following properties of $k_d(x^n)$.
	\begin{itemize}
		\item $0\leq k_d(x^n)\leq\binom{n}{d}$.
		\item For $1\leq d\leq n$, the sum $\sum_{x^n\in{\cal C}_m}{k_d(x^n)}$ is twice the number of $d$-away codeword-pairs in ${\cal C}_m$. Therefore,
		\begin{align}\label{eq:sumk1}
			\sum_{d=1}^{n}{\sum_{x^n\in{\cal C}_m}{k_d(x^n)}} = 2\binom{|{\cal C}_m|}{2} = |{\cal C}_m|^2 - |{\cal C}_m|.
		\end{align}
		\item If we define $k_0(x^n)=1$, then $\sum_{d=0}^{n}{k_d(x^n)} = |{\cal C}_{m}|$ and further
		\begin{align}\label{eq:sumk}
			\sum_{x^n\in{\cal C}_m}{\sum_{d=0}^{n}{k_d(x^n)}} = |{\cal C}_m|^2.
		\end{align}
	\end{itemize}	
\end{lemma}

\begin{definition}[Coset HDS]
	The HDS of the $m$-th coset is
	\begin{align*}
		\phi_m(d) \triangleq \sum_{x^n\in{\cal C}_m}{\Pr(X^n=x^n|X^n\in{\cal C}_m)\cdot k_d(x^n)}.
	\end{align*}
\end{definition}

\begin{definition}[Code HDS]
	For $0\leq d\leq n$, the HDS of an overlapped arithmetic code is defined as
	\begin{align*}
		\psi(d;n) \triangleq \sum_{x^n\in\mathbb{B}^n}{\Pr(X^n=x^n)\cdot k_d(x^n)}.
	\end{align*}
\end{definition}

For uniform binary sources, \textit{i.e.}, $\Pr(X^n=x^n)\equiv2^{-n}$, we have
\begin{align}\label{eq:phidef}
	\phi_m(d) = \tfrac{1}{|{\cal C}_m|}\sum_{x^n\in{\cal C}_m}{k_d(x^n)}
\end{align}
and
\begin{align}\label{eq:psinddef}
	\psi(d;n) = 2^{-n}\sum_{x^n\in\mathbb{B}^n}{k_d(x^n)} 
	&= 2^{-n}\sum_{m=0}^{2^{nr}-1}\sum_{x^n\in{\cal C}_m}{k_d(x^n)}\nonumber\\
	&= 2^{-n}\sum_{m=0}^{2^{nr}-1}|{\cal C}_m|\cdot\phi_m(d).
\end{align}
From \eqref{eq:sumk} and \eqref{eq:phidef}, we have $\sum_{d=0}^{n}\phi_m(d)=|{\cal C}_m|$.

\begin{definition}[Asymptotic Code HDS]
	For $0\leq d\leq n$, the asymptotic HDS of an overlapped arithmetic code is
	\begin{align}
		\psi(d) \triangleq \lim_{n\rightarrow\infty}{\psi(d;n)}.\nonumber
	\end{align}
\end{definition}

For linear block codes, since $k_d(x^n)=k_d(0^n)$ for every $x^n\in\mathbb{B}^n$, it is easy to obtain $\psi(d;n)=k_d(0^n)$ for every $d\in[0:n]$. That is, the HDS of a linear block code is equivalent to its weight spectrum. However, for overlapped arithmetic codes, the problem is much more complex.

\begin{table*}[!t]
	\caption{Example of Codeword HDS}\label{tab:cwhds}\centering
	\begin{tabular}{c||c||c||c||c||c||c}
		\hline
		$x^n$  					& $s(x^n)$ 					& $m$ 			   & $k_1(x^n)$		& $k_2(x^n)$		& $k_3(x^n)$		 & $k_4(x^n)$\\			
		\hline\hline
		\color{blue}{$0000$} 	& \color{blue}{$0.0000$} 	& \color{blue}{0} 	&\color{blue}{0} 	&\color{blue}{0} 	 &\color{blue}{0} 	&\color{blue}{0}\\
		\hline
		\color{red}{$0001$} 	& \color{red}{$0.4142$} 	& \color{red}{1}	&\color{red}{1}		&\color{red}{2}		&\color{red}{0}		 &\color{red}{0}\\
		\hline
		\color{red}{$0010$} 	& \color{red}{$0.5858$} 	& \color{red}{1} 	&\color{red}{1}		&\color{red}{2}		&\color{red}{0}		 &\color{red}{0}\\
		\hline
		\color{red}{$0011$} 	& \color{red}{$1.0000$} 	& \color{red}{1}	&\color{red}{2}		&\color{red}{0}		&\color{red}{1}		 &\color{red}{0}\\
		\hline
		\color{red}{$0100$} 	& \color{red}{$0.8284$} 	& \color{red}{1}	&\color{red}{0}		&\color{red}{2}		&\color{red}{1}		 &\color{red}{0}\\
		\hline
		\color{green}{$0101$} 	& \color{green}{$1.2426$} 	& \color{green}{2}	&\color{green}{1}	&\color{green}{3}	 &\color{green}{1}	&\color{green}{1}\\
		\hline
		\color{green}{$0110$} 	& \color{green}{$1.4142$} 	& \color{green}{2}	&\color{green}{1}	&\color{green}{3}	 &\color{green}{1}	&\color{green}{1}\\
		\hline
		\color{green}{$0111$} 	& \color{green}{$1.8284$} 	& \color{green}{2}	&\color{green}{2}	&\color{green}{0}	 &\color{green}{3}	&\color{green}{1}\\
		\hline
		\color{green}{$1000$} 	& \color{green}{$1.1716$}	& \color{green}{2}	&\color{green}{3}	&\color{green}{0}	 &\color{green}{2}	&\color{green}{1}\\
		\hline
		\color{green}{$1001$} 	& \color{green}{$1.5858$} 	& \color{green}{2}	&\color{green}{1}	&\color{green}{3}	 &\color{green}{1}	&\color{green}{1}\\
		\hline
		\color{green}{$1010$} 	& \color{green}{$1.7574$} 	& \color{green}{2}	&\color{green}{1}	&\color{green}{3}	 &\color{green}{1}	&\color{green}{1}\\
		\hline
		\color{cyan}{$1011$} 	& \color{cyan}{$2.1716$} 	& \color{cyan}{3}	&\color{cyan}{1}	&\color{cyan}{2}	 &\color{cyan}{0}	&\color{cyan}{0}\\
		\hline	
		\color{green}{$1100$} 	& \color{green}{$2.0000$} 	& \color{green}{2}	&\color{green}{1}	&\color{green}{4}	 &\color{green}{1}	&\color{green}{0}\\
		\hline
		\color{cyan}{$1101$} 	& \color{cyan}{$2.4142$} 	& \color{cyan}{3}	&\color{cyan}{1}	&\color{cyan}{2}	 &\color{cyan}{0}	&\color{cyan}{0}\\
		\hline	
		\color{cyan}{$1110$} 	& \color{cyan}{$2.5858$} 	& \color{cyan}{3}	&\color{cyan}{1}	&\color{cyan}{2}	 &\color{cyan}{0}	&\color{cyan}{0}\\
		\hline	
		\color{cyan}{$1111$} 	& \color{cyan}{$3.0000$} 	& \color{cyan}{3}	&\color{cyan}{3}	&\color{cyan}{0}	 &\color{cyan}{0}	&\color{cyan}{0}\\
		\hline
		Sum 					& --- 						& ---				&20					&28					&12					 &6\\
		\hline
		$\psi(d;n)$ 		& --- 						& ---									&$\frac{20}{16}=\frac{5}{4}$					&$\frac{28}{16}=\frac{7}{4}$		&$\frac{12}{16}=\frac{3}{4}$			&$\frac{6}{16}=\frac{3}{8}$\\
		\hline
	\end{tabular}
	\label{tab:codewordhds}
\end{table*}	

\begin{table}[!t]
	\caption{Example of Coset HDS and Code HDS}\label{tab:cbhds}\centering
	\begin{tabular}{c||c||c||c||c||c||c}
		\hline
		Term 			& $d=0$ & $d=1$  	& $d=2$		& $d=3$		& $d=4$		& Sum\\				
		\hline\hline
		$\phi_0(d)$		& $1$ 	& $0$ 	 	& $0$		& $0$		& $0$		& $1$\\
		\hline
		$\phi_1(d)$ 	& $1$ 	& $4/4$	 	& $6/4$		& $2/4$		& $0$		& $4$\\
		\hline
		$\phi_2(d)$ 	& $1$ 	& $10/7$ 	& $16/7$	& $10/7$	& $6/7$		& $7$\\
		\hline
		$\phi_3(d)$ 	& $1$ 	& $6/4$	 	& $6/4$		& $0$		& $0$		& $4$\\
		\hline
		$\psi(d;n)$ & $1$ 	& $20/16$ 	& $28/16$ 	& $12/16$ 	& $6/16$ 	& $5.125$\\
		\hline
	\end{tabular}
	\label{tab:codehds}
\end{table}	

\begin{example}
	Let us use the code in \cref{fig:part}, where $n=4$ and $r=0.5$, to explain the HDS. This code partitions source space $\mathbb{B}^4$ into 4 cosets:
	\begin{itemize}
		\item ${\cal C}_0 = \{\underline{0000}\}$,
		\item ${\cal C}_1 = \{\underline{0001}, \underline{0010}, \underline{0011}, \underline{0100}\}$,
		\item ${\cal C}_2 = \{\underline{0101}, \underline{0110}, \underline{0111}, \underline{1000}, \underline{1001}, \underline{1010}, \underline{1100}\}$, and
		\item ${\cal C}_3 = \{\underline{1011}, \underline{1101}, \underline{1110}, \underline{1111}\}$,
	\end{itemize}
	where $\underline{b_1\cdots b_n}\in\mathbb{B}^n$ denotes a codeword of $n$ bits. Let us take ${\cal C}_1$ as an example. As shown by \cref{tab:cwhds}, 
	\begin{itemize}
		\item $k_1(\underline{0001})=1$, $k_2(\underline{0001})=2$, $k_3(\underline{0001})=0$, and $k_4(\underline{0001})=0$; 
		\item $k_1(\underline{0010})=1$, $k_2(\underline{0010})=2$, $k_3(\underline{0010})=0$, and $k_4(\underline{0010})=0$; 
		\item $k_1(\underline{0011})=2$, $k_2(\underline{0011})=0$, $k_3(\underline{0011})=1$, and $k_4(\underline{0011})=0$; 
		\item $k_1(\underline{0100})=0$, $k_2(\underline{0100})=2$, $k_3(\underline{0100})=1$, and $k_4(\underline{0100})=0$.
	\end{itemize}
	Let $k_0(x^n)=1$. Then $\sum_{d=0}^{4}{k_d(x^n)}=|{\cal C}_1|=4$ for every $x^n\in{\cal C}_1$. In ${\cal C}_1$,
	\begin{itemize}
		\item there are two $1$-away codeword pairs, \textit{i.e.}, 
		\begin{align*}
			(\underline{0001},\underline{0011}) \;{\rm and}\; (\underline{0010}, \underline{0011}), 
		\end{align*}
		while we have 
		\begin{align*}
			\sum_{x^n\in{\cal C}_1}{k_1(x^n)}=1+1+2+0=4, 
		\end{align*}
		twice the number of $1$-away codeword-pairs; 
		\item there are three $2$-away codeword-pairs, \textit{i.e.}, 
		\begin{align*}
			(\underline{0001},\underline{0010}), (\underline{0001},\underline{0100}), \;{\rm and}\;  
			(\underline{0010}, \underline{0100}), 
		\end{align*}
		while we have
		\begin{align*}
			\sum_{x^n\in{\cal C}_1}{k_2(x^n)}=2+2+0+2=6, 
		\end{align*}
		twice the number of $2$-away codeword-pairs; and 
		\item there is one $3$-away codeword-pair (\underline{0011},\underline{0100}), while
		\begin{align*}
			\sum_{x^n\in{\cal C}_1}{k_3(x^n)}=0+0+1+1=2, 
		\end{align*}
		twice the number of $3$-away codeword-pairs. 
	\end{itemize}
	Furthermore, we have
	\begin{align}
		\sum_{d=1}^{4}\sum_{x^n\in{\cal C}_1}{k_d(x^n)}=4+6+2+0=12=2\binom{|{\cal C}_1|}{2}=|{\cal C}_1|^2-|{\cal C}_1|,\nonumber
	\end{align}
	verifying \eqref{eq:sumk1}, and 
	\begin{align}
		\sum_{d=0}^{4}\sum_{x^n\in{\cal C}_1}{k_d(x^n)}=4+12=16=|{\cal C}_1|^2,\nonumber
	\end{align}
	verifying \eqref{eq:sumk}. According to \eqref{eq:psinddef}, we can easily obtain $\psi(d;n)$, as shown in the last rows of \cref{tab:codewordhds} and \cref{tab:codehds}. Hence,
	\begin{align*}
		\sum_{d=0}^{4}{\psi(d;4)}=82/16=5.125>4=2^{4(1-1/2)}. 
	\end{align*}
	In this example, we find $\sum_{d=0}^{n}{\psi(d;n)} > 2^{n(1-r)}$, which is not a coincidence, just as shown below.
\end{example}

\begin{lemma}[Sum of HDS]\label{lem:sumhds}
	Let $f(u)$ be the asymptotic initial CCS given by \eqref{eq:asympt}. Then
	\begin{align}
		\lim_{n\to\infty}\frac{\sum_{d=0}^{n}{\psi(d;n)}}{2^{n(1-r)}} = \int_{0}^{1}{f^2(u)du}.
	\end{align}
\end{lemma}
\begin{proof}
	According to \eqref{eq:sumk} and \eqref{eq:psinddef}, we will obtain
	\begin{align}
		\sum_{d=0}^{n}{\psi(d;n)} 
		&= 2^{-n}\sum_{m=0}^{2^{nr}-1}{|{\cal C}_m|^2}\nonumber\\
		&= 2^{n(1-r)}\sum_{m=0}^{2^{nr}-1}{\left(\frac{|{\cal C}_m|}{2^{n(1-r)}}\right)^2}2^{-nr}.\nonumber
	\end{align}
	It can be written as
	\begin{align}
		\lim_{n\to\infty}\frac{\sum_{d=0}^{n}{\psi(d;n)}}{2^{n(1-r)}} 
		= \lim_{n\to\infty}\sum_{m=0}^{2^{nr}-1}{\left(\frac{|{\cal C}_m|}{2^{n(1-r)}}\right)^2}2^{-nr}
		\stackrel{(a)}{=} \int_{0}^{1}{f^2(u)du},\nonumber
	\end{align}
	where $(a)$ comes from the definition of $f(u)$.
\end{proof}	

\begin{lemma}[Convexity of HDS]
	For an overlapped arithmetic code with length $n$ and rate $r$, we have 
	\begin{align}\label{eq:conv}
		\sum_{d=0}^{n}{\psi(d;n)} \geq 2^{n(1-r)},
	\end{align}
	where the equality holds iff source space $\mathbb{B}^n$ is equally partitioned into $2^{nr}$ cosets of cardinality $2^{n(1-r)}$.
\end{lemma}
\begin{proof}
	Since $\int_{0}^{1}{f^2(u)du}$ is a nonnegative and convex function in $f(u)$, we have $\int_{0}^{1}{f^2(u)du}\geq1$ and the equality holds iff $f(u)$ is uniform over $[0,1)$, \textit{i.e.}, source space $\mathbb{B}^n$ is equally partitioned into $2^{nr}$ cosets of cardinality $2^{n(1-r)}$. This lemma then follows \cref{lem:sumhds} immediately.
\end{proof}

\section{Two Toy Methods to Calculate HDS}\label{subsec:toy}
\begin{remark}[Exhaustive Enumeration]
	As shown by \eqref{eq:psinddef}, to obtain $\psi(d;n)$ for all $d\in[0:n]$, we should try every $x^n\in\mathbb{B}^n$ and every $y^n\in\mathbb{B}^n$, so the total complexity is ${\cal O}(2^n\times2^n)={\cal O}(4^n)$. More concretely, the computing complexity of $\psi(d;n)$ varies for different $d$. Let $z^n=x^n\oplus y^n\in\mathbb{B}^n$ and $|x^n\oplus y^n|=|z^n|=d$. For convenience, we define $z_{j^d}^n\in\mathbb{B}^n$ as a length-$n$ binary block with $z_{j^d}=1^d$ and $z_{[n]\setminus j^d}=0^{n-d}$, where $j^d$ and $[n]$ are defined by \eqref{eq:jd}. Let ${\bf 1}_A$ be the well-known indicator function equal to either $1$ if $A$ is true, or $0$ if $A$ is false. From \eqref{eq:psinddef}, we have
	\begin{align}\label{eq:psind}
		\psi(d;n) = 2^{-n}\sum_{j^d\in{\cal J}_{n,d}} \sum_{x^n\in\mathbb{B}^n}{\bf 1}_{m(x^n)=m(x^n\oplus z_{j^d}^n)},
	\end{align}
	where ${\cal J}_{n,d}$ is defined by \eqref{eq:Jnd}. Obviously, the computing complexity of \eqref{eq:psind} is ${\cal O}(2^n\binom{n}{d})$, following the binomial distribution w.r.t. $d$, extremely huge and unacceptable for every $d\in[0:n]$. Hence, we are badly in need of a fast method to calculate $\psi(d;n)$.
\end{remark}

\begin{theorem}[Binomial Approximation of HDS]\label{thm:coarsehds}
	If we take overlapped arithmetic codes as random codes, then $\psi(d;n)$ for $0\leq d\leq n$ obeys the following binomial distribution 
	\begin{align}\label{eq:coarsehds}
		\psi(d;n) \approx \binom{n}{d}\cdot 2^{-nr} \cdot \int_{0}^{1}{f^2(u)\,du}.
	\end{align}
\end{theorem}
\begin{proof}
	As shown by \cref{lem:sumhds},
	\begin{align}
		\sum_{d=0}^{n}{\psi(d;n)} \approx 2^{n(1-r)}\cdot \int_{0}^{1}{f^2(u)du}.\nonumber
	\end{align}
	If $\psi(d;n)$ for $0\leq d\leq n$ obeys the binomial distribution, then
	\begin{align}
		\psi(d;n) \approx \frac{\binom{n}{d}\cdot 2^{n(1-r)}\cdot \int_{0}^{1}{f^2(u)du}}{\sum_{d=0}^{n}\binom{n}{d}}.\nonumber
	\end{align}
	As we know, $\sum_{d=0}^{n}\binom{n}{d}=2^n$. Now \eqref{eq:coarsehds} follows immediately.
\end{proof}

According to \eqref{eq:coarsehds}, the complexity of $\psi(d;n)$ is ${\cal O}(1)$, hence the total complexity is ${\cal O}(n)$ for all $d\in[0:n]$. Despite its low complexity, one serious drawback of \eqref{eq:coarsehds} is its low accuracy. Because overlapped arithmetic codes are not random codes, $\psi(d;n)$ for $0\leq d\leq n$ does not strictly follow the binomial distribution, especially at the two ends of $[0:n]$. Hence with \eqref{eq:coarsehds}, we can obtain only a coarse approximation of $\psi(d;n)$. For example, \eqref{eq:coarsehds} returns $2^{-nr}\int_{0}^{1}{f^2(u)\,du}$ for $d=0$, while $\psi(0;n)=1$ actually. It will be verified in \cref{sec:hdsexample} that \eqref{eq:coarsehds} works well only for $d\approx n/2$, while performs poorly in other cases, which motivates us to look for a more accurate method. 

\section{Soft Approximation of HDS for $d<n$}\label{subsec:hdscal}
\begin{theorem}[Soft Approximation of HDS for $d<n$]\label{thm:softhds}
	Let us define $j^d$ as \eqref{eq:jd} and ${\cal J}_{n,d}$ as \eqref{eq:Jnd}. For almost every $0<r<1$ and for every $d<n$, as $n\to\infty$, we have
	\begin{align}\label{eq:softhds}
		\psi(d;n)
		\to 2^{-d}\sum_{j^d\in{\cal J}_{n,d}}{\sum_{b^d\in\mathbb{B}^d}\left(1-|\tau(j^d,b^d)|\right)^+},
	\end{align}
	where $\tau(j^d,b^d)$ is defined by \eqref{eq:tau} and $(\cdot)^+\triangleq\max(0,\cdot)$.
\end{theorem}
\begin{proof}
Let us define the conditional indicator function as
\begin{align*}
	{\bf 1}_{A|B} \triangleq 
	\begin{cases}
		1, & A{\rm~is~true~given}~B\\
		0, & A{\rm~is~false~give}~B.
	\end{cases}
\end{align*}
Then \eqref{eq:psind} can be written as
\begin{align*}
	\psi(d;n) = 2^{-d}\sum_{j^d\in{\cal J}_{n,d}}\sum_{b^d\in\mathbb{B}^d}{\psi(d;n|j^d,b^d)},
\end{align*}
where
\begin{align*}
	\psi(d;n|j^d,b^d) \triangleq 2^{-(n-d)}\sum_{x_{[n]\setminus j^d}\in\mathbb{B}^{n-d}}
	{\bf 1}_{m(x^n)=m(x^n\oplus z_{j^d}^n) | x_{j^d}=b^d}.
\end{align*}
Let us define the following binary random variable
\begin{align*}
	V(j^d,b^d) \triangleq {\bf 1}_{m(X^n)=m(X^n\oplus z_{j^d}^n) | X_{j^d}=b^d}.
\end{align*}
Then ${\bf 1}_{m(x^n)=m(x^n\oplus z_{j^d}^n) | x_{j^d}=b^d}$ is a realization of $V(j^d,b^d)$. According to \cref{thm:equiv}, we have
\begin{align*}
	\left\{m(x^n)=m(x^n\oplus z_{j^d}^n)\middle|x_{j^d}=b^d\right\} \leftrightarrow 
	\left\{s(x^n|x_{j^d}=b^d) \in {\cal I}(j^d,b^d)\right\},
\end{align*}
where $\{\cdot\}\leftrightarrow\{\cdot\}$ denotes two equivalent events, so $V(j^d,b^d)$ is a binary random variable with bias probability $\Pr\left\{E(j^d,b^d)\in{\cal I}(j^d,b^d)\right\}$, where $E(j^d,b^d)\triangleq s(X^n|X_{j^d}=b^d)$, as defined by \eqref{eq:E}. 
\begin{itemize}
	\item As $(n-d)\to\infty$, we have infinitely many independent realizations of $V(j^d,b^d)$ and thus according to the law of large numbers,
\begin{align}\label{eq:approx}
	\lim_{(n-d)\to\infty}\psi(d;n|j^d,b^d)
	&= \lim_{(n-d)\to\infty}\Pr\left\{E(j^d,b^d)\in{\cal I}(j^d,b^d)\right\}\nonumber\\
	&\stackrel{(a)}{=} \left(1-|\tau(j^d,b^d)|\right)^+,
\end{align}
where $(a)$ comes from \cref{thm:prob}. 

	\item For $1\leq (n-d)<\infty$, there are only finitely many realizations of $V(j^d,b^d)$, so \eqref{eq:approx} cannot be applied. However, by the generalized law of large numbers, for every $j^d\in{\cal J}_{n,d}$,
\begin{align}\label{eq:sumpsidnjdbd}
	\lim_{n\to\infty}\frac{\sum_{b^d\in\mathbb{B}^d}{\psi(d;n|j^d,b^d)}}{ \sum_{b^d\in\mathbb{B}^d}{\left(1-|\tau(j^d,b^d)|\right)^+}} = 1.
\end{align}
\end{itemize}
However, note that $V(j^d,b^d)$ is a function w.r.t. $(n-d)$ random variables $X_{[n]\setminus j^d}$ and will degenerate into a deterministic constant for $d=n$. Therefore, \eqref{eq:sumpsidnjdbd} does not hold for $d=n$. 
\end{proof}

We name \eqref{eq:softhds} as the \textit{Soft Approximation} of $\psi(d;n)$, to distinguish it from the \textit{Hard Approximation} of $\psi(d;n)$ named in the next sub-section. 

\begin{remark}[Accuracy]
Though \eqref{eq:softhds} is accurate for almost every $d$, there is an exception $\psi(n;n)$, which cannot be approximated by \eqref{eq:softhds} otherwise there will be a large deviation.
\end{remark}

\begin{remark}[Complexity]
	As we know, the cardinality of ${\cal J}_{n,d}$ is $\binom{n}{d}$, so the complexity to calculate $\psi(d;n)$ by \eqref{eq:softhds} is ${\cal O}(2^d\binom{n}{d})$. Compared with \eqref{eq:psind}, whose complexity is ${\cal O}(2^n\binom{n}{d})$, the complexity of \eqref{eq:softhds} is reduced by $2^{n-d}$ times. However, the complexity of \eqref{eq:softhds} goes up hyper-exponentially as $d$ increases. Therefore, even though \eqref{eq:softhds} is accurate for almost every $d$, it is feasible only for small $d$ in practice. It will be very desirable if the complexity of \eqref{eq:softhds} can be reduced for large $d$. Since $\sum_{d=0}^{n}2^d\binom{n}{d}=3^n$, the total complexity to calculate $\psi(d;n)$ for all $d\in[0:n]$ by \eqref{eq:softhds} is ${\cal O}(3^n)$. As a comparison, remember that the total complexity of the exhaustive enumeration \eqref{eq:psind} is ${\cal O}(4^n)$, and the total complexity of the binomial approximation \eqref{eq:coarsehds} is ${\cal O}(n)$. 
\end{remark}

\section{Hard Approximation of HDS for $1\ll d<n$}\label{sec:hdsn}
This sub-section will propose a variant of \eqref{eq:softhds} to calculate $\psi(d;n)$, which is called \textit{Hard Approximation}, just to distinguish it from the \textit{Soft Approximation} defined by \eqref{eq:softhds}. In essence, the {\em Hard Approximation} is an approximation of the {\em Soft Approximation} for $d\gg 1$. However, there is still an exception when $d=n$, which will be particularly discussed in the next sub-section. 
\begin{definition}[$j^d$-Active Set]
	The active set associated with $j^d\in{\cal J}_{n,d}$ is defined as
	\begin{align}
		{\cal B}_{j^d} \triangleq \{b^d: (b^d\in\mathbb{B}^d) \wedge (|\tau(j^d,b^d)|<1)\} \subseteq \mathbb{B}^d.
	\end{align}
\end{definition}

\begin{definition}[$j^d$-Universal Sequence and $j^d$-Active Sequence]
For every $j^d\in{\cal J}_{n,d}$, we define two sequences. One is the $j^d$-universal sequence
\begin{align}\label{eq:omegajd}
	\omega_{j^d} \triangleq \left(\tau(j^d,b^d)\right)_{b^d\in\mathbb{B}^d} \in (-2^{nr},2^{nr})^{|\omega_{j^d}|},
\end{align} 
and the other is the $j^d$-active sequence
\begin{align}\label{eq:omegajdb}
	\omega_{j^d,{\cal B}} \triangleq \left(\tau(j^d,b^d)\right)_{b^d\in{\cal B}_{j^d}} \in (-1,1)^{|\omega_{j^d,{\cal B}}|},
\end{align}
where
\begin{align*}
	|\omega_{j^d,{\cal B}}| = |{\cal B}_{j^d}| = \left(\sum_{b^d\in\mathbb{B}^d}{{\bf 1}_{|\tau(j^d,b^d)|<1}}\right) \leq |\omega_{j^d}| = 2^d.
\end{align*}
\end{definition}

\begin{definition}[$d$-Universal Sequence and $d$-Active Sequence]
For every $d\in[0:n]$, we define two sequences. One is the $d$-universal sequence
\begin{align}\label{eq:omegad}
	\omega_d \triangleq \left(\omega_{j^d}\right)_{j^d\in{\cal J}_{n,d}} \in (-2^{nr},2^{nr})^{|\omega_d|},	
\end{align}
and the other is the $d$-active sequence
\begin{align}\label{eq:omegadb}
	\omega_{d,{\cal B}} \triangleq \left(\omega_{j^d,{\cal B}}\right)_{j^d\in{\cal J}_{n,d}} \in (-1,1)^{|\omega_{d,{\cal B}}|},
\end{align}
where
\begin{align*}
	|\omega_{d,{\cal B}}| = \left(\sum_{j^d\in{\cal J}_{n,d}}|{\cal B}_{j^d}|\right) \leq |\omega_d| = \binom{n}{d}2^d.
\end{align*}
\end{definition}

We use $\omega_{j^d,{\cal B}}\subseteq\omega_{j^d}$ to denote that $\omega_{j^d,{\cal B}}$ is a sub-sequence of $\omega_{j^d}$. Similarly, $\omega_{d,{\cal B}}\subseteq\omega_d$ denotes that $\omega_{d,{\cal B}}$ is a sub-sequence of $\omega_d$. Clearly, iff $d=n$, we have $\omega_{j^d,{\cal B}}=\omega_{j^d}$ and $\omega_{d,{\cal B}}=\omega_d$. In addition, $|\omega_{j^d,{\cal B}}|\to\infty$ and $|\omega_{d,{\cal B}}|\to\infty$ as $d\to\infty$. Since $\tau(j^d,b^d)$ is defined by \eqref{eq:tau}, the following lemmas hold obviously. 
\begin{lemma}[Conditional Distribution of Shift Function]
	For almost every $0<r<1$, as $d\to\infty$, both $\omega_{j^d,{\cal B}}$ and $\omega_{d,{\cal B}}$ will be u.d. over $(-1,1)$.
\end{lemma}
\begin{lemma}[Average Length of Coexisting Interval]
	For almost every $0<r<1$ and for every $j^d\in{\cal J}_{n,d}$, we have
	\begin{align}\label{eq:sumBjd}
		\lim_{d\to\infty}\frac{\sum_{b^d\in{\cal B}_{j^d}}{\left(1-|\tau(j^d,b^d)|\right)}}{|{\cal B}_{j^d}|} = 1/2.
	\end{align}
\end{lemma}
\begin{theorem}[Hard Approximation of HDS for $1\ll d<n$]\label{thm:th3a}
	For almost every $0<r<1$ and for $d<n$, as $d\to\infty$,
	\begin{align}\label{eq:hardhds}
		\psi(d;n) \to 2^{-(d+1)} \sum_{j^d\in{\cal J}_{n,d}} \sum_{b^d\in\mathbb{B}^d} {{\bf 1}_{|\tau(j^d,b^d)|<1}}.
	\end{align}	
\end{theorem}
\begin{proof}
As $d\to\infty$, by using ${\cal B}_{j^d}$, we can rewrite \eqref{eq:softhds} as 
\begin{align*}
	\psi(d;n)
	&\to 2^{-d}\sum_{j^d\in{\cal J}_{n,d}}{\sum_{b^d\in{\cal B}_{j^d}}\left(1-|\tau(j^d,b^d)|\right)}\\
	&\overset{(a)}{\to} 2^{-(d+1)}\sum_{j^d\in{\cal J}_{n,d}}|{\cal B}_{j^d}|,
\end{align*}
where $(a)$ comes from \eqref{eq:sumBjd}. Then \eqref{eq:hardhds} follows immediately.
\end{proof}

\begin{remark}[Comparison of \cref{thm:th3a} with \cref{thm:softhds}]
	It can be seen that \eqref{eq:softhds} and \eqref{eq:hardhds} are very similar to each other, and \eqref{eq:hardhds} can be taken as a binary approximation of \eqref{eq:softhds}. We would like to highlight the following points.
	\begin{itemize}
		\item They have the same complexity ${\cal O}(2^d\binom{n}{d})$.
		\item As $d$ increases, \eqref{eq:softhds} and \eqref{eq:hardhds} will approach each other gradually, according to the generalized law of large numbers. 
		\item Neither \eqref{eq:softhds} nor \eqref{eq:hardhds} applies to $d=n$ because $E{(j^d,b^d)}$ will degenerate into a deterministic constant. 
		\item The condition of \eqref{eq:softhds} is $n\to\infty$, while the condition of \eqref{eq:hardhds} is $d\to\infty$. Since $d<n$, it is obvious that $d\to\infty$ is a stricter condition than $n\to\infty$. Hence, \eqref{eq:softhds} has a broader scope of application than \eqref{eq:hardhds}.
		\item For small $d$, the accuracy of \eqref{eq:hardhds} is poor because its cornerstone \eqref{eq:sumBjd} does not hold; while \eqref{eq:softhds} is always accurate no matter for small $d$ or large $d$ (but except $d=n$).
	\end{itemize}
\end{remark}

\section{Discussion on $d=n$}
As analyzed before, neither \eqref{eq:softhds} nor \eqref{eq:hardhds} applies to $d=n$ because $E{(j^d,b^d)}$ will degenerate into a deterministic constant. Now we want to see what will happen when $d=n$.

\begin{lemma}[$n$-away Codeword Pairs]\label{lem:naway}
	Let $n$ be code length. If a pair of $n$-away codewords coexist in the same coset, then this pair of codewords must belong to the $2^{nr-1}$-th coset. Conversely speaking, the $2^{nr-1}$-th coset ${\cal C}_{2^{nr-1}}$ includes all pairs of $n$-away codewords.
\end{lemma}
\begin{proof}
	According to \eqref{eq:tau}, the definition of $\tau$-function, we have
	\begin{align}
		\tau(j^n,x^n)
		&= \underbrace{(1-2^{-r})\sum_{i=1}^{n}{2^{ir}}}_{2^{nr}-1} - 2\underbrace{(1-2^{-r})\sum_{i=1}^{n}{x_i2^{ir}}}_{s(x^n)}\nonumber\\
		&= (2^{nr}-1) - 2s(x^n),\nonumber
	\end{align}
	which is followed by
	\begin{align}
		s(x^n) = 2^{nr-1}-\frac{1+\tau(j^n,x^n)}{2}.\nonumber
	\end{align}
	For $y^n=x^n\oplus1^n$, we have
	\begin{align}
		s(y^n) 
		&= s(x^n\oplus1^n) = s(x^n) + \tau(j^n,x^n) \nonumber\\
		&= 2^{nr-1} - \frac{1-\tau(j^n,x^n)}{2}.\nonumber
	\end{align}
	If $|\tau(j^n,x^n)|<1$, then $0<\frac{1\pm\tau(j^n,x^n)}{2}<1$ and
	\begin{align*}
		\begin{cases}
			(2^{nr-1}-1)<s(x^n)<2^{nr-1}\\
			(2^{nr-1}-1)<s(y^n)<2^{nr-1},	
		\end{cases}
	\end{align*} 
	implying $\lceil s(x^n)\rceil\equiv\lceil s(y^n)\rceil\equiv2^{nr-1}$. In other words, $x^n$ and $y^n=x^n\oplus1^n$ must coexist in the $2^{nr-1}$-th coset ${\cal C}_{2^{nr-1}}$.
\end{proof}

\begin{lemma}[Particularity of $d=n$]\label{corol:d=n}
	The necessary and sufficient condition for the event that $x^n$ and $(x^n\oplus1^n)$ coexist in the same coset is $|\tau(j^n,x^n)|<1$.
\end{lemma}

As a comparison, let us recall \cref{corol:neccoe}, which states that $|\tau(j^d,b^d)|<1$ is the necessary condition for the coexistence of $x^n$ and $y^n=x^n\oplus z^n$ in the same coset, where $x_{j^d}=b^d$, $z_{j^d}=1^d$, and $z_{[n]\setminus j^d}=0^{n-d}$. Especially, \cref{corol:d=n} says that, if $d=n$, the necessary condition $|\tau(j^n,b^n)|<1$ is also the sufficient condition for the coexistence of $x^n=b^n$ and $y^n=x^n\oplus 1^n$ in the same coset.

\begin{theorem}[$\psi(n;n)$]\label{thm:th3b}
	Let $j^n=\{1,\dots,n\}$. Then we have
	\begin{align}\label{eq:psinapprox}
		\psi(n;n) 
		= 2^{-n}\cdot|{\cal B}_{j^n}| = 2^{-n}\sum_{b^n\in\mathbb{B}^n}{{\bf 1}_{|\tau(j^n,b^n)|<1}}.
	\end{align}	
\end{theorem}
\begin{proof}
	There is only one element $j^n$ in the set ${\cal J}_{n,n}$. For every $b^n\in\mathbb{B}^n$, if $|\tau(j^n,b^n)|<1$, then both $b^n$ and $b^n\oplus1^n$ coexist in the same coset. Hence, \eqref{eq:psinapprox} holds naturally.
\end{proof}

Let us compare \eqref{eq:psinapprox} with \eqref{eq:hardhds}. By \eqref{eq:hardhds}, as $d=n\to\infty$,
\begin{align}
	\psi(n;n) \to 2^{-(n+1)}\cdot|{\cal B}_{j^n}| = 2^{-(n+1)}\sum_{b^n\in\mathbb{B}^n}{{\bf 1}_{|\tau(j^n,b^n)|<1}},
\end{align}
which is just half of \eqref{eq:psinapprox}! The definition of $E(j^d,b^d)$ explains why this phenomenon happens. According to \eqref{eq:E}, if $d<n$, $E(j^d,b^d)$ is a random variable; while when $d=n$, $E(j^n,b^n)$ will degenerate into a deterministic constant, and hence \eqref{eq:hardhds} does not apply. 

In addition, also please notice another subtle difference between \eqref{eq:psinapprox} and \eqref{eq:hardhds}. That is, \eqref{eq:hardhds} gives only an asymptotic value of $\psi(d;n)$ for $d<n$, while \eqref{eq:psinapprox} gives an exact value of $\psi(n;n)$.

For clarity, we integrate \cref{thm:th3a} and \cref{thm:th3b} into the following theorem.
\begin{theorem}[Hard Approximation of HDS]\label{thm:hard}
	Let $\alpha\triangleq{\bf 1}_{(d=n)}$. For almost every $0<r<1$ and for $1\ll d\leq n$, we have
	\begin{align}\label{eq:psith3}
		\psi(d;n) 
		&\approxeq 2^{\alpha-d-1}|\omega_{d,{\cal B}}| = 2^{\alpha-d-1}\sum_{j^d\in{\cal J}_{n,d}}|\omega_{j^d,{\cal B}}|\nonumber\\
		&= 2^{\alpha-d-1}\sum_{j^d\in{\cal J}_{n,d}}|{\cal B}_{j^d}|\nonumber\\
		& = 2^{\alpha-d-1}\sum_{j^d\in{\cal J}_{n,d}}\sum_{b^d\in\mathbb{B}^d}{{\bf 1}_{|\tau(j^d,b^d)|<1}},
	\end{align}
	where the approximation becomes equality if $d=n$.	
\end{theorem}

Meanwhile, by taking the case of $d=n$ into consideration, we will get a general form of \cref{thm:softhds} as below.
\begin{theorem}[Soft Approximation of HDS]\label{thm:soft}
	Let $\alpha\triangleq{\bf 1}_{(d=n)}$. For almost every $0<r<1$ and for $1\leq d\leq n$, we have
	\begin{align}\label{eq:psith2}
		\psi(d;n) 
		&\approx 2^{\alpha-d}\sum_{j^d\in{\cal J}_{n,d}}\sum_{b^d\in{\cal B}_{j^d}}{\left(1-|\tau(j^d,b^d)|\right)}\nonumber\\
		&= 2^{\alpha-d}\sum_{j^d\in{\cal J}_{n,d}}\sum_{b^d\in\mathbb{B}^d}{\left(1-|\tau(j^d,b^d)|\right)^+}.
	\end{align}	
\end{theorem}  

\section{Fast Approximation of HDS for $d\approx n$}
As show by \cref{thm:hard} and \cref{thm:soft}, for both {\em Hard Approximation} and {\em Soft Approximation}, the complexity is too high to be acceptable for large $d$. This section will derive a fast method to calculate $\psi(d;n)$ for $d\approx n$ based on the close affinity between CCS and HDS, which is named as \textit{Fast Approximation}, whose complexity is ${\cal O}(1)$, the same as that of the {\em Binomial Approximation} defined by \eqref{eq:coarsehds}. Through this work, we bridge HDS with CCS \cite{FangTIT24}.

\subsection{Normalized Shift Function}
It can be seen from \eqref{eq:psith3} that to calculate $\psi(d;n)$, the key is to find $|\omega_{d,{\cal B}}|$, the number of codewords making the shift function $\tau(j^d,b^d)$ fall into the interval $(-1,1)$. Actually, \eqref{eq:psith3} suggests to calculate $\psi(d;n)$ with exhaustive enumeration, whose complexity is ${\cal O}(2^d\binom{n}{d})$, very high for large $d$. In the following, instead of exhaustive enumeration, we try to find a simple method to calculate $\psi(d;n)$ for large $d$. Our analysis is based on the close affinity between CCS and HDS. More concretely, for every $d\approx n$, we can derive the asymptotic distribution of $\omega_d$ according to CCS. Further, according to the distribution of $\omega_d$, we can obtain $|\omega_{d,{\cal B}}|$. This is a very interesting finding.

Before deriving the distribution of $\omega_d$, we should know the distribution of $\omega_{j^d}$. To derive the distribution of $\omega_{j^d}$, let us define a random variable $\tau(j^d,X^d)$ according to \eqref{eq:tau}. Obviously, for uniform binary sources, the distribution of the random variable $\tau(j^d,X^d)$ is just the distribution of the sequence $\omega_{j^d}$. 

However, note that $\tau(j^d,X^d)$ is discretely distributed over the interval $(-2^{nr},2^{nr})$. As $n\to\infty$, the shift function $\tau(j^d,X^d)$ will not be well defined because the interval $(-2^{nr},2^{nr})$ will become the real field $\mathbb{R}$. Therefore, it is very difficult to derive the distribution of $\tau(j^d,X^d)$ directly. To overcome this difficulty, for a given $j^d\in{\cal J}_{n,d}$, according to \eqref{eq:tau}, we define a normalized random variable for $\tau(j^d,X^d)$.
\begin{definition}[Normalized Shift Function]
	For every $j^d\in{\cal J}_{n,d}$, we define the normalized shift function as
	\begin{align}
		W(j^d) \triangleq 2^{-nr}\tau(j^d,X^d) = c(j^d) - 2V(j^d),
	\end{align}
	where
	\begin{align*}
		c(j^d) \triangleq 2^{-nr}(1-2^{-r})\sum_{d'=1}^{d}{2^{rj_{d'}}}
	\end{align*}
	and	
	\begin{align}\label{eq:Vjd}
		V(j^d) \triangleq 2^{-nr}(1-2^{-r})\sum_{d'=1}^{d}{X_{d'}2^{rj_{d'}}}.
	\end{align}
\end{definition}
It can be seen that $c(j^d)$ is a constant, while $V(j^d)$ is a random variable for $d\geq1$. Clearly, $0\leq V(j^d) \leq c(j^d)$ and $-c(j^d)\leq W(j^d) \leq c(j^d)$. It is easy to know $0\leq c(j^d)<1$. Therefore, $V(j^d)$ is defined over $[0,1)$ and $W(j^d)$ is defined over $(-1,1)$. 

Let $f_{W|j^d}(w)$, where $-1<w<1$, be the pdf of $W(j^d)$, and $f_{V|j^d}(v)$, where $0\leq v<1$, be the pdf of $V(j^d)$. According to the property of pdf, it is easy to obtain
\begin{align}\label{eq:wvpdf}
	f_{W|j^d}(w) = \tfrac{1}{2}f_{V|j^d}(\tfrac{c(j^d)-w}{2}).
\end{align}
Both $V(j^d)$ and $W(j^d)$ are tractable because they are well defined. Once $f_{W|j^d}(w)$ is obtained, the distribution of $\omega_{j^d}$ can be easily derived. Further, we define
\begin{align}\label{eq:fWd}
	f_{W|d}(w) \triangleq \frac{\sum_{j^d\in{\cal J}_{n,d}}{f_{W|j^d}(w)}}{\binom{n}{d}}.
\end{align}
Once $f_{W|d}(w)$ is obtained, the distribution of $\omega_{d}$ can be easily derived. 

\subsection{Distribution of Normalized Shift Function}
We begin with $d=n$. There is only one element $j^n=\{1,\dots,n\}$ in ${\cal J}_{n,n}$, so we abbreviate $V(j^n)$ to $V$, and $W(j^n)$ to $W$ for conciseness. Similarly, $f_{W|j^n}(w)$ is shortened to $f_W(w)$, and $f_{V|j^n}(v)$ to $f_V(v)$.
\begin{theorem}[Asymptotic Distribution of Normalized Shift Function for $d=n$]\label{thm:wn}
	As $n\to\infty$, the pdf of $V$ is $f_V(v)=f(v)$, where $f(u)$ is the asymptotic CCS defined by \eqref{eq:asympt}, and the pdf of $W$ is 
	\begin{align}\label{eq:fWfV}
		f_W(w)=f_V(\tfrac{1-w}{2})/2 = f(\tfrac{1-w}{2})/2.
	\end{align}
\end{theorem}
\begin{proof}
	For $j^n=\{1,\dots,n\}\in{\cal J}_{n,n}$, we have 
	\begin{align}\label{eq:V}
		V = (1-2^{-r})\sum_{i=1}^{n}{X_i2^{(i-n)r}} \overset{d}{=} (2^r-1)\sum_{i=1}^{n}{X_i2^{-ir}}.
	\end{align}
	Surprisingly, by comparing \eqref{eq:V} with \eqref{eq:U}, we find $\lim_{n\to\infty}{V} = U$, where $U$ is the asymptotic projection. Thereby, as $n\to\infty$, the pdf of $V$ will be exactly equal to the asymptotic CCS $f(u)$ defined by \eqref{eq:asympt}. On knowing the pdf of $V$, the pdf of $W$ can be easily obtained from \eqref{eq:wvpdf}. Since $\lim_{n\to\infty}c(j^n)=1$, we have $f_W(w) \to f(\frac{1-w}{2})/2$. 
\end{proof}

We proceed to the case of $d=(n-1)$. There are $n$ sequences in ${\cal J}_{n,n-1}$, and for every sequence, the pdf of $W(j^{n-1})$ is different. We have the following theorem and corollaries (see \cite{FangTIT24} for the proofs).
\begin{theorem}[Asymptotic Distribution of Normalized Shift Function for $d=(n-1)$]\label{thm:wn1}
	For $1\leq k\leq n$, we define
	\begin{align}\label{eq:jn1}
		j^{n-1} 
		&= [n]\setminus (n-k+1)\nonumber\\
		&=\{1,\dots,n-k,n-k+2,\dots,n\}\in {\cal J}_{n,n-1}. 
	\end{align}
	Let $f(u)$ be the asymptotic CCS defined by \eqref{eq:asympt}. 
	\begin{itemize}
		\item For $k<\infty$, as $n\to\infty$,
		\begin{align}\label{eq:Vjn1}
			f_{V|j^{n-1}}(v) \to 2^{1-k(1-r)}\sum_{x^{k-1}\in\mathbb{B}^{k-1}}f((v-l(x^{k-1}))2^{kr})
		\end{align}
		and
		\begin{align}
			f_{W|j^{n-1}}(w) \to 2^{-k(1-r)}
			\sum_{x^{k-1}\in\mathbb{B}^{k-1}}f((\tfrac{c(j^{n-1})-w}{2}-l(x^{k-1}))2^{kr}),
		\end{align}
		where $c(j^{n-1}) = 1-(2^r-1)2^{-kr}$ and $l(x^i)$ is defined by \eqref{eq:lXi}. 
		
		\item $f_{V|j^{n-1}}(v)\to f(v)$ and $f_{W|j^{n-1}}(w)\to f(\frac{1-w}{2})/2$ as $k\to\infty$.  
	\end{itemize}
\end{theorem}

\begin{corollary}[Simple Relation Between $f_{V|j^{n-1}}(v)$ and $f(u)$]
	Let $f(u)$ be the asymptotic CCS defined by \eqref{eq:asympt}. Let $j^{n-1}=[n]\setminus (n-k+1)\in {\cal J}_{n,n-1}$, where $1\leq k\leq n$. Then
	\begin{align}\label{eq:fvfVv}
		f(v) = \tfrac{1}{2}\lim_{n\to\infty}\left(f_{V|j^{n-1}}(v) + f_{V|j^{n-1}}(v-(2^r-1)2^{-kr})\right).
	\end{align}
\end{corollary}

Finally, let us discuss $f_{W|(n-1)}(w)$, $-1<w<1$, which is defined by \eqref{eq:fWd}, where $d=(n-1)$. The following corollary gives an interesting result about $f_{W|(n-1)}(w)$.
\begin{corollary}[Asymptotic Form of $f_{W|{(n-1)}}(w)$]\label{corol:fWn1w}
	As $n\to\infty$, we have $f_{W|(n-1)}(w)\to f(\frac{1-w}{2})$, where $f(u)$ is the asymptotic CCS defined by \eqref{eq:asympt}. In other words, for large $n$, we have $f_{W|(n-1)}(w)\approx f(\frac{1-w}{2})$.
\end{corollary}
\begin{proof}
	There are $n$ sequences in ${\cal J}_{n,n-1}$, and each sequence can be written as $j^{n-1}={[n]\setminus (n-k+1)}\in {\cal J}_{n,n-1}$, where $1\leq k\leq n$. As stated by \cref{thm:wn1}, as $k\to\infty$, we have $f_{W|j^{n-1}}(w)\to f(\frac{1-w}{2})/2$. According to the definition \eqref{eq:fWd}, $f_{W|(n-1)}(w)$ is actually the average of $n$ functions, and these functions converge to $f(\frac{1-w}{2})/2$. Hence, this corollary holds obviously.
\end{proof}

So far, we have solved the asymptotic distribution problem of $W(j^d)$ for $d=n$ and $(n-1)$. It is certain that the above methodology can be easily extended to the general case of $d\approx n$. There are $\binom{n}{d}$ sequences in ${\cal J}_{n,d}$. For each sequence $j^d\in{\cal J}_{n,d}$, we calculate the asymptotic pdf of $V(j^d)$ and then derive the asymptotic pdf of $W(j^d)$. Since the procedure is very complex and boring, we would like to stop here for this issue. However, it deserves being spotted that \cref{corol:fWn1w} can be extended to the general case $d\approx n$.
\begin{corollary}[Asymptotic Form of $f_{W|d}(w)$]\label{corol:fWd}
	Given $(n-d)<\infty$, as $n\to\infty$, we have $f_{W|d}(w)\to f(\frac{1-w}{2})$, where $f(u)$ is the asymptotic CCS defined by \eqref{eq:asympt}. In other words, $f_{W|d}(w)\approx f(\frac{1-w}{2})$ for $d\approx n$.
\end{corollary}

\subsection{Bridging HDS with CCS}
Given the pdf of $W(j^d)$, we can easily derive how the sequence $\omega_{j^d}$ is distributed over $(-2^{nr},2^{nr})$. The following lemma answers how many terms in the sequence $\omega_{j^d}$ fall into the interval $(-1,1)$.
\begin{lemma}[Length of $j^d$-Active Sequence]\label{lem:caractset}
	Given $d\approx n$, for $j^d\in{\cal J}_{n,d}$, the number of codewords $b^d\in\mathbb{B}^d$ making $|\tau(j^d,b^d)|<1$ is $|{\cal B}_{j^d}|\approx 2^{d+1-nr}f_{W|j^d}(0)$, where $f_{W|j^d}(w)$ is the pdf of $W(j^d)$.
\end{lemma}
\begin{proof}
	Every $j^d\in{\cal J}_{n,d}$ corresponds to $2^d$ codewords in total, and for every  $b^d\in\mathbb{B}^d$, the shift function $\tau(j^d,b^d)$ is distributed over $(-2^{nr},2^{nr})$. According to the definition of $W(j^d)$, we have
	\begin{align}
		|{\cal B}_{j^d}| 
		&= 2^d\int_{-2^{-nr}}^{2^{-nr}}{f_{W|j^d}(w)}\,dw \nonumber\\
		&\overset{(a)}\approx 2^d\cdot 2^{1-nr} \cdot f_{W|j^d}(0) \nonumber\\
		&= 2^{d+1-nr}f_{W|j^d}(0),
	\end{align}
	where $(a)$ comes from the fact that $f_{W|j^d}(w)$ tends to be uniform over $(-2^{-nr},2^{-nr})$ as $n\to\infty$. 
\end{proof}

In turn, given $f_{W|d}(w)$, we can easily derive how the sequence $\omega_d$ is distributed over $(-2^{nr},2^{nr})$. The following lemma answers how many elements in the sequence $\omega_d$ fall into the interval $(-1,1)$.
\begin{lemma}[Length of $d$-Active Sequence]\label{lem:caractsetd}
	For large $n$, given $d\approx n$, the number of codewords $b^d\in\mathbb{B}^d$ making $|\tau(j^d,b^d)|<1$ is  
	\begin{align}
		\sum_{j^d\in{\cal J}_{n,d}}{|{\cal B}_{j^d}|} 
		&\approx \binom{n}{d}2^{d+1-nr}f_{W|d}(0)\nonumber\\ 
		&\approx \binom{n}{d}2^{d-nr}f(1/2), 
	\end{align}	
	where $f_{W|d}(w)$ is defined by \eqref{eq:fWd} and $f(u)$ is the asymptotic CCS defined by \eqref{eq:asympt}. 
\end{lemma}

\begin{theorem}[Fast Approximation of HDS]\label{thm:fast}
	For $d\approx n$, we have 
	\begin{align}\label{eq:fast}
		\psi(d;n) \approx \binom{n}{d}2^{\alpha-nr-1}f(1/2),
	\end{align}
	where $\alpha = {\bf 1}_{(d=n)}$ and $f(u)$ is the asymptotic CCS defined by \eqref{eq:asympt}.
\end{theorem}
\begin{proof}
	According to \eqref{eq:psinapprox}, we have
	\begin{align}\label{eq:nn}
		\psi(n;n) 
		&\approx 2^{-n}\cdot|{\cal B}_{j^d}| \stackrel{(a)}\approx 2^{-n}\cdot2^{n+1-nr}f_W(0)\nonumber\\
		&= 2^{1-nr}f_W(0) = 2^{1-nr}f_V(1/2)/2 \nonumber\\
		&\stackrel{(b)}{=} 2^{-nr}f(1/2),
	\end{align}
	where $(a)$ comes from \cref{lem:caractset} and $(b)$ comes from \cref{thm:wn}. According to \eqref{eq:hardhds}, for $d<n$,
	\begin{align}\label{eq:dn}
		\psi(d;n) 
		&\approx 2^{-(d+1)}\sum_{j^d\in{\cal J}_{n,d}}{|{\cal B}_{j^d}|}\nonumber\\
		&\stackrel{(a)}{\approx} 2^{-(d+1)} \binom{n}{d}2^{d-nr}f(1/2)\nonumber\\
		&= \binom{n}{d}2^{-nr-1}f(1/2),
	\end{align}	
	where $(a)$ comes from \cref{lem:caractsetd}. Combining \eqref{eq:nn} with \eqref{eq:dn}, we will obtain \eqref{eq:fast}. 	
\end{proof}

After comparing \cref{thm:soft} and \cref{thm:hard} with \cref{thm:fast}, it can be found that the computing complexity of $\psi(d;n)$ is reduced from ${\cal O}(2^d\binom{n}{d})$ to ${\cal O}(1)$. 

Finally, we give a summary on the approximate formulas of $\psi(d;n)$. In \cref{tab:summary}, there are four formulas, tagged with TH-1, TH-2, TH-3, and TH-4, respectively, where $\alpha = {\bf 1}_{(d=n)}$. The complexity and scope of application of these formulas are also provided.

\begin{table*}
	\small\centering
	\caption{A Summary on Approximate Formulas of $\psi(d;n)$}
	\begin{tabular}{c||c||c||c}
		\hline
		Tag &Formula &Scope &Complexity\\
		\hline
		\hline
		TH-1
		&$\displaystyle\binom{n}{d}\cdot 2^{-nr} \cdot \int_{0}^{1}{f^2(u)\,du}$
		&$d\approx \tfrac{n}{2}$ &${\cal O}(1)$\\
		\hline
		TH-2
		&$\displaystyle2^{\alpha-d} \sum_{b^d\in\mathbb{B}^d} \sum_{j^d\in{\cal J}_{n,d}}{\left(1-|\tau(j^d,b^d)|\right)^+}$
		&any $d$ &${\cal O}(2^d\binom{n}{d})$\\
		\hline
		TH-3
		&$\displaystyle2^{\alpha-d-1}\sum_{b^d\in\mathbb{B}^d}\sum_{j^d\in{\cal J}_{n,d}}{{\bf 1}_{|\tau(j^d,b^d)|<1}}$
		&$d\gg 1$ &${\cal O}(2^d\binom{n}{d})$\\
		\hline
		TH-4
		&$\displaystyle\binom{n}{d}\cdot 2^{\alpha-nr-1}\cdot f(1/2)$ 
		&$d\approx n$ &${\cal O}(1)$\\
		\hline
	\end{tabular}
	\label{tab:summary}
\end{table*}

\section{Convergence of HDS}\label{subsec:converge}
Given $1\leq j_1<j_2<\cdots<j_d$, another form of \eqref{eq:softhds} is 
\begin{align}\label{eq:psid}
	\psi(d) = 2^{-d} \sum_{b^d\in\mathbb{B}^d}\sum_{j_d=d}^{\infty}\sum_{j_{d-1}=(d-1)}^{(j_d-1)}\cdots\sum_{j_1=1}^{(j_2-1)}{\left(1-|\tau(j^d,b^d)|\right)^+}.
\end{align}
Since $\psi(d)$ is the sum of infinite terms, an interesting and important problem is whether $\psi(d)<\infty$ or not? If $\psi(d;n)<\infty$ as $n\to\infty$, we say that $\psi(d;n)$ is \textit{convergent}; otherwise, we say that $\psi(d;n)$ is \textit{divergent}. It was shown in \cite{FangTCOM16a} that $\psi(1)<\infty$ and $\psi(2)<\infty$, but for $d\geq3$, $\psi(d;n)$ may or may not converge as $n\to\infty$. Below, we will first give the concrete closed forms of $\psi(1)$ and $\psi(2)$, and then give the necessary and sufficient condition for the convergence of $\psi(3)$.

\subsection{$\psi(1)$ and $\psi(2)$ are Convergent}
\label{subsec:psi12}
\begin{corollary}[Convergence of $\psi(1)$]\label{corol:psi1}
	As $n\to\infty$,
	\begin{align*}
		\psi(1) = \sum_{i=1}^{J_1}{\left(1-(1-2^{-r})2^{ir}\right)}<\infty,
	\end{align*}	
	where 	
	\begin{align}\label{eq:J1}
		J_1 \triangleq -\left\lfloor\tfrac{1}{r}\log_2{(2^r-1)}\right\rfloor < \infty.
	\end{align} 	
\end{corollary}
\begin{proof}
	As we know, ${\cal J}_{n,1}=\{\{1\},\dots,\{n\}\}$, so when $d=1$, 
	\begin{align*}
		\psi(1) = \left(\psi(1|0)+\psi(1|1)\right)/2,
	\end{align*}
	where 
	\begin{align}
		\psi(1|b) = \sum_{i=1}^{\infty}{\left(1-|\tau(i,b)|\right)^+}.\nonumber
	\end{align}
	Due to the symmetry, $\psi(1|0)=\psi(1|1)$, so we consider only $\psi(1|0)$ below. According to \eqref{eq:tau}, we have $\tau(i,0) = (1-2^{-r}){2^{ir}}>0$, which is monotonously increasing w.r.t. $i$. After solving $\tau(i,0)<1$, we obtain
	\begin{align}
		i \leq J_1 
		&\triangleq \left\lceil-\tfrac{1}{r}\log_2{(1-2^{-r})}\right\rceil-1 \nonumber\\
		&= -\left\lfloor\tfrac{1}{r}\log_2{(2^r-1)}\right\rfloor < \infty.\nonumber
	\end{align} 
	Finally, $\psi(1) = \psi(1|0) = \psi(1|1)$.
\end{proof}

\begin{corollary}[Convergence of $\psi(2)$]\label{corol:psi2}
	The closed form of $\psi(2|0^2)$ is
	\begin{align}
		\psi(2|0^2) = \sum_{i=1}^{J_{2,1}}\sum_{k=1}^{\kappa_1(i)}{(1-(1-2^{-r})2^{ir}(2^{kr}+1))},\nonumber
	\end{align}
	where 
	\begin{align}\label{eq:J21}
		J_{2,1} \triangleq -\left\lfloor\tfrac{1}{r}\log_2{(4^r-1)}\right\rfloor
	\end{align}
	and $\kappa_1(i)$ is a function w.r.t. $i$ defined as
	\begin{align}
		\kappa_1(i) &\triangleq 
		\left\lceil\tfrac{1}{r}\left(\log_2{(2^{-ir}-1+2^{-r})}-\log_2{(2^r-1)}\right)\right\rceil \nonumber\\
		&\leq \kappa_1(1) = \left\lceil\tfrac{1}{r}\left(\log_2{(2^{1-r}-1)}-\log_2{(2^r-1)}\right)\right\rceil.\nonumber
	\end{align}
	The closed form of $\psi(2|\underline{10})$ is
	\begin{align}
		\psi(2|\underline{10}) = \sum_{i=1}^{J_{2,2}}\sum_{k=1}^{\kappa_2(i)}{(1-(1-2^{-r})2^{ir}(2^{kr}-1))},\nonumber
	\end{align}
	where
	\begin{align}\label{eq:J22}
		J_{2,2} \triangleq -\left\lfloor\tfrac{2}{r}\log_2{(2^r-1)}\right\rfloor	
	\end{align}
	and $\kappa_2(i)$ is a function w.r.t. $i$ defined as
	\begin{align}
		\kappa_2(i) &\triangleq 
		\left\lceil\tfrac{1}{r}\left(\log_2{(2^{-ir}+1-2^{-r})} - \log_2{(2^r-1)}\right)\right\rceil \nonumber\\
		&\leq \kappa_2(1) = J_1.\nonumber
	\end{align}
	As $n\to\infty$, $\psi(2;n)$ will converge to \cite{FangTIT24}
	\begin{align}
		\psi(2) = \frac{\psi(2|0^2) + \psi(2|\underline{10})}{2}<\infty.\nonumber
	\end{align}
\end{corollary}

\begin{corollary}[Relation Between $J_1$ and $J_{2,2}$]
	Given $J_1$ and $J_{2,2}$ defined by \eqref{eq:J1} and \eqref{eq:J22}, respectively, we have 
	\begin{align}\label{eq:j22j1}
		J_{2,2} = \begin{cases}
			2J_1, 	&\!\!\!\!{\rm if}~-J_1\leq\tfrac{1}{r}\log_2{(2^r-1)}<-J_1+1/2\\
			2J_1-1, &\!\!\!\!{\rm if}~-J_1+1/2\leq\tfrac{1}{r}\log_2{(2^r-1)}<-J_1+1.
		\end{cases}
	\end{align}
\end{corollary}
\begin{proof}
	For simplicity, let $x\triangleq\tfrac{1}{r}\log_2{(2^r-1)}$. Then $\lfloor x\rfloor=-J_1$ and $x\in[-J_1,-J_1+1)$. If $x\in[-J_1,-J_1+1/2)$, then $2x\in[-2J_1,-2J_1+1)$ and $\lfloor 2x\rfloor=-2J_1$. If $x\in[-J_1+1/2,-J_1+1)$, then $2x\in[-2J_1+1,-2J_1+2)$ and $\lfloor 2x\rfloor=-2J_1+1$. Since $J_1=-\lfloor x\rfloor$ and $J_{2,2}=-\lfloor 2x\rfloor$, \eqref{eq:j22j1} holds obviously. 
\end{proof}

To understand the above corollaries, the curves of $J_1$, $J_{2,1}$, and $J_{2,2}$ w.r.t. $r$ are plotted in \cref{subfig:J12} and the curves of $\psi(1)$ and $\psi(2)$ w.r.t. $r$ are plotted in \cref{subfig:psi12}. It is easy to verify that $J_1=J_{2,2}=0$ when $r=1$. As $r$ decreases from $1$, $J_1$ and $J_{2,2}$ will jump from $0$ to $1$ immediately. Then at $r\approx0.8114$, which corresponds to $2^r-1=2^{-r/2}$, $J_{2,2}$ jumps from $1$ to $2$; at $r=\log_2{\varphi}\approx 0.6942$, where $\varphi\approx1.618$ is the golden ratio, $J_1$ jumps from $1$ to $2$ and $J_{2,2}$ jumps from $2$ to $3$. The simple relation between $J_1$ and $J_{2,2}$, \textit{i.e.}, $J_{2,2}=2J_1$ or $2J_1-1$, is also verified. As for $J_{2,1}$, according to \eqref{eq:J21}, we have $J_{2,1}=-1$ for $r\geq\log_2{\varphi}$ and $J_{2,1}=0$ for $0.5\leq r<\log_2{\varphi}$. However, negative $J_{2,1}$ makes no sense, so we lower bound $J_{2,1}$ by $0$ in \cref{subfig:J12}.

As for $\psi(d)$, it can be found from \cref{subfig:psi12} that as $r$ decreases, $\psi(d)$ will strictly go up, coinciding with our intuition. However, the curves of $\psi(d)$ are not always smooth and there are many turning points, roughly corresponding to the jump points of $J_1$, $J_{2,1}$, and $J_{2,2}$. By intuition, there should be $\psi(2)>\psi(1)$. However, surprisingly, we find that $\psi(2)>\psi(1)$ does not hold always.

\begin{corollary}[An Exception]
	Let $r_0\approx0.8114$ be the root of $2^{-r/2}=2^r-1$. Then at least for $r_0\leq r<1$, we have $\psi(1)>\psi(2)$.
\end{corollary}
\begin{proof}
	For $r_0\leq r<1$, it can be found from \cref{subfig:J12} that $J_1=J_{2,2}=1$ and $J_{2,1}\leq 0$, so we have
	\begin{align}
		\psi(1) = 1-(1-2^{-r})2^{r} = 1-(2^r-1) = 1-x\nonumber
	\end{align}	
	and 
	\begin{align}
		\psi(2) = \frac{\psi(2|\underline{10})}{2} 
		&= \frac{1}{2}\sum_{k=1}^{\kappa_2(1)}{(1-(1-2^{-r})2^{r}(2^{kr}-1))}\nonumber\\
		&= \frac{1-(1-2^{-r})2^{r}(2^{r}-1)}{2}\nonumber\\
		&= \frac{1-(2^{r}-1)^2}{2} = \frac{1-x^2}{2},\nonumber
	\end{align}
	where $x\triangleq(2^r-1)\in[0,1]$. The ratio between $\psi(2)$ and $\psi(1)$ is
	\begin{align}
		\frac{\psi(2)}{\psi(1)} = \frac{1+x}{2} = 2^{r-1} \leq 1,\nonumber
	\end{align}
	and the equality holds iff $x=r=1$.
\end{proof}

\begin{figure*}[!t]
	\centering
	\subfigure[]{\includegraphics[width=.5\linewidth]{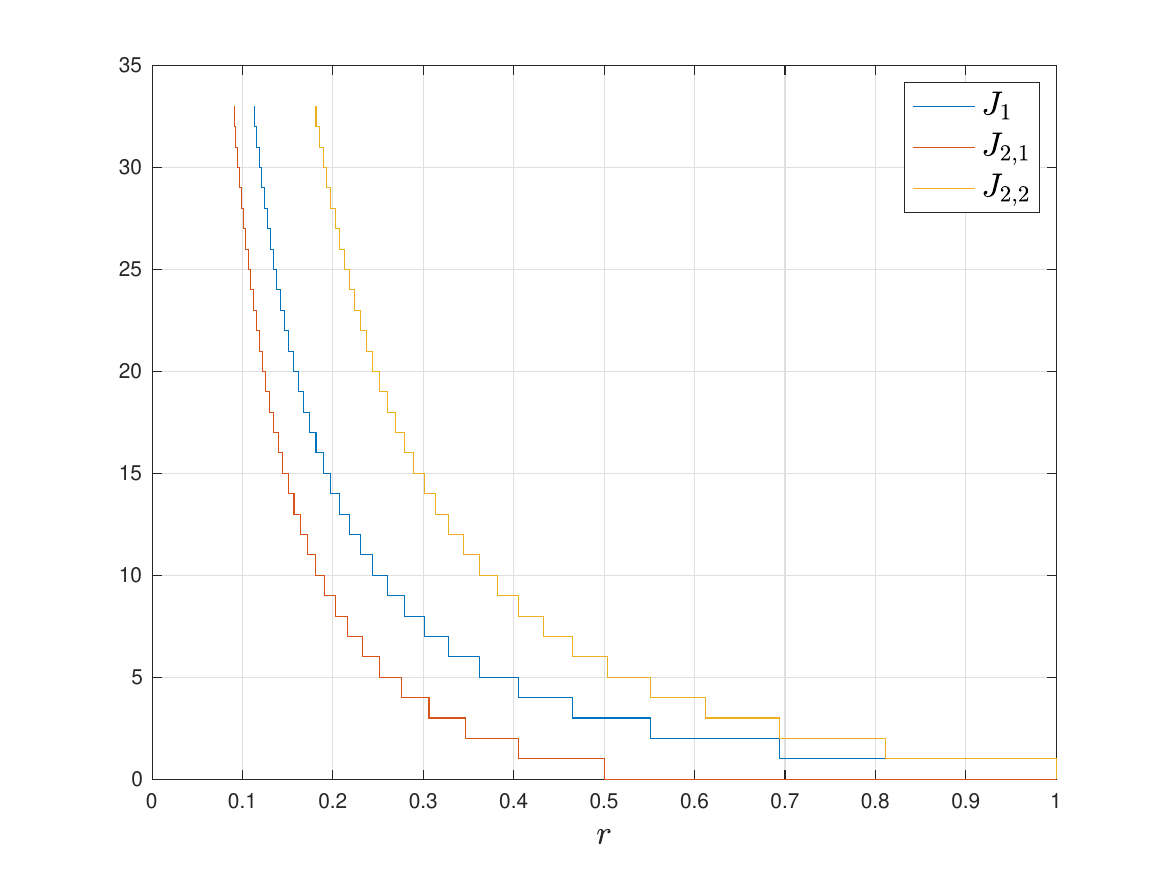}\label{subfig:J12}}%
	\subfigure[]{\includegraphics[width=.5\linewidth]{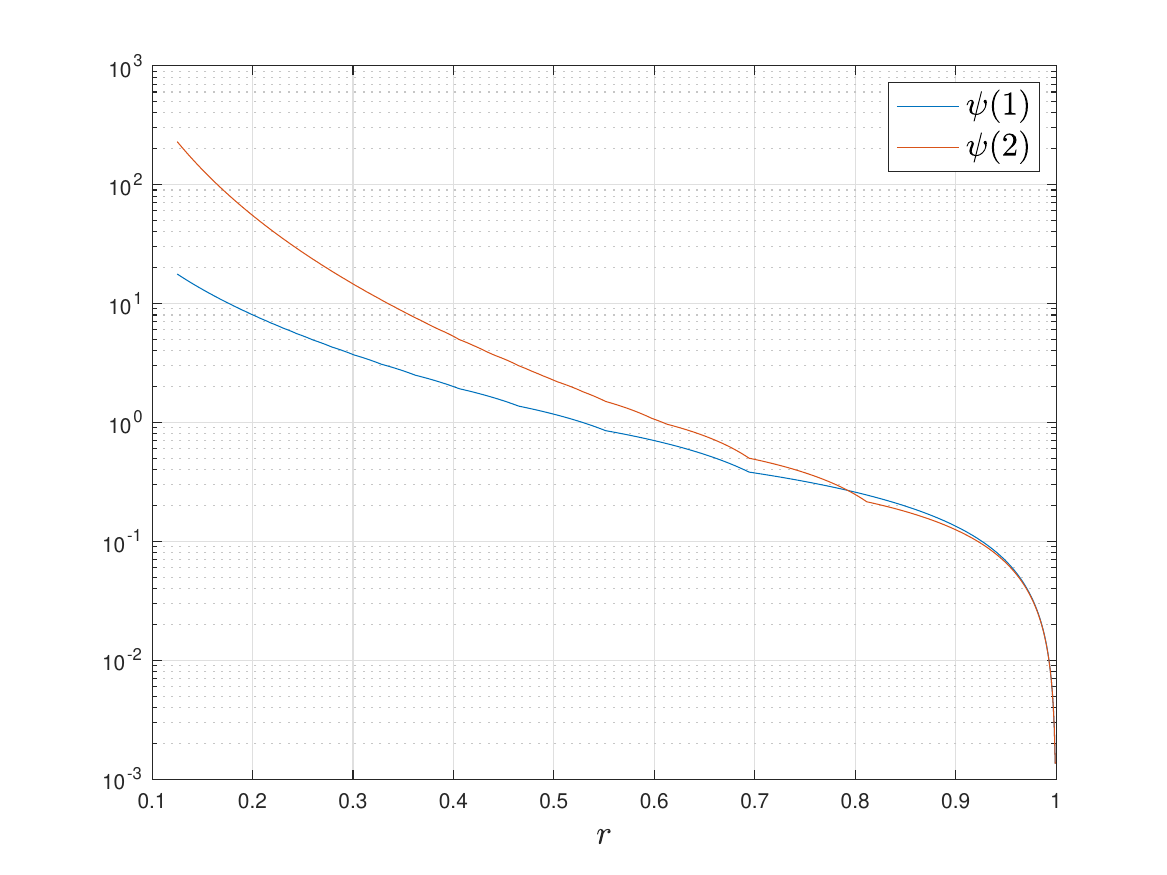}\label{subfig:psi12}}
	\caption{(a) Variations of $J_1$, $J_{2,1}$, and $J_{2,2}$ w.r.t. overlapping factor $r$, where $J_1$, $J_{2,1}$, and $J_{2,2}$ are defined by \eqref{eq:J1},  \eqref{eq:J21}, and \eqref{eq:J22}, respectively. Note that $J_{2,1}$ is lower bounded by $0$ because negative $J_{2,1}$ makes no sense. (b) Variations of $\psi(1)$ and $\psi(2)$ w.r.t. $r$, which are given by \cref{corol:psi1} and \cref{corol:psi2}, respectively.}
	\label{fig:hds12}
\end{figure*}

\subsection{Convergence of $\psi(3)$}\label{subsec:psi3}
In theory, if only $\psi(d)<\infty$, it can be calculated following the same methodology developed for $\psi(1)$ and $\psi(2)$. However, the procedure will be more and more complex. What's worse, $\psi(d)=\infty$ often happens for $d\geq 3$, as observed in \cite{FangTCOM16a}. The following theorem gives the necessary and sufficient condition for the convergence of $\psi(3)$.

\begin{theorem}[Necessary and Sufficient Condition for the Convergence of $\psi(3)$]\label{thm:psi3}
	If there is no pair of integers $i\geq1$ and $j\geq1$ such that $2^{ir}(2^{jr}-1)=1$, then $\psi(3)<\infty$; otherwise, $\psi(3)=\infty$.
\end{theorem}
\begin{proof}
	Since the proof is overly long, it is not included in this monograph. The reader may refer to \cite{FangTIT24} for a detail proof.
\end{proof}

\begin{corollary}[Sufficient Condition for the Convergence of $\psi(3)$]
	If $2^r$ is a transcendental number, $\psi(3)<\infty$.
\end{corollary}
\begin{proof}
	This is a direct result of \cref{thm:psi3}.
\end{proof}

We have given the necessary and sufficient condition for the convergence of $\psi(3)$. Similarly, it is possible to deduce the necessary and sufficient condition for the convergence of $\psi(d)$ for $d>3$. However, the procedure will become more and more complex. Actually, we strongly believe that if $2^r$ is a transcendental number, then $\psi(d)<\infty$ for any $d$, not just for $d=3$. However, we are not able to provide a strict proof at present, so we remain it as future work.

\begin{figure*}[!t]
	\centering
	\subfigure[]{\includegraphics[width=.5\linewidth]{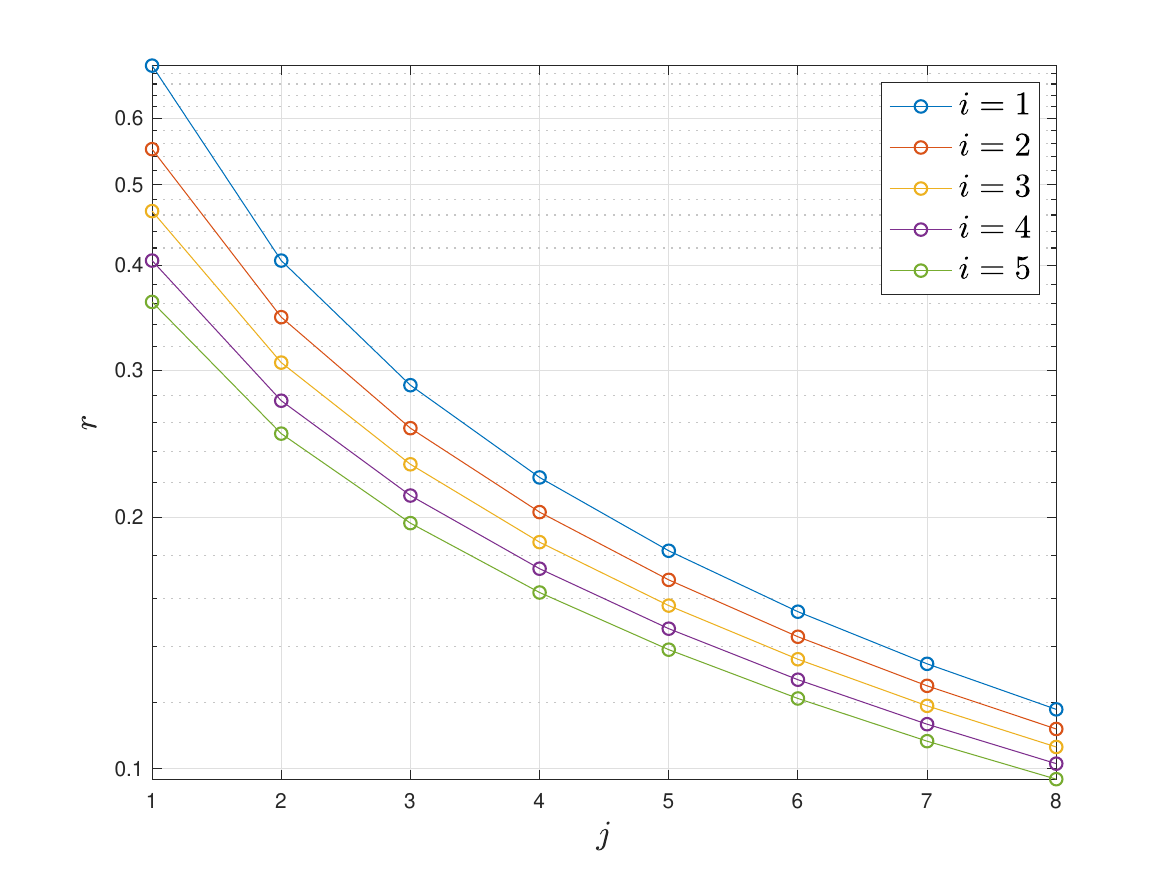}\label{subfig:d3a}}%
	\subfigure[]{\includegraphics[width=.5\linewidth]{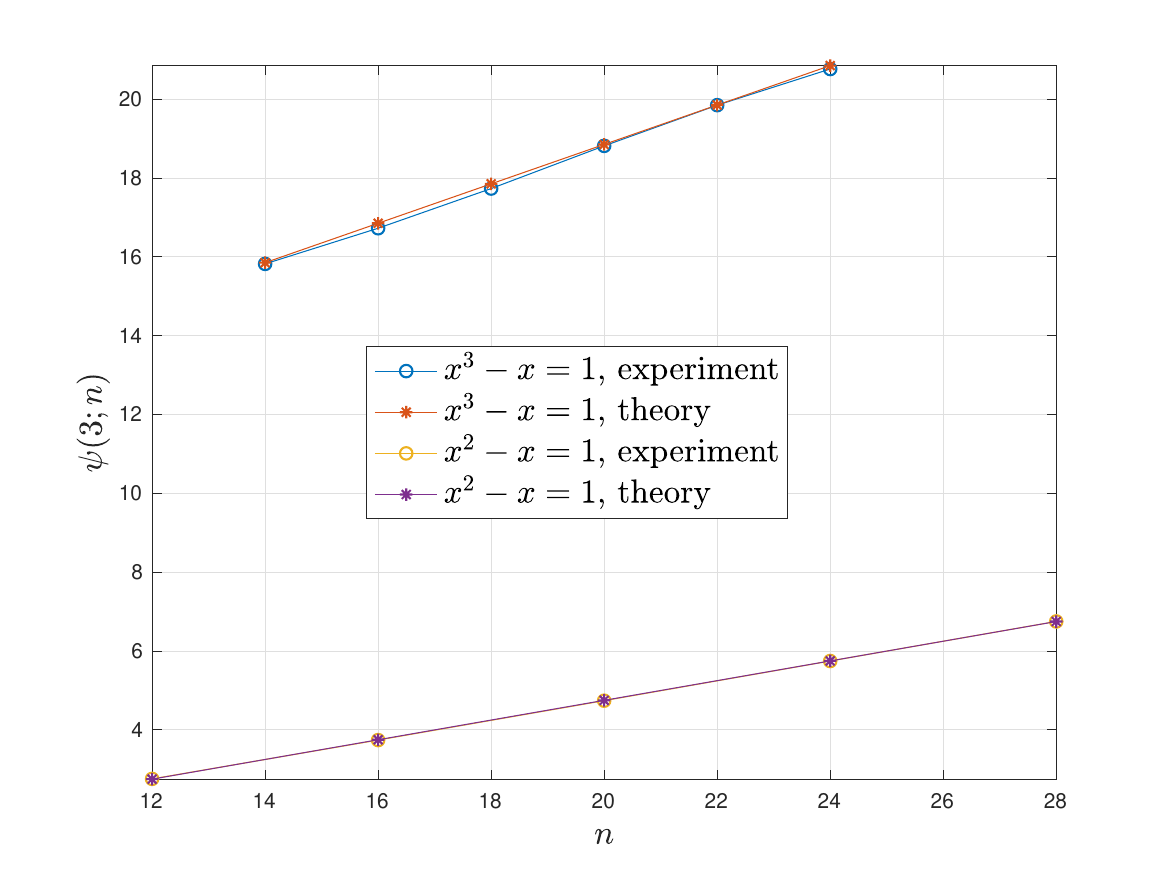}\label{subfig:d3b}}
	\caption{(a) Each circle corresponds to a tuple of $(i,j,r)$ making $2^{ir}(2^{jr}-1)=1$, where $i,j\in\mathbb{Z}$ and $r\in[0,1]$. Note that for some specific $r$, there may be more than one $(i,j)$-pairs satisfying $2^{ir}(2^{jr}-1)=1$. For example, if $2^r(2^{2r}-1)=1$, then $2^{4r}(2^r-1)=1$. (b) Theoretical and experimental results of $\psi(3;n)$ for two divergent settings of $x=2^r$.}
	\label{fig:d3}
\end{figure*}

\section{Closed-Form Expression of Divergent $\psi(3;n)$}\label{subsec:diverge}
For every pair of $i\geq1$ and $j\geq1$, there must be an overlapping factor $r\in[0,1]$ making $2^{ir}(2^{jr}-1)=1$. Some examples of $(i,j)$-pairs and the corresponding overlapping factors satisfying $2^{ir}(2^{jr}-1)=1$ are given in \cref{subfig:d3a}. It can be seen that when $2^r$ is the golden ratio, $i=j=1$ will make $2^{ir}(2^{jr}-1)=1$. Actually, this is the largest $r$ making $\psi(3)=\infty$, as proved by the following lemma. As $i$ or $j$ increases, the overlapping factor making $2^{ir}(2^{jr}-1)=1$ will be smaller. 

\begin{lemma}[Sign of $(2^{j_3r}-2^{j_2r}-2^{j_1r})$]\label{lem:signum}
	Let $\varphi$ denote the golden ratio. Depending on the relation between $2^r$ and $\varphi$, the value of $(2^{j_3r}-2^{j_2r}-2^{j_1r})$ will have different signs.
	\begin{itemize}
		\item $2^r=\varphi$: We have $(2^{j_3r}-2^{j_2r}-2^{j_1r})\geq0$ and the equality holds iff $(j_3-j_2)=(j_2-j_1)=1$. 
		\item $2^r>\varphi$: We have $(2^{j_3r}-2^{j_2r}-2^{j_1r})>0$ for every $1\leq j_1<j_2<j_3$ and hence $\psi(3)<\infty$.
		\item $2^r<\varphi$: The value of $(2^{j_3r}-2^{j_2r}-2^{j_1r})$ can be positive or negative, and for some special $2^r$, may be $0$.
	\end{itemize}
\end{lemma}
\begin{proof}
	We can write $(2^{j_3r}-2^{j_2r}-2^{j_1r})$ as
	\begin{align*}
		&(2^{j_3r}-2^{j_2r}-2^{j_1r}) 
		= 2^{j_1r}(2^{(j_3-j_1)r}-2^{(j_2-j_1)r}-1)\\
		&\qquad\qquad= 2^{j_1r}(2^{(j_2-j_1)r}(2^{((j_3-j_1)-(j_2-j_1))r}-1)-1)\\
		&\qquad\qquad= 2^{ir}(2^{jr}(2^{kr}-1)-1),
	\end{align*}
	where $i=j_1\geq1$, $j=(j_2-j_1)\geq1$, and $k=((j_3-j_1)-(j_2-j_1))=(j_3-j_2)\geq1$. Clearly, $2^{jr}(2^{kr}-1)>0$ and it is strictly increasing w.r.t. both $j$ and $k$. Thus $2^{jr}(2^{kr}-1)\geq2^r(2^r-1)$. After solving $2^r(2^r-1)=1$, we obtain $2^r=\varphi\approx1.618$ and $r\approx 0.6942$. Hence, if $2^r=\varphi$, we have
	\begin{align*}
		(2^{j_3r}-2^{j_2r}-2^{j_1r}) 
		&= 2^{ir}(2^{jr}(2^{kr}-1)-1)\\
		&\geq (2^{jr}(2^{kr}-1)-1)\\
		&\geq (2^r(2^r-1)-1) = 0,
	\end{align*}
	where the equality holds iff $k=(j_3-j_2)=j=(j_2-j_1)=1$.
	
	Since $2^r(2^r-1)$ is strictly increasing w.r.t. $2^r$ and $r$, if $2^r>\varphi$, then $(2^{j_3r}-2^{j_2r}-2^{j_1r})\geq (2^r(2^r-1)-1)>0$ for any $1\leq j_1<j_2<j_3$ and thus $\psi(3)<\infty$. Similarly, the third bullet point also holds.
\end{proof}

Now we wonder whether it is possible to find the analytical form of $\psi(3;n)$ even if $\psi(3;n)$ does not converge? Below we first give the general form of $\psi(3;n)$ when it does not converge, and then derive the concrete forms of $\psi(3;n)$ in two special cases. The proofs of the following theorem and corollaries can be found in \cite{FangTIT24}.

\begin{theorem}[Linearity of $\psi(3;n)$ for Large $n$]\label{thm:psi3linear}
	If there exist one or more pairs of integers $i\geq1$ and $j\geq1$ such that $2^{ir}(2^{jr}-1)=1$, then for large $n$, $\psi(3;n)$ will increase linearly. Let ${\cal P}$ denote the set of all pairs of integers $i\geq1$ and $j\geq1$ satisfying $2^{ir}(2^{jr}-1)=1$. Every $(i,j)\in{\cal P}$ will cause a divergent term $(n-(i+j))$, and
	\begin{align}\label{eq:psi3ngeneral}
		\psi(3;n) &\approx c_0 + (1/4)\sum_{(i,j)\in{\cal P}}{(n-(i+j))}\nonumber\\
		&= c_1 + n|{\cal P}|/4,
	\end{align}	
	where $c_0$ is the sum of convergent terms, $|{\cal P}|$ is the cardinality of ${\cal P}$, and $c_1=c_0-\frac{1}{4}\sum_{(i,j)\in{\cal P}}{(i+j)}$.
\end{theorem}

As shown by \eqref{eq:psi3ngeneral}, to obtain $\psi(3;n)$ for a specific $r$, the key is to determine $c_0$, the sum of convergent terms. Following are two examples to show how to calculate $c_0$ for a specific $r$. In the first example, the set ${\cal P}$ includes only one $(i,j)$-pair, while in the second example, the set ${\cal P}$ includes two $(i,j)$-pairs. 
\begin{corollary}[$\psi(3;n)$ for Golden Ratio]\label{corol:psi3}
	If $(2^{2r}-2^r-1)=0$, then $\psi(3;n)\approx (n-1)/4$ for $n\geq5$. 
\end{corollary}

As shown by the third bullet point of \cref{lem:signum}, if $2^r$ is smaller than the golden ratio, $(2^{j_3r}-2^{j_2r}-2^{j_1r})$ has zero-crossing points. An interesting thing is that for some special overlapping factors, there may be more than one pairs of $i\geq1$ and $j\geq1$ such that $2^{ir}(2^{jr}-1)=1$. For example, in \cref{subfig:d3a}, we find an overlapping factor $r\approx0.4057$ that results in $2^r(2^{2r}-1)=2^{4r}(2^r-1)=1$. To prove this point, let us define $x\triangleq 2^r\in(1,2)$ for convenience. Then $2^r(2^{2r}-1)=x(x^2-1)=1$, which is followed by $\alpha\triangleq(x^3-x-1)=0$. We call $\alpha$ as the \textbf{zero element} just as in finite field. With polynomial division,
\begin{align}
	2^{4r}(2^r-1) = x^4(x-1) = (x^5-x^4)\bmod\alpha = 1.\nonumber
\end{align}
The closed form of $\psi(3;n)$ for $\alpha\triangleq(x^3-x-1)=0$ is given below.

\begin{corollary}[$\psi(3;n)$ for $(x^3-x-1)=0$]\label{corol:psi3b}
	Let $x\triangleq 2^r\in(1,2)$. If $\alpha\triangleq(x^3-x-1)=0$, then for $n\geq14$,
	\begin{align}
		\psi(3;n)\approx \frac{-12x^2-17x+79}{4} + n/2.
	\end{align}
\end{corollary}

With the above two examples, we show that the analytical form of $\psi(3;n)$ is calculable for every $n$ even though it does not converge. Of course, following the methodology developed in \cite{FangTIT24}, the reader can also derive the mathematical expression of $\psi(3;n)$ in other divergent cases. However, the procedure will be very complex.

To verify the correctness of \cref{corol:psi3} and \cref{corol:psi3b}, some results are given in \cref{subfig:d3b}. It can be seen that theoretical curves almost coincide with  experimental curves.

\begin{remark}[Propagation of Divergence]
	Finally, we would like to end this sub-section with an interesting phenomenon. If for some overlapping factor $r$, there are two or more pairs of $i\geq1$ and $j\geq1$ such that $2^{ir}(2^{jr}-1)=1$, then $\psi(4;n)$ and $\psi(6;n)$ may not converge. For example, if
	$x^3-x-1=x^5-x^4-1=0$, then 
	\begin{align*}
		x^2(x^5-x^4-1)+(x^3-x-1) = x^7-x^6+x^3-x^2-x-1=0,
	\end{align*}
	which will cause a divergent term $(n-7)$ and further make $\psi(6;n)$ divergent. Meanwhile, we have
	\begin{align*}
		(x^5-x^4)-(x^3-x) 
		&= x^5-x^4-x^3+x\\
		&= x(x^4-x^3-x^2+1) = 0,
	\end{align*}
	and thus $(x^4-x^3-x^2+1)=0$, which will cause a divergent term $(n-4)$ and further make $\psi(4;n)$ divergent. By repeatedly combining the above zero elements, we can even observe that $\psi(d;n)$ is divergent for other $d$. This phenomenon is coined as \textit{Propagation of Divergence} in \cite{FangTIT24} and still remains a difficult open question.
\end{remark}

\section{Experimental Verification}\label{sec:hdsexample}
We run experiments under the special setting of $n=20$ and $r=1/2$ to verify \cref{tab:summary}. We choose $r=1/2$ because the asymptotic CCS $f(u)$ is given by \eqref{eq:closedForm_halfRate}. We choose $n=20$ because the complexity is too high for larger $n$. This section will first discuss the distribution of shift function $\tau(j^d,b^d)$ and then demonstrate $\psi(d;n)$. 

\subsection{Distribution of Shift Function}\label{subsec:shift}
For conciseness, we give the theoretical results only for $d=n$ and $(n-1)$, while the theoretical results for any $d\approx n$ can be obtained in a similar way. For $d=n$, there is only one size-$n$ set $j^n=\{1,\dots,n\}\in{\cal J}_{n,n}$. According to \cref{thm:wn}, we have $f_{V|n}(v) = f_{V|j^n}(v) \approx f(v)$, where $f(u)$ is the asymptotic CCS given by \eqref{eq:closedForm_halfRate}, and 
\begin{align}\label{eq:fWn}
	f_{W|n}(w) &= f_{W|j^n}(w) \approx f(\tfrac{1-w}{2})/2 \nonumber\\
	&= 
	\begin{cases}
		\frac{1+w}{12\sqrt{2}-16}, 	&-1 \leq w < 2\sqrt{2}-3\\
		\frac{1}{4-2\sqrt{2}}, 		&2\sqrt{2}-3 \leq w < 3-2\sqrt{2}\\
		\frac{1-w}{12\sqrt{2}-16},	&3-2\sqrt{2} \leq w < 1.%
	\end{cases}
\end{align}
Then we discuss $d=(n-1)$. Though there are $\binom{n}{n-1}=n$ size-$(n-1)$ sets $j^{n-1}\in{\cal J}_{n,n-1}$, only two subcases are considered below.
\begin{itemize}
	\item If $j^{n-1}=\{1,\dots,n-1\}$, \textit{i.e.}, $k=1$ in \cref{thm:wn1}, we have
	\begin{align*}
		&f_{V|j^{n-1}}(v) \approx \sqrt{2}f(\sqrt{2}v) = \\
		&\qquad\begin{cases}
			\frac{2v}{3\sqrt{2}-4}, 		&0 \leq v < 1-1/\sqrt{2}\\
			\frac{\sqrt{2}}{2-\sqrt{2}}, 	&1-1/\sqrt{2} \leq v < \sqrt{2}-1\\
			\frac{\sqrt{2}-2v}{3\sqrt{2}-4},&\sqrt{2}-1 \leq v < 1/\sqrt{2}%
		\end{cases}
	\end{align*}
	and
	\begin{align}\label{eq:fWn1k1}
		&f_{W|j^{n-1}}(w) \approx 2^{-1/2}f(1/2-w/\sqrt{2}) = \nonumber\\
		&\quad\begin{cases}
			\frac{1/2+w/\sqrt{2}}{6-4\sqrt{2}}, 		&-1/\sqrt{2} \leq w < 2-3/\sqrt{2}\\
			\frac{1}{2\sqrt{2}-2}, 		&2-3/\sqrt{2} \leq w < 3/\sqrt{2}-2\\
			\frac{1/2-w/\sqrt{2}}{6-4\sqrt{2}},	&3/\sqrt{2}-2 \leq w < 1/\sqrt{2}.%
		\end{cases}
	\end{align}
	
	\item If $j^{n-1}=\{1,\dots,n-2,n\}$, \textit{i.e.}, $k=2$ in \cref{thm:wn1}, we have $l(0)=0$ and $l(1)=(1-2^{-r})=(1-1/\sqrt{2})$. Hence,
	\begin{align*}
		f_{V|j^{n-1}}(v) &\approx f(2v) + f(2(v-(1-1/\sqrt{2}))).
	\end{align*}
	After arrangement, we obtain
	\begin{align*}
		f_{V|j^{n-1}}(v) \approx		
		\begin{cases}
			\frac{2v}{3\sqrt{2}-4}, 			&0 \leq v < (\sqrt{2}-1)/2\\
			\frac{1}{2-\sqrt{2}}, 				&(\sqrt{2}-1)/2 \leq v < (2-\sqrt{2})\\
			\frac{(3-\sqrt{2})-2v}{3\sqrt{2}-4},&(2-\sqrt{2}) \leq v < (3-\sqrt{2})/2.%
		\end{cases}
	\end{align*}
	As for $f_{W|j^{n-1}}(w)$, since $1-(2^r-1)2^{-kr} = \frac{3-\sqrt{2}}{2}$, we have
	\begin{align*}
		2f_{W|j^{n-1}}(w)
		&\approx 
		f(\tfrac{3-\sqrt{2}}{2}-w) + f(\tfrac{3-\sqrt{2}}{2}-w-(2-\sqrt{2}))\\
		&= f(\tfrac{3-\sqrt{2}}{2}-w) + f(\tfrac{\sqrt{2}-1}{2}-w).
	\end{align*}
	After arrangement, we obtain
	\begin{align}\label{eq:fWn1k2}
		f_{W|j^{n-1}}(w) \approx 
		\begin{cases}
			\frac{\tfrac{3-\sqrt{2}}{2}+w}{6\sqrt{2}-8}, & -\tfrac{3-\sqrt{2}}{2} < w\leq -\tfrac{5-3\sqrt{2}}{2}\\
			\frac{1}{4-2\sqrt{2}}, & -\tfrac{5-3\sqrt{2}}{2} \leq w \leq \tfrac{5-3\sqrt{2}}{2}\\
			\frac{\tfrac{3-\sqrt{2}}{2}-w}{6\sqrt{2}-8}, & \tfrac{5-3\sqrt{2}}{2} \leq w < \tfrac{3-\sqrt{2}}{2}.
		\end{cases}
	\end{align}
\end{itemize}
For other $j^{n-1}\in{\cal J}_{n,n-1}$, we can derive $f_{V|j^{n-1}}(v)$ and $f_{W|j^{n-1}}(w)$ according to \cref{thm:wn1} in a similar way.

To obtain experimental results, we run an exhaustive search, whose total complexity is ${\cal O}(3^n)$, where $3^n=\sum_{d=0}^{n}\binom{n}{d}2^d$. For every $j^d\in{\cal J}_{n,d}$ and every $b^n\in\mathbb{B}^d$, we define 
\begin{align}
	t(j^d,b^d)\triangleq{\rm sgn}(\tau(j^d,b^d)) \cdot \lceil|\tau(j^d,b^d)|\rceil,\nonumber
\end{align}
where ${\rm sgn}(\cdot)$ denotes the sign function. According to the properties of shift function, it is easy to know $t(j^d,b^d)\in(-2^{nr}:2^{nr})$. For a specific $j^d\in{\cal J}_{n,d}$, let $c(x;j^d)$, where $x\in(-2^{nr}:2^{nr})$, denote the number of $b^d\in\mathbb{B}^d$ such that $t(j^d,b^d)=x$. Obviously, $c(x;j^d)=c(-x;j^d)$ and $\sum_{x=-(2^{nr}-1)}^{2^{nr}-1}{c(x;j^d)}=2^d$, which can be rewritten as 
\begin{align}
	\sum_{x=-(2^{nr}-1)}^{2^{nr}-1}{c(x;j^d)2^{-d}} = 1.\nonumber
\end{align}
Now we need to build the connection between $c(x;j^d)$ for $x\in(-2^{nr}:2^{nr})$ and $f_{W|j^d}(w)$ for $w\in(-1,1)$. According to the definition of normalized shift function, we have
\begin{align}
	\sum_{x=-(2^{nr}-1)}^{2^{nr}-1}f_{W|j^d}(x2^{-nr})2^{-nr} \approx \int_{-1}^{1}f_{W|j^d}(w)\,dw=1.\nonumber
\end{align}
Hence $f_{W|j^d}(x2^{-nr})2^{-nr} \approx c(x;j^d)2^{-d}$ and for any $x\in(-2^{nr}:2^{nr})$,
\begin{align}
	f_{W|j^d}(x2^{-nr})=f_{W|j^d}(-x2^{-nr})\approx c(x;j^d)2^{nr-d}.\nonumber
\end{align}
Further, we define $c(x;d)\triangleq\sum_{j^d\in{\cal J}_{n,d}}{c(x;j^d)}$. Obviously,
\begin{align*}
	\sum_{x=-(2^{nr}-1)}^{2^{nr}-1}c(x;d)=\binom{n}{d}2^d. 
\end{align*}
Similarly, $f_{W|d}(w)$ for $w\in(-1,1)$ and $c(x;d)$ for $x\in(-2^{nr}:2^{nr})$ can be connected by
\begin{align}
	f_{W|d}(x2^{-nr}) = f_{W|d}(-x2^{-nr}) \approx \frac{c(x;d)2^{nr}}{\binom{n}{d}2^d}.\nonumber
\end{align}

We plot the results for $n=20$ and $r=1/2$ in \cref{fig:fW}. Let us check the correctness of theoretical analyses one by one. 
\begin{itemize}
	\item For $d=n$, we have $f_{W|n}(w)=f_{W|j^n}(w)=f(\frac{1-w}{2})/2$, where $j^n=\{1,\dots,n\}$ is the only element in ${\cal J}_{n,n}$. The theoretical result of $f_{W|n}(w)$ (and also $f_{W|j^n}(w)$), which is given by \eqref{eq:fWn}, corresponds to the TH curve in \cref{subfig:fWn}, while the experimental result of $f_{W|n}(w)$ (and also $f_{W|j^n}(w)$) corresponds to the EXP curve in \cref{subfig:fWn}. It can be observed from \cref{subfig:fWn} that these two curves almost coincide with each other, strongly confirming the correctness of \cref{thm:wn}. 
	
	\item For $d=(n-1)$, we have derived the expression of $f_{W|j^{n-1}}(w)$ for $k=1$ and $k=2$ according to \cref{thm:wn1}, as given by \eqref{eq:fWn1k1} and \eqref{eq:fWn1k2}, respectively. We compare the theoretical results of $f_{W|j^{n-1}}(w)$ for $k=1$ and $k=2$ with their experimental results in \cref{subfig:fWjn1_1} and \cref{subfig:fWjn1_2}, respectively. It can be observed that the theoretical curves almost coincide with the corresponding experimental curves, perfectly confirming the correctness of \cref{thm:wn1}. 
	
	\item It is also declared by \cref{thm:wn1} that $f_{W|j^{n-1}}(w)$ will converge to $f(\frac{1-w}{2})/2$, where $f(u)$ is the asymptotic initial CCS given by \eqref{eq:asympt}, as $k$ increases. To confirm this prediction, we plot several curves of $f_{W|j^{n-1}}(w)$ for $k=4$, $8$, $12$, and $16$ in \cref{subfig:fWjn1}, where the TH curve is $f(\frac{1-w}{2})/2$. It can be seen that as $k$ increases, $f_{W|j^{n-1}}(w)$ does converge to the TH curve $f(\frac{1-w}{2})/2$, verifying the correctness of \cref{thm:wn1}.
\end{itemize}

\begin{figure*}[!t]
	\subfigure[]{\includegraphics[width=.5\linewidth]{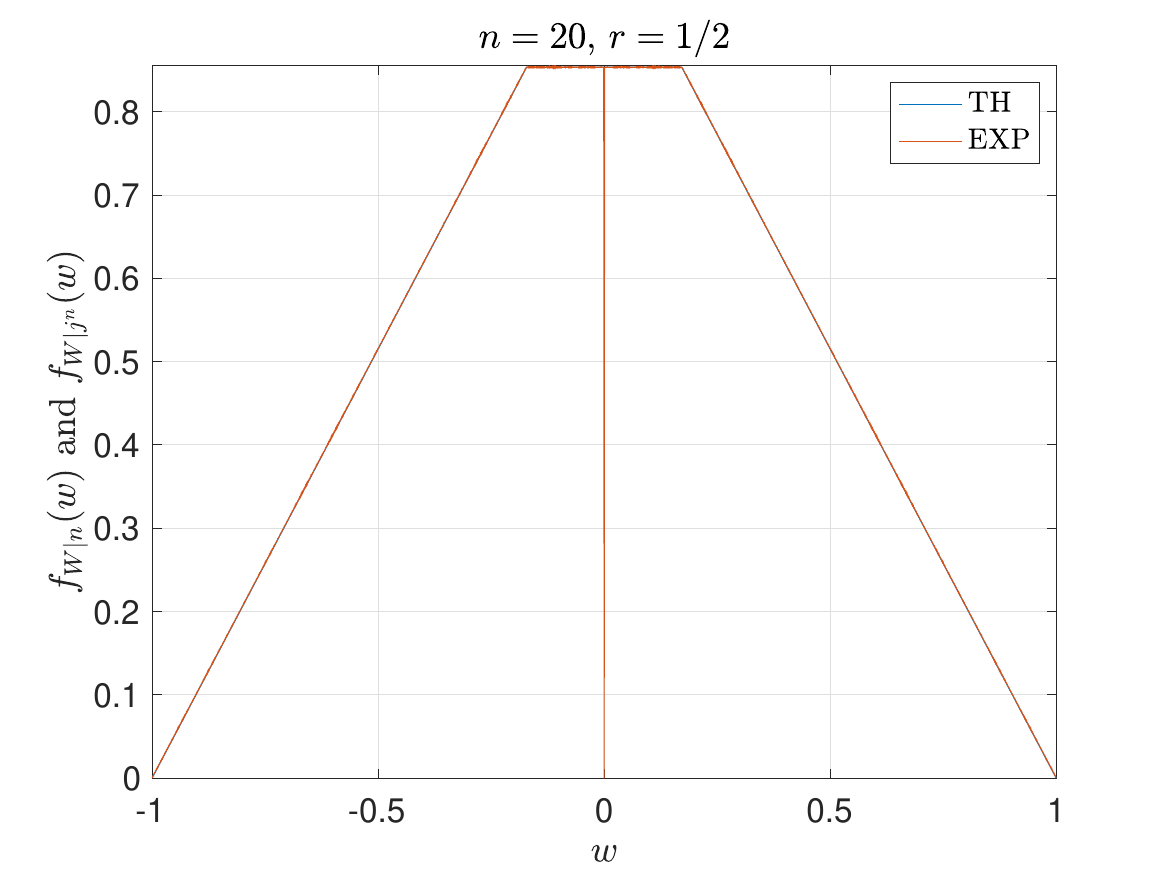}\label{subfig:fWn}}%
	\subfigure[]{\includegraphics[width=.5\linewidth]{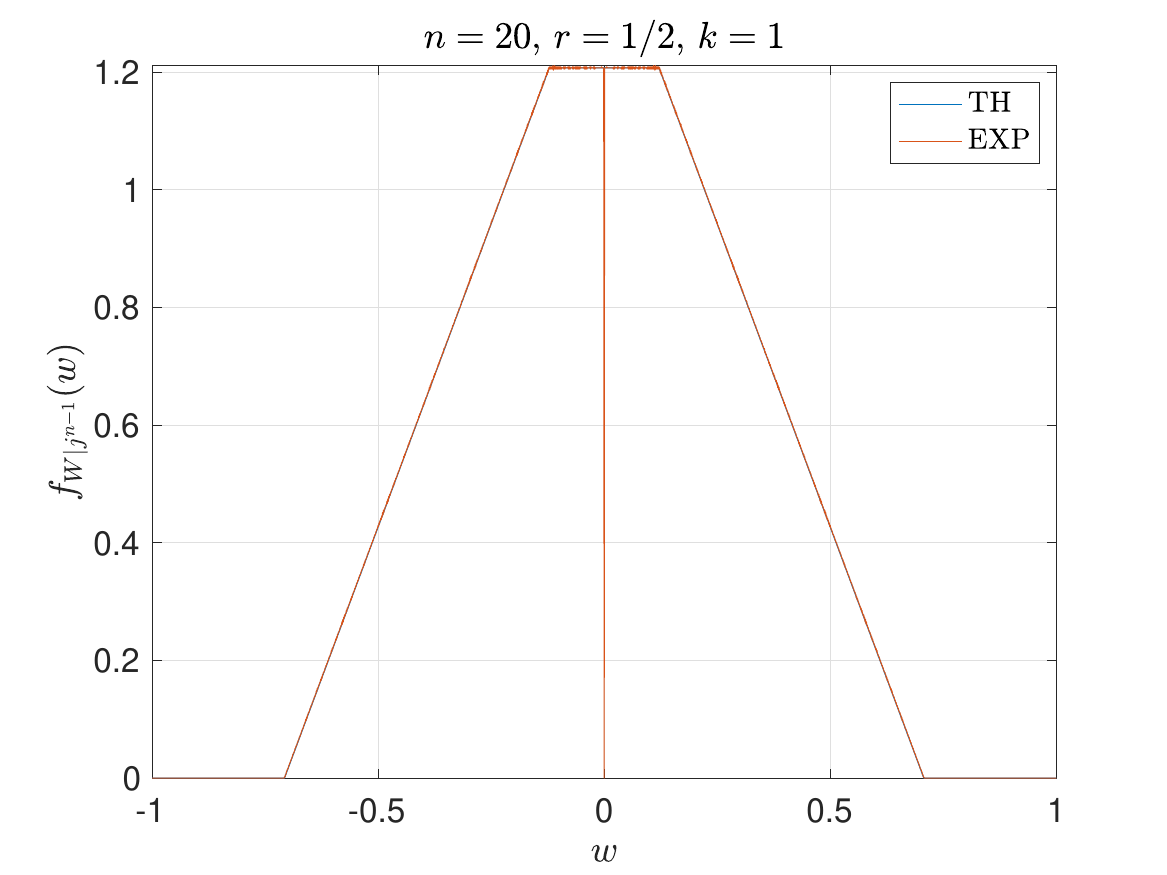}\label{subfig:fWjn1_1}}\\
	\subfigure[]{\includegraphics[width=.5\linewidth]{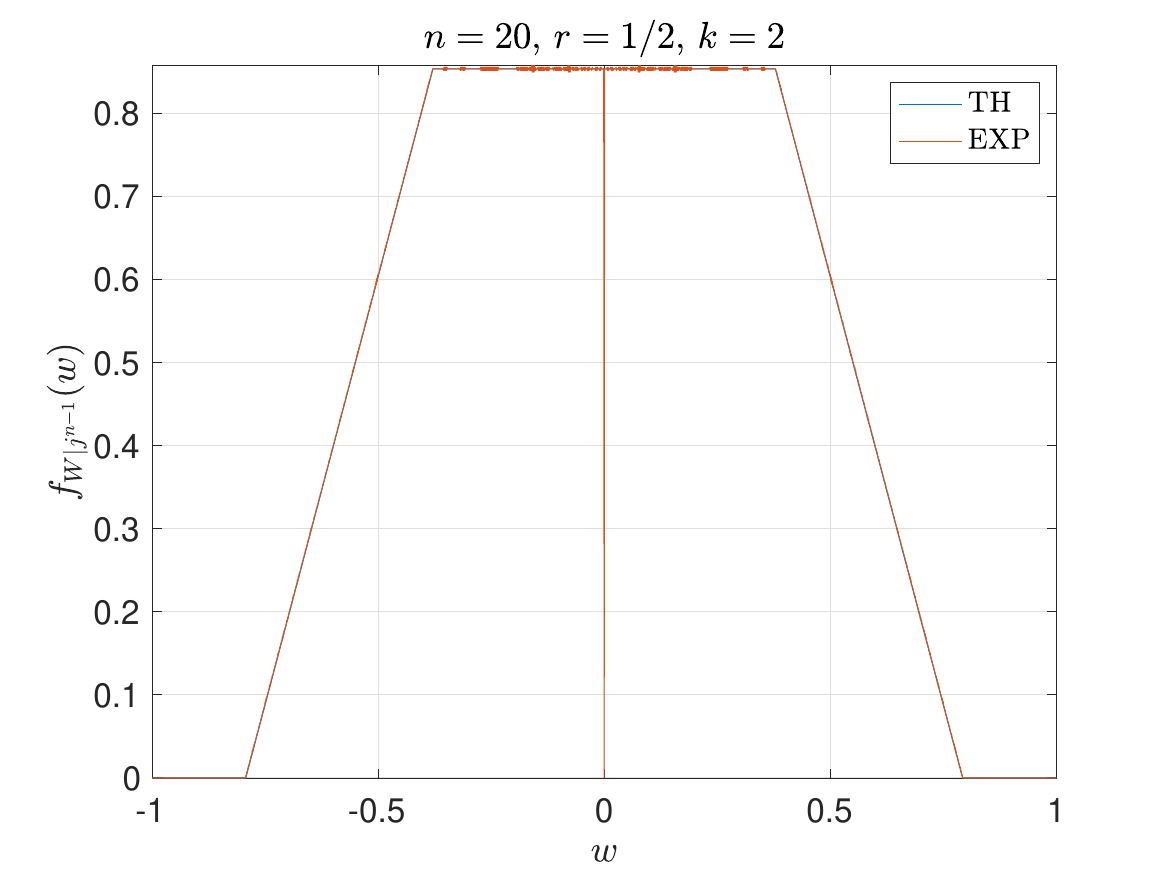}\label{subfig:fWjn1_2}}%
	\subfigure[]{\includegraphics[width=.5\linewidth]{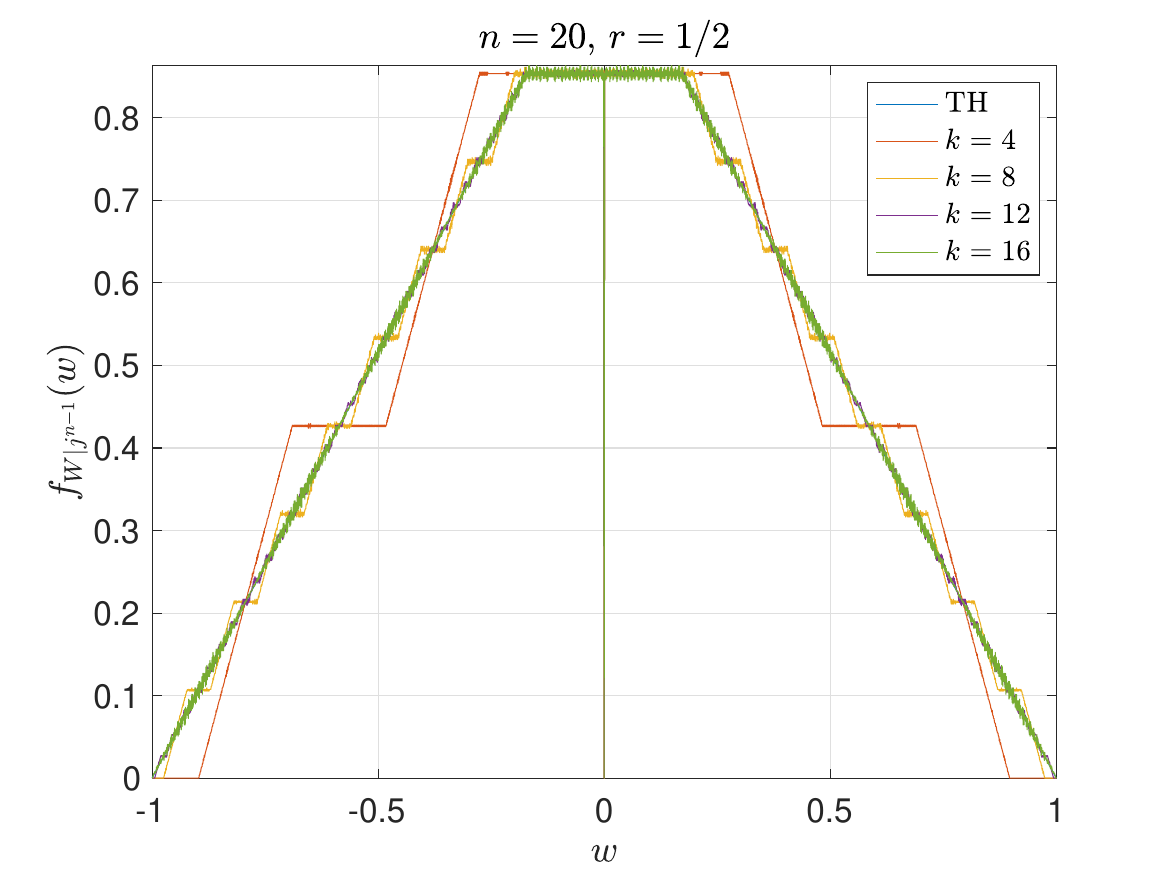}\label{subfig:fWjn1}}
	\caption{Distribution of normalized shift function. The TH curves in (a) and (d) are given by \eqref{eq:fWn}, while the TH curves in (b) and (c) are given by \eqref{eq:fWn1k1} and \eqref{eq:fWn1k2}, respectively. In (b), (c), and (d), the parameter $k$ is defined by \eqref{eq:jn1}. Note that for every EXP/experimental curve, there is a zero point at $w=0$, because $\tau(j^d,b^d)\neq 0$ for every $j^d\in{\cal J}_{n,d}$ and every $b^d\in\mathbb{B}^d$ when $r=1/2$.}
	\label{fig:fW}
\end{figure*}

\begin{figure*}[!t]
	\centering
	\subfigure[]{\includegraphics[width=.5\linewidth]{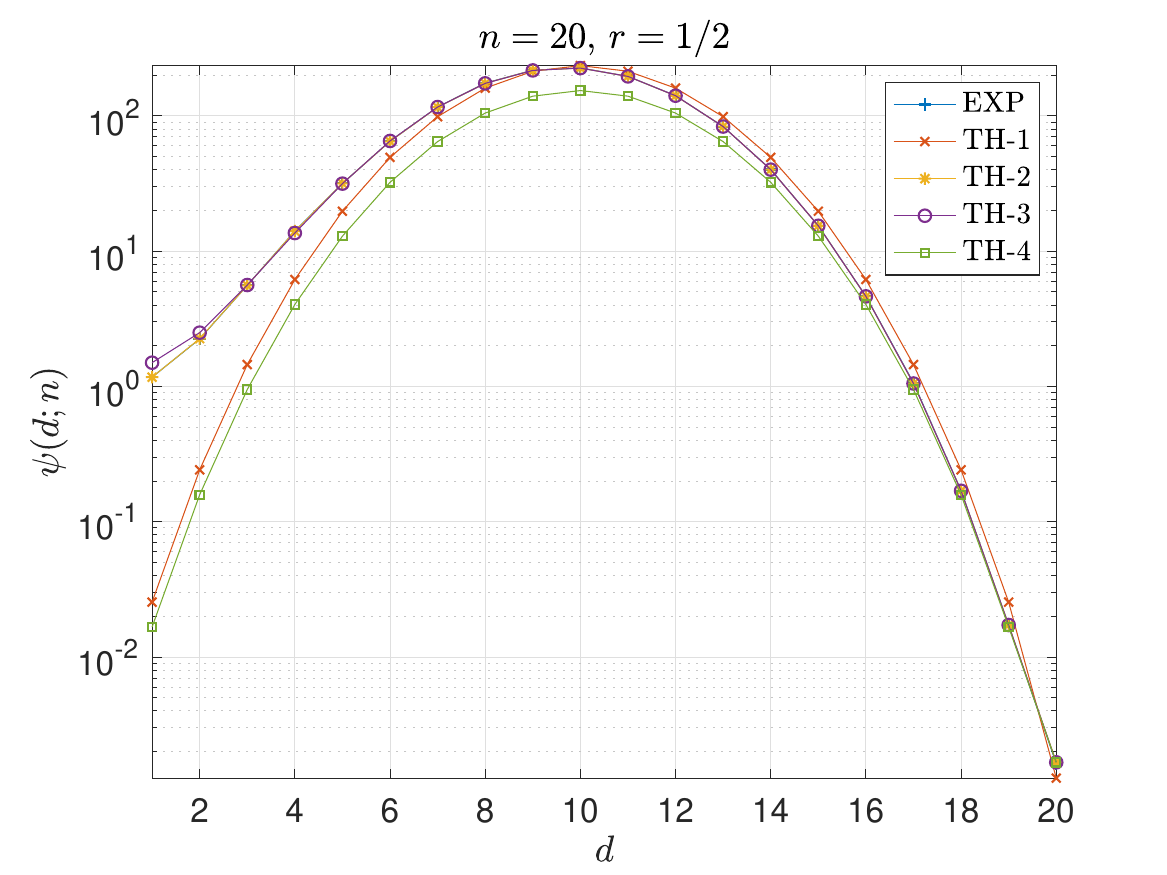}\label{subfig:hds}}%
	\subfigure[]{\includegraphics[width=.5\linewidth]{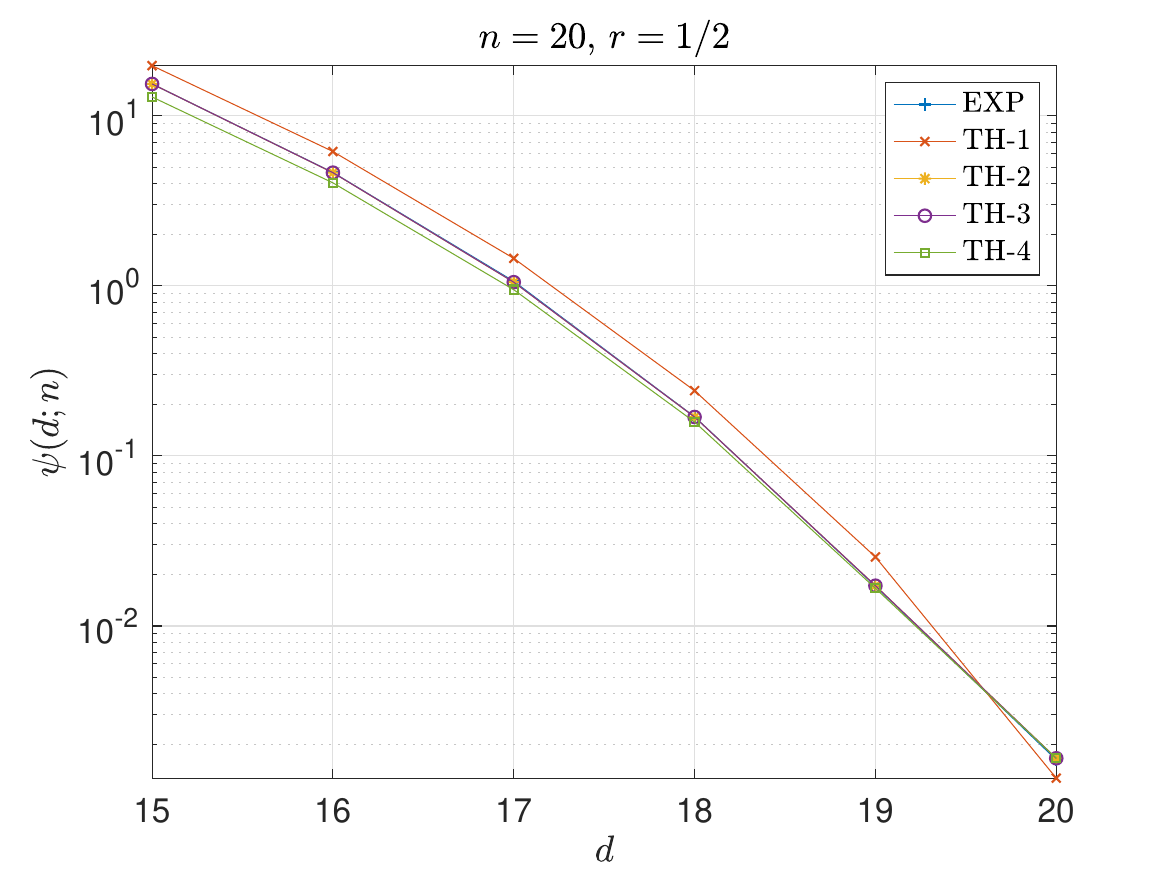}\label{subfig:hds_zoom}}
	\caption{HDS obtained by different approximate formulas, where TH-1, TH-2, TH-3, and TH-4 are defined by \cref{tab:summary}, and EXP refers to experimental curves.}
	\label{fig:hds}
\end{figure*}

\subsection{HDS Obtained by Different Methods}\label{subsec:larged}
Now we verify the four approximate formulas of $\psi(d;n)$ listed in \cref{tab:summary}. When $r=1/2$, we have $f(1/2)=\frac{1}{2-\sqrt{2}}\approx 1.7$, so for the TH-4/fast approximate formula \eqref{eq:fast}, $\psi(n;n) \approx 1.7\times 2^{-n/2}\binom{n}{n} \approx 0.0017$ and $\psi(d;n) \approx 1.7\times 2^{-n/2-1}\binom{n}{d}$ for $d<n$, where $n=20$; while for the TH-1/binomial approximate formula \eqref{eq:coarsehds}, we have
\begin{align*}
	\int_0^1{f^2(u)\,du} 
	&= \frac{2}{(3\sqrt{2}-4)^2}\int_{0}^{\sqrt{2}-1}u^2\,du + \frac{1}{(2-\sqrt{2})^2} \int_{\sqrt{2}-1}^{2-\sqrt{2}}1\,du\\
	&= \frac{2(\sqrt{2}-1)^3}{3(3\sqrt{2}-4)^2} + \frac{3-2\sqrt{2}}{(2-\sqrt{2})^2}
	= \frac{1}{3(\sqrt{2}-1)} + 1/2 \approx 1.3
\end{align*}
and hence $\psi(d;n)\approx 1.3\times 2^{-n/2}\binom{n}{d}$ for every $d$. The whole HDS for all $d\in[1:n]$ is shown by \cref{subfig:hds}, while for clarity, \cref{subfig:hds_zoom} zooms in on the partial HDS for large $d$. We have the following findings. 
\begin{itemize}
	\item The TH-1/binomial formula \eqref{eq:coarsehds} is only a coarse approximation of the EXP/experimental curve. As a rule of thumb, with \eqref{eq:coarsehds}, we will get a smaller value of $\psi(d;n)$ than its real value for $d<n/2$, and a larger value than its real value for $d>n/2$, implying that overlapped arithmetic codes are worse than random codes \cite{FangTCOM16a}. In other words, \eqref{eq:coarsehds} is relatively accurate only for $d\approx n/2$. 
	
	\item The TH-2/soft approximate formula \eqref{eq:psith2} perfectly coincides with the EXP/experimental curve for all $[1:n]$, however the cost is high complexity for large $d$. 
	
	\item The TH-3/hard approximate formula \eqref{eq:psith3} almost coincides with the EXP/experimental curve for large $d$, but there will be a large deviation as $d$ decreases. 
	
	\item The TH-4/fast approximate formula \eqref{eq:fast} coincides with the EXP/experimental curve perfectly for $d\approx n$, but as $d$ decreases, it will be much lower than the EXP/experimental curve. 
\end{itemize}
In one word, all theoretical predictions summarized in \cref{tab:summary} are perfectly verified by \cref{fig:hds}.

Finally, based on theoretical analyses and experimental results, we give the following suggestion. That is, if we want to calculate the HDS for an overlapped arithmetic code with a good compromise between accuracy and complexity, the best way is to use the TH-2/soft formula \eqref{eq:psith2} for $d\approx 1$ and the TH-4/fast formula \eqref{eq:fast} for $d\approx n$, while use the TH-1/binomial formula \eqref{eq:coarsehds} for other $d$.

\chapter{Tailed Overlapped Arithmetic Codes}\label{c-tailed}
\vspace{-25ex}%
From \cref{fig:ccsdec}, the reader may have found that the performance of overlapped arithmetic coeds is very poor and the residual decoding errors cannot be removed even though the source and side information are strongly correlated \cite{FangTCOM16a}. This finding is very similar to the well-known error-floor phenomenon of LDPC codes, and can be theoretically explained with HDS. From \cref{tab:summary}, the reader may have found an important property of overlapped arithmetic codes: For small Hamming distance $d$, \textit{e.g.}, $1$, $\psi(d;n)$ does not converge to $0$ as block length $n\to\infty$. Since the minimum Hamming distance of overlapped arithmetic codes is always $1$, closely-spaced codewords in the same coset cannot be removed by increasing code length, which explains why the performance of overlapped arithmetic codes is so poor. 

The performance of overlapped arithmetic codes can be improved by mapping the last few source symbols of a block, called {\em tail},  to non-overlapped intervals \cite{GrangettoCL07,GrangettoTSP09}. This improvement can be theoretically explained from the viewpoint of HDS: By mapping the tail of a block to non-overlapped intervals, the minimum Hamming distance of overlapped arithmetic codes may be increased and closely-spaced codewords in the same coset can be removed \cite{FangTCOM16b}.

\section{Principle}\label{sec:tdacencoding}
Now we divide each source block $x^n$ into body $x^{n-t}$ and tail $x_{n-t+1}^n$. Every tail symbol is always coded at rate $1$, so the mapping rule is: $0\rightarrow [0,2^{-1})$ and $1\rightarrow[2^{-1},1)$. To compress $x^n$ at the average rate $R$, every body symbol should be coded at rate 
\begin{align}\label{eq:r}
	r = \frac{nR-t}{n-t}\leq R, 
\end{align} 
so the mapping rule is: $0 \rightarrow [0,2^{-r})$ and $1\rightarrow [(1-2^{-r}),1)$. Obviously, $r$ is strictly decreasing w.r.t. $t$. To guarantee $r\geq 0$, $t$ can not be larger than $nR$, so $t\in[0:nR]$, where $nR$ is assumed to be an integer for simplicity. If $t=0$ and $r=R$, it is called a {\em tailless} overlapped arithmetic code; while if $t>0$ and $r<R$, it is called a {\em tailed} overlapped arithmetic code. According to \eqref{eq:lXi}, the final interval after coding $x^n$ is $[l(x^n),l(x^n)+2^{-nR})$, where
\begin{align}
	l(x^n) 
	&= (2^r-1)\sum_{i=1}^{n-t}{x_i 2^{-ir}} + 2^{-(n-t)r}\sum_{i=1}^{t}{x_{i+(n-t)} 2^{-i}}\nonumber\\
	&\overset{(a)}{=} (2^r-1)\sum_{i=1}^{n-t}{x_i 2^{-ir}} + 2^{n(1-R)}\sum_{i=n-t+1}^{n}{x_i 2^{-i}},
\end{align}
where $(a)$ comes from $(n-t)(1-r)=n(1-R)$. Since $nR=(n-t)r+t$, according to \eqref{eq:ell}, we have
\begin{align}\label{eq:newsxn}
	s(x^n) = 2^{nR}l(x^n) \overset{d}{=} 2^t(1-2^{-r})\sum_{i=1}^{n-t}{x_i2^{ir}} + \sum_{i=1}^{t}{x_{i+(n-t)}2^{i-1}}.
\end{align}
The normalized final interval is $[s(x^n),s(x^n)+1)$ and $s(x^n) \leq (2^{nR}-1)$. It can also be seen that $m(x^n)\triangleq\lceil s(x^n)\rceil\in[0:2^{nR})$.

\section{Hamming Distance Spectrum}\label{sec:encoding}
Let $x_i\in\mathbb{B}$, $y_i\in\mathbb{B}$, and $z_i=x_i\oplus y_i\in\mathbb{B}$. Let $\{\pm 1\}\triangleq\{-1,1\}$ and $\{0,\pm1\}\triangleq \{-1,0,1\}$. We define the following ternary variable 
\begin{align*}
	w_i \triangleq (x_i\oplus y_i)(1-2x_i) = z_i(1-2x_i) \in \{0,\pm 1\}.
\end{align*}
Clearly, $w_i=0$ if $z_i=0$ and $w_i=\pm 1$ if $z_i=1$. Conversely, if $w_i=\pm 1$, we have $z_i=1$ and $x_i=(1-w_i)/2$; otherwise, \textit{i.e.}, $w_i=0$, we can get $z_i=0$, but $x_i$ is unknown. Thus, given $w^n$, $z^n$ is fully determined, but $x^n$ is partially determined. Let $d=|w^n|$ be the support (number of nonzero elements) of $w^n$ and let
\begin{align*}
	j^d = \{j_1,\dots,j_d\} \subseteq [n]\triangleq\{1,\dots,n\},
\end{align*}
where $1\leq j_1<\dots<j_d\leq n$, be the set of the indices of nonzero elements of $w^n$. In other words, $w_i=\pm 1$ if $i\in j^d$ and $w_i=0$ if $i\in[n]\setminus j^d$. Actually, every $w^n$ leads $2^{n-d}$ different $x^n$'s. From $x^n$, we draw all elements indexed by $j^d$ to form a sub-vector $x_{j^d}=b^d\in\mathbb{B}^d$. Then $w_{j^d} = z_{j^d}(1-2x_{j^d}) = (1-2x_{j^d})\in\{\pm 1\}^d$. Now we define
\begin{align}\label{eq:newtau}
	\tau_{n,r}(w^n) 
	&\triangleq (1-2^{-r})\sum_{i=1}^{n}{w_i 2^{ir}} = (1-2^{-r})\sum_{i\in j^d}{w_i 2^{ir}}\nonumber\\
	&= (1-2^{-r})\sum_{i\in j^d}{(1-2x_i)2^{ir}} = \tau(j^d,b^d),	
\end{align}
where $\tau(j^d,b^d)$ is defined by \eqref{eq:tau}. Following \eqref{eq:newsxn}, we define
\begin{align}\label{eq:tauTailed}
	\tau(w^n;n,t,r) 
	&\triangleq 2^t(1-2^{-r})\sum_{i=1}^{n-t}{w_i 2^{ir}} + 
	\sum_{i=1}^{t}{w_{i+(n-t)}2^{i-1}}\nonumber\\
	&= 2^t \tau_{n-t,r}(w^{n-t}) + \tau_{t,1}(w_{n-t+1}^n).
\end{align}
According to \eqref{eq:softhds}, the following theorem holds obviously.
\begin{theorem}[HDS of Tailed Overlapped Arithmetic Codes]
	Let $t$ be tail length of an overlapped arithmetic code. If body symbols are coded at rate $r$, then as code length $n\to\infty$, 
\begin{align}
	\psi(d;n) 
	&\to 2^{\alpha-d}\sum_{w^n:|w^n|=d}{\left(1-|\tau(w^n;n,t,r)|\right)^+}\nonumber\\
	&= 2^{\alpha-d}\sum_{w^n:|w^n|=d}{\left(1-|2^t \tau_{n-t,r}(w^{n-t}) + \tau_{t,1}(w_{n-t+1}^n)|\right)^+},
\end{align}
where $\alpha = {\bf 1}_{d=n}$ and $(\cdot)^+ = \max(0,\cdot)$.
\end{theorem}

\section{Closely-Spaced Mate Codewords}
Experiments show that tailed overlapped arithmetic codes are better than tailless overlapped arithmetic codes \cite{GrangettoCL07,GrangettoTSP09}. Below, we will reveal that closely-spaced mate codewords can be removed by increasing tail length, which explains the superiority of tailed overlapped arithmetic codes over tailless overlapped arithmetic codes. Due to computation complexity and to keep the exposition simple, we consider below only $\psi(1;n)$ and $\psi(2;n)$. To begin with, we give the following lemma.

\begin{lemma}\label{lem:tail}
	If $|z^{n-t}|=0$ and $|z_{n-t+1}^n|>0$, then $m(x^n) \neq m(x^n\oplus z^n)$. Conversely, if $|z^n|>0$ and $m(x^n) = m(x^n\oplus z^n)$, then $|z^{n-t}|>0$.
\end{lemma}
\begin{proof}
	If $|z^{n-t}|=0$, then $\tau(w^n;n,t,r)=\tau_{t,1}(w_{n-t+1}^n)\in\mathbb{Z}$, implying $m(x^n) \neq m(x^n\oplus z^n)$ for any $|z_{n-t+1}^n|>0$.
\end{proof}

In plain words, if a pair of codewords coexist in the same coset, it is impossible that these two codewords differ from each other only in tails. With the help of \cref{lem:tail}, we explain below why $\psi(1;n)$ and $\psi(2;n)$ will tend to $0$ as $t$ increases.

\begin{theorem}[$1$-Away Mate Codeword Pairs]
	As tail length $t$ increases, $1$-away mate codeword pairs will be removed and hence the minimum Hamming distance of overlapped arithmetic codes is at least $2$.
\end{theorem}
\begin{proof}
Given $|w^n|=1$, there are only two cases: 
\begin{itemize}
	\item $|w^{n-t}|=0$ and $|w_{n-t+1}^n|=1$; 
	\item $|w^{n-t}|=1$ and $|w_{n-t+1}^n|=0$. 
\end{itemize}
According to \cref{lem:tail}, in the former case, $m(x^n) \neq m(x^n\oplus z^n)$ holds always. Hence, we need to consider only the later case, which is followed by $\tau(w^n;n,t,r) = 2^t\tau_{n-t,r}(w^{n-t})$. Given $|w^{n-t}|=1$, we can obtain from \eqref{eq:tauTailed} that
\begin{align*}
	|\tau_{n-t,r}(w^{n-t})| \geq (1-2^{-r})2^r = (2^r-1)
\end{align*}
and further $|\tau(w^n;n,t,r)| \geq 2^t(2^r -1)$. If we suppose $t \leq {nR \over 2}$, it can be deduced from \eqref{eq:r} that $r \geq {R \over 2-R}$ and thus
\begin{align}\label{eq:temp}
	\tau(w^n;n,t,r) = 2^t (2^r -1) \geq 2^t (2^{R/(2-R)} -1). 
\end{align}
Increasing $t$ will certainly make the right hand side of \eqref{eq:temp} greater than $1$. An example of $2^t(2^r-1)$ is given in \cref{subfig:variation} for $n=64$ and $R=0.5$, which shows that $2^t(2^r-1)$ increases monotonously for small $t$. Hence, it is possible to make $|\tau(w^n;n,t,r)|\geq 1$ hold always for all $w^n$'s satisfying $|w^n|=1$ by increasing $t$. In other words, $1$-away mate codeword pairs can be removed by simply increasing tail length $t$, even for finite code length $n$.
\end{proof}

\begin{theorem}[$2$-Away Mate Codewords Pairs]\label{thm:2away}
	Given $n/t=\infty$, at least for $R>\log_2\frac{3}{2}\approx 0.585$, $2$-away mate codeword pairs can be removed by simply increasing tail length $t$ and the minimum Hamming distance of overlapped arithmetic codes is at least $3$.
\end{theorem}
\begin{proof}
Given $|w^n|=2$, there are only three cases: 
\begin{itemize}
	\item $|w^{n-t}|=0$ and $|w_{n-t+1}^n|=2$; 
	\item $|w^{n-t}|=2$ and $|w_{n-t+1}^n|=0$; 
	\item $|w^{n-t}|=|w_{n-t+1}^n|=1$. 
\end{itemize}
According to \cref{lem:tail}, in the first case, $m(x^n)\neq m(x^n\oplus z^n)$, so we need to consider only the later two cases.

In the case of $|w^{n-t}|=2$ and $|w_{n-t+1}^n|=0$, we have $\tau(w^n;n,t,r) = 2^t\tau_{n-t,r}(w^{n-t})$. It is easy to get from \eqref{eq:tauTailed} that
\begin{align*}
	|\tau_{n-t,r}(w^{n-t})| 
	&\geq (1-2^{-r})(2^{2r}-2^{r})\\
	&= (2^r-1)^2.
\end{align*}
Still suppose $t \leq {nR \over 2}$, which is followed by 
\begin{align}\label{eq:2t}
	2^t(2^r-1)^2 \geq 2^t(2^{R/(2-R)}-1)^2.
\end{align}
Increasing $t$ will make the right hand side of \eqref{eq:2t} not smaller than $1$ and in turn $|\tau(w^n;n,t,r)|\geq 2^t(2^r-1)^2\geq 1$ holds always for all $w^n$'s satisfying $|w^{n-t}|=2$ and $|w_{n-t+1}^n|=0$. An example of $2^t(2^r-1)^2$ is shown in \cref{subfig:variation}, where $n=64$ and $R=0.5$.

Finally, we consider the case of $|w^{n-t}|=|w_{n-t+1}^n|=1$. We have 
\begin{itemize}
	\item $|\tau_{n-t,r}(w^{n-t})| = (2^r-1)2^{ir}$, where $i\in[0:(n-t))$; and
	\item $|\tau_{t,1}(w_{n-t+1}^n)| = 2^{j}$, where $j\in[0:t)$. 
\end{itemize}
For simplicity, let $\gamma_i \triangleq (2^r-1)2^{t+ir}$. Then $|\tau(w^n;n,t,r)| = |\gamma_i \pm 2^{j}|$, where $(i,j)\in[0:(n-t))\times[0:t)$. 
\begin{itemize}
	\item Since $\gamma_i>0$ and $2^{j}\geq 1$, we have $|\gamma_i + 2^{j}| > 1$.	
	\item Since $\gamma_i\geq (2^r-1)2^t$ and $2^j\leq 2^{t-1}$, we have 	
	\begin{align*}
		\gamma_i-2^{j} \geq (2^r-1)2^t-2^{t-1} \geq 2^t(2^r-3/2). 
	\end{align*}
	Given $n/t=\infty$, we have $r=R$. If $(2^R-3/2)>0$, we can always make $2^t(2^R-3/2)\geq1$ by increasing $t$.
\end{itemize}
Based on the above analyses, we conclude that $2$-away mate codeword pairs can be removed by simply increasing tail length $t$.
\end{proof}

We can understand \cref{thm:2away} from a more intuitive viewpoint. There are $(n-t)\times t$ possible values for $(i,j)\in[0:(n-t))\times[0:t)$. For $n$ very large and $t\ll n$, we have $r\approx R$ and thus $(\gamma_{i+1} - \gamma_i)$ increases w.r.t. $t$, \textit{i.e.}, $\gamma_i$'s tend to be sparser as $t$ increases. Hence as $t$ increases, it is less likely that $\gamma_i$'s fall within $(2^j-1,2^j+1)$'s. An example is given in \cref{subfig:delta} to confirm this point, where
\begin{align}\label{eq:delta}
	\Delta_i \triangleq \min(\gamma_i - 2^{\lfloor\log_2{\gamma_i}\rfloor}, 2^{\lceil\log_2{\gamma_i}\rceil} - \gamma_i).
\end{align}
As shown by \cref{subfig:delta}, as $t$ increases, fewer $\Delta_i$'s will be less than $1$ and when $t>6$, there is no $\Delta_i$ less than $1$. Therefore, it is possible to make $|\tau(w^n;n,t,r)| = |\gamma_i\pm 2^j| \geq 1$ hold always for all $w^n$'s satisfying $|w^{n-t}|=|w_{n-t}^n|=1$ by increasing $t$.

\begin{figure*}[!t]
	\centering
	\subfigure[]{\includegraphics[width=.5\linewidth]{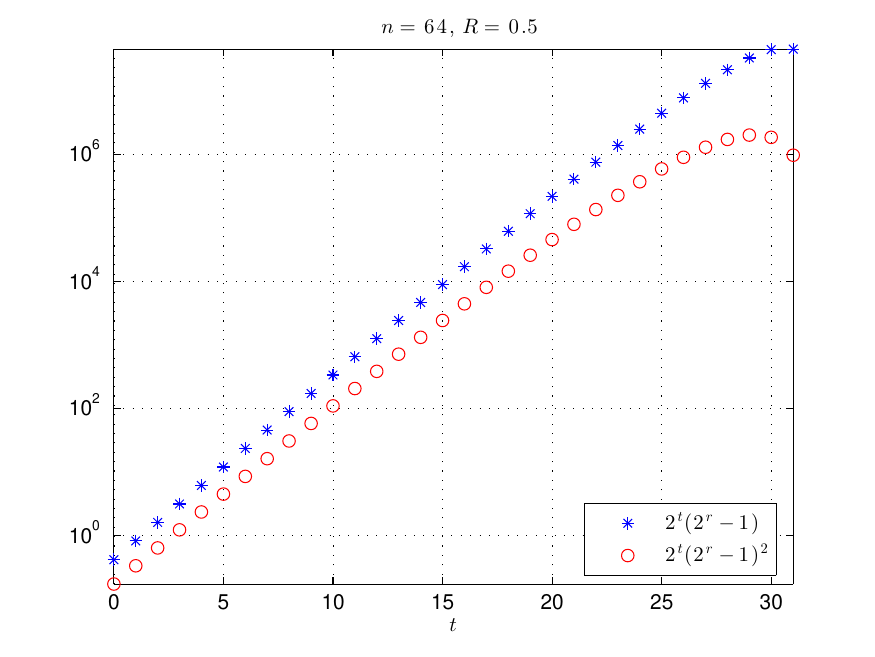}\label{subfig:variation}}%
	\subfigure[]{\includegraphics[width=.5\linewidth]{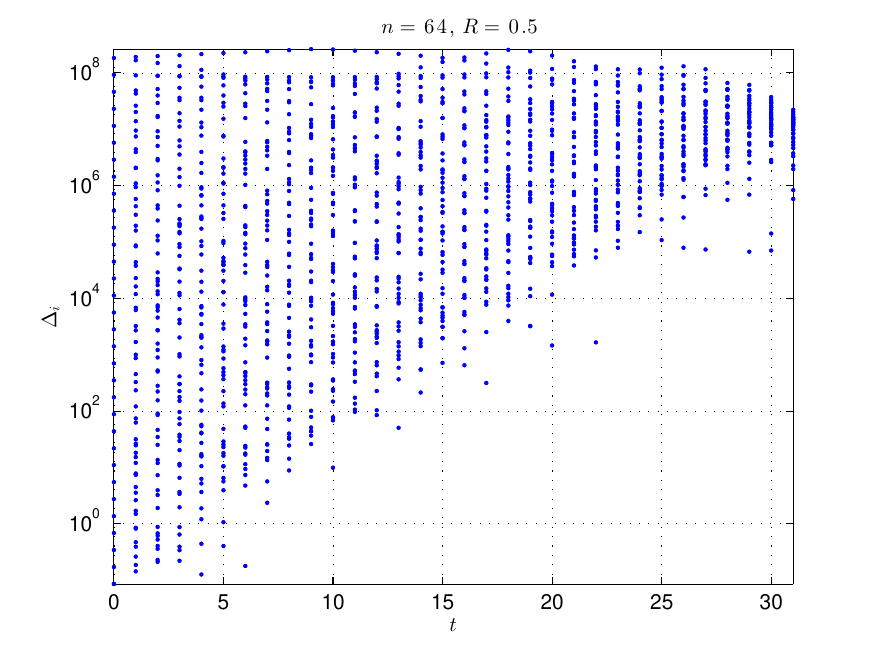}\label{subfig:delta}}
	\caption{
		(a) $2^t(2^r-1)$ and $2^t(2^r-1)^2$ versus tail length $t$. 
		(b) Distribution of $\Delta_i$, which is defined by \eqref{eq:delta}, versus tail length $t$.
	}
	\label{fig:removing}
\end{figure*}

The analysis of the general case $d>2$ is more complex, but the principle is similar to the above cases. For $n$ sufficiently large and $d\ll n$, $\psi(d;n)$ will tend to $0$ as $t$ increases. However, note that as $d$ increases, larger $t$ is required to make $\psi(d;n) = 0$.

\section{Effects of Tail Length}\label{subsec:lengthVSloss}
Following \eqref{eq:Uin}, we define the bitstream projection of tailed overlapped arithmetic codes. According to \eqref{eq:Ui}, we have
\begin{align*}
	\begin{cases}
		U_{i-1} = 2^{-r}U_i + X_i(1-2^{-r}), & i\in[1:(n-t)]\\
		U_{i-1} = (U_i+X_i)/2, 				 & i\in[(n-t+1):n].
	\end{cases}
\end{align*}
Especially, $U_0=2^{-nR}\lceil s(x^n)\rceil$ and $U_n=\lceil s(x^n)\rceil-s(x^n)$. We also define CCS for tailed overlapped arithmetic codes. Following the analyses in \cref{sec:ccs}, we can easily know $f_i(u)=\Pi(u)$ for $i\in[(n-t):n]$ and
\begin{align}\label{eq:ccstailed}
	f_{i-1}(u) = 2^{r-1}\left(f_i(u2^r) + f_i((u-(1-2^{-r}))2^{r})\right)
\end{align}
for $i\in[1:(n-t)]$. The overall metric of the $x^i$-path can still be given by \eqref{eq:logmusi}. Especially for $i\in[(n-t):n]$, the term $\log{f_i(u(x^i))}$ can be ignored because $f_i(u)=\Pi(u)$, and hence \eqref{eq:logmusi} can be reduced to
\begin{align*}
	\sigma(x^i) = |x^i|\cdot\log{\epsilon} + (i-|x^i|)\cdot\log(1-\epsilon).
\end{align*}
Similarly, following the analyses in \cref{sec:infthm}, the block rate loss of tailed overlapped arithmetic codes is 
\begin{align*}
	\lim_{n\to\infty}(nR-H(M)) = -h(U) = \int_{0}^{1}f(u)\log_2{f(u)}\,du \geq 0,
\end{align*}
where $U$ is the asymptotic initial projection defined by \eqref{eq:U} and $f(u)$ is the asymptotic initial CCS defined by \eqref{eq:asympt}. For $n$ sufficiently large, the cardinality of ${\cal C}_m$ is
\begin{align*}
	|{\cal C}_m| \approx f(m2^{-nR}) \cdot 2^{n(1-R)} = f(m2^{-nR}) \cdot 2^{(n-t)(1-r)}
\end{align*}
and the {\em Expected Coset Cardinality} (ECC) is
\begin{align*}
	\mathbb{E}[|{\cal C}_M|] \approx 2^{n(1-R)}\int_{0}^{1}{f^2(u)\,du} = 2^{(n-t)(1-r)}\int_{0}^{1}{f^2(u)\,du}.
\end{align*}
The normalized ECC is
\begin{align*}
	\lim_{n\to\infty}\frac{\mathbb{E}[|{\cal C}_M|]}{2^{n(1-R)}} = \lim_{n\to\infty}\frac{\mathbb{E}[|{\cal C}_M|]}{2^{(n-t)(1-r)}} = \int_{0}^{1}{f^2(u)\,du} \geq 1.
\end{align*}

We have learned that from the viewpoint of HDS, larger tail length $t$ has a positive effect of removing closely-spaced mate codewords. Now we discuss the effect of tail length $t$ from the viewpoint of CCS. For fairness, the average rate $R$ of $x^n$ should be fixed by coding the body $x^{n-t}$ at rate $r=\frac{nR-t}{n-t}$, which is monotonously decreasing w.r.t. $t$. As analyzed in \cref{sec:ccs}, the asymptotic initial CCS $f(u)$ will be spikier and spikier as $r$ decreases. Two extreme cases are $f(u)=\Pi(u)$ for $r=1$ and $f(u)=\delta(u-0.5)$ for $r=0$. Hence both $-h(U)$ and $\int_{0}^{1}{f^2(u)\,du}$ are monotonously decreasing w.r.t. $r$ and in turn monotonously increasing w.r.t. $t$. That means, greater $t$ will lead to larger rate loss and higher decoding complexity (as we know, the expected number of codewords searched by a full-search decoder is equal to the ECC). Thus from the viewpoint of CCS, larger tail length $t$ has a negative effect of increasing rate loss and decoding complexity. In a word, $t$ should be carefully selected to reach a compromise between HDS and CCS.

\begin{remark}[A Weird Phenomenon When $t=nR$]
If $t=nR$, we have $r=(1-2^{-r})=0$. In this extreme case, $x_{n-t+1}^n$ is transmitted directly without coding, while $x^{n-t}$ is not transmitted. Now \eqref{eq:ccstailed} becomes $f_{i-1}(u)=f_i(u)$ and thus $f_0(u)=\dots=f_n(n)=\Pi(u)$, which leads to $h(U)=0$ and $\int_{0}^{1}{f^2(u)\,du}=1$. Hence, there is no rate loss and the complexity of full-search decoder is minimized.   

Meanwhile, since $r=0$, we have 
\begin{align*}
	\tau(w^n;n,t,r) = \tau(w^n;n,nR,0) = \tau_{t,1}(w_{n-t+1}^n) \in \mathbb{Z}, 
\end{align*}
showing that $\tau(w^n;n,t,r)$ purely depends on $w_{n-t+1}^n$ and has nothing to do with $w^{n-t}$. In addition, since $\tau(w^n;n,t,r)$ is always an integer, $x^n$ and $y^n=(x^n\oplus z^n)$ cannot coexist in the same coset if $z_{n-t+1}^n\neq 0^t$. In other words, $\psi(d;n)=0$ for $d>(n-t)=n(1-R)$. Finally,
\begin{align}\label{eq:extremehds}
	\psi(d;n) = 
	\begin{cases}
		\binom{n(1-R)}{d}, & d\in[0:n(1-R)]\\
		0, & d\in(n(1-R):n],
	\end{cases}		
\end{align}
showing that $\psi(d;n)$ for small $d$ does not converge to $0$, \textit{i.e.}, closely-spaced mate codewords cannot be removed. 
\end{remark}

\section{Experimental Verification}
This section presents four experiments to verify theoretical analyses from different aspects. We use the first experiment to show how tail length $t$ changes the HDS of overlapped arithmetic codes. Then the second experiment reveals how $\psi(d;n)$ for small $d$ varies w.r.t. tail length $t$. After that, the third experiment illustrates the impact of tail length $t$ on rate loss and decoding complexity. Finally, the last experiment demonstrates how the performance of overlapped arithmetic codes can be improved by increasing tail length $t$.

\subsection{Hamming Distance Spectrum}
\begin{figure*}[!t]
	\centering
	\subfigure[]{\includegraphics[width=.5\linewidth]{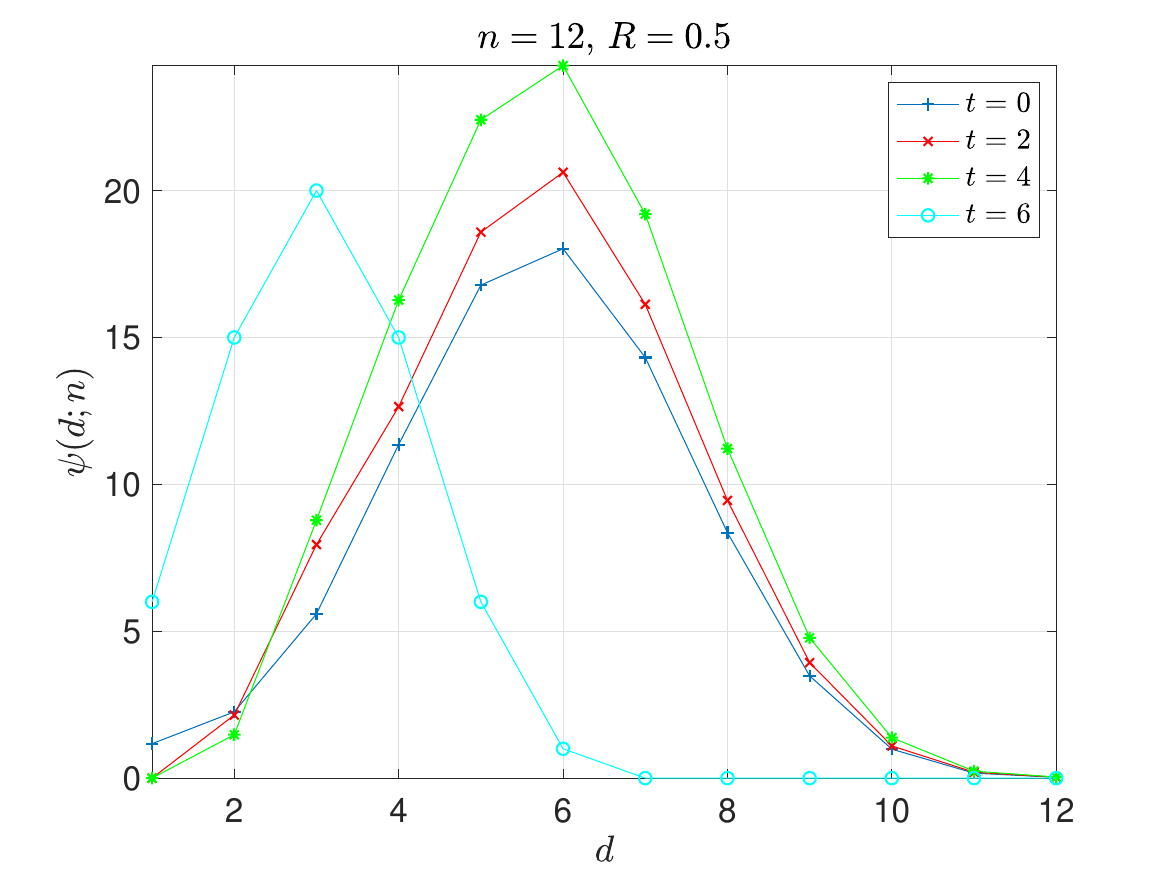}\label{subfig:hds5}}%
	\subfigure[]{\includegraphics[width=.5\linewidth]{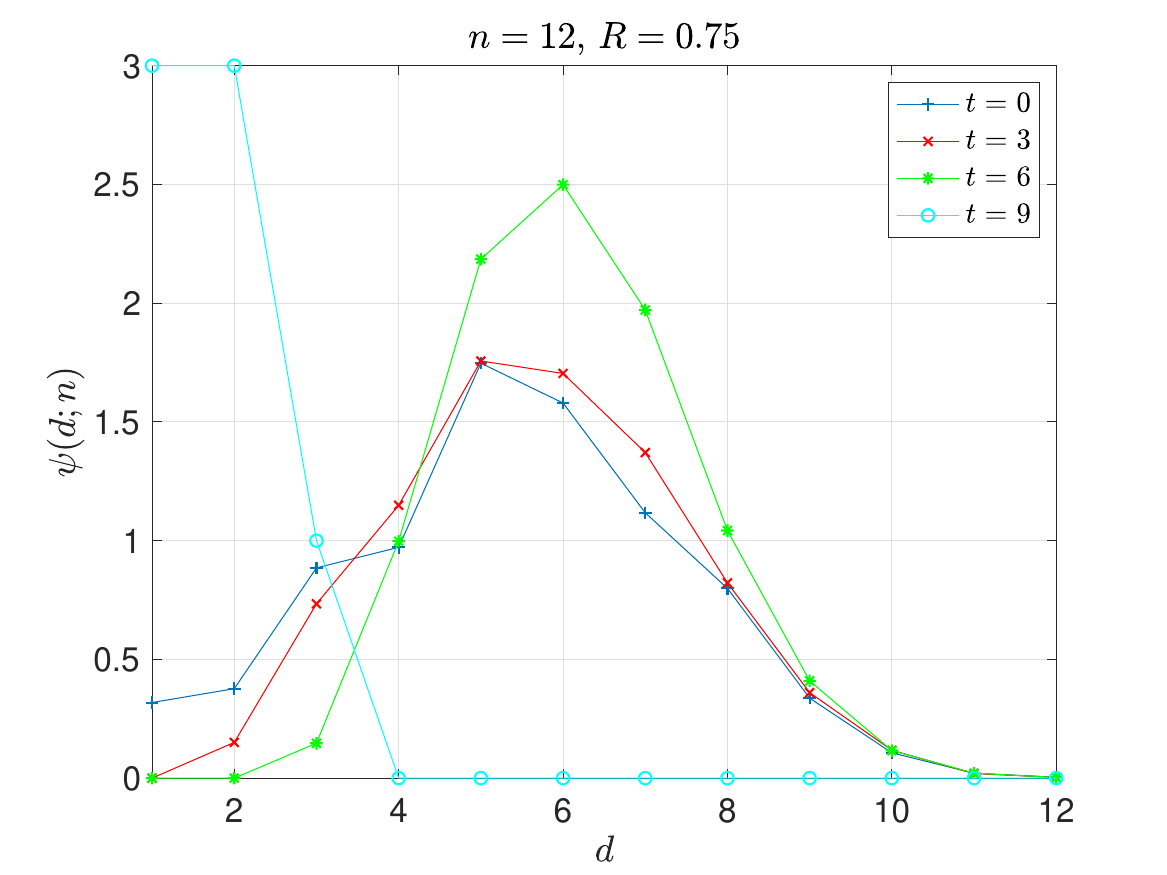}\label{subfig:hds75}}
	\caption{Examples of HDS for different tail lengths.}
	\label{fig:thds}
\end{figure*}

The first experiment studies the effect of tail length $t$ on the HDS of overlapped arithmetic codes. Some examples are given in \cref{fig:thds}, where code length $n=12$. In \cref{subfig:hds5}, we set $R=0.5$ and thus $t\in[0:nR]=[0:6]$; while in \cref{subfig:hds75}, we set $R=0.75$ and thus $t\in[0:nR]=[0:9]$. It can be seen that, as $t$ increases (but not very close to bitstream length $nR$), $\psi(d;n)$ will tend to be smaller (and even become $0$ when $t$ is large) for small $d$ but tend to be larger for large $d$. Note that the HDS when $t=nR$ is very different from the HDS in other cases, confirming the correctness of \eqref{eq:extremehds}.

\subsection{Closely-Spaced Mate Codewords}
\begin{figure*}[!t]
	\centering
	\subfigure[]{\includegraphics[width=.5\linewidth]{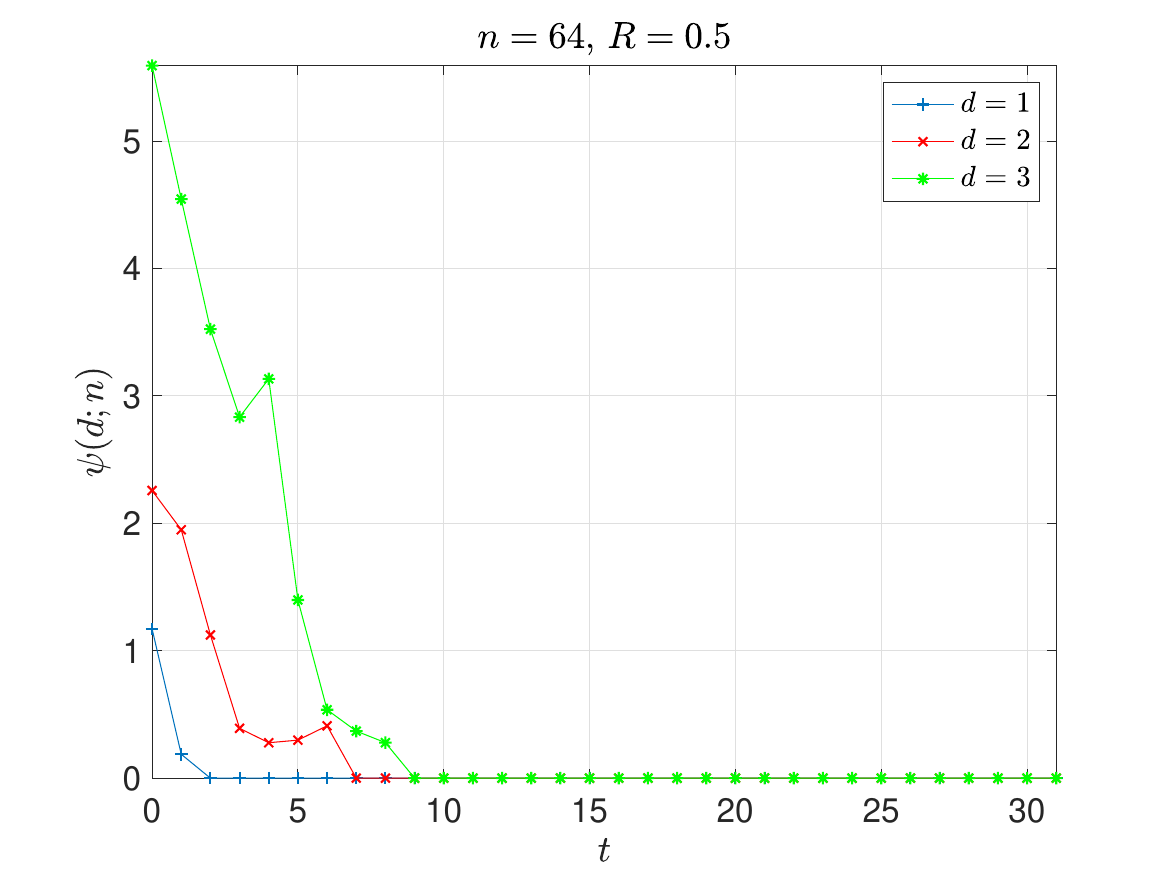}\label{subfig:hdsvstH}}%
	\subfigure[]{\includegraphics[width=.5\linewidth]{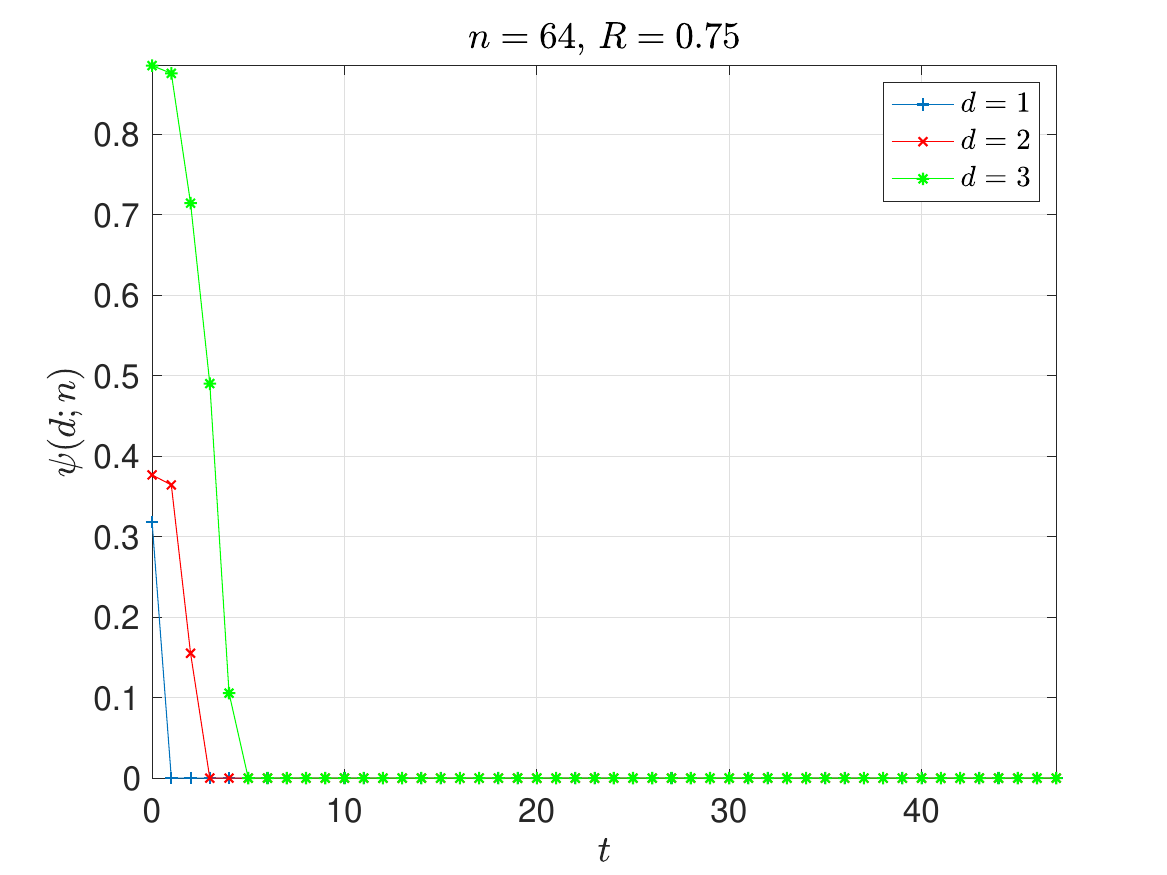}\label{subfig:hdsvst3Q}}
	\caption{Effect of tail length $t$ on $\psi(d;n)$ for small $d$.}
	\label{fig:psivst}
\end{figure*}

The second experiment studies the effect of tail length $t$ on $\psi(d;n)$ for small $d$. Two examples are given in \cref{fig:psivst}, where code length $n=64$ and code rate $R=0.5$ or $0.75$. It can be seen that in general, $\psi(d;n)$ for small $d$ tends to be smaller as $t$ increases, implying that closely-spaced mate codewords can be removed by increasing tail length $t$ (but not very close to bitstream length $nR$). However, it must be pointed out that the decrease of $\psi(d;n)$ for small $d$ w.r.t. $t$ is not strictly monotonous, so tail length $t$ should be carefully chosen in practice to optimize the overall performance.

\subsection{Rate Loss and Decoding Complexity}
\begin{figure*}[!t]
	\centering
	\subfigure[]{\includegraphics[width=.5\linewidth]{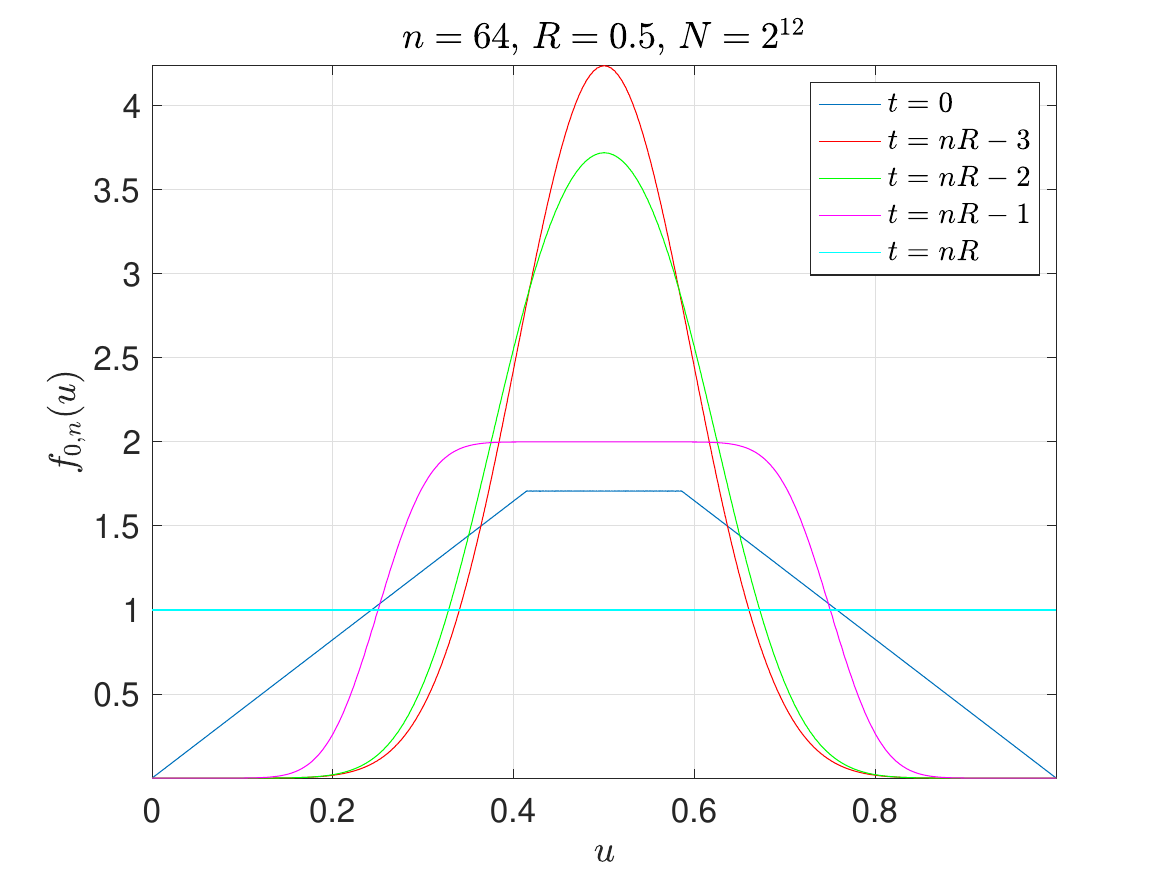}\label{subfig:ccs}}%
	\subfigure[]{\includegraphics[width=.5\linewidth]{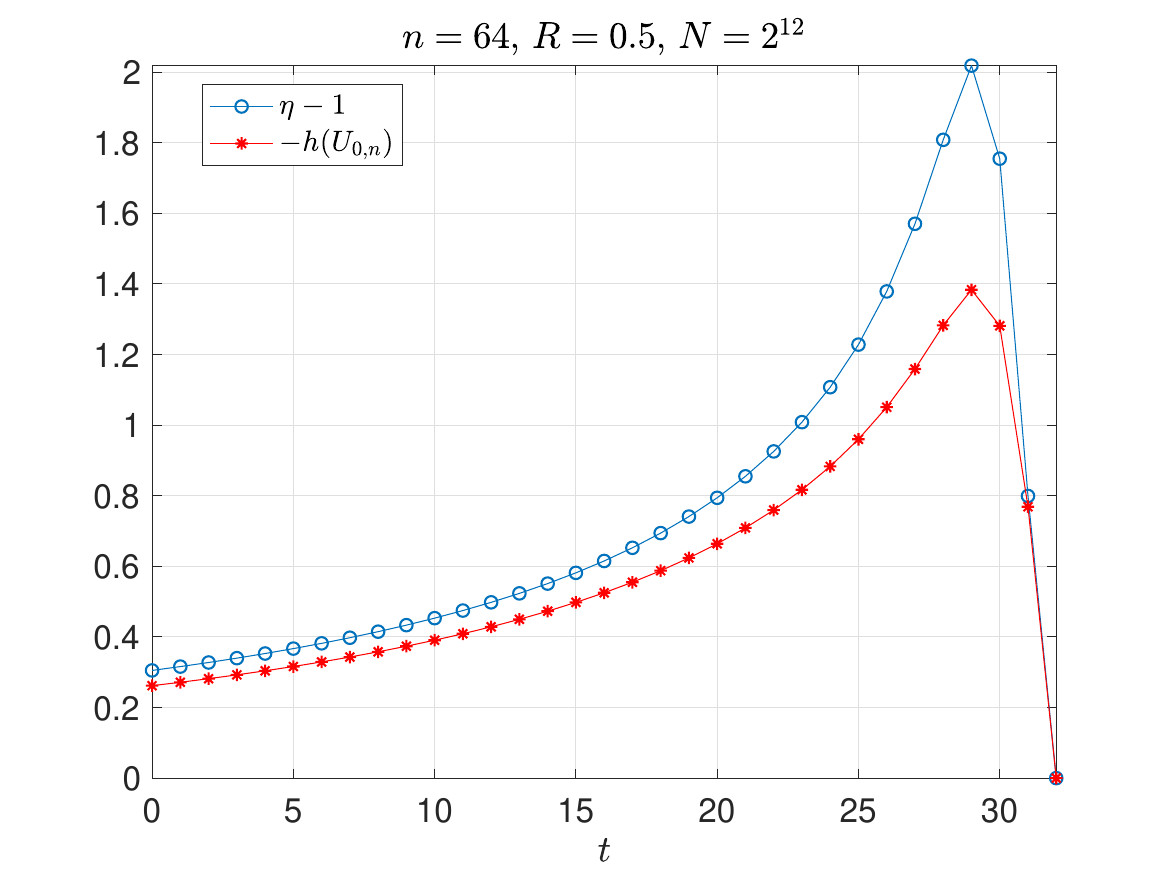}\label{subfig:loss}}
	\caption{Examples of initial CCS, normalized ECC, and rate loss of tailed overlapped arithmetic codes. The results are obtained by the numerical algorithm given in \cref{sec:numerical} with segment number $N=2^{12}$. (a) Initial CCS versus tail length. (b) Rate loss and normalized ECC versus tail length.}
	\label{fig:ccsloss}
\end{figure*}

The third experiment studies the effect of tail length $t$ on rate loss (equal to $-h(U_0)$) and decoding complexity (measured by normalized ECC). Since both rate loss and normalized ECC are closely related to the initial CCS, some examples of $f_0(u)$ are given in \cref{subfig:ccs}, where code length $n=64$ and code rate $R=0.5$. By comparing the curve of $t=0$ with the curve of $t=(nR-3)$, it can be seen that for $t\ll nR$, the initial CCS does tend to be spikier as $t$ increases. However, as $t\to nR$, there is an opposite trend, \textit{i.e.}, $f_0(u)$ will tend to be flatter. In the extreme case of $t=nR$, $f_0(u)$ is uniform over $[0,1)$.

Then some examples of $\eta\triangleq\int_{0}^{1}{f_0^2(u)\,du}$ and $-h(U_0)$ versus tail length $t$ are given in \cref{subfig:loss}, where code length $n=64$ and code rate $R=0.5$. It can be seen that, for $t\ll nR$, both $\eta$ and $-h(U_0)$ will go up as $t$ increases, implying higher decoding complexity and greater rate loss. This is the negative effect of larger tail length. As $t$ approaches $nR$, there is an opposite trend, \textit{i.e.}, both $\eta$ and $-h(U_0)$ will go down. In the extreme case $t=nR$, we have $\eta=1$ and $-h(U_0)=0$, implying zero rate loss and minimized decoding complexity. However, as shown by \eqref{eq:extremehds} and \cref{fig:thds}, the HDS is very bad in this case.

\subsection{Frame Error Rate}
\begin{figure*}[!t]
	\centering
	\subfigure[]{\includegraphics[width=.5\linewidth]{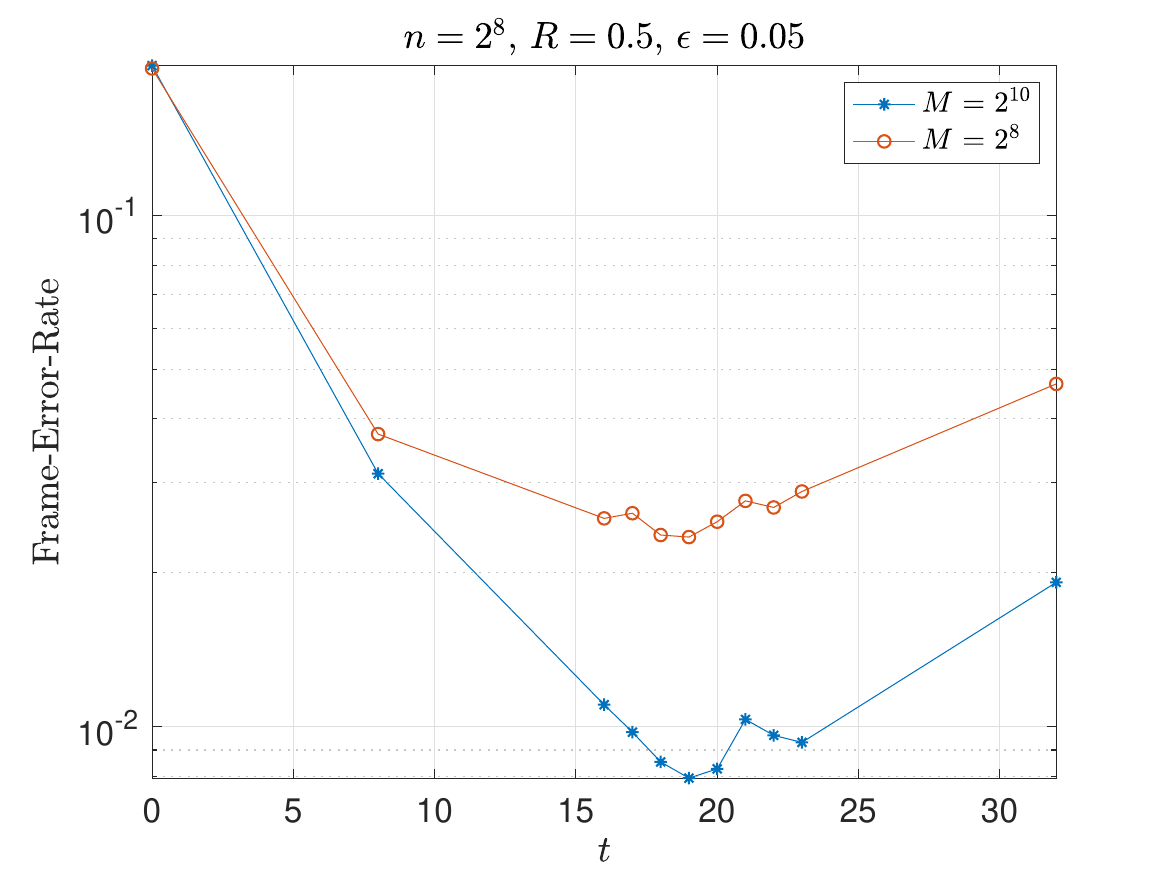}\label{subfig:fer5}}%
	\subfigure[]{\includegraphics[width=.5\linewidth]{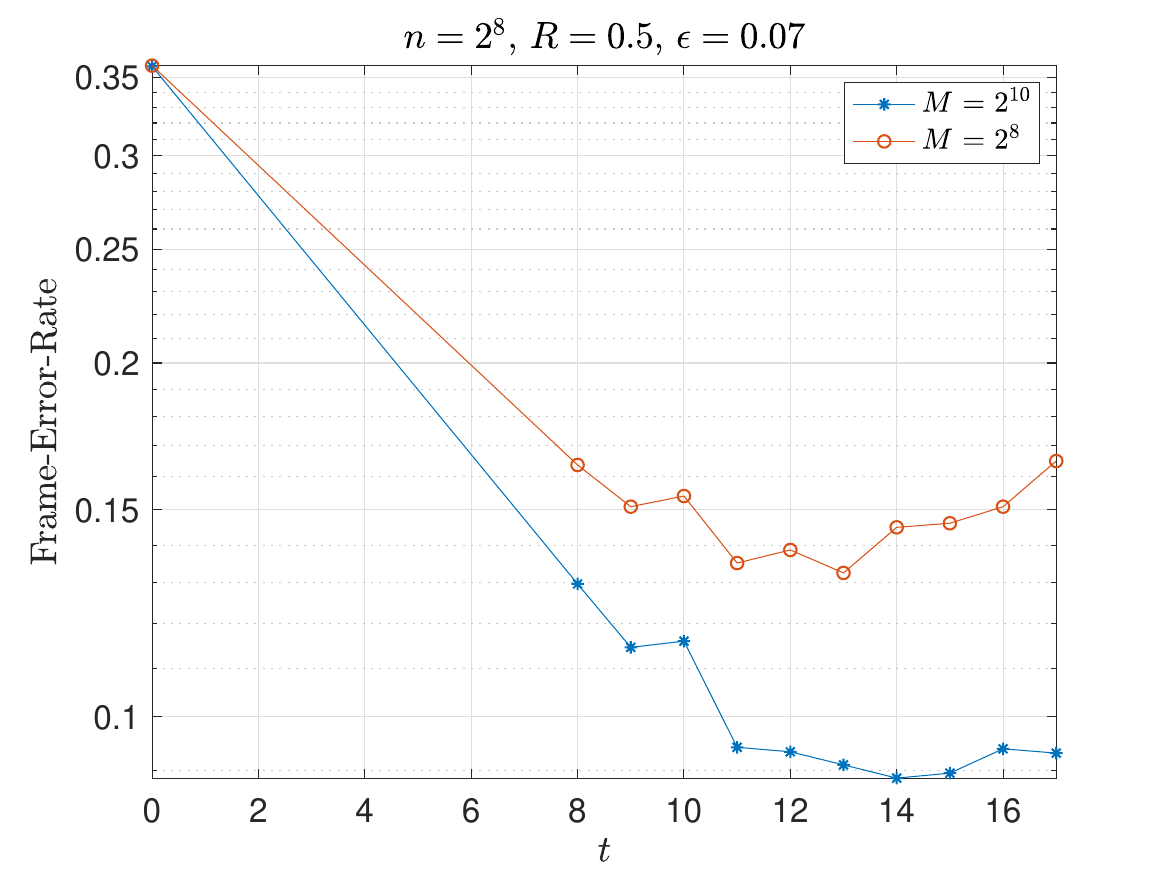}\label{subfig:fer7}}
	\caption{Examples to show how tail length $t$ impacts {\em frame-error-rate}.}
	\label{fig:fer}
\end{figure*}
The last experiment studies the effect of tail length $t$ on {\em frame-error-rate} (FER). Two examples are given in \cref{fig:fer}. The experimental settings can be found in the titles and legends, where $M$ is the maximum number of survival paths after pruning and $\epsilon$ is the crossover probability between the source and side information. As shown in \cref{subsec:decoding}, for large $M$, no significant gain can be achieved by making use of CCS during decoding, so all results in \cref{fig:fer} are obtained without using CCS. It can be seen that the FER of tailless overlapped arithmetic codes ($t=0$) is very high. As tail length $t$ increases, the FER will first drop sharply, then fluctuate slightly, and finally rebound gradually. The lowest FER is observed at $t\approx 19$ for $\epsilon=0.05$ and at $t\approx 14$ for $\epsilon=0.07$. In general, $t\approx 16$ is a good choice for almost every setting.

\chapter{Conclusion}\label{c-conclusion}
\vspace{-25ex}%
This monograph begins with arithmetic codes by emphasizing two implementation issues---{\em underflow handling} and {\em bitstream termination}. After a brief description of overlapped arithmetic codes, this monograph introduces related theoretical advances. It is revealed that overlapped arithmetic codes are actually a kind of nonlinear coset codes partitioning source space into unequal-sized cosets, and hence have two fundamental properties---{\em Coset Cardinality Spectrum} (CCS), which is a continuous function describing how coset cardinality is distributed, and {\em Hamming Distance Spectrum} (HDS), which is a discrete function describing how inter-mate Hamming distance is distributed, where mates refer to codewords belonging to the same coset. For both CCS and HDS, theoretical formulas are derived, and their tight connection is found; while for CCS, numerical algorithms are also given. On the basis of CCS, the decoding algorithm of overlapped arithmetic codes is derived. Experimental results show that aided by CCS, the performance of low-complexity decoder can be significantly improved. To deduce HDS, a powerful analysis tool---{\em Coexisting Interval} is developed, which can also be used to deduce the frame error rate of overlapped arithmetic codes.

The performance of overlapped arithmetic codes is very poor, if all source symbols of a block are mapped to partially overlapped intervals, because the minimum Hamming distance is always 1. The performance of overlapped arithmetic codes can be improved by mapping the last few source symbols of a block, called {\em tail}, to non-overlapped intervals. This improvement can be theoretically explained from the viewpoint of HDS, {\em i.e.}, the minimum Hamming distance of overlapped arithmetic codes may be increased by mapping the {\em tail} of a block to non-overlapped intervals. This monograph discusses the effect of tail length on overlapped arithmetic codes from the aspects of both HDS and CCS. It is shown that closely-spaced mate codewords can be removed by increasing tail length, which perfectly explains the superiority of tailed overlapped arithmetic codes over tailless overlapped arithmetic codes. It is also shown that increasing tail length will cause higher decoding complexity and larger rate loss meanwhile. These findings indicate that increasing tail length usually brings both positive (removing closely-spaced mate codewords) and negative (higher decoding complexity and larger rate loss) effects, so tail length should be tuned carefully to optimize the overall performance. At the same time, it is also found that the effects of tail length are sometimes weird, so how to optimize overlapped arithmetic codes still remains a very difficult open question.

However, it should be pointed out that this monograph is limited to uniformly distributed binary sources. For overlapped arithmetic codes, there have been some advances on non-uniform binary sources \cite{FangTIT20,FangTIT21} and uniform non-binary sources \cite{FangTIT23}.

All theoretical analyses of this monograph are perfectly verified by experiments. A software package of this monograph for result reproduction is available on \url{https://github.com/fy79/NOW-FTCIT}.

\backmatter  

\printbibliography

\end{document}